\documentclass[11pt,letterpaper]{article}

\newcommand{\abnote}[1]{}
\newcommand{\snote}[1]{}

 \usepackage[T1]{fontenc}    %
 \usepackage{booktabs}       %

\usepackage{tcs}

\usepackage{mysymbols}

\newcommand{\indic}[1]{ \mathbb{I} \tlprn {#1}}

\newcommand\apnorm[2]{\left|\!\left|#1 \right|\!\right|_{#2}}

\newcommand{\mustar}{\mu^*}
\newcommand{\muhat}{\hat{\mu}}
\newcommand{\muprime}{\mu'}

\newcommand{\thetastar}{\theta^*}
\newcommand{\thetahat}{\hat{\theta}}
\newcommand{\thetaprime}{\theta'}
\newcommand{\thetals}{\theta_{\mathrm{LS}}}

\newcommand{\card}[1]{\lvert#1\rvert}
\newcommand{\Card}[1]{\left\lvert#1\right\rvert}

\newcommand\MYcurrentlabel{xxx}
\newcommand{\MYstore}[2]{%
  \global\expandafter \def \csname MYMEMORY #1 \endcsname{#2}%
}
\newcommand{\MYload}[1]{%
  \csname MYMEMORY #1 \endcsname%
}
\newcommand{\MYnewlabel}[1]{%
  \renewcommand\MYcurrentlabel{#1}%
  \MYoldlabel{#1}%
}
\newcommand{\MYdummylabel}[1]{}
\newcommand{\torestate}[1]{%
  \let\MYoldlabel\label%
  \let\label\MYnewlabel%
  #1%
  \MYstore{\MYcurrentlabel}{#1}%
  \let\label\MYoldlabel%
}
\newcommand{\restatetheorem}[1]{%
  \let\MYoldlabel\label
  \let\label\MYdummylabel
  \begin{theorem*}[Restatement of \cref{#1}]
    \MYload{#1}
  \end{theorem*}
  \let\label\MYoldlabel
}
\newcommand{\restatelemma}[1]{%
  \let\MYoldlabel\label
  \let\label\MYdummylabel
  \begin{lemma*}[Restatement of \cref{#1}]
    \MYload{#1}
  \end{lemma*}
  \let\label\MYoldlabel
}
\newcommand{\restateprop}[1]{%
  \let\MYoldlabel\label
  \let\label\MYdummylabel
  \begin{proposition*}[Restatement of \cref{#1}]
    \MYload{#1}
  \end{proposition*}
  \let\label\MYoldlabel
}
\newcommand{\restatefact}[1]{%
  \let\MYoldlabel\label
  \let\label\MYdummylabel
  \begin{fact*}[Restatement of \cref{#1}]
    \MYload{#1}
  \end{fact*}
  \let\label\MYoldlabel
}
\newcommand{\restatecorr}[1]{%
  \let\MYoldlabel\label
  \let\label\MYdummylabel
  \begin{corollary*}[Restatement of \cref{#1}]
    \MYload{#1}
  \end{corollary*}
  \let\label\MYoldlabel
}
\newcommand{\restate}[1]{%
  \let\MYoldlabel\label
  \let\label\MYdummylabel
  \MYload{#1}
  \let\label\MYoldlabel
}

\newcommand{\suchthat}{\;\middle\vert\;}
\newcommand{\cS}{\mathcal{S}}

\title{Sample-Optimal Private  Regression\\
in Polynomial Time}

\author{
Prashanti Anderson\\
\texttt{paanders@csail.mit.edu} \\
MIT
\and
Ainesh Bakshi\footnote{Supported by the NSF TRIPODS program (award DMS-2022448).} \\
\texttt{ainesh@mit.edu} \\
MIT
\and
Mahbod Majid \\
\texttt{mahbod@mit.edu} \\
MIT
\and
Stefan Tiegel\footnote{Supported by the European Union’s Horizon research and innovation programme (grant agreement no. 815464).} \\
\texttt{stefan.tiegel@inf.ethz.ch} \\
ETH Z\"urich
}

\date{}
\begin{document}

\maketitle

\begin{abstract}
We consider the task of privately obtaining prediction error guarantees in ordinary least-squares regression problems with Gaussian covariates (with unknown covariance structure).
We provide the first sample-optimal polynomial time algorithm for this task under both pure and approximate differential privacy.
We show that any improvement to the sample complexity of our algorithm would violate either statistical-query or information-theoretic lower bounds.
Additionally, our algorithm is robust to a small fraction of arbitrary outliers and achieves optimal error rates as a function of the fraction of outliers.
In contrast, all prior efficient algorithms either incurred sample complexities with sub-optimal dimension dependence, scaling with the condition number of the covariates, or obtained a polynomially worse dependence on the privacy parameters.

Our technical contributions are two-fold:
first, we leverage \textit{resilience} guarantees of Gaussians within the sum-of-squares framework. As a consequence, we obtain efficient sum-of-squares algorithms for regression with optimal robustness rates and sample complexity. Second, we generalize the recent robustness-to-privacy framework~\cite{Hopkins2023Robustness} to account for the geometry induced by the covariance of the input samples. This framework crucially relies on the robust estimators to be sum-of-squares algorithms, and combining the two steps yields a sample-optimal private regression algorithm. We believe our techniques are of independent interest, and we demonstrate this by obtaining an efficient algorithm for covariance-aware mean estimation, with an optimal dependence on the privacy parameters.

\end{abstract}

\thispagestyle{empty}
\clearpage
\newpage

\microtypesetup{protrusion=false}
\thispagestyle{empty}
\tableofcontents{}
\thispagestyle{empty}
\microtypesetup{protrusion=true}
\clearpage
\newpage

\setcounter{page}{1}

\section{Introduction}

In the increasingly data-driven fields of machine learning, econometrics, drug design and quantitative social sciences, the use of linear regression models plays a pivotal role in understanding and predicting relationships among variables. The models can often shed insight into complex phenomena such as economic growth, public health trends and social behavior. 

The data collected for the aforementioned applications is often sensitive to the identity of the individuals involved. 
Unauthorized access or exposure of this data could lead to significant privacy violations and potential harm to individuals, including discrimination and identity theft. Furthermore, trust in public and academic institutions may be eroded if personal data are mismanaged, leading to reluctance among individuals to participate in future studies or share information.

Complementary to privacy concerns of the individuals involved, the data collected for econometric and medical studies is often noisy to systematic collection errors and individuals who behave as outliers. Least-squares regression has received significant attention from both the differential privacy~\cite{alabi2022differentially, varshney2022nearly, liu2023near,asi2023robustness,brown2024insufficient} and robust statistics~\cite{diakonikolas2019efficient, klivans2018efficient, prasad2020robust, bakshi2021robust, pensia2024robust} communities in recent years, with a plethora of efficient and inefficient estimators. 

A recent line of work bridges these areas together, demonstrating that estimators that obtain the \emph{right} privacy guarantees are automatically robust to outliers~\cite{georgiev2022privacy}. Going in the opposite direction, efficient algorithms for estimating mean and covariance in the presence of arbitrary outliers were used to obtain sample-optimal private algorithms for estimating a Gaussian in total-variation distance~\cite{Hopkins2023Robustness}. These results, as well as other (inefficient) reductions between robustness and privacy~\cite{asi2023robustness} suggest that privacy and robustness are often complementary guarantees.

Focusing on ordinary least squares regression, nearly sample-optimal robust estimators based on spectral filtering are well-established and can be implemented in polynomial time~\cite{diakonikolas2019efficient}. However, a significant gap remains in the literature concerning private regression, as it is unclear whether filtering algorithms can be adapted for privacy. Currently, the best-known upper and lower bounds on the sample complexity of efficient private regression differ by polynomial factors in the dimension~\cite{brown2024insufficient}, require additional condition-number factors~\cite{varshney2022nearly, liu2023near}, and offer weak privacy guarantees. Developing sample-optimal, efficient algorithms for private regression remains a challenging endeavor, resulting in numerous private estimators with sub-optimal and often incomparable guarantees. The status of estimators for private regression remains nuanced even in one dimension~\cite{alabi2022differentially}. This leads us to the following central question:

\begin{quote}
\begin{center}
   \emph{Does there exist an efficient sample-optimal estimator for private regression?    }
\end{center}
\end{quote}

\subsection{Our Results}

We resolve this question affirmatively by presenting the first efficient, sample-optimal algorithm for private regression.\footnote{We use the term ``sample-optimal'' to mean that improving the sample complexity would violate existing lower bounds that are either information theoretic or computational (in our case in the SQ model).} Our estimator achieves both \emph{pure differential privacy} (or short, pure DP), widely regarded as the gold standard for privacy, as well as \emph{approximate differential privacy} (approx DP) simultaneously. We begin by defining our notion of privacy:

\begin{definition}[Pure Differential Privacy~\cite{dwork2006calibrating}]
Let $\calX$ be the set of all finite length input strings and let $\calO$ be the set of output strings. A randomized mechanism $\calM: \calX \to \calO$ is $(\eps, \delta)$-differentially private if for every pair of input strings $X, X' \in \calX$ with Hamming distance $1$, and every subset $S \subseteq \calO$, 
\begin{equation*}
    \prob{ \calM(X) \in S   } \leq e^{\eps} \cdot \prob{ \calM(X' ) \in S   } + \delta. 
\end{equation*}
If $\delta =0$ the mechanism satisfied pure DP and otherwise it satisfies approx DP.
\end{definition}

In addition to satisfying pure/approx differential privacy, our estimator also allows for data contamination via the most stringent model for handling outliers, known as the strong contamination model (see~\cite{diakonikolas2019recent} for a survey). 

\begin{definition}[Strong Contamination]
Fix a distribution $\calD$, a corruption rate $0<\eta<1/2$ and a set of n i.i.d. samples $\Set{x_1^*, x_2^*, \ldots, x_n^*}$ from $\calD$.
A set of points $\Set{x_1, x_2, \ldots, x_n}$ is called an \emph{$\eta$-corrupted} sample from $\calD$ (or $n$ $\eta$-corrupted samples), if $x_i = x_i^*$, for at least $(1-\eta) n$ indices.  
    
\end{definition}
Intuitively, this model allows for the adversary to be computationally unbounded and have access to the algorithm description.  We can now formally define the statistical model we consider:

\begin{problem}[Robust Regression]
\label{model:robust_regression}
Let $\Sigma \in \R^{d \times d}$ be an unknown PSD matrix and $\theta \in \R^d$ be an unknown vector. Let $\Set{(x_i^* , y_i^*)}_{i \in [n]}$ be $n$ i.i.d. samples generated as follows: $$y_i^* = \Iprod{ \theta , x_i^*} + \zeta_i ,$$ where $x_i^* \sim \calN(0, \Sigma)$, $\zeta_i \sim \calN(0,1)$. 
We refer to $\theta$ as the optimal hyperplane. The input is then an $\eta$-corruption of the samples $\Set{(x_i^* , y_i^*)}_{i \in [n]}$, denoted by $\Set{(x_i , y_i)}_{i \in [n]}$. The goal is to output $\hat{\theta}$ such that 
\begin{equation*}
    \expecf{(x,y)\sim \calD}{ \Paren{ y- \Iprod{\hat{\theta}, x} }^2 }  \leq  \expecf{(x,y)\sim \calD}{ \Paren{ y- \Iprod{\theta, x} }^2 }  + \alpha^2, 
\end{equation*}
where $\alpha \geq \eta \log(1/\eta)$.
\end{problem}

Note, the guarantee above corresponds to the \textit{ generalization error} of the estimator $\hat{\theta}$, and any computationally bounded algorithm can only achieve rate $\alpha \geq \eta \log(1/\eta)$~\cite{diakonikolas2019efficient}. 
We also note that getting the \textit{right} sample complexity for the private problem automatically implies that the estimator is robust to adversarial corruptions~\cite{georgiev2022privacy}, handling them explicitly is simply for ease of exposition. Further, in the above model, the adversary is allowed to corrupt both the samples and the labels.

\paragraph{Sample-optimal Private Regression.} Our main contribution is an efficient, sample-optimal algorithm for regression, satisfying pure differential privacy and robustness to adversarial corruptions.

\begin{theorem}[Optimal Private Regression (informal, see \cref{thm:main_private_regression,} )]
\label{thm:pure-dp-regression-informal}
Let $\theta, \Sigma$ be such that $\norm{\theta}_2 \leq R$ and $\Sigma \preceq \covscale I$ for some $R,L$.
Given $0< \alpha, \eps  <1$, and $n$ $\eta$-corrupted samples (as defined in \cref{model:robust_regression}) with parameters $\theta$ and $\Sigma$, there exists an $\eps$-differentially private algorithm which runs in $\poly(n, \log L, \log R)$-time and outputs $\hat{\theta}$ such that with probability $1-\beta$, 
\begin{equation*}
    \expecf{(x,y)\sim \calD}{ \Paren{ y- \Iprod{\hat{\theta}, x} }^2 }  \leq  \expecf{(x,y)\sim \calD}{ \Paren{ y- \Iprod{\theta, x} }^2 }  + \alpha^2, 
\end{equation*}
as long as $\alpha\geq \eta \log(1/\eta)$ and
\begin{equation*}
    n \geq \tilde{\Omega}\Paren{ \frac{d^2 +\log^2(1/\beta)}{\alpha^2} + \frac{d + \log(1/\beta)}{\alpha \eps} +  \frac{d\log(R \sqrt{\covscale} )}{\eps}  }. 
\end{equation*}
In the same setting,\footnote{Without requiring the bound on the covariance.} there exists an $(\eps,\delta)$-differentially private algorithm which runs in $\poly(n, \log L, \log R)$ time whenever 
\begin{equation*}
    n \geq \tilde{\Omega}\Paren{ \frac{d^2 +\log^2(1/\beta)}{\alpha^2} + \frac{d + \log(1/\beta)}{\alpha \eps} +  \frac{ \log(1/\delta) )}{\eps}  }   \,.
\end{equation*}
\end{theorem}

\begin{remark}[On Optimality.]
We show that each individual term appearing in the sample complexity above is necessary (up to log-factors).
Even in the regime where $\alpha$ is a fixed constant, the  $d^2$ term is necessary for any computationally bounded algorithm since any improvement to this term would break the statistical query (SQ) lower bound for robust regression.
In particular, for $\varepsilon = \Omega(1), \beta = 2^{-\Omega(n)}$ any $\varepsilon$-DP algorithm is robust to a small constant fraction of corruptions \cite[Theorem 3.1]{georgiev2022privacy} and hence requires $\Omega(d^2)$ samples~\cite{diakonikolas2019efficient}.\footnote{We remark that their lower bound (Theorem 3.1 in~\cite{diakonikolas2019efficient}) states the lower bound for a regression model in which the variance of the noise is unknown but bounded by 1. However, inspecting their lower bound construction, the variance of the noise can actually be taken to be in the interval $[1-\eta, 1]$. Our algorithms continue to work when the variance of the noise is promised to be in this interval (see the remarks in~\cref{sec:robust_regression_oneshot}).}
The same argument applies to approx-DP algorithms as long as $\delta = 2^{-\Omega(n)}$.
We also show that the terms $$\frac{d + \log(1/\beta)}{\alpha \eps} +  \frac{d\log(R \sqrt{\covscale} )}{\eps}$$ are information-theoretically necessary (see \cref{thm:pure-dp-lb-regresion} for a formal statement), the last term for pure DP. 
The $\log^2(1/\beta)/\alpha^2$ term can likely be improved to $\log(1/\beta)/\alpha^2$, which is again information-theoretically necessary but we did not try to optimize this particular factor. 
Therefore, improving the sample complexity of our estimator (ignoring $\log^2(1/\beta)$) would either break SQ or information-theoretic lower bounds.
As far as we are aware, there is no known lower bound for the $\Omega( \tfrac {\log (1/\delta)} \varepsilon)$ term.
\end{remark}

Prior to our work, there was no known computationally efficient algorithm for regression under pure-DP, not even when allowing sub-optimal sample complexities. A natural baseline would be to privately output the optimal regressor given by the normal equations, i.e., $\hat{\theta} = \Paren{ X^\top X }^{-1}  X^\top y$, where $X$ is a $n\times d$ matrix with each row being a sample, and the $y$ is the corresponding response vector. First, observe that $\norm{ X^\top y}_2 = \norm{ \Paren{ X^\top X } \beta }_2$ can scale as $\Omega(LR)$, where $\Sigma \preceq L I$ and $\norm{\beta}_2 \leq R$. Therefore, to isotropize the samples, we would have to learn the covariance up to accuracy $1/\bigO{LR}$ (using the estimator in~\cite{Hopkins2023Robustness}), which requires the sample complexity to scale polynomially in $L$ and $R$.   In contrast, the informational-theoretically optimal estimators scale logarithmically in $R$ and $L$ and therefore we consider any estimator that scales polynomially in $R$ and $L$ to be inefficient.  

Restricting to computationally inefficient estimators,  Asi, Ullman and Zakynthinou~\cite{asi2023robustness} obtain an estimator for pure-DP regression which requires sample complexity\footnote{We supress the dependence on the failure probability for clarity here.}
\begin{equation*}
    n \geq \Omega\Paren{\frac{d }{\alpha^2} + \frac{ \Paren{ d \log(R + \kappa)}  }{\alpha\eps}  },
\end{equation*}
where $\kappa$ is the condition number of the covariance $\Sigma$.
In contrast, our estimator is efficient, can handle arbitrarily ill-conditioned covariances as long as $\Sigma \preceq\covscale I$ (note that this does not imply a bound on the condition number, since there is no lower bound on $\Sigma$) and the $\log(R)$ term does not scale with $\alpha$. Apriori, the upper bound on the covariance may seem asymmetric and scale-invariant, but note that the following linear models are equivalent: $y_i = \Iprod{\theta , x} = \Iprod{\Sigma^{-1/2} \theta , g}$, where $x \sim \calN(0, \Sigma)$ and $g \sim \calN(0, I)$. Our combined assumption on $\Sigma$ and $\theta$ implies that the parameter $\Sigma^{1/2} \theta$ is bounded by $R\cdot\covscale$ and this is necessary for pure-DP via packing lower bounds. 

Previous efficient private estimators for regression all work in the weaker approx-DP setting and their techniques cannot achieve pure DP. Furthermore, they are also not robust to corruptions in the covariates, while our algorithm is robust to corruptions in both the covariates and labels. In particular, Varshney, Thakurta, and Jain~\cite{varshney2022nearly} obtain a sample complexity of $\Omega\Paren{\frac{d}{\alpha^2} +
\frac{\kappa \sqrt{\covscale}R d\log \delta^{-1}}{\alpha\eps} }$ and are not robust to outliers.  Liu, Jain, Kong, Oh and Suggala~\cite{liu2023near} improve it to sample complexity $\Omega\Paren{ \frac{d}{\alpha^2} + \frac{\sqrt{\kappa} d \log(1/\delta) }{ \alpha \eps} }$ and are only robust to label noise, which allows for the leading term to scale linearly in $d$. In contrast, we get rid of the condition number dependence and decouple $d/\alpha$ from the $\log(1/\delta)$ term. Finally, Brown, Hayase, Hopkins, Kong, Liu, Oh, Perdomo and Smith~\cite{brown2024insufficient} obtain sample complexity $\Omega\Paren{ \frac{d}{\alpha^2} + \frac{d \sqrt{\log(1/\delta)} }{\alpha \eps}  + \frac{d \log(1/\delta)^2}{\eps^2} } $. We decouple all $d$ and $\log(1/\delta)$ factors, and get a linear dependence on $1/\eps$.
We summarize the comparisons in \cref{tab:regression} and defer further discussion of these estimators to~\cref{sec:related-work}.

\begin{table}[h]
\begin{center}
\begin{tabular}{| c | c | c | c | c |} 
\hline
Paper & Sample Complexity & Poly-time? & Privacy \\
\hline \hline
\cite{liu2022differential} & $\frac{d}{\alpha^2} + \frac{d + \log \delta^{-1}}{\alpha \eps}$ & No & Approximate
\\ \hline
\cite{varshney2022nearly} & $\frac{d}{\alpha^2} +
\frac{\kappa \sqrt{\covscale}R d\log \delta^{-1}}{\alpha\eps}$ & Yes & Approximate
\\ \hline
\cite{liu2023near} & $\frac{d}{\alpha^2}+ \frac{\sqrt{\kappa} d \log \delta^{-1}}{\alpha \eps}$ & Yes & Approximate
\\ \hline
\cite{asi2023robustness} & $\frac{d}{\alpha^2} + \frac{d \log\paren{R + \sqrt{\kappa}}}{\alpha \eps}$ & No & Pure
\\ \hline
\cite{brown2024insufficient} &$\frac{d}{\alpha^2} + \frac{d \sqrt{\log \delta^{-1}}}{\alpha \eps} + \frac{d \log^2 \delta^{-1}}{\eps^2}$ & Yes & Approximate
\\ \hline
\Cref{thm:main_private_regression} & $\frac{d^2}{\alpha^2} + \frac{d}{\alpha \eps} + \frac{d \log R \covscale}{\eps}$ & Yes & Pure
\\ \hline
\Cref{thm:approx_dp_regression_full} & $\frac{d^2}{\alpha^2} + \frac{d}{\alpha \eps} + \frac{\log \delta^{-1}}{\eps} $ & Yes & Approximate
\\ \hline
\end{tabular}
\caption{\label{tab:regression} Private regression, omitting logarithmic factors. } 
\end{center}
\end{table}

At a high level, our algorithm proceeds via a robustness-to-privacy reduction. Similar to \cite{Hopkins2023Robustness}, we consider a convex relaxation of the score function for the exponential mechanism and relate it to a robust algorithm for regression, based on the \textit{sum-of-squares} hierarchy. However, there are two significant challenges with this approach: (a) there is no known \textit{sum-of-squares} based algorithm for robust regression that achieves optimal rates in polynomial time. The state-of-the-art robust algorithms require quasi-polynomial sample complexity and running time to obtain optimal rates, even for Gaussians~\cite{klivans2018efficient, bakshi2021robust}. And (b) unlike in \cite{Hopkins2023Robustness} we estimate parameters in an error metric specified by the unknown covariance which would be too costly to estimate privately. In particular, naive applications of~\cite{Hopkins2023Robustness} give algorithms with sample complexity scaling either with $d^2/(\alpha \eps)$ or $d \log(\covscale)/(\alpha \eps)$, where $\covscale$ is the upper bound on the covariance. We discuss these challenges in greater detail in \cref{sec:tech-overview}.

Our contributions include a new \textit{sum-of-squares} based algorithm for robust regression, which leverages \textit{resilience} properties of Gaussian samples in the sum-of-squares proof system. This statement was previously suspected to be impossible to certify (see "Resilience is likely not sos-izable" on page 6 in~\cite{kothari2022polynomial}). We show that it is not necessary to certify resilience for indeterminates in the proof system to obtain optimal rates for Gaussian regression. Our key insight is to show it suffices to invoke resilience on sum-of-squares indicated subsets of samples, which circumvents the hardness of certifying resilience. Next, we show that we can significantly extend the privacy framework of Hopkins, Kamath, Majid and Narayanan~\cite{Hopkins2023Robustness} to give private algorithms which are accurate with respect to the geometry induced by the covariance of the samples, as opposed to the Euclidean geometry. This enables us to match the information-theoretic lower bound of $d/(\alpha\eps)$.

\subsection{Application: Covariance-Aware Mean Estimation}

We believe our techniques for developing robust estimators and transforming them into private estimators are broadly applicable, enabling progress on various related problems. We demonstrate this by applying our framework to obtain sample-optimal estimators covariance-aware mean estimation under both pure-DP and approx-DP, another canonical problem in (robust) statistics.

In this problem, we obtain an $\eta$-corruption of $n$ samples drawn from $\calN(\mu, \Sigma)$. As in regression, a natural approach is to simply learn the covariance privately, using the optimal estimators from~\cite{Hopkins2023Robustness}. However, their algorithm is a two-step approach where they first learn the covariance, isotropize and then learn the mean. Observe, any such two-step algorithm must publish the covariance, incurring an $\Omega(d^2/(\alpha \eps))$ dependence on the privacy parameters. However, learning the mean in Mahalanobis distance only requires outputting a $d$-dimensional vector, and therefore the information-theoretic lower bound scales like $d/(\alpha \eps)$. We demonstrate that we can match this sample complexity with an efficient estimator:

\begin{theorem}[Covariance-Aware Mean Estimation (informal \cref{thm:main_private_cov_mean_est})]
\label{thm:informal-mean-est}
Let $\mu, \Sigma$ be such that $\norm{\mu}_2 \leq R$ and $ \Sigma \succeq \tfrac 1 {\covscale} \cdot I $ for some $R,L$.
Given $0< \alpha, \eps  <1$, and $n$ $\eta$-corrupted samples from $\calN(\mu,\Sigma)$ there exists a $\eps$-differentially private algorithm which runs in $\poly(n,\log L, \log R)$ time and outputs $\hat{\mu}$ such that with probability $1-\beta$, 
\begin{equation*}
     \norm{\Sigma^{-1/2}\Paren{\mu - \hat{\mu} } }_2 \leq \alpha,  
\end{equation*}
as long as  $\alpha \geq \eta \sqrt{\log(1/\eta)}$ and 
\begin{equation*}
    n \geq \Omega\Paren{ \frac{d^2 + \log^2 (1/\beta)}{\alpha^2} + \frac{d + \log(1/\beta)}{\alpha \eps} +  \frac{d\log(R \sqrt{\covscale} )}{\eps}  }\, .
\end{equation*}
In the same setting,\footnote{Without requiring the bound on the covariance.} there exists an $(\eps,\delta)$-differentially private algorithm which runs in $\poly(n, \log R, \log L)$ time whenever 
\begin{equation*}
    n \geq \tilde{\Omega}\Paren{ \frac{d^2 +\log^2(1/\beta)}{\alpha^2} + \frac{d + \log(1/\beta)}{\alpha \eps} +  \frac{ \log(1/\delta) )}{\eps}  }   \,.
\end{equation*}
\end{theorem}

\begin{remark}[On Optimality]
We prove a matching information-theoretic lower bound of $\frac{d + \log(1/\beta)}{\alpha \eps} +  \frac{d\log(R \sqrt{\covscale} )}{\eps} $ for pure-DP (see~\cref{thm:pure-dp-lb-mean-etimation}). 
As earlier, we know that any such estimator is inherently robust and thus requires at least $\Omega(d^2)$ samples (again for $\varepsilon = \Omega(1), \beta = 2^{-\Omega(n)}$, and $\delta = 0$ or $\delta = 2^{-\Omega(n)}$), see also \cite[Section 7.3.2]{diakonikolas2024sos}.
Under Approximate-DP we know that the $d/\alpha \eps + \log(1/\delta) / \eps$ is necessary and the $\log(1/\beta)/\alpha \eps$ term is potentially improvable to $\min(\log(1/\delta),\log(1/\beta))/\alpha \eps$ even in the known covariance case from fingerprinting lower bounds \cite{karwa2018finite,kamath2019privately,kamath2022new}.
\end{remark}

We note that no efficient, sample-optimal estimators were known for covariance-aware mean estimation under pure-DP. Asi, Ullman and Zakynthinou~\cite{asi2023robustness} obtained an inefficient estimator that requires the number of samples to scale as 
\begin{equation*}
    n \geq \Omega\Paren{\frac{d }{\alpha^2} + \frac{ \Paren{ d \log(R + \kappa)}  }{\alpha\eps}  },
\end{equation*}
where $\kappa$ is the condition number. In contrast, we can handle arbitrarily ill-conditioned covariances.
The estimators obtained by Hopkins, Kamath, Majid and Narayanan have the sample complexity scale as 
\begin{equation*}
    n \geq \Omega\Paren{\frac{d^2 }{\alpha^2} + \frac{ d^2 }{\alpha\eps}   },
\end{equation*}
since a key step in their approach is to isotropize the samples. This requires privately releasing the covariance, which is a $d^2$-dimensional object. In contrast, we demonstrate that isotropizing the samples is unnecessary and we can extend the exponential mechanism to sample from a space where the induced metric is $\norm{\cdot}_{\Sigma^{-1/2}}$, even though $\Sigma$ is unknown. There are several estimators that work in the approx-DP setting and we compare to them in \cref{tab:mean-estimation}.

\begin{table}[h]
\begin{center}
\begin{tabular}{| c | c | c | c | c |} 
\hline
Paper & Sample Complexity & Poly-time? & Privacy \\
\hline \hline
\cite{kamath2019privately} & $\frac{d}{\alpha^2}+\frac{d}{\alpha \rho^{1 / 2}}+\frac{\sqrt{d} \log ^{1 / 2} R}{\rho^{1 / 2}} + \frac{d^{3 / 2} \log ^{1 / 2}\kappa}{\rho^{1 / 2}}$ & Yes & Concentrated
\\ \hline
\cite{liu2022differential} & $\frac{d}{\alpha^2} + \frac{d + \log\delta^{-1}}{\alpha \eps}$ & No & Approximate
\\ \hline
\cite{brown2021covariance} & $\frac{d}{\alpha^2} + \frac{d}{\alpha \eps} + \frac{\log \delta^{-1}}{\eps}$ & No & Approximate 
\\ \hline
\cite{brown2023fast} & $\frac{d}{\alpha^2} + \frac{d\sqrt{\log \delta^{-1}}}{\alpha \eps} + \frac{d \log \delta^{-1}}{\eps}$ & Yes & Approximate
\\ \hline
\cite{kuditipudi2023pretty} & $\frac{d}{\alpha^2} + \frac{d \log \delta^{-1}}{\alpha \eps} + \frac{d \log^2\delta^{-1}}{\eps^2}$ & Yes & Approximate
\\ \hline
\cite{asi2023robustness} & 
$\frac{d}{\alpha^2} + \frac{d\log(R+\sqrt{\kappa})}{\alpha\eps}$
& No & Pure \\
\hline
\cite{Hopkins2023Robustness} & $\frac{d^2}{\alpha^2} + \frac{d^2}{\alpha \eps} + \frac{ d^2 \log \kappa}{\eps} +\frac{d \log R}{\eps}$ & Yes & Pure 
\\ \hline
\cite{Hopkins2023Robustness} & $\frac{d^2}{\alpha^2} + \frac{d^2}{\alpha \eps} + \frac{\log 1/\delta}{\eps}$ & Yes & Approximate
\\ \hline
\Cref{thm:main_private_cov_mean_est} & $\frac{d^2}{\alpha^2} + \frac{d}{\alpha \eps} + \frac{ d \log R\sqrt{L}}{\eps}$ & Yes & Pure 
\\ \hline
\Cref{thm:approx_dp_mean_full} & $ \frac{d^2}{\alpha^2} + \frac{d}{\alpha \eps} + \frac{ \log \delta^{-1}}{\eps}$  & Yes & Approximate \\
\hline
\end{tabular}
\caption{\label{tab:mean-estimation} Private covariance-aware mean estimation of Gaussians, omitting logarithmic factors.}
\end{center}
\end{table}

\section{Technical Overview}
\label{sec:tech-overview}

Our approach consists of two key components. First, we present an efficient \textit{sum-of-squares} algorithm for robust regression that achieves optimal rates as a function of the outliers and optimal sample complexity. Second, we extend the robustness-to-privacy framework of Hopkins, Kamath, Majid and Narayanan~\cite{Hopkins2023Robustness} to accommodate the geometry induced by the covariance of the input data. In this section, we describe the main bottlenecks for private regression and key ideas we introduce to overcome them. 

\subsection{Optimal Robust Regression for Gaussians}
Filtering-based approaches are well-known to obtain optimal rates for robust regression, when the underlying covariates are drawn from a Gaussian~\cite{diakonikolas2019efficient}. However, such algorithms are unsuitable for a privacy-to-robustness reduction. On the other hand, a \textit{sum-of-squares} based estimator for mean and covariance estimation was recently shown to be private~\cite{Hopkins2023Robustness}. Therefore, a natural approach is to construct a SoS estimator for regression with optimal rates. The state-of-the-art SoS estimators for robust regression by Bakshi and Prasad~\cite{bakshi2021robust}, exploit $k$-th moment information to get an error rate of $\eta^{1-1/k}$, where $\eta$ is the fraction of outliers. This estimator uses information about order-$k$ moments and therefore requires $\Omega(d^k)$ samples. To achieve an error rate of $\tilde{\mathcal{O}}\Paren{ \eta }$, which is optimal for Gaussians, Bakshi and Prasad's estimator would require quasi-polynomially many samples and running time.

\paragraph{Robust Covariance estimation.}
We use the task of robust covariance estimation as a running example, since it is a key sub-routine in both regression and covariance-aware mean estimation. Given a set, $X^* = \Set{x_i^*}_{i \in [n]}$, of $n$ i.i.d. samples from $\calN(0, \Sigma)$, and an $\eta$-corruption denoted by $X = \Set{x_i}_{i \in [n]}$, we can demonstrate a first-cut information-theoretic argument for covariance estimation, which mimics a SoS relaxation. Consider a coupling, denoted by indicators $r_i = \indic{ x_i^* = x_i }$, i.e. the $i$-th sample remains uncorrupted. Then, the empirical covariance of the indicator selected points, $\tilde{\Sigma} = \frac{1}{n} \sum_{i \in [n]} r_i \cdot x_i x_i^\top$ , serves as a robust estimator: 
\begin{equation}
\label{eqn:hypercontractivity-computation-intro}
    \begin{split}
        \Iprod{ \Sigma - \tilde{\Sigma}, vv^\top  }^2 & \approx  \Iprod{  \frac{1}{n} \sum_{i \in [n]} (1-r_i) \cdot x_i^* {x^*_i}^\top, vv^\top  }^2   \\
        &\leq \Paren{ \frac{1}{n} \sum_{i \in [n]} (1-r_i)^2 } \Paren{ \frac{1}{n} \sum_{i \in [n]} \Iprod{x_i^*, v}^4 } \\
        & \leq \bigO{\eta} \cdot (v^\top \Sigma v)^2, 
    \end{split}
\end{equation}
where the first equality conflates the empirical covariance of the $x_i$'s and the true covariance, the first inequality is Cauchy-Schwarz and the second uses hypercontractivity of Gaussians (see~\cref{fact:hypercontractivity}). Note, this argument only ever uses the fourth-moment of the distribution, and as a result estimates the covariance only up to $\bigO{\sqrt{\eta}}$-error. To obtain better rates, we consider the notion of \textit{resilience}.

A set of samples $X^* = \Set{x^*_i}_{i \in [n]}$ is \textit{resilient}  if for every subset $\calT \subset [n]$ of size $\eta \cdot n$, for all directions $v\in\mathbb{R}^d$, 
\begin{equation}
\label{eqn:resilience-intro}
\begin{split}
      & \frac{1}{  n} \sum_{i \in \calT} \Iprod{ x^*_i ,  v}^2    \leq   \Paren{\eta \log(1/\eta) } v^\top \Sigma v, \\
      \textrm{and} \hspace{0.2in}  & \frac{1}{  n} \sum_{i \in \calT} \Iprod{ x^*_i ,  v}^4    \leq   \Paren{\eta \log(1/\eta)^2 } \Paren{v^\top \Sigma v}^2. 
\end{split}
\end{equation}
Invoking resilience instead of hypercontractivity, in a calculation similar to \cref{eqn:hypercontractivity-computation-intro}, we get improved rates:
\begin{equation*}
\begin{split}
    \Iprod{ \Sigma - \tilde{\Sigma}, vv^\top  }^2 &  \leq \Paren{ \frac{1}{n} \sum_{i \in [n]} (1-r_i) } \Paren{ \frac{1}{n} \sum_{i \in [n]} (1-r_i) \Iprod{x_i, v}^4 } \\
        & \leq \bigO{\eta^2 \log^2(1/\eta)} \cdot (v^\top \Sigma v)^2,
\end{split}
\end{equation*}
where the second inequality follows from applying \cref{eqn:resilience-intro} to the set $\calT$ indicated by $1-r_i$.  Of course, we do not know the uncorrupted set, and a brute-force search would require exponential time. The standard SoS approach, introduced by Kothari, Steinhardt and Steurer~\cite{kothari2018robust}, considers a relaxation that searches for the indicators of the uncorrupted set, subject to hypercontractivity constraints, and gets stuck at $\sqrt{\eta}$, much like the computation in \cref{eqn:hypercontractivity-computation-intro}. Encoding \textit{resilience} as a constraint instead would solve the issue, but only at the expense of exponential running time, since there are exponentially many subsets of size $\eta n$. 

Kothari, Manohar and Zhang~\cite{kothari2022polynomial} sidestep the issue of certifying \textit{resilience} in the SoS proof system by working with pseudo-expectations instead, and show that the rounded object must be close to the true covariance. This approach inherently outputs a $d^2$ dimensional object, the description of the estimated covariance and hence results in a two-shot algorithm for mean estimation (or even a hypothetical extension to regression). While this suffices for privately learning Gaussians in total variation distance~\cite{Hopkins2023Robustness}, this step loses an additional factor of $d$, getting sample complexity that scales as $d^2/(\alpha \eps)$, for covariance-aware mean estimation and regression. To avoid outputting a covariance, and isotropizing the samples, we require the SoS proof system to implicitly construct an object that is close to the true covariance, which in turn seems to require certifying resilience, bringing us back to where we started. We note that the techniques of~\cite{kothari2022polynomial} precisely avoid showing the existence of such sum-of-squares proofs.

\paragraph{Certifying Resilience in Sum-of-Squares.} We consider the following polynomial system (which is identical to the one in~\cite{kothari2022polynomial}): given an $\eta$-corrupted sample $\Set{x_i}_{i\in[n]}$, let $\Set{x_i', w_i}_{i\in [n]}$ be indeterminates and 
\begin{equation}
\label{eqn:cov-constraint-intro}
    \begin{split}
    \calA_\eta = \left \{
    \begin{aligned}
      &\forall i\in [n].
      & w_i^2
      & = w_i\hspace{0.1in}; \hspace{0.2in}  w_i x_i'  =  w_i x_i \\
      && \sum_{i \in [n]} w_i & = (1-\eta)n  \\
      & \forall v\in\mathbb{R}^d &  \frac{1}{n}\sum_{i \in [n]}  {\langle x'_i , v\rangle^{4}} 
      &\leq  \Paren{3+\eta \log^2 (1/\eta)} \left(  \frac{1}{n}\sum_{i \in [n]}  {\langle x'_i , v\rangle^{2}} \right)^{2}
    \end{aligned}
  \right \}
    \end{split}
\end{equation}
This system searches for a set of points that satisfies the hypercontractivity constraint. This constraint admits a succinct description, whenever the hypercontractivity inequality admits a sum-of-squares proof (see \cref{subsec:sos-background} for relevant background).  Our key insight is that while it is likely hard to certify resilience along all directions (cf.~\cite{hopkins2019hard}), and for all subsets, we can certify it for any fixed direction, and for the subsets that are indicated by the indeterminates $w_i$. We formalize this using the Selector Lemma (see \cref{lemma:subset_selection}): given non-negative scalars $\Set{a_i}_{i \in [n]}$ such that for all subsets $\calT$ of size $k$, $\sum_{i \in \calT} a_i \leq B$, there's a degree-two proof of the following:
\begin{equation}
\label{eqn:selector-lemma-intro}
\tag{Selector Lemma}
    \Set{ \forall i \in [n] \colon  w_i^2 = w_i ;  \sum_{i} w_i = k n } \sststile{}{} \Set{ \sum_{i \in [n]} w_i a_i \leq B  }.
\end{equation}
The \ref{eqn:selector-lemma-intro} allows us to certify resilience, as long as the only indeterminates are the SoS indicators. We use it as follows: let $\Sigma'= \frac{1}{n} \sum_{i \in [n]} x_i' (x_i')^\top$. Then, for any fixed direction $u$, we have

\begin{equation}
\label{eqn:cov-est-intro}
    \begin{split}
       \calA_\eta \sststile{}{w, x'} \Biggl\{ \Iprod{ \Sigma - \Sigma', u u^\top }^2 & = \Paren{    \frac{1}{n}\sum_{i \in [n]} \Iprod{x'_i, u}^2 - v^\top \Sigma v }^2   \\
       & \leq 2 \Paren{  \frac{1}{n}\sum_{i \in [n]} r_i w_i \Iprod{x'_i, u}^2    }^2 + 2 \Paren{ \frac{1}{n}\sum_{i \in [n]} (1-r_i w_i) \Iprod{x'_i, u}^2  } \Biggr\},
    \end{split}
\end{equation}
where the indicators $r_i$ denote whether the $i$-th sample was uncorrupted. As a consequence, $r_i w_i x_i' = r_i w_i x_i^*$, the uncorrupted samples, and therefore, the first term in \eqref{eqn:cov-est-intro} can be bounded as follows:
\begin{equation*}
\begin{split}
    \calA_\eta \sststile{}{} \Biggl\{ \Paren{  \frac{1}{n}\sum_{i \in [n]} r_i w_i \Iprod{x'_i, u}^2    }^2  & = \Paren{  \frac{1}{n}\sum_{i \in [n]} r_i w_i \Iprod{x^*_i, u}^2    }^2   \leq \bigO{\eta }  \Paren{ \frac{1}{n} \sum_{i \in [n]} w_i \cdot \Iprod{x_i^*, u}^4 }  \leq \tilde{\mathcal{O}}(\eta^2),     \Biggr\}
\end{split}
\end{equation*}
where the last inequality follows from applying resilience and the  \ref{eqn:selector-lemma-intro} to the SoS-indicators. The rest of the terms require careful analysis, but can be handled in a similar manner (see \cref{lem:main_identity_cov_sos_proof} for a complete proof). We emphasize here that the direction $u$ above is not allowed to be an indeterminate, as the Selector Lemma does not extend to this setting due to the hardness inherited from Hopkins and Li~\cite{hopkins2019hard}. However, we can still show that for all fixed $u$, 
\begin{equation}
\label{eqn:covariance-estimation-sos-intro}
   \calA_\eta \sststile{}{w,x'} \Set{ \Iprod{ \Sigma' - \Sigma , uu^\top }^2 \leq  \bigO{\eta^2 \log^2(1/\eta)} \Iprod{\Sigma, uu^\top}^2 }.
\end{equation}
This is precisely the kind of statement we require to execute a single shot regression and covariance-aware mean estimation algorithm, since we use $\Sigma'$ as an implicit proxy that captures the geometry of the samples. 

\paragraph{Single-shot Robust Regression.} We now have all the ingredients we need in order to obtain a robust regression algorithm with optimal rates and sample complexity. Here, our input is $\Set{(x_i, y_i)}_{i \in [n]}$ generated according to \cref{model:robust_regression}. Our constraint system resembles that of Bakshi and Prasad~\cite{bakshi2021robust}, but the hypercontractivity constant is set to $\Paren{3+\eta \log^2 (1/\eta)}$, as described in \cref{eqn:cov-constraint-intro}. In addition, we search for thee regression vector using an indeterminate $\theta'$ and require that $\frac{1}{n}\sum_{i \in [n]} \Paren{ \Iprod{ x_i', \theta'  } - y_i' } x_i' =0$. This constraint encodes the first-order optimality condition for the least-squares objective and was introduced as the \textit{gradient condition} in ~\cite{bakshi2021robust}. The key SoS inequality we show is as follows: 
\begin{equation}
    \calA_\eta \sststile{}{w,\theta',x'} \Set{  \Iprod{u, \Sigma\Paren{\theta' - \theta } }^2 \leq \bigO{\eta^2 \log^2(1/\eta)}  u^\top \Sigma u } 
\end{equation}
The proof is quite involved and requires applying the Selector Lemma several times (see \cref{lem:main_sos_proof_regression}). Next, observe suffices to obtain such a statement for a fixed direction $u$. To see this, let $\mu$ be a pseudo-distribution consistent with $\calA$, then for any fixed $u$, applying the pseudo-expectation operator, 
\begin{equation}
\label{eqn:pseduo-expec-regression-closeness}
    \Iprod{u, \Sigma\Paren{\pexpecf{\mu}{\theta'} - \theta } }^2  \leq \pexpecf{\mu}{ \Iprod{u, \Sigma\Paren{\theta' - \theta } }^2}  \leq \bigO{\eta^2 \log^2(1/\eta)}  u^\top \Sigma u .
\end{equation}
Setting $u = \pexpecf{\mu}{\theta'} - \theta$, we can conclude that $\Norm{\Sigma^{1/2} \Paren{  \pexpecf{}{ \theta' } - \theta  } }_2^2 \leq \bigO{\eta^2 \log^2(1/\eta)}$, as desired. The sample complexity of our algorithm is dominated by the concentration required for \textit{resilience}, and it follows from~\cite[Lemma 6.3]{Hopkins2023Robustness} that the sample complexity scales as $\bigO{\frac{d^2 + \log^2(1/\beta)}{\alpha^2}}$, where $\beta$ is the failure probability. In the regime where $\alpha$ is a constant, the SQ lower bound in \cite{diakonikolas2019robust} implies that $\Omega(d^2)$ samples are necessary.  Executing a similar argument also yields an algorithm for robust covariance-aware mean estimation (see \cref{thm:mean-est-sample-opt} for a complete proof).

\subsection{Geometry-Aware Exponential Mechanism}

For the remainder of the section, we describe how to use the single-shot robust estimator to execute an efficient version of the exponential mechanism~\cite{mcsherry2007mechanism}.

\paragraph{Idealized Exponential Mechnanism.} Consider the scenario where we know the covariance $\Sigma$ and are allowed to run in exponential time. Given a dataset $X=\Set{(x_i, y_i)}_{i \in [n]}$, let $\hat{\theta}(X)$ be a robust estimator for $X$, i.e. $\Norm{\Sigma^{1/2} \Paren{ \hat{\theta}(X) - \theta } }_2 \leq \alpha$.  We could then construct the following idealized exponential mechanism: given $0< \eps<1$ and samples $X$, we want to sample a parameter $\tilde{\theta}$ with probability 
\begin{equation*}
    \prob{\tilde{\theta} } \propto \exp\Paren{-\eps \cdot \textrm{score}( \tilde{\theta} ) } \textrm{ where }  \textrm{score}( \tilde{\theta} ) = \min\Paren{ d(X, X') \colon \Norm{ \Sigma^{1/2} \Paren{ \hat{\theta}(X') -  \tilde{\theta} } }_2 \leq \alpha },
\end{equation*}
where $\alpha \geq \eta \log(1/\eta)$, and $d(X, X')$ is the Hamming distance between datasets $X$ and $X'$. This score function assigns each  $\tilde{\theta}$ a score that is equal to the minimum number of samples that need to be changed in order for the robust estimator on the modified dataset $X'$ to be $\alpha$-close. It is folklore to show that such a score function enjoys optimal privacy guarantees. Using arguments similar to Hopkins, Kamath, Majid and Narayanan~\cite{Hopkins2023Robustness}, combined with our robust regression estimator, we can show that the score is quasi-convex and efficiently sampleable, as long as the covariance is known. However, when the covariance is unknown, it is unclear even how to check whether two candidates are close in the $\Sigma^{1/2}$ norm. A simple fix would be to ignore the geometry, and define the closeness to be $\Norm{ \Paren{ \hat{\theta}(X') -  \tilde{\theta} } }_2 \leq \alpha/\sqrt{\kappa}$, where $\kappa$ is the condition number of the covariance. However, this reduces the volume of points with low score by a factor of the condition number, and results in a sample complexity that scales proportional to $d\log(\kappa)/(\alpha\eps)$.

\paragraph{Towards an efficiently computable score function.} Following~\cite{Hopkins2023Robustness}, we can define a relaxation of the score function. Recall the system from \cref{eqn:cov-constraint-intro}, instantiated with $\eta =t/n$. Then, a natural relaxation to consider is 

\begin{equation}
\label{eqn:score-attempt}
\tag{First Cut Score}
    \textrm{score}( \tilde{\theta} )  = \min t \colon \exists  \textrm{ a degree $\mathcal{O}(1)$ } \pexpecf{}{\cdot} \textrm{s.t.}  \pexpecf{}{\cdot} \sdtstile{}{} \calA_{t/n} \textrm{ and } \pexpecf{}{  \Iprod{ \Sigma',  \Paren{\theta' - \tilde{\theta}} \Paren{\theta' - \tilde{\theta}}^\top }  }  \leq \alpha ,
\end{equation}
i.e. there exists a constant-degree pseudo-expectation operator that after changing $t$ points in $X$ satisfies the constraints, and the resulting robust estimator is close to $\tilde{\theta}$, wrt the SoS covariance $\Sigma'$. Intuitively, we established in \cref{eqn:covariance-estimation-sos-intro} that $\Sigma'$ and $\Sigma$ are close, so the closeness constraint captures $\tilde{\theta}$ and $\theta'$ being close in $\Sigma^{1/2}$ norm. The desiderata for any score function to work is threefold: it should have (a) bounded sensitivity, (b) should admit an efficient separation oracle and (c) should be quasi-convex. 
It is not hard to show that the score function in  \cref{eqn:score-attempt} satisfies bounded sensitivity and it admits an efficient separation oracle since all the constraints are linear in $\pexpecf{}{\cdot}$. However, this score function is not quasi-convex in $\tilde{\theta}$ and therefore there may be no efficient sampling procedure for it. 

One way to make the score function quasi-convex in $\tilde{\theta}$ is to first fix some $\hat{\Sigma}$ and then use this estimate to define the score function as follows:
\begin{equation}
\label{eqn:score-attempt-two}
\tag{Second Cut Score}
    \textrm{score}( \tilde{\theta} )  = \min t \colon \exists  \textrm{ a degree $\mathcal{O}(1)$ } \pexpecf{}{\cdot} \textrm{s.t.}  \pexpecf{}{\cdot} \sdtstile{}{} \calA_{t/n} \textrm{ and }  
    \left\Vert\pexpecf{}{\hat{\Sigma}^{1/2}(\theta' - \tilde{\theta})}\right\Vert^2   \leq \alpha^2 ,
\end{equation}
However, if we estimate $\hat{\Sigma}$ via a robust estimator and then fix this estimate, the score function no longer admits a non-trivial bound on sensitivity. Maintaining bounded sensitivity with this score function requires estimating $\hat{\Sigma}$ privately, which incurs a sample complexity of $d^2/(\alpha\eps)$. Thus, we cannot estimate $\hat{\Sigma}$ as a first step and can only work with implicit representations within our program. 

\paragraph{Handling Matrix Inverses in SoS.} To design a quasi-convex score, it suffices to decouple $\Sigma'$ from $\tilde{\theta}$ and again, if we were given the true covariance $\Sigma$, we could work with the following closeness constraint instead: for all $v$, 
\begin{equation}
\label{eqn:ideal-closeness-constraint}
    \Iprod{ v, \pexpecf{}{\theta'} - \tilde{\theta} }^2 \leq \alpha^2 \cdot v^\top \Sigma^{-1} v. 
\end{equation}
Observe, this constraint is equivalent to \cref{eqn:pseduo-expec-regression-closeness}, after substituting $u = \Sigma^{-1}v$. We show that this constraint is convex in $\tilde{\theta}$: given $\tilde{\theta}_1$ and $\tilde{\theta}_2$ such that \cref{eqn:ideal-closeness-constraint} holds for $\pexpecf{1}{\cdot}$ and $\pexpecf{2}{\cdot}$ respectively, for any convex combination $\lambda \tilde{\theta}_1 + (1-\lambda) \tilde{\theta}_2$, let $\pexpecf{3}{\cdot} = \lambda \pexpecf{1}{\cdot} + (1-\lambda)   \pexpecf{2}{\cdot}$. Then,
\begin{equation}
\label{eqn:quasi-convexity-intro}
    \begin{split}
        \Iprod{v, \pexpecf{3}{\theta'} - \lambda \tilde{\theta}_1 + (1-\lambda) \tilde{\theta}_2  }^2 & = \Iprod{v, \lambda\Paren{ \pexpecf{1}{\theta'} -  \tilde{\theta}_1 } +  (1-\lambda) \Paren{  \pexpecf{2}{\theta'} - \tilde{\theta}_2 }  }^2 \\
        & \leq \lambda \Iprod{v, \pexpecf{1}{\theta'} -  \tilde{\theta}_1 }^2 + (1-\lambda) \Iprod{v, \pexpecf{2}{\theta'} - \tilde{\theta}_2}^2 \\
        & \leq \alpha^2 \cdot v^\top \Sigma^{-1} v,
    \end{split}
\end{equation}
where we use convexity of $f(z)=z^2$. 
Therefore, it would suffice for us to construct an implicit representation of $\Sigma^{-1}$ in the SoS proof system. A priori, such a representation might seem implausible, since the inverse cannot be approximated by a low-degree polynomial, specially for ill-conditioned matrices. However, we need an implicit representation that is only a constant approximation in L\"owner ordering. We consider the following auxiliary constraints: 

\begin{equation}
    \calB = \Set{ \Sigma' = \frac{1}{n}\sum_{i\in[n]} x_i' (x_i')^\top, \hspace{0.1in}  Q \succeq 0 ,\hspace{0.1in} Q\Sigma' = I , \hspace{0.1in}  \Sigma' Q= I ,\hspace{0.1in}  Q \Sigma' Q = Q },
\end{equation}
where $\Sigma'$ is the covariance of the indeterminates, $Q$ is the implicit representation for $\Sigma^{-1}$ and the constraints satisfy that $Q$ is a left and right inverse for $\Sigma'$. This system is feasible for any $Q, \Sigma'$ that are symmetric and inverses of each other. Further, we can derive from the hypercontractivity constraint that $\Sigma'$ and $\Sigma$ are close (as we did in \cref{eqn:covariance-estimation-sos-intro}). We manage to show the following inequality (\cref{lem:sos_inverse_covariance}):
\begin{equation}
\label{eqn:q-closeness-intro}
    \calA \cup \calB \sststile{}{u, \Sigma' ,Q} \Set{ \Iprod{Q, uu^\top } \leq \bigO{1} \Iprod{ \Sigma^{-1}, uu^\top } },
\end{equation}
and the proof appears in \cref{sec:deferred_privacy_sos_proofs}. Our proof reduces relating $Q$ and $\Sigma^{-1}$ to bounding the closeness of $\Sigma$ and $\Sigma'$ under all bi-linear forms, i.e. 
\begin{equation*}
    \calA  \sststile{}{u,v \Sigma'} \Set{ \Iprod{ \Sigma - \Sigma', u v^\top}^2 \leq \bigO{\eta} \paren{u^\top \Sigma u}^2 + \bigO{\eta} \paren{v^\top \Sigma v}^2  }\,.
\end{equation*}
Now that we have an implicit representation of $\Sigma^{-1}$, we can rewrite the idealized closeness constraint from \cref{eqn:ideal-closeness-constraint} as follows: for all $v \in \mathbb{R}^d$, 
\begin{equation}
    \Iprod{ v, \pexpecf{}{\theta'} - \tilde{\theta} }^2 \leq \alpha^2 \cdot \Iprod{ Q, vv^\top }.
\end{equation}
This constraint admits a concise description: $\alpha^2 \pexpecf{}{Q} - \Paren{ \pexpecf{}{\theta'} - \tilde{\theta}}  \Paren{ \pexpecf{}{\theta'} - \tilde{\theta}}^\top \succeq 0$. Now, we are ready to define our geometry-aware score function.

\paragraph{Geometry-aware score.} We are finally ready to describe our SoS relaxation of the score: 
\begin{equation}
\label{eqn:geometry-aware-score}
\tag{Geometry-Aware Score}
\begin{split}
    \textrm{score}( \tilde{\theta} )   = \min t & \colon \exists  \textrm{ a degree $\mathcal{O}(1)$ } \pexpecf{}{\cdot} \textrm{s.t.}  \pexpecf{}{\cdot} \sdtstile{}{} \calA_{t/n} \textrm{ and }  \\
    & \alpha^2 \pexpecf{}{Q} - \Paren{ \pexpecf{}{\theta'} - \tilde{\theta}}  \Paren{ \pexpecf{}{\theta'} - \tilde{\theta}}^\top \succeq 0 ,
\end{split}
\end{equation}
It is easy to show that this score satisfies bounded sensitivity. Since we decoupled $Q$ and $\theta'$, it is also straightforward to show that \ref{eqn:geometry-aware-score} is convex, in fact, the argument is almost identical to \cref{eqn:quasi-convexity-intro}. Next, we show that our relaxation outputs an accurate estimate of the true regression vector:
\begin{equation}
\label{eqn:utility-intro}
\begin{split}
    \Norm{\Sigma^{1/2}\Paren{  \tilde{\theta}  - \theta } }_2 & = \Norm{\Sigma^{1/2}\Paren{  \tilde{\theta}  - \theta \pm \pexpecf{}{\theta'}  } }_2  \\
    &\leq  \Norm{\Sigma^{1/2}\Paren{  \pexpecf{}{\theta'}  - \theta    } }_2  +  \Norm{\Sigma^{1/2}\Paren{  \pexpecf{}{\theta'}  - \tilde{\theta}    } }_2 \\
    & \leq \alpha + \frac{T\log(n/T) }{n} +  \underbrace{ \Norm{\Sigma^{1/2}\Paren{  \pexpecf{}{\theta'}  - \tilde{\theta}    } }_2}_{\eqref{eqn:utility-intro}.(1)}
\end{split}
\end{equation}
We use the closeness formulation in \ref{eqn:geometry-aware-score} to bound term \eqref{eqn:utility-intro}.(1). Observe, taking quadratic forms, we have that for all $u$,
\begin{equation*}
    \Iprod{u , \pexpecf{}{ \theta'} - \tilde{\theta} }^2 \leq \alpha^2 \cdot u^\top \pexpecf{}{Q} u \leq \bigO{\alpha^2} u^{\top}\Sigma^{-1}u,
\end{equation*}
where the last inequality follows from \cref{eqn:q-closeness-intro}. In particular, setting $u \gets \Sigma^{1/2} \Paren{\pexpecf{}{ \theta'} - \tilde{\theta}} $, we have 
\begin{equation*}
    \Norm{  \Sigma^{1/2} \Paren{ \pexpecf{}{ \theta'} - \tilde{\theta}} }^2_2 \leq \bigO{\alpha^2}, 
\end{equation*}
which concludes the utility argument. 

\paragraph{Efficient Separation Oracle.}
A keen reader may have noticed that the score function is no longer linear in $\pexpecf{}{\cdot}$.  We show that the constraint 
$$\alpha^2 \pexpecf{}{Q} - \Paren{ \pexpecf{}{\theta'} - \tilde{\theta}}  \Paren{ \pexpecf{}{\theta'} - \tilde{\theta}}^\top \succeq 0$$
admits an efficiently computable hyperplane, even though it is quadratic in $\pexpecf{}{\cdot}$. Restricting to the special case of one dimension, and setting $y = \pexpecf{}{Q}$ and $x= \pexpecf{}{\theta'}-\tilde{\theta}$, we want an efficient separation oracle for the set $\alpha^2 y - x^2\geq 0$. The feasible region for this set is clearly convex since it corresponds to the area above the parabola. The separating hyperplane is then just the appropriate tangent to the parabola.  
Specifically, if $(x_0, y_0)$ does not satisfy the constraint then the separating hyperplane is 
\[ 2x_0 x - \alpha^2 y = x_0^2 \,.\]
Furthermore, we show that this argument has a natural extension to higher dimensions, and we provide a complete proof in \cref{lem:closeness-separation-oracle}.

\section{Related Work}
\label{sec:related-work}
\paragraph{Private Linear Regression.}
There has been significant work on private linear regression, see, e.g.~\cite{kifer2012private, bassily2014private, wang2015privacy, sheffet2017differentially, wang2018revisiting}. Here we will focus on a subset of the most relevant work.
Employing a high-dimensional propose-test-release approach 
\cite{liu2022differential} studies this problem under the constraint of approximate differential privacy and obtains an inefficient algorithm requiring $\frac{d}{\alpha^2} + \frac{d + \log\paren{1/\delta}}{\alpha \eps}$ many samples. This algorithm satisfies robustness to constant fraction corruptions.

Efficient algorithms for this problem that require sample complexity linear in the dimension have been studied under the constraint of approximate differential privacy. \cite{varshney2022nearly} gives an algorithm based on gradient descent for this problem, requiring sample complexity linear in $d$ but quadratic in $\kappa$, the condition number of the covariance of the input distribution. This algorithm is not robust. \cite{liu2023near} gives an algorithm for this problem using $\frac{d}{\alpha^2}+ \frac{\sqrt{\kappa}d \log\paren{1/\delta}}{\alpha \eps}$ many samples. This algorithm is robust to label noise but is not robust to corruptions in the covariates. \cite{brown2024insufficient} gives the current best known algorithm for this problem under approximate differential privacy, removing the dependence on the condition number and using $\frac{d}{\alpha^2} + \frac{d \sqrt{\log\paren{1/\delta}}}{\alpha \eps} + \frac{d \log (1/\delta)^2}{\eps^2}$ many samples. However, it requires a quadratic dependence on the privacy parameter $1/\eps$. This algorithm also does not satisfy robustness to constant fraction corruptions.
Under the assumption of pure differential privacy, \cite{asi2023robustness} gives an inefficient algorithm based on a reduction to robust algorithms, requiring $\frac{d}{\alpha^2} + \frac{d \log\paren{R + \sqrt{\kappa}}}{\alpha \eps}$ many samples.

\paragraph{Private Covariance-Aware Mean Estimation.}

In covariance-aware mean estimation, we are interested in estimating the mean of a distribution using the Mahalanobis error metric, which measures the error relative to the covariance of the distribution. In this work, our focus is on the setting where the distribution is Gaussian.
Under the assumption of approximate differential privacy, \cite{kamath2019privately} gives an efficient algorithm for this problem that requires estimating the covariance spectrally, leading to a sample complexity of $\Omega\paren{d^{1.5}}$. Overcoming this super-linear dependence on the dimension $d$ has been the focus of subsequent work.
\cite{liu2022differential} and \cite{brown2021covariance} give inefficient algorithms for this problem that require a sample size linear in the dimension: $\frac{d}{\alpha^2} + \frac{d + \log\paren{1/\delta}}{\alpha \eps}$ and $\frac{d}{\alpha^2} + \frac{d}{\alpha \eps} + \frac{\log\paren{1/\delta}}{\eps}$, respectively. Two concurrent works, \cite{brown2023fast} and \cite{kuditipudi2023pretty}, obtain fast algorithms with similar guarantees. \cite{kuditipudi2023pretty} gives an algorithm for this problem that takes $\frac{d}{\alpha^2} + \frac{d \log\paren{1/\delta}}{\alpha \eps} + \frac{d \log^2\paren{1/\delta}}{\eps^2}$ samples. \cite{brown2023fast} provides an algorithm with slightly better sample complexity of $\frac{d}{\alpha^2} + \frac{d \sqrt{\log\paren{1/\delta}}}{\alpha \eps} + \frac{d \log\paren{1/\delta}}{\eps}$, improving the dependence on $\log\paren{1/\delta}$ and $1/\eps$. Neither of these algorithms is robust to constant fraction corruptions.
Similar to the regression setting, under the assumption of pure differential privacy, \cite{asi2023robustness} gives an inefficient algorithm based on a reduction to robust algorithms, requiring $\frac{d}{\alpha^2} + \frac{d \log\paren{R + \sqrt{\kappa}}}{\alpha \eps}$ samples.

\section{Preliminaries}

\subsection{Notation}

We use $\tilde{O}(\cdot)$ and $\tilde{\Omega}(\cdot)$ to suppress logarithmic factors in $\tfrac 1 \alpha, \tfrac 1 \eta$.
Importantly, this does not suppress logarithmic factors in $\tfrac 1 \beta$, where $\beta$ is the failure probability of our algorithm.

We also use the following definition.
\begin{definition}
    \label{def:m-ball}
    For any matrix $M$, let the $M$-ball of radius $r$ around $x$ be all points $y$ such that $\norm{M(y-x)}_2 \leq r$.
\end{definition}

\subsection{Concentration Bounds}

We need the following concentration bounds.
The first set will be used for mean estimation and linear regression in case the data is isotropic.
The second (stronger) set will be used for covariance estimation in relative Frobenius norm.
While the standard Gaussian distribution will satisfy all the conditions below, they also apply to more general sub-Gaussian distributions.
In particular, the following conditions are sufficient (see the end of the second section for a more general definition in terms of what concentration properties are sufficient).
\begin{definition}
    \label{def:moment_match_sub_gaussian}
    We call a distribution $\mathcal{D}$ \emph{fourth moment matching reasonable} sub-Gaussian, if it is mean-zero and identity-covariance, and has independent $O(1)$-sub-Gaussian coordinates that have fourth moment equal to 3.
\end{definition}

\paragraph{First moment subset concentration bounds.}

We prove the bounds presented here in \cref{sec:eps_goodness_proof}.
\begin{fact}[Mean and Covariance Bounds]
\torestate{
\label{fact:mean_and_cov_of_1-eps}
Let $\eta \leq 1/e, \psi \leq O(\eta \sqrt{\log(1/\eta)})$ and $n \geq \Omega(\tfrac{d + \log(1/\beta)}{\eta^2 \log(1/\eta)})$ and $\Sigma$ such that $(1-\psi)I_d \preceq \Sigma \preceq (1+\psi) I_d$.
Given a set $\calS = \{x_1^*, x_2^*, \ldots, x_n^*\}$ of \iid samples from a distribution $\calD$ such that the distribution $\Sigma^{-1/2}(\mathcal{D} - \mu)$ is a fourth moment matching reasonable sub-Gaussian distribution.
Then, with probability $1- \tfrac \beta 2$ it holds that, for all subsets $\calT \subset \calS$ of size $(1-\eta)n$ we have the following bound‚
\begin{equation*}
    \Norm{ \frac{1}{(1-\eta)n}\sum_{x_i^* \in \calT} x_i^* - \mu }_2 \leq \bigO{\eta \sqrt{\log(1/\eta)}} %
\end{equation*}
and also
\begin{equation*}
    \Norm{\frac{1}{(1-\eta) n} \sum_{x_i^* \in \calT} \Paren{x_i^* - \mu} \Paren{x_i^* - \mu}^\top  - I_d }_{\textrm{op}} \leq  \bigO{\eta \log(1/\eta) + \psi} \,. %
\end{equation*}
}
\end{fact}

\begin{fact}[Covariance of all small subsets]
\torestate{
\label{fact:cov_small_subset}
Let $\eta \leq 1/e, \psi \leq O(\eta \sqrt{\log(1/\eta)})$ and $n \geq \Omega(\tfrac{d + \log(1/\beta)}{\eta^2 \log(1/\eta)})$ and $\Sigma$ such that $(1-\psi)I_d \preceq \Sigma \preceq (1+\psi) I_d$.
Given a set $\calS = \{x_1^*, x_2^*, \ldots, x_n^*\}$ of \iid samples from a distribution $\calD$ such that the distribution $\Sigma^{-1/2}(\mathcal{D} - \mu)$ is a fourth moment matching reasonable sub-Gaussian distribution.
Then, with probability at least $1-\tfrac \beta 2$ it holds that, for all subsets $\calT \subset \calS$ of size $\eta n$, we have the following bound
\begin{equation*}
    \Norm{\frac{1}{(1-\eta) n} \sum_{x_i^* \in \calT} \Paren{x_i^* - \mu} \Paren{x_i^* - \mu}^\top   }_{\textrm{op}} \leq  \bigO{\eta \log(1/\eta)} \,.%
\end{equation*}
}
\end{fact}

\begin{definition}
\label{def:eps_goodness}
We denote a set $\calS = \{x_1^*, x_2^*, \ldots, x_n^*\}$ of $n$ \iid samples as \emph{$(\eta,\psi)$
-good}, if $\calS$ satisfies the bounds in \cref{fact:mean_and_cov_of_1-eps,fact:cov_small_subset}.
If $\psi \leq O(\eta \log(1/\eta))$, we simply call the set \emph{$\eta$-good}.
\end{definition}

Note that~\cref{fact:mean_and_cov_of_1-eps,fact:cov_small_subset} imply that a set $\calS$ of $n \geq \Omega(\tfrac{d + \log(1/\beta)}{\eta^2 \log(1/\eta)})$ \iid samples from a fourth moment matching reasonable sub-Gaussian distribution with mean $\mu$ and covariance $(1-\psi) I_d \preceq \Sigma \preceq (1+\psi) I_d$  is $(\eta,\psi)$-good with probability at least $1-\beta$.

\paragraph{Second moment subset concentration bounds.}

The following bounds appeared in~\cite[Lemma 6.3]{Hopkins2023Robustness}.\footnote{Note that Lemma 6.3 as stated in~\cite{Hopkins2023Robustness} only applies to symmetric matrices $P$. We can reduce the case of asymmetric $P$ to this by noting that $\iprod{x_i x_i^\top - I_d}{P} = \iprod{x_i x_i^\top - I_d}{\tfrac 1 2 (P + P^\top)}$ and $\norm{\tfrac 1 2 (P + P^\top)}_F \leq \norm{P}_F$. } 
We remark that they stated the result only for samples from the standard Gaussian distribution.
Inspecting their proof, it becomes clear that the result hold for any mean-zero, isotropic distribution $\mathcal{D}$ such that the following conditions are met:
\begin{itemize}
    \item \textbf{Hanson-Wright:} There is a universal constant $c > 0$ such that for all symmetric matrices $P$ and all $t > 0$ it holds that $\mathbb{P}_{X \sim \mathcal{D}} (\lvert\iprod{XX^T - I_d}{P}\rvert > t) \leq 2\exp(-c \min\{\tfrac {t^2} {\norm{P}_F^2}, \tfrac t {\norm{P}_{\mathrm{op}}}\})$.
    \item \textbf{Gram Matrix Concentration:} For any $m$ large enough and $X_1, \ldots, X_m$ i.i.d.\ sampled from $D$, let $A \in \R^{m \times d}$ be the matrix with rows $X_i$. Then, there exists a constant $c > 0$ such that $\mathbb{P}(\norm{A}_{\mathrm{op}} > c \cdot (\sqrt{m} + \sqrt{d} + t) \leq 2 e^{-\Omega(t^2)}$.
    \item \textbf{Fourth Moment Condition:} For every matrix $P$, it holds that $\E_{X \sim \mathcal{D}} \iprod{XX^T - I_d}{P}^2 = 2 \pm O(\eta \log(1/\eta))$.
\end{itemize}
It follows by standard results~\cite{rudelson2013hanson,Vershynin_2018}, that fourth moment matching reasonable sub-Gaussian distributions satisfy all three conditions (the last one follows since the fourth moment tensor of fourth moment matching reasonable sub-Gaussian distributions is the same as that of the standard Gaussian, for which the expectation is exactly 2).

\begin{lemma}
\label{lem:higher-order-stability}
Let $\eta > 0$ and $n \geq \tilde{\Omega}(\tfrac{d^2 + \log^2(1/\beta)}{\eta^2})$.
Let $x_1^*, \ldots, x_n^*$ be \iid samples from a mean-zero fourth moment matching reasonable sub-Gaussian distribution.
Then with probability at least $1-\beta$ the following statements hold (simultaneously):
For all  $P \in \R^{d\times d}$ of Frobenius norm at most 1, it holds that
\begin{enumerate}
    \item $\Abs{\frac{1}{n}\sum_{i \in [n]} \iprod{x_i^*(x_i^*)^\top - I_d}{P}} \leq O(\eta \log (1 /\eta))$,
    \item $\Abs{\frac{1}{n}\sum_{i \in [n]} \iprod{x_i^*(x_i^*)^\top - I_d}{P}^2 - 2} \leq O(\eta \log^2 ( 1 /\eta))$,
    \item For any $\calT \subseteq [n]$ of size at most $\eta n$, it holds that (a) $\Abs{\frac{1}{n}\sum_{i \in \calT} \iprod{x_i^*(x_i^*)^\top - I_d}{P}} \leq O(\eta \log(1/\eta))$ and (b) $\Abs{\frac{1}{n}\sum_{i \in \calT} \iprod{x_i^*(x_i^*)^\top - I_d}{P}^2} \leq O(\eta \log^2(1/\eta))$,
    \item $\frac{1}{n}\sum_{i \in [n]} \Abs{\iprod{x_i^*(x_i^*)^\top - I_d}{P}} \leq O(1)$.
\end{enumerate}
\end{lemma}
\begin{definition}
    \label{def:higher_order_goodness}
    We denote a set $\calS = \{x_1^*, x_2^*, \ldots, x_n^*\}$ of $n$ \iid samples as \emph{$\eta$
-higher-order-good}, if the points in $\calS$ satisfy conclusions 1-4 in~\cref{lem:higher-order-stability}.
\end{definition}

\begin{fact}[Gaussian Hypercontractivity]
\label{fact:hypercontractivity}
For any $k \in \mathbb{N}$, 
\begin{equation*}
    \expecf{x \sim \calN(0,\Sigma)}{ \Iprod{x, v}^{2k} } \leq \mathcal{O}_{k}\Paren{1} \Paren{v^\top \Sigma v }^{k}.
\end{equation*}
\end{fact}

\subsection{Sum-of-Squares Background}
\label{subsec:sos-background}

\paragraph{The sum-of-squares framework.}
\label{subsec:sos-framework}

We now provide an overview of the sum-of-squares proof system.
We closely follow the exposition as it appears in the lecture notes of Barak~\cite{barak2016proofs}.   

\paragraph{Pseudo-Distributions.}
A discrete probability distribution over $\R^m$ is defined by its probability mass function, $D\from \R^m \to \R$, which must satisfy $\sum_{x \in \mathrm{supp}(D)} D(x) = 1$ and $D \geq 0$.
We extend this definition by relaxing the non-negativity constraint to merely requiring that $D$ passes certain low-degree non-negativity tests.
We call the resulting object a pseudo-distribution.

\begin{definition}[Pseudo-distribution]
A \emph{degree-$\ell$ pseudo-distribution} is a finitely-supported function $D:\R^m \rightarrow \R$ such that $\sum_{x} D(x) = 1$ and $\sum_{x} D(x) p(x)^2 \geq 0$ for every polynomial $p$ of degree at most $\ell/2$, where the summation is over all $x$ in the support of $D$.
\end{definition}
Next, we define the related notion of pseudo-expectation.
\begin{definition}[Pseudo-expectation]
The \emph{pseudo-expectation} of a function $f$ on $\R^m$ with respect to a pseudo-distribution $\mu$, denoted by $\pexpecf{\mu(x)}{f(x)}$,  is defined as
\begin{equation*}
    \pexpecf{\mu(x)}{f(x)} = \sum_{x} \mu(x) f(x).
\end{equation*}
\end{definition}
We use the notation $\pexpecf{\mu(x)}{(1,x_1, x_2,\ldots, x_m)^{\otimes \ell}}$ to denote the degree-$\ell$ moment tensor of the pseudo-distribution $\mu$.
In particular, each entry in the moment tensor corresponds to the pseudo-expectation of a monomial of degree at most $\ell$ in $x$. 

\begin{definition}[Constrained pseudo-distributions]
\label{def:constrained-pseudo-distributions}
Let $\calA = \Set{ p_1\geq 0 , p_2\geq0 , \dots, p_r\geq 0}$ be a system of $r$ polynomial inequality constraints of degree at most $d$ in $m$ variables.
Let $\mu$ be a degree-$\ell$ pseudo-distribution over $\mathbb{R}^m$.
We say that $\mu$ \emph{satisfies} $\calA$ at degree $\ell \ge1$ if for every subset $\calS \subset [r]$ and every sum-of-squares polynomial $q$ such that $\deg(q) + \sum_{i \in \calS } \max\Paren{ \deg(p_i), d} \leq \ell$, $\pexpecf{\mu}{ q \prod_{i \in \calS} p_i } \geq 0$.
Further, we say that $\mu$ \emph{approximately satisfies} the system of constraints $\calA$ if the above inequalities are satisfied up to additive error $\pexpecf{\mu}{ q \prod_{i \in \calS} p_i } \geq -2^{-n^{\ell} } \norm{q} \prod_{i \in \calS} \norm{p_i}$, where $\norm{\cdot}$ denotes the Euclidean norm of the coefficients of the polynomial, represented in the monomial basis.  
\end{definition}

Crucially, there's an efficient separation oracle for moment tensors of constrained pseudo-distributions. 
Below gives the unconstrained statement; the constraint statement follows analogously.

\begin{fact}[\cite{shor1987approach, nesterov2000squared, parrilo2000structured, grigoriev2001complexity}]
    \label{fact:sos-separation-efficient}
    For any $m,\ell \in \N$, the following convex set has a $m^{\bigO{\ell}}$-time weak separation oracle, in the sense of \cite{grotschel1981ellipsoid}:\footnote{
        A separation oracle of a convex set $S \subset \R^M$ is an algorithm that can decide whether a vector $v \in \R^M$ is in the set, and if not, provide a hyperplane between $v$ and $S$.
        Roughly, a weak separation oracle is a separation oracle that allows for some $\eta$ slack in this decision.
    }:
    \begin{equation*}
        \Set{  \pexpecf{\mu(x)} { (1,x_1, x_2, \ldots, x_m)^{\otimes \ell } } \Big\vert \text{ $\mu$ is a degree-$\ell$ pseudo-distribution over $\R^m$}}
    \end{equation*}
\end{fact}
This fact, together with the equivalence of weak separation and optimization \cite{grotschel1981ellipsoid} forms the basis of the sum-of-squares algorithm, as it allows us to efficiently approximately optimize over pseudo-distributions. 

Given a system of polynomial constraints, denoted by $ \calA$, we say that it is \emph{explicitly bounded} if it contains a constraint of the form $\{ \|x\|^2 \leq 1\}$. Then, the following fact follows from  \cref{fact:sos-separation-efficient} and \cite{grotschel1981ellipsoid}:

\begin{theorem}[Efficient optimization over pseudo-distributions]
    \label{fact:eff-pseudo-distribution}
There exists an $(m+r)^{O(\ell)} $-time algorithm that, given any explicitly bounded and satisfiable system $ \calA$ of $r$ polynomial constraints in $m$ variables, outputs a degree-$\ell$ pseudo-distribution that satisfies $ \calA$ approximately, in the sense of~\cref{def:constrained-pseudo-distributions}.\footnote{
    Here, we assume that the bit complexity of the constraints in $ \calA$ is $(m+t)^{O(1)}$.
}
\end{theorem}

We now state some standard facts for pseudo-distributions, which extend facts that hold for standard probability distributions.
These can be found in the prior works listed above.

\begin{fact}[Cauchy--Schwarz for pseudo-distributions]
Let $f,g$ be polynomials of degree at most $d$ in the variables $x \in \R^m$. Then, for any degree-$d$ pseudo-distribution $\mu$,
$\pexpecf{\mu}{fg}  \leq \sqrt{ \pexpecf{\mu}{ f^2} } \cdot \sqrt{ \pexpecf{\mu}{ g^2} }$.
    \label{fact:pseudo-expectation-cauchy-schwarz}
\end{fact}

\begin{fact}[Hölder's inequality for pseudo-distributions] \label{fact:pseudo-expectation-holder}
Let $f,g$ be polynomials of degree at most $d$ in the variables $x \in \R^m$. 
Fix $t \in \N$. Then, for any degree-$dt$ pseudo-distribution $\mu$,
\begin{equation*}
    \pexpecf{\mu}{ f^{t-1}  g} \leq \Paren{ \pexpecf{\mu}{ f^t }  }^{\frac{t-1}{t}} \cdot  \Paren{  \pexpecf{\mu }{ g^t }  }^{\frac{1}{t}}.  
\end{equation*}
In particular, when $t$ is even,
$\pexpecf{\mu}{f}^t \leq \pexpecf{\mu}{ f^t }$.
\end{fact}

\paragraph{Sum-of-squares proofs.}

Let $f_1, f_2, \ldots, f_r$ and $g$ be multivariate polynomials in the indeterminates $x \in \R^m$.
Given the constraints $\{f_1 \geq 0, \ldots, f_r \geq 0\}$, a \emph{sum-of-squares proof} of the identity $\{g \geq 0\}$ is a set of polynomials $\{p_S\}_{S \subseteq [r]}$ such that
\begin{equation*}
    g = \sum_{S \subseteq [r]} p^2_S \cdot \prod_{i \in S} f_i.
\end{equation*}
As its name suggests, the existence of such an SoS proof shows that if the constraints $\{f_i \geq 0 \mid i \in [r]\}$ are satisfied, then the identity $g \geq 0$ is also satisfied.
We say that this SoS proof has \emph{degree $\ell$} if for every set $S \subseteq [r]$, the polynomial $p^2_S \Pi_{i \in S} f_i$ has degree at most $\ell$.
If there is a degree-$\ell$ SoS proof that $\{f_i \geq 0 \mid i \in [r]\}$ implies $\{g \geq 0\}$, we write
\begin{equation}
    \{f_i \geq 0 \mid i \in [r]\} \sststile{\ell}{x}\{g \geq 0\}.
\end{equation}
We will sometimes drop the indeterminate in $\sststile{\ell}{x}$ when this causes no confusion.
For all polynomials $f,g\colon\R^m \to \R$ and for all coordinate-wise polynomials $F\colon \R^m \to \R^{m_F}$, $G\colon \R^m \to \R^{m_G}$, $H\colon \R^{m_H} \to \R^m$, we have the following inference rules.\footnote{
    This notation should be read in the following way: given the proofs above the bar line, we can derive the proof below the bar line.
}
\begin{figure}[h]
    \renewcommand{\arraystretch}{1.5}
    \begin{center}
    \begin{tabular}{c c}
        Addition Rule & Multiplication Rule \\
        \vspace{1em}
        $\displaystyle\frac{ \calA \sststile{\ell}{} \{f \geq 0, g \geq 0 \} } { \calA \sststile{\ell}{} \{f + g \geq 0\}}$
        & $\displaystyle\frac{ \calA \sststile{\ell}{} \{f \geq 0\},\quad \calA \sststile{\ell'}{} \{g \geq 0\}} { \calA \sststile{\ell+\ell'}{} \{f \cdot g \geq 0\}}$ \\
        Transitivity Rule & Substitution Rule \\
        $\displaystyle\frac{ \calA \sststile{\ell}{}  \calB,\quad \calB \sststile{\ell'}{} C}{ \calA \sststile{\ell \cdot \ell'}{} C}$
        & $\displaystyle\frac{\{F \geq 0\} \sststile{\ell}{} \{G \geq 0\}}{\{F(H) \geq 0\} \sststile{\ell \cdot \deg(H)} {} \{G(H) \geq 0\}}$
    \end{tabular}
    \end{center}
    \renewcommand{\arraystretch}{1}
\end{figure}

Sum-of-squares proofs allow us to deduce properties of pseudo-distributions that satisfy some constraints.
\begin{fact}[Soundness]
  \label{fact:sos-soundness}
  Let $\mu$ be a degree-$\ell$ pseudo-distribution.
  If $\mu$ is consistent with the set of degree-$d_A$ polynomial constraints $\calA$, denoted $\mu \sdtstile{d_A}{} \calA$, and there is a degree-$d_B$ sum-of-squares proof that $ \calA \sststile{d_B}{} \calB$, and $\ell \geq d_Ad_B$, then $\mu \sdtstile{d_A d_B}{}  \calB$.
\end{fact}

We also have a converse to \cref{fact:sos-soundness}: every property of low-level pseudo-distributions can be derived by low-degree sum-of-squares proofs.
\begin{fact}[Completeness]
  \label{fact:sos-completeness}
  Let $d \geq r \geq r'$.
  Suppose $ \calA$ is a collection of polynomial constraints with degree at most $r$, and $ \calA \sststile{}{x} \{ \sum_{i = 1}^m x_i^2 \leq 1\}$.   Let $\{g \geq 0 \}$ be a polynomial constraint.
  If every degree-$d$ pseudo-distribution that satisfies $D \sdtstile{r}{}  \calA$ also satisfies $D \sdtstile{r'}{} \{g \geq 0 \}$, then for every $\eta > 0$, there is a sum-of-squares proof $ \calA \sststile{d}{} \{g \geq - \eta \}$.
\end{fact}

\paragraph{Basic sum-of-squares proofs.}

Now, we recall some basic facts about sum-of-squares proofs.
First, any univariate polynomial inequality admits a sum-of-squares proof over the reals.
\begin{fact}[Univariate polynomial inequalities admit SoS proofs~\cite{laurent2009sums}]
\label{fact:univariate-sos-proofs}
Let $p$ be a polynomial of degree $d$.
If $p(x)\geq 0$ for all $x \geq 0$, we have $\sststile{d}{x} \Set{ p(x) \geq 0 } $.
If $p(x)\geq 0$ for all $x \in [a, b]$, then $\Set{ x\geq a, x\leq b } \sststile{d}{x} \Set{ p(x) \geq 0 }$.
\end{fact}

Second, if $p \geq 0$ and $p$ is a quadratic, then this admits a sum-of-squares proof.
\begin{fact}[Quadratic polynomial inequalities admit SoS proofs]
    \label{fact:nonnegative-quadratic}
Let $p$ be a polynomial in the indeterminates $x \in \R^m$ such that $p$ has degree $2$ and $p \geq 0$ for all $x \in \mathbb{R}^m$. Then $\sststile{2}{x} \Set{ p(x) \geq 0  }$. 
\end{fact}
\begin{proof}
Let $M$ be the unique $(m + 1) \times (m+1 )$ Hermitian matrix such that, for $v(x) = (1,x_1, \dots , x_m)^\dagger$,
\[
p(x_1, \dots , x_m) = v(x)^\dagger M v(x) \,.
\]
The inequality $p \geq 0$ implies that $M$ is PSD: consider a vector $v  = (v_1, \dots , v_{m+1}) \in \R^{m+1}$.
If $v_1 \neq 0$, then $v^\dagger M v = p(w) \geq 0$ for $w_1 = v_2/v_1, \dots , w_m = v_{m+1}/v_1$.
If $v_1 = 0$, then $v^\dagger M v = \lim_{c \to \infty} p(c \cdot w) \geq 0$ for $w_1 = v_2, \dots , w_m = v_{m+1}$.
This shows that $M$ must be PSD, so we can write $M = \sum_{i = 1}^{m+1} u_iu_i^\dagger$ for some vectors $u_i \in \R^{m+1}$.  Thus,
\[
p(x_1, \dots , x_m) = v(x)^\dagger M v(x) = \sum_{i = 1}^{m+1} \langle u_i, v(x)\rangle^2
\]
which is a degree-$2$ SoS polynomial and we are done.
\end{proof}

We also use the following basic sum-of-squares proofs.
For further details, we refer the reader to a recent monograph~\cite{fleming2019semialgebraic}.

\begin{fact}[Operator norm bound]
\label{fact:operator_norm}
For a symmetric matrix $A \in \R^{d \times d}$ and a vector $v \in \mathbb{R}^d$,
\[
\sststile{2}{v} \bracks[\Big]{ v^\dagger A v \leq \norm{A}\norm{v}^2 }.
\]
\end{fact}

\begin{fact}[SoS Cross Terms Identity]
\label{eqn:cross_term_id}
For all $i \in [n]$, let $a_i, b_i $, be non-negative indeterminates, then 
\[
\Set{ \forall i \in [n],  a_i \geq 0 , b_i \geq 0 } \sststile{2}{} \Set{\sum_{i \in [n]} a_i b_i \leq  \Paren{\sum_{i \in [n]} a_i} \Paren{\sum_{i \in [n]} b_i} }  
\]
\end{fact}

\begin{fact}[Almost triangle inequality] \label{fact:sos-almost-triangle}
Let $f_1, f_2, \ldots, f_r$ be indeterminates. Then
\[
\sststile{2t}{f_1, f_2,\dots,f_r} \Set{ \parens[\Big]{\sum_{i\leq r} f_i}^{2t} \leq r^{2t-1} \parens[\Big]{\sum_{i =1}^r f_i^{2t}}}.
\]
\end{fact}

\begin{fact}[SoS Cauchy-Schwarz]
\label{fact:sos_cs}
It holds that
\[
\sststile{2}{x_1,\ldots,x_n,y_1,\ldots,y_n} \Set{\paren{\sum_{i=1}^n x_i y_i}^2 \leq \Paren{\sum_{i=1}^n x_i^2}\Paren{\sum_{i=1}^n y_i^2}} \,.
\]
\end{fact}
\begin{proof}
    This holds since
    \[
    \Paren{\sum_{i=1}^n x_i^2}\Paren{\sum_{i=1}^n y_i^2} - \paren{\sum_{i=1}^n x_i y_i}^2 = \sum_{i,j=1}^n x_i^2 y_j^2 - x_i y_i x_j y_j = \sum_{1 = i < j = n} \paren{x_i y_j - x_j y_i}^2 \,.
    \]
\end{proof}

\begin{fact}[SoS Hölder's inequality]\label{fact:sos-holder}
Let $w_1, \ldots w_n$ be indeterminates and let $f_1,\ldots f_n$ be polynomials of degree $d$ in the variables $x \in \R^m$. 
Let $k$ be a power of 2.  
Then
\[
\Set{w_i^2 = w_i, \forall i\in[n] } \sststile{2kd}{x,w} \Set{  \parens[\Big]{\frac{1}{n} \sum_{i = 1}^n w_i f_i}^{k} \leq \parens[\Big]{\frac{1}{n} \sum_{i = 1}^n w_i}^{k-1} \parens[\Big]{\frac{1}{n} \sum_{i = 1}^n f_i^k}}. 
\]
\end{fact}

\begin{fact}
    \label{lem:square_root_sos}
    Let $a$ be an indeterminate and $C > 0$ be a constant.
    Then
    \[
        \Set{a^2 \leq C^2} \sststile{2}{a} \Set{a \leq C}
    \]
\end{fact}
\begin{proof}
    \[
       \Set{a^2 \leq C^2} \sststile{2}{a} \Set{a = \Paren{\frac a {\sqrt{C}}} \sqrt{C} \leq \frac 1 2 \Paren{\frac {a^2} C + C} \leq C} \,.
    \]
\end{proof}
\begin{fact}
    \label{fact:cancellation_two_sos}
    Let $a$ be an indeterminate and $C > 0$ be a constant.
    Then
    \[
        \Set{a^4 \leq C a^2} \sststile{4}{a} \Set{a^2 \leq C} \,.
    \]
\end{fact}
\begin{proof}
    \[
       \Set{a^4 \leq C a^2} \sststile{4}{a}  \Set{a^2 =  \Paren{ \frac {a^2} {\sqrt{C}} } \sqrt{C} \leq \frac 1 2 \Paren{\frac {a^4} C + C} \leq \frac {a^2} 2 + \frac{C}{2}  } \,.
    \]
    The inequality follows by rearranging.
\end{proof}

\begin{fact}
\label{fact:cov-hack-tmp}
    Let $a, b$ be indeterminates. Then we have that
    \[ \sststile{}{a, b} \Set{a^3 b \leq a^4 + b^4}\,.\]
\end{fact}
\begin{proof}
    By applying the SoS inequality that $ab \leq \frac{1}{2}a^2 + \frac{1}{2} b^2$ twice we have that
    \[ \sststile{}{a, b} \Set{a^3 b \leq \frac{1}{2}a^4 + \frac{1}{2}a^2b^2 \leq \frac{3}{4}a^4 + \frac{1}{4}b^4}\,.\]
\end{proof}

\begin{fact}
\label{fact:sos-squaring}
    Let $a,b$ be indeterminates. Then we have that 
    \[ \Set{a \leq b, a\geq 0} \sststile{}{a,b} \Set{a^2 \leq b^2}\,.\]
\end{fact}
\begin{proof}
    Note that 
    \[ \Set{a \leq b, a\geq 0} \sststile{}{a,b} \Set{b-a \geq 0, b+a \geq 0}\,.\]
    The conclusion then follows direction from the fact that if there is an SoS proof that $f \geq 0, g \geq 0$ then there is also an SoS proof that $fg \geq 0$.
\end{proof}

\begin{lemma}[SoS Selector Lemma]
\torestate{
\label{lemma:subset_selection}
Let $B, a_1, \ldots, a_n \in \mathbb{R}$, and $k \in \mathbb{N}$ such that for all $T \subseteq [n]$ of size $k$ we have $\sum_{i \in T} a_i \leq B$, then also the following holds: 
\begin{equation*}
    \Set{\forall i \in [n] \colon 0 \leq z_i \leq 1, \sum_{i=1}^n z_i = k} \sststile{2}{z_1, \ldots, z_n} \Set{\sum_{i=1}^n z_i a_i \leq B } \,.
\end{equation*}
Further, if the $a_i$ are non-negative, then the same conclusion holds also if the constraints only include $\sum_{i=1}^n z_i \leq k$.
}
\end{lemma}
\begin{proof}
We focuse on the case when our constraints include $\sum_{i=1}^n z_i = k$.
The case for non-negative $a_i$ will follow.
Without loss of generality, assume that $a_1 \geq \ldots \geq a_n$.
Define $l = k - \sum_{i \leq k} z_i$.
First, note that our constraints imply at degree 2 that $\sum_{i \leq k} z_i a_i \leq \sum_{i \leq k} a_i - l \cdot a_k$ since
\begin{align*}
    \sum_{i \leq k} a_i - l \cdot a_k - \sum_{i \leq k} z_i a_i = \sum_{i \leq k} (a_i - a_k) - \sum_{i\leq k} z_i \cdot (a_i - a_k) = \sum_{i \leq k} (1-z_i) \cdot (a_i - a_k) \geq 0 \,,
\end{align*}
where we used that $z_i \leq 1$.
Second, the constraints also yield, again at degree 2, that $\sum_{i > k} z_i a_i \leq a_{k+1} l$ since using $a_i \leq a_{k+1}$ for $i \geq k + 1$ implies
\[
    a_{k+1} l - \sum_{i > k} z_i a_i \geq a_{k+1} \cdot \Paren{k - \sum_{i \leq k} z_i - \sum_{i > k} z_i} = 0 \,,
\]
where we also used that $z_i \geq 0$.
Note that when our constraints only include $\sum_{i=1}^n z_i \leq k$ but the $a_i$ are non-negative, the last term is still at least 0.
This is the only place where the proof needs to be adapted.

Together, this implies that
\[
  \Set{\forall i \in [n] \colon 0 \leq z_i \leq 1, \sum_{i=1}^n z_i \leq k} \sststile{2}{z_1, \ldots, z_n}  \Set{\sum_{i=1}^n z_i a_i \leq \sum_{i \leq k} a_i - l \cdot (a_k - a_{k+1}) \leq B} \,.
\]
\end{proof}

\begin{fact}[Cauchy-Schwarz for Pseudo-Expectations]
\label{fact:pe_cs}
    Let $\pE$ be a degree-$d$ pseudo-expectations and $p,q$ be two polynomials of degree at most $d/2$.
    Then it holds that $\paren{\pE p q}^2 \leq \pE p^2 \pE q^2$.
    In particular, setting $q \equiv 1$ it holds that $\paren{\pE p}^2 \leq \pE p^2$.
\end{fact}
\begin{proof}
Let $M$ be the moment matrix associated with $\pE$ and $\Vec{p},\Vec{q}$ the vectors containing the coefficients of $p$ and $q$, respectively.
Since $M$ is positive semi-definite, there exists a symmetric matrix $M^{1/2}$ such that $M^{1/2} M^{1/2} = M$.
It follows by the standard Cauchy-Schwarz Inequality that
\[
\paren{\pE p q}^2  = \Iprod{\Vec{p}, M\Vec{q}}^2 =\Iprod{M^{1/2}\Vec{p},M^{1/2}\Vec{q}}^2 \leq \Norm{M^{1/2}\Vec{p}}_2^2\Norm{M^{1/2}\Vec{q}}_2^2 = \pE p^2 \pE q^2 \,.
\]
\end{proof}

We will also need the following definition for the reduction between robustness and privacy.
\begin{definition}[Approximately Satisfying Linear Operator]
    \label{def:tau_relaxed_system}
    Let $n \in \N$ and $T \in [n]$.
    Let $\calA$ be a system of polynomial inequalities in variables $w = (w_1, \ldots, w_n)$ and $z$ (potentially many) such that
    \[
        \calA_T = \Set{q_1(w,z) \geq 0, \ldots, q_m(w,z) \geq 0} \cup \Set{\sum_{i=1}^n w_i \geq n - T} \,.
    \]
    Let $\tau > 0$.
    We say that a linear operator $\calL$ $\tau$-approximately satisfies $\calA_T$ at degree $D$ if the following hold
    \begin{enumerate}
        \item $\calL 1 = 1$.
        \item For all polynomials $p$ such that $\mathrm{deg}(p^2) \leq D$ and $\norm{\calR(p)}_2 \leq 1$ (where $\calR(p)$ is the vector representation of the coefficients of $p$) it holds that $\calL p^2 \geq - \tau T$.
        \item For all $i  = 1, \ldots, m$ and polynomials $p$ such that $\mathrm{deg}(p^2 \cdot q_i) \leq D$ and $\norm{\calR(p)}_2 \leq 1$ (where $\calR(p)$ is the vector representation of the coefficients of $p$) it holds that $\calL p^2 q_i \geq - \tau T$.
        \item  For every polynomial $p$ such that $\mathrm{deg}(p^2 \cdot (\sum_{i=1}^n w_i -n + T))\leq D$ and $\norm{\calR(p)}_2 \leq 1$ (where $\calR(p)$ is the vector representation of the coefficients of $p$) it holds that $\calL (\sum_{i=1}^n) w_i -n + T)p^2 \geq - 5 \tau T n$.
    \end{enumerate}
\end{definition}

\subsection{Robustness to Privacy}

We apply the following theorem in order to transform robust algorithms into private ones.

\begin{theorem}[\cite{Hopkins2023Robustness}]
\label{thm:pure_dp_reduction}
    Let $0 < \eta, r < 1 < R$ be fixed parameters.
    Suppose we have a score function $\cS(\theta, \mathcal{Y}) \in [0, n]$ that takes as input a dataset $\mathcal{Y} = \{y_1, \dots, y_n\}$ and a parameter $\theta \in \Theta \subset \mathbb{B}(R)^d$ (where $\Theta$ is convex and contained in a ball of radius $R$), with the following properties:
\begin{itemize}
    \item (Bounded Sensitivity) For any two adjacent datasets $\mathcal{Y}, \mathcal{Y}'$ and any $\theta \in \Theta$, $|\cS(\theta, \mathcal{Y})-\cS(\theta, \mathcal{Y}')| \le 1.$
    \item (Quasi-Convexity) For any fixed dataset $\mathcal{Y}$, any $\theta, \theta' \in \Theta$, and any $0 \le \lambda \le 1$, we have that $\cS(\lambda \theta + (1-\lambda) \theta', \mathcal{Y}) \le \max(\cS(\theta, \mathcal{Y}), \cS(\theta', \mathcal{Y}))$.
    \item (Efficiently Computable) For any given $\theta \in \Theta$ and dataset $\mathcal{Y}$, we can compute $\cS(\theta, \mathcal{Y})$ up to error $\gamma$ in $\poly(n, d, \log \frac{R}{r}, \log \gamma^{-1})$ time for any $\gamma > 0$.
    \item (Robust algorithm finds low-scoring point) For a given dataset $\mathcal{Y}$, let $T = \min_{\theta_0 \in \Theta} \cS(\theta_0, \mathcal{Y})$. Then, we can find some point $\theta$ such that for all $\theta'$ within distance $r$ of $\theta$, $\cS(\theta', \mathcal{Y}) \le T+1$, in time $\poly(n, d, \log \frac{R}{r})$.
    \item (Volume) For any given dataset $\mathcal{Y}$ and $\eta' \ge \eta$, let $V_{\eta'}(\mathcal{Y})$ represent the $d$-dimensional volume of points $\theta \in \Theta \subset \R^d$ with score at most $\eta' n$. (Note that $V_1(\mathcal{Y})$ is the full volume of $\Theta$).
\end{itemize}
    Then, 
    we have a pure $\eps$-DP algorithm $\mathcal{A}$ on datasets of size $n$, that runs in $\poly(n, d, \log \frac{R}{r})$ time, with the following property. For any dataset $\mathcal{Y}$, if there exists $\theta$ with $\cS(\theta, \mathcal{Y}) \le \eta n$ and if 
    \[n \ge \Omega\left(\max\limits_{\eta': \eta \le \eta' \le 1}\frac{\log(V_{\eta'}(\mathcal{Y})/V_{\eta}(\mathcal{Y})) + \log (1/(\beta \cdot \eta))}{\eps \cdot \eta'}\right)\,,\]
    then $\mathcal{A}(\mathcal{Y})$ outputs some $\theta \in \Theta$ of score at most $2\eta n$ with probability $1-\beta$.
\end{theorem}

We will also need the following theorem for transforming robust algorithms into algorithms that satisfy approximate differential privacy.

\begin{theorem}
\label{thm:approx-dp-reduction}
    Let $0 < \eta < 0.1$ and $r< 1 <R$ be fixed parameters. Suppose we have a score function $\calS(\theta, \calY) \in [0, \infty)$ that takes as input a dataset $\calY = \{y_1, \ldots , y_n\}$ and a parameter $\theta \in \Theta \subset \mathbb{R}^d$ (where $\Theta$ is convex and contained in a ball of radius $R$), with the same properties as in~\cref{thm:pure_dp_reduction}. In addition, fix some parameter $\eta^* \in [10\eta, 1]$. Suppose that $n \geq \Omega{\left(\frac{\log (1/\delta) + \log(V_{\eta^*}(\calY)/V_{0.8\eta^*}(\calY))}{\eps \cdot \eta^*}\right)}$ for all $\calY$ suc that there exists $\theta$ with $\calS(\theta, \calY) \leq 0.7\eta^*n$. Then, we have that there exists an $(\eps, \delta)$-DP algorithm $\calA$, that runs in $\poly(n, d, \log (\tfrac R r))$ time, such that for any dataset $\calY$, if there exists $\theta$ with $\calS(\theta, \calY) \leq \eta n$ and if $n \geq \Omega \left(\max_{\eta':\eta \leq \eta'\leq \eta^*} \frac{\log (V_{\eta'}(\calY)/V_{\eta}(\calY)) + \log (1/(\beta \cdot \eta))}{\eps \cdot \eta}\right)$
\end{theorem}

\section{Private and Robust Covariance-Aware Regression}
\label{sec:robust_regression_oneshot}

In this section, we will present a oneshot algorithm, i.e., that outputs an estimator for the regression vector without computing an estiamte for the unknown covariance first, for robust regression with (nearly) optimal error (in terms of the corruptions level).
In~\cref{sec:private_regression} we will use this algorithm and~\cref{thm:pure_dp_reduction} to deduce our main theorem for private regression.
The main private regression theorem we obtain is as follows: 
\begin{theorem}[Sample-Optimal Private Regression, Full Version of~\cref{thm:pure-dp-regression-informal}]
\torestate{
\label{thm:main_private_regression}
Let $\theta \in \R^d$ such that $\norm{\theta} \leq R$ and $\Sigma$ such that $\Sigma \preceq \covscale I_d$.
Let $0 <\eta$ be less than a sufficiently small constant and let $0 < \alpha, \beta, \eps$ and $\alpha < 1$.
Let $\cX = \Set{(x_1, y_1), \ldots, (x_n,y_n)}$ be $n \geq n_0$ samples following~\cref{model:robust_regression} (robust regression model) with corruption level $\eta$ and optimal hyperplane $\theta$.
There exists an  $(\eps, 0)$-differentially private algorithm that, given $\eta, \alpha, \eps$, and $\cX$, runs in time $\poly(n,\log L, \log R)$ and with probability at least $1-\beta$ outputs an estimate $\thetahat$ satisfying
    \[ \Norm{\Sigma^{1/2}\left(\thetahat - \theta\right)} \leq \bigO{ \alpha } \,,
    \]
    whenever 
    $$n_0 = \tilde{\Omega}\Paren{ \frac{d^2 + \log^2(1/\beta)}{\alpha^2} + \frac{d + \log(1/\beta)}{\alpha \eps} + \frac{d \log(R) + d \log(\covscale)}{\eps} }, $$
and $\alpha \geq \Omega( \eta \log(1/\eta)) $.
}
\end{theorem}

As discussed earlier, all terms in the sample-complexity are (nearly) optimal up to under either SQ or information-theoretic lower bounds, except for the $\log^2(1/\beta)$ term which could potentially be improved to $\log(1/\beta)$.
Note that our error guarantees imply small generalization error as stated in~\cref{thm:pure-dp-regression-informal}.
Indeed, this follows since
\begin{align*}
    \E \Brac{(y- \Iprod{\hat{\theta}, x})^2} &= \E \Brac{(y- \Iprod{\hat{\theta} - \theta + \theta, x})^2} \\
    &= \E\Brac{(y - \Iprod{\theta, x})^2} + \E\Brac{\Iprod{\hat{\theta} - \theta, x}^2} + 2\E\Brac{(y - \Iprod{\theta, x})\Iprod{\hat{\theta} - \theta, x}} \\
    &= \E\Brac{(y - \Iprod{\theta, x})^2} + \Norm{\Sigma^{1/2}(\hat{\theta} - \theta)}^2 \,.
\end{align*}

Along the way, we show  the following theorem about robust regression
\begin{theorem}[Sample-Optimal Robust Regression]    \label{thm:main_robust_regression}
    Let $\theta \in \R^d$ and $0 <\eta$ be less than a sufficiently small constant.
    Let $0 < \alpha, \beta$.
Let $\cX = \Set{(x_1, y_1), \ldots, (x_n,y_n)}$ be $n \geq n_0$ samples following~\cref{model:robust_regression} (robust regression model) with corruption level $\eta$ and optimal hyperplane $\theta$.
There exists an algorithm that, given $\eta, \alpha, \eps$, and $\cX$, runs in time $n^{\mathcal{O}(1)}$ and with probability at least $1-\beta$ outputs an estimate $\thetahat$ satisfying
    \[
        \Norm{\Sigma^{1/2}\left(\thetahat - \theta\right)} \leq \bigO{\eta \log(1/\eta)} \,,
    \]
    whenever $n_0 = \tilde{\Omega}\Paren{ \frac{d^2 + \log^2(1/\beta)}{\eta^2}}$.
\end{theorem}

We will prove~\cref{thm:main_robust_regression} in this section and~\cref{thm:main_private_regression} in the next.

\paragraph{Intuition for robust estimator.}
Note that in the non-robust setting, when we are given $n$ \iid samples $(x_1^*, y_1^*), \ldots, (x_n^*,y_n^*)$ from the above model, the least-squares estimator defined as
\[
    \thetals = \argmin_{\tilde{\theta} \in \R^d} \tfrac 1 n \sum_{i=1}^n \paren{y_i^* - \Iprod{x_i^*, \tilde{\theta}}}^2
\]
satisfies $\norm{\Sigma^{1/2}(\thetals - \theta)} \leq \bigO{\sqrt{\tfrac {d + \log(1/\beta)}{n}}}$
Note that this error is at most $O(\eta \log(1/\eta))$ for our choice of $n$.
On a very high level, our algorithm can be interpreted as searching for a $(1-\eta)$-fraction of the input samples satisfying concentration bounds that would be satisfied by the Gaussian distribution and outputting the associated least-squares estimator.

\paragraph{Sum-of-Squares algorithm.}

Our algorithm is based on the sum-of-squares hierarchy.
We will use the following notation:
We denote by $\{x_1, x_2 , \ldots x_n\}$ the $\eta$-corrupted version of $\cS = \{x_1^*, \ldots, x_n^*, \}$ we observe.
Further, let $r_i \in \{0,1\}$, be 1 if the $i$-th sample is uncorrupted, i.e., $(x_i^*,y_i^*) = (x_i,y_i)$, and 0 otherwise.
We will use the following constraint system in vector-valued variables $\thetaprime, x_1', \ldots, x_n'$, scalar-valued variables $y_1', \ldots, y_n'$ and $w_1, \ldots, x_n$, and the matrix-valued variable $R$.
Note that the last constraint is a one-dimensional version of the second-last one about Gaussian concentration.
The third-last constraint encodes that $\thetaprime$ solves the least-squares optimization problem on pairs $(x_1',y_1'), \ldots, (x_n',y_n')$.
\begin{equation*}
\calA_{\eta}=
  \left \{
    \begin{aligned}
      &\forall i\in [n].
      & w_i^2
      & = w_i \\
      & &\sum_{i\in [n]} w_i &\geq (1-\eta)n \\
      &&\tfrac{1} n \sum_{i \in [n]} x_i'(x_i')^\top &= \Sigma' \\
      &\forall i\in [n] & w_i(x'_i - x_i) = 0&\,,\; w_i(y'_i - y_i) =0 \\
      &&\tfrac{1} n \sum_{i \in [n]} \Paren{\Iprod{x_i', \thetaprime} - y_i'}x_i' &= 0 \\
      &&\tfrac{1} n \sum_{i \in [n]} \Paren{y_i' - \Iprod{\thetaprime, x_i'}}^2 &\leq 1+O(\eta) \\
      &\forall v \in \mathbb{R}^d
      & \frac{1}{n}\sum_{i \in [n]}  {\langle x'_i , v\rangle^{4}} 
      &\leq  \Paren{3+\eta \log^2 (1/\eta)} \left(  v^\top \Sigma' v \right)^{2} \\
      & & \tfrac{1} n \sum_{i \in [n]} \Paren{y_i' - \Iprod{\thetaprime, x_i'}}^4 &\leq \bigO{1}
    \end{aligned}
  \right \}
\end{equation*}

Note that for the inequalities that must be satisfied for all $v$, we require that there is an SoS proof of the inequality for any $v$. These constraints can be efficiently encoded as polynomial inequalities in our program variables by using auxiliary variables (see~\cite{fleming2019semialgebraic}). 

\paragraph{Feasibility.} The feasibility of the above program follows by setting $w_i = r_i, x_i' = x_i^*$ and $\thetaprime$ to be the least-squares solution for pairs $(x_1^*,y_1^*), \ldots, (x_n^*,y_n^*)$.
We defer a formal proof of feasibility to \cref{sec:feasibility_regression}.
 Our algorithm will be the following:

\begin{mdframed}
  \begin{algorithm}[Optimal Robust Regression]
    \label{algo:robust_regression}\mbox{}
    \begin{description}
    \item[Input:] $\eta$-corrupted sample $(x_1,y_1), \ldots, (x_n,y_n)$, corruption level $\eta$.
    \item[Output:] An estimate $\thetahat$ attaining the guarantees of \cref{thm:main_robust_regression}.
    
    \item[Operations:]\mbox{}
    \begin{enumerate}
    \item Find a degree-$\bigO{1}$ pseudo-distribution $\zeta$ satisfying $\calA_{\eta}$.
    \item Output $ \thetahat = \pE_{\zeta}[\thetaprime]$.
    \end{enumerate}
    \end{description}
  \end{algorithm}
\end{mdframed}

Our main technical lemma will be the following which we will prove at the end of this section.
Throughout the rest of this section, we use $\tilde{\Sigma}$ to denote $\tfrac 1 n \sum_{i=1}^n x_i^* (x_i^*)^\top$.
\begin{lemma}
\label{lem:main_sos_proof_regression}
Let $0 < \eta $ be smaller than a sufficiently small constant and let $\thetals$ be the least-squares solution associated with $(x_1^*,y_1^*), \ldots, (x_n^*, y_n^*)$.
Suppose $x_1^*, \ldots, x_n^*$ and $y_1^*, \ldots, y_n^*$ are $\eta$-good and $\eta$-higher-order-good (cf.~\cref{def:eps_goodness,def:higher_order_goodness})\footnote{In the case of $y_i^*$ we mean by $\eta$-good the case of $d=1$ in \cref{def:eps_goodness}.}.
For any fixed vector $u$, i.e., that is not an indeterminate, we have that
\begin{align*}
    \calA_\eta \sststile{\bigO{1}}{ w, \theta'} \Biggl\{ \Iprod{u, \tilde{\Sigma} \Paren{\thetaprime - \thetals} }^2 \leq \bigO{\eta^2\log^2(1/\eta)}  \Paren{u^\top \Sigma u} \Biggr\}
\end{align*}
and that
\begin{align*}
    \calA_\eta \sststile{\bigO{1}}{ w, \theta'} \Biggl\{ \Iprod{u, \thetaprime - \thetals}^2 \leq \bigO{\eta^2\log^2(1/\eta)}  \Paren{u^\top \Sigma^{-1} u} \Biggr\} \,.
\end{align*}
\end{lemma}
Note that the second claim is almost equivalent to the first via the transformation $u \mapsto \Sigma^{-1} u$.
The second formulation will be more useful for the reduction to privacy.

\paragraph{Handling unknown variance of the noise between $1-\eta$ and 1.}
We state all of our results in the setting when the variance of the noise is known and assumed to be 1.
Yet, our algorithms also work in the (very slightly harder) setting when the variance of the noise is unknown but promised to be in $[1-\eta,1]$.
The proofs of feasibility are not affected since the constraints on the $y_i' - \Iprod{x_i',\theta'}$ are only easier to satisfy.
We will point out in the SoS proofs below why the proof of~\cref{lem:main_sos_proof_regression} still holds.

We can deduce \cref{thm:main_robust_regression} from~\cref{lem:main_sos_proof_regression} as follows.
\begin{proof}[Proof of \cref{thm:main_robust_regression}]
    First, note that with probability at least $1-\beta$ the least-squares estimator with respect to $(x_1^*,y_1^*), \ldots, (x_n^*, y_n^*)$ satisfy $\eta$-goodness and $\eta$-higher-order-goodness~\cref{fact:mean_and_cov_of_1-eps,fact:cov_small_subset,lem:higher-order-stability}.
    We condition on this event.
    Recall that for our choice of $n$
    \[
        \Norm{\Sigma^{1/2}(\thetals - \theta)} \leq \bigO{\sqrt{\frac {d + \log(1/\beta)} n}} \leq \eta \,.
    \]
    Thus, by triangle inequality, it is enough to bound $\Norm{\Sigma^{1/2}(\thetahat - \thetals)}$.
    Further, let $\tilde{\Sigma} = \tfrac{1}{n} \sum_{i=1}^n x_i^* (x_i^*)^\top$.
    We know that $\Sigma$ and $\tilde{\Sigma}$ are $O(\eta\sqrt{\log(1/\eta)})$-close in every direction.
    Thus, it is enough to show that $\Norm{\tilde{\Sigma}^{1/2}(\thetahat - \thetals)}$
    
    Let $\zeta$ be a degree-$\bigO{1}$ pseudo-expectation satisfying $\calA_{\eta}$.
    By~\cref{fact:eff-pseudo-distribution}, this can be computed in time $(n \cdot d)^{\mathcal{O}(1)}$.
    Let $\thetahat = \pE_\zeta \thetaprime$ and $u = \thetahat - \thetals$.
    By~\cref{lem:main_sos_proof_regression} and Cauchy-Schwarz for pseudo-expectations (cf.~\cref{fact:pe_cs}) it holds that
    \begin{align*}
        \Norm{\tilde{\Sigma}^{1/2}(\thetahat - \thetals)}^4 &= \Brac{\pE_\zeta \Iprod{u, \tilde{\Sigma} \Paren{\thetaprime - \thetals}}}^2 \leq \pE_\zeta \Iprod{u, \Paren{\tfrac 1 n \sum_{i=1}^n x_i^* (x_i^*)^\top} \Paren{\thetaprime - \thetals}}^2 \\
        &\leq \bigO{\eta^2 \log^2(1/\eta)} \cdot u^\top \Sigma u  \leq \bigO{\eta^2 \log^2(1/\eta)} \Norm{\tilde{\Sigma}^{1/2}(\thetahat - \thetals)}^2\,,
    \end{align*}
    Canceling the factors of  $\Norm{\tilde{\Sigma}^{1/2}(\thetahat - \thetals)}^2$ and taking square roots completes the proof.

\end{proof}

\paragraph{Obtain the main SoS guarantee from intermediate lemmas.}

The first part of \cref{lem:main_sos_proof_regression} will directly follow by combining the following two lemmas.
\begin{lemma}
\label{lem:sos_regression_without_bootstrapping}
Let $0 < \eta $ be smaller than a sufficiently small constant and let $\thetals$ be the least-squares solution associated with $(x_1^*,y_1^*), \ldots, (x_n^*, y_n^*)$.
Suppose $x_1^*, \ldots, x_n^*$ and $y_1^*, \ldots, y_n^*$ are $\eta$-good and $\eta$-higher-order-good.
For any fixed vector $u$, i.e., that is not an indeterminate, we have that
\begin{align*}
    \calA_\eta \sststile{\bigO{1}}{ w, \theta'} \Biggl\{ &\Iprod{u, \tilde{\Sigma} \Paren{\thetaprime - \thetals} }^2 \\
    &\leq \bigO{\eta^2\log^2(1/\eta)}  \Paren{u^\top \Sigma u} + \bigO{\eta \log(1/\eta)}\Norm{\tilde{\Sigma}^{1/2}\Paren{\thetaprime - \thetals}}^2  \Paren{u^\top \Sigma u}\Biggr\}  \,.
\end{align*}
\end{lemma}
\begin{lemma}
\label{lem:sos_regression_bootstrapping}
Let $0 < \eta $ be smaller than a sufficiently small constant and let $\thetals$ be the least-squares solution associated with $(x_1^*,y_1^*), \ldots, (x_n^*, y_n^*)$.
Suppose $x_1^*, \ldots, x_n^*$ and $y_1^*, \ldots, y_n^*$ are $\eta$-good and $\eta$-higher-order-good.
It holds that 
For any fixed vector $u$, i.e., that is not an indeterminate, we have that
\begin{align*}
    \calA_\eta \sststile{\bigO{1}}{ w, \theta'} \Biggl\{ \Norm{\tilde{\Sigma}^{1/2}\Paren{\thetaprime - \thetals}}^2 \leq \bigO{\eta} \Biggr\}  \,.
\end{align*}
\end{lemma}

We next show how to derive the second part of~\cref{lem:main_sos_proof_regression} from the first.
Via the transformation $u \mapsto \Sigma^{-1} u$, we obtain that
\begin{align*}
    \calA_\eta \sststile{\bigO{1}}{ w, \theta'} \Biggl\{ \Iprod{u, \thetaprime - \thetals}^2 \leq \bigO{\eta^2\log^2(1/\eta)}  \Paren{u^\top \Sigma^{-1} \tilde{\Sigma} \Sigma^{-1} u} \Biggr\} \,.
\end{align*}
By $\eta$-goodness, it follows that $\tilde{\Sigma} \preceq \bigO{1} \Sigma$ which implies the claim.

\paragraph{Proof of coarse SoS guarantee.}

We start by proving~\cref{lem:sos_regression_without_bootstrapping}.
This proof does not use the last constraint (bound on fourth moments of the noise), we will need it for the proof of~\cref{lem:bootstrapping_regression} however.
\begin{proof}[Proof of \cref{lem:sos_regression_without_bootstrapping}]
    First, note that as a minimizer of the squared loss, the least-squares estimator satisfies
    \[
        \tfrac 1 n \sum_{i=1}^n x_i^* \Paren{\Iprod{x_i^*, \thetals} - y_i^*} = 0 \,.
    \]
    For $i \in [n]$, let $r_i = 1$ if the $i$-th input point was not corrupted and 0 if it was.
    Note that our constraints imply that $r_i w_i x_i' = r_i w_i x_i^*$ and similarly $r_i w_i y_i' = r_i w_i y_i^*$.
    Hence, using the above and our constraint that $\tfrac{1} n \sum_{i \in [n]} x_i' \cdot \Paren{\Iprod{x_i', \thetaprime} - y_i'} = 0$ it follows that
    \begin{align*}
        \calA_\eta \sststile{}{} \Biggl\{\tilde{\Sigma}(\theta' - \thetals) &= \tfrac 1 n \sum_{i=1}^n x_i^* \cdot \Iprod{x_i^*, \thetaprime - \thetals} = \tfrac 1 n \sum_{i=1}^n x_i^* \cdot \Paren{\Iprod{x_i^*, \thetaprime } - y_i^*} + \tfrac 1 n \sum_{i=1}^n x_i^* \cdot \Paren{y_i^* - \Iprod{x_i^*, \thetals}} \\
        &= \tfrac 1 n \sum_{i=1}^n x_i^* \cdot \Paren{\Iprod{x_i^*, \thetaprime } - y_i^*} - \tfrac{1} n \sum_{i \in [n]} x_i' \cdot \Paren{\Iprod{x_i', \thetaprime} - y_i'} \\
        &= \tfrac 1 n \sum_{i=1}^n (1- r_i w_i) x_i^* \cdot \Paren{\Iprod{x_i^*, \thetaprime } - y_i^*} - \tfrac{1} n \sum_{i \in [n]} (1- r_i w_i) x_i' \cdot \Paren{\Iprod{x_i', \thetaprime} - y_i'} \Biggr\}\,.
    \end{align*}
    Let $u$ be an arbitrary unit vector.
    Substituting the equality above in the expression below and using the SoS triangle inequality (cf. \cref{fact:sos-almost-triangle}) we obtain
    \begin{align}
        \nonumber \calA_\eta \sststile{}{} \Biggl\{ &\Iprod{u, \tilde{\Sigma} \Paren{\thetaprime - \thetals} }^2  \\
        & \nonumber = \Paren{\Iprod{u, \tfrac 1 n \sum_{i=1}^n (1- r_i w_i) x_i^* \cdot \Paren{\Iprod{x_i^*, \thetaprime } - y_i^*}} - \Iprod{u, \tfrac{1} n \sum_{i \in [n]} (1- r_i w_i) x_i' \cdot \Paren{\Iprod{x_i', \thetaprime} - y_i'}}}^2 \\
        &\leq 2\underbrace{\Paren{\tfrac 1 n \sum_{i=1}^n (1- r_i w_i) \Iprod{u, x_i^*} \cdot \Paren{\Iprod{x_i^*, \thetaprime } - y_i^*}}^2}_{\text{Term A}} + 2\underbrace{\Paren{ \tfrac 1 n \sum_{i=1}^n (1- r_i w_i) \Iprod{u, x_i'} \cdot \Paren{\Iprod{x_i', \thetaprime } - y_i'} }^2}_{\text{Term B}} \label{eq:sos_robust_regression} \Biggr\}
    \end{align}

    \paragraph*{Bounding Term A in \cref{eq:sos_robust_regression}}
    We start with Term A in \cref{eq:sos_robust_regression}.
    By the SoS Cauchy-Schwarz Inequality (cf. \cref{fact:sos_cs}) it follows that
    \begin{equation}
    \begin{split} 
        \calA \sststile{}{} \Biggl\{ & \Paren{\tfrac 1 n \sum_{i=1}^n (1- r_i w_i) \Iprod{u, x_i^*} \cdot \Paren{\Iprod{x_i^*, \thetaprime } - y_i^*}}^2 \\
        & \leq \underbrace{\Paren{\tfrac 1 n \sum_{i=1}^n (1- r_i w_i) \Iprod{u, x_i^*}^2}}_{\text{Term A.0}} \underbrace{\Paren{\tfrac 1 n \sum_{i=1}^n (1- r_i w_i) \Paren{\Iprod{x_i^*, \thetaprime } - y_i^*}^2}}_{\text{Term A.1}}  \Biggr\}
    \end{split}
    \end{equation}
    We first bound Term A.0.
    By \cref{fact:empirical_cov_bound} it holds that $\tfrac 1 n \sum_{i=1}^n \Iprod{u,x_i^*}^2 \leq (1+ \bigO{\eta \log(1/\eta)}\left(u^\top \Sigma u\right)$.
    Further, by \cref{lemma:small_subset} (the second part with $\mustar = 0$) we obtain
    \[
        \calA_\eta \sststile{}{} \Set{ \tfrac 1 n \sum_{i=1}^n r_i w_i \Iprod{u, x_i^*}^2\geq \left(1 - \bigO{\eta \log(1/\eta)}\right)\left(u^\top \Sigma u\right) } \,.
    \]
    Hence, $\calA_\eta \sststile{}{} \Set{\tfrac 1 n \sum_{i=1}^n (1- r_i w_i) \Iprod{u, x_i^*}^2 \leq \bigO{\eta \log(1/\eta)} \left(u^\top \Sigma u\right)}$. 
    We now bound Term A.1. By SoS Triangle Inequality it holds that
    \begin{align*}
        \calA_\eta \sststile{}{} \Biggl\{ \tfrac 1 n \sum_{i=1}^n (1- r_i w_i) \Paren{\Iprod{x_i^*, \thetaprime } - y_i^*}^2 &\leq  \tfrac 2 n \sum_{i=1}^n (1- r_i w_i) \Paren{\Iprod{x_i^*, \theta } - y_i^*}^2\\
        &+ \tfrac 2 n \sum_{i=1}^n (1- r_i w_i) \Iprod{x_i^*, \thetaprime -\theta}^2 \Biggr\}
    \end{align*}
    The first of these two terms is at most $\bigO{\eta \log(1/\eta)}$ in the same way as in Term A.0 (applying the same reasoning with $d =1$)\footnote{If we only have an upper bound of 1 on the variance of the ``uncorrupted noise'', this upper bound still holds.}.
    Again by \cref{fact:empirical_cov_bound} and since $\theta$ is sufficiently close to $\thetals$, we can bound the second term as follows
    \begin{align*}
    \calA_\eta \sststile{}{} \Biggl\{ \tfrac 2 n \sum_{i=1}^n (1- r_i w_i) \Iprod{x_i^*, \thetaprime -\theta}^2 &\leq \tfrac 2 n \sum_{i=1}^n \Iprod{x_i^*, \thetaprime -\theta}^2 \leq 4 \Norm{\tilde{\Sigma}^{1/2} (\thetaprime - \thetals)}^2\,.
    \end{align*}
    Thus, we conclude that Term A is at most
    \begin{align*}
        \calA_\eta \sststile{}{} \Biggl\{ &\Paren{\tfrac 1 n \sum_{i=1}^n (1- r_i w_i) \Iprod{u, x_i^*} \cdot \Paren{\Iprod{x_i^*, \thetaprime } - y_i^*}}^2 \\
        &\leq \bigO{\eta^2 \log^2(1/\eta)} \left(u^\top \Sigma u\right) + \bigO{\eta \log(1/\eta)} \cdot \Norm{\tilde{\Sigma}^{1/2}(\thetaprime - \thetals)}^2 \left(u^\top \Sigma u\right) \Biggr\}  \,.
    \end{align*}
    
    \paragraph*{Bounding Term B in \cref{eq:sos_robust_regression}}
    Next, we bound Term B in \cref{eq:sos_robust_regression}.
    The strategy is analogous to the one for Term A but slightly more complex.
    First, by the Cauchy-Schwarz Inequality it follows that
    \begin{equation}
    \begin{split}
        \calA_\eta &\sststile{}{} \Biggl\{ \Paren{ \tfrac 1 n \sum_{i=1}^n (1- r_i w_i) \Iprod{u, x_i'} \cdot \Paren{\Iprod{x_i', \thetaprime } - y_i'} }^2 \\
        &\leq \underbrace{\Paren{\tfrac 1 n \sum_{i=1}^n (1- r_i w_i) \Iprod{u, x_i'}^2}}_{\text{Term B.0}} \cdot \underbrace{\Paren{\tfrac 1 n \sum_{i=1}^n (1- r_i w_i) \cdot \Paren{\Iprod{x_i', \theta' } - y_i'}^2}}_{\text{Term B.1}} \,
    \end{split}
    \end{equation}
    We start with Term B.0:
    Note that we have that
    \[
    \Paren{\tfrac 1 n \sum_{i=1}^n (1- r_i w_i) \Iprod{u, x_i'}^2} = u^\top \Sigma' u - \tfrac{1} n  \sum r_iw_i \Iprod{u, x_i^*}^2\,.
    \]
    By~\cref{fact:cov_small_subset} and~\cref{lemma:subset_selection} we have that
    \[ \calA_{\eta} \sststile{}{} \Set{\tfrac{1} n  \sum r_iw_i \Iprod{u, x_i^*}^2 \geq (1-\eta\log(1/\eta) u^\top \Sigma u} \,.\]
    Thus, we can bound Term B.0 as follows:
    \[ \calA_{\eta} \sststile{}{} \Set{\Paren{\tfrac 1 n \sum_{i=1}^n (1- r_i w_i) \Iprod{u, x_i'}^2} \leq \Iprod{\Sigma - \Sigma', uu^\top} + \eta \log (1/\eta) \left(u^\top \Sigma u \right) \leq \eta \log (1/\eta) \left(u^\top \Sigma u\right)} \,.\]
    Where in the last inequality we used that by~\cref{lem:main_identity_cov_sos_proof} we have that 
    \begin{align*}
        \calA_\eta \sststile{6}{w} \Set{ \Iprod{ \Sigma' - \Sigma,   uu^\top }^2  \leq \bigO{\eta^2 \log^2 (1/\eta)}  \Paren{ u^\top \Sigma u }^2 } \,.
    \end{align*}
    And that by~\cref{lem:square_root_sos} this implies that $\calA_\eta \sststile{6}{w} \Set{ \Iprod{ \Sigma' - \Sigma,   uu^\top }  \leq \bigO{\eta \log (1/\eta)} \Paren{u^\top \Sigma u }}$.
    Note that we have that term B.1 is bounded by
    \begin{equation}
    \label{eqn:b1-term-regression-expansion}
        \begin{split}
            \calA_\eta \sststile{}{} \Biggl\{\tfrac 1 n \sum_{i=1}^n (1- r_i w_i) \cdot \Paren{\Iprod{x_i', \theta' } - y_i'}^2 &= \tfrac 1 n \sum_{i=1}^n\cdot \Paren{\Iprod{x_i', \theta' } - y_i'}^2 - \tfrac 1 n \sum_{i=1}^n r_i w_i \cdot \Paren{\Iprod{x_i', \theta' } - y_i'}^2 \\
            &\leq 1 + \eta - \tfrac 1 n \sum_{i=1}^n r_i w_i \cdot \Paren{\Iprod{x_i^*, \theta' } - y_i^*}^2\,, \Biggr\}
        \end{split}
    \end{equation}
    where the second inequality comes from our constraint on the covariance of the noise.
    Now observe that we can write the second term of Equation~\eqref{eqn:b1-term-regression-expansion} as 
    \begin{equation}
    \label{eqn:tmp-think-of-something}
        \begin{split}
            \sststile{}{} \Biggl\{ & - \tfrac 1 n \sum_{i=1}^n r_i w_i \cdot \Paren{\Iprod{x_i^*, \theta' } - y_i^*}^2  \\
            &= - \tfrac 1 n \sum_{i \in [n]}  \Paren{\Iprod{x_i^*, \theta' } - y_i^*}^2 + \tfrac 1 n \sum_{i=1}^n (1-r_i w_i) \cdot \Paren{\Iprod{x_i^*, \theta' } - y_i^*}^2 \Biggr\} \,. 
        \end{split}
    \end{equation}
    We bound the first and second sum of Equation~\eqref{eqn:tmp-think-of-something} separately.
    Note that the first sum is equal to
    \begin{align*}
       \sststile{}{} \Set{ \tfrac 1 n \sum_{i \in [n]}  \Paren{\Iprod{x_i^*, \thetals } - y_i^*}^2 + \tfrac 1 n \sum_{i \in [n]} \Iprod{x_i^*, \thetals - \theta'}^2 + \tfrac 2 n \sum_{i \in [n]} \left(\Iprod{x_i^*, \thetals} - y_i^*\right)\left(\Iprod{x_i^*, \thetals - \theta'}\right) }\,.
    \end{align*}
    The last term in this sum is equal to $0$ by the gradient condition on the uncorrupted samples and $\thetals$. Furthermore, the second term is equal to $\norm{\tilde{\Sigma}^{1/2} (\theta' - \thetals)}^2 \geq 0$. Finally, we note that the first term contains no indeterminates and is bounded by
    \begin{align*}
        \calA_\eta \sststile{}{} \Biggl\{  & \tfrac 1 n \sum_{i \in [n]}  \Paren{\Iprod{x_i^*, \thetals } - y_i^*}^2 
        \\
        &= \tfrac 1 n \sum_{i \in [n]}  \Paren{\Iprod{x_i^*, \thetals } - y_i^*}\Paren{\Iprod{x_i^*, \theta } - y_i^*} - \tfrac 1 n \sum_{i \in [n]}  \Paren{\Iprod{x_i^*, \thetals } - y_i^*}\Paren{\Iprod{x_i^*, \theta - \thetals }}\\
        &= \tfrac 1 n \sum_{i \in [n]}  \Paren{\Iprod{x_i^*, \thetals } - y_i^*}\Paren{\Iprod{x_i^*, \theta } - y_i^*} \\
        &= \tfrac 1 n \sum_{i \in [n]} \Paren{\Iprod{x_i^*, \theta } - y_i^*}^2 - \tfrac 1 n \sum_{i \in [n]}  \Paren{\Iprod{x_i^*, \theta - \thetals}}\Paren{\Iprod{x_i^*, \theta } - y_i^*} \\
        &= \tfrac 1 n \sum_{i \in [n]} \Paren{\Iprod{x_i^*, \theta } - y_i^*}^2 - \tfrac 1 n \sum_{i \in [n]} \Iprod{x_i^*, \theta - \thetals}^2 \\
        &\geq 1 - \eta  \Biggr\}  \,,
    \end{align*} 
    using the gradient condition on the uncorrupted samples, bounds on the one-dimensional noise\footnote{If the variance of the ``uncorrupted noise'' is only promised to be at least $1-\eta$, the lower bound only changes by an additional additive factor of $\eta$ that doesn't affect the rest of the proof.}, and that $\norm{\tilde{\Sigma}^{1/2}(\theta - \thetals)}^2 \leq \bigO{\eta^2}$.
    Thus, we conclude that 
    \[ \sststile{}{} \Set{ - \tfrac 1 n \sum_{i \in [n]}  \Paren{\Iprod{x_i^*, \theta' } - y_i}^2 \leq -1 + \eta  } \,.\]
    We now bound the second sum in Equation~\eqref{eqn:tmp-think-of-something}. We have by SoS Almost Triangle Inequality and~\cref{fact:cov_small_subset} that
    \begin{align*}
        \calA_{\eta} \sststile{}{} \Biggl\{ \tfrac 1 n \sum_{i=1}^n (1-r_i w_i) \cdot \Paren{\Iprod{x_i^*, \theta' } - y_i^*}^2 &\leq \tfrac 2 n \sum_{i=1}^n (1-r_i w_i) \cdot \Paren{\Iprod{x_i^*, \theta } - y_i^*}^2 + \tfrac 2 n \sum_{i=1}^n (1-r_iw_i) \langle x_i^*, \theta' - \theta\rangle^2 \\
        &\leq \eta \log (1/\eta) + 2 \norm{\tilde{\Sigma}^{1/2} (\theta'-\theta)}^2 \Biggr\} \,.
    \end{align*}
    Thus, we can derive the following upper bound on~\eqref{eqn:tmp-think-of-something}
    \[
        -1 + \bigO{\eta \log(1/\eta)} + \bigO{1} \norm{\tilde{\Sigma}^{1/2}(\theta' - \theta)} \,.
    \]
    Combining these facts we conclude that Term B.1 is at most $\bigO{\eta \log(1/\eta)} + 2 \norm{\tilde{\Sigma}^{1/2} (\theta'-\theta)}^2$. Thus, we have the following bound on term B:
    \begin{align*}
        \calA_{\eta} \sststile{}{} \Biggl\{&\Paren{ \tfrac 1 n \sum_{i=1}^n (1- r_i w_i) \Iprod{u, x_i'} \cdot \Paren{\Iprod{x_i', \thetaprime } - y_i'} }^2 \\
        &\leq \bigO{\eta^2 \log^2 (1/\eta)}\cdot (u^\top \Sigma u) + \bigO{\eta \log (1/\eta)} \cdot \norm{\tilde{\Sigma}^{1/2}(\theta' - \thetals}^2\cdot (u^\top \Sigma u)\Biggr\} \,.
    \end{align*}

    \paragraph*{Putting Things Together}
    Using the bounds on Term A and Term B obtained above, we conclude that
    \begin{align*}
        \calA_\eta \sststile{2}{u, w, \theta'} \Biggl\{ &\Iprod{u, \Paren{\tfrac 1 n \sum_{i=1}^n x_i^* (x_i^*)^\top} \Paren{\thetaprime - \thetals} }^2 \\
        &\leq \bigO{\eta^2\log^2(1/\eta)}  \Paren{u^\top \Sigma u} + \bigO{\eta \log(1/\eta)}\Norm{\tilde{\Sigma}^{1/2}\Paren{\thetaprime - \thetals}}^2  \Paren{u^\top \Sigma u}\Biggr\}  \,.
    \end{align*}
\end{proof}

\subsection{Bootstrapping}

It remains to prove~\cref{lem:sos_regression_bootstrapping}.
This follows from a result of~\cite{bakshi2021robust} which gives a $\sqrt{\eta}$ rate for our model.
We explicitly reprove it below for completeness using the terminology introduced in our work.
The main lemma we need is as follows:
\begin{lemma}[Bootstrapping Regression]
\label{lem:bootstrapping_regression}
Suppose $x_1^*, \ldots, x_n^*$ and $y_1^*, \ldots, y_n^*$ are $\eta$-good and $\eta$-higher-order-good.
Then it holds that
\begin{equation*}
\calA_\eta \sststile{O(1)}{} \Set{ 
\norm{ \tilde{\Sigma}^{1/2} \Paren{ \theta' - \thetals}  }_2^2
 \leq \bigO{ \eta  }}   \,.
\end{equation*}
\end{lemma}
We will need the following intermediate lemma which we will prove at the end of this section.
Note that the difference to the results in~\cref{sec:cov-estimation-spectral} is that we allow $u$ to be an SoS variable instead of just a fixed vector.
This is essential for obtaining~\cref{lem:bootstrapping_regression}.
\begin{lemma}
\label{lem:boostrapping_covariance}
    Suppose $x_1^*, \ldots, x_n^*$ and $y_1^*, \ldots, y_n^*$ are $\eta$-good and $\eta$-higher-order-good.
    Then it holds that
    \begin{equation*}
    \calA_\eta \sststile{O(1)}{u} \Set{ \Paren{u^\top \Paren{\Sigma' - \tilde{\Sigma}} u}^2 \leq \bigO{ \eta  } \Paren{u^\top \tilde{\Sigma} u}^2  } \,.
    \end{equation*}
\end{lemma}

\begin{proof}

We start by giving the proof of~\cref{lem:bootstrapping_regression}
\begin{proof}[Proof of~\cref{lem:bootstrapping_regression}]
    Note that since for $C > 0$, $\Set{a^4 \leq C a^2} \sststile{\bigO{1}}{a} \Set{a^2 \leq C}$ and $\Set{a^2 \leq C} \sststile{\bigO{1}}{a} \Set{a \leq \sqrt{C}}$ (cf.~\cref{lem:square_root_sos,fact:cancellation_two_sos}), it is enough to show that
    \[
        \calA_\eta \sststile{O(1)}{} \Paren{(\theta' - \thetals)^\top \tilde{\Sigma} (\theta' - \thetals)}^4 \leq \bigO{\eta^2} \Paren{(\theta' - \thetals)^\top \tilde{\Sigma} (\theta' - \thetals)}^2 \,.
    \]
    Analogously as in the proof of~\cref{lem:sos_regression_without_bootstrapping} it holds that
    \[
        \calA_\eta \sststile{\bigO{1}}{} \Biggl\{ \tilde{\Sigma} (\theta' - \thetals)  = \tfrac 1 n \sum_{i=1}^n (1- r_i w_i) x_i^* \cdot \Paren{\Iprod{x_i^*, \thetaprime } - y_i^*} - \tfrac{1} n \sum_{i \in [n]} (1- r_i w_i) x_i' \cdot \Paren{\Iprod{x_i', \thetaprime} - y_i'} \Biggr\} \,.
    \]
    Adding and subtracting $\thetals$ in the inner product of the first sum, we obtain that $\calA$ implies at constant degree that
    \begin{align*}
        \Biggl\{ \tilde{\Sigma} (\theta' - \thetals)  &= \tfrac 1 n \sum_{i=1}^n (1- r_i w_i) x_i^* \cdot \Paren{\Iprod{x_i^*, \thetals } - y_i^*} - \tfrac{1} n \sum_{i \in [n]} (1- r_i w_i) x_i' \cdot \Paren{\Iprod{x_i', \thetaprime} - y_i'}  \\
        &+ \tfrac 1 n \sum_{i=1}^n (1- r_i w_i) x_i^* \cdot \Iprod{x_i^*, \theta' - \thetals }\Biggr\} \,.
    \end{align*}
    Thus, by SoS triangle inequality we obtain that
    \begin{align*}
        \calA_\eta \sststile{\bigO{1}}{} \Biggl\{ \Paren{(\theta' - \thetals)^\top \tilde{\Sigma} (\theta' - \thetals)}^4 &\leq \bigO{1} \underbrace{\Paren{\tfrac 1 n \sum_{i=1}^n (1-r_i w_i) \Iprod{x_i^*, \theta' - \thetals} \Paren{\Iprod{x_i^*,\thetals} - y_i^*}}^4}_{\text{Term A}} \\
        &+\bigO{1} \underbrace{\Paren{\tfrac 1 n \sum_{i=1}^n (1-r_i w_i) \Iprod{x_i', \theta' - \thetals} \Paren{\Iprod{x_i',\theta'} - y_i'}}^4}_{\text{Term B}} \\
        &+\bigO{1} \underbrace{\Paren{\tfrac 1 n \sum_{i=1}^n (1-r_i w_i) \Iprod{x_i^*, \theta' - \thetals}^2}^4}_{\text{Term C}}\Biggr\} \,.
    \end{align*}
\end{proof}
By SoS triangle inequality and $\eta$-higher-order-goodness and~\cref{fact:cert_bounded_moments_from_concentration} (SoS proof of hypercontractive fourth moments) it follows that Term C above is at most $\eta^2 \Paren{(\theta' - \thetals)^\top \tilde{\Sigma}(\theta' - \thetals)}^4$ which is a small multiple of the left-hand side.
Thus, we can ignore the last summand by rearranging.

\paragraph{Bounding Terms A and B.}
We continue to bound Terms A and B.
For Term A, we obtain using again~\cref{fact:cert_bounded_moments_from_concentration}, SoS Cauchy Schwarz, and that $ \tfrac 1 n \sum_{i=1}^n \Paren{\Iprod{x_i^*,\thetals} - y_i^*}^4 \leq \bigO{1}$ (cf.~\cref{sec:feasibility_regression})
\begin{align*}
    \calA_\eta &\sststile{\bigO{1}}{} \Biggl\{ \Paren{\tfrac 1 n \sum_{i=1}^n (1-r_i w_i) \Iprod{x_i^*, \theta' - \thetals} \Paren{\Iprod{x_i^*,\thetals} - y_i^*}}^4 \\
    &\leq \eta ^2 \Paren{\tfrac 1 n \sum_{i=1}^n \Iprod{x_i^*, \theta' - \thetals}^2 \Paren{\Iprod{x_i^*,\thetals} - y_i^*}^2}^2 \\
    &\leq \eta ^2 \Paren{\tfrac 1 n \sum_{i=1}^n \Iprod{x_i^*, \theta' - \thetals}^4} \Paren{ \tfrac 1 n \sum_{i=1}^n \Paren{\Iprod{x_i^*,\thetals} - y_i^*}^4}  \\
    &\leq \bigO{\eta^2} \Paren{(\theta' - \thetals)^\top \tilde{\Sigma} (\theta' - \thetals)}^2\Biggr\}
\end{align*}

Similarly, using the constraints that there is an SoS proof in variables $v$ that $\tfrac{1}{n} \sum_{i=1}^n \Iprod{x_i',v}^4 \leq \bigO{1} \Iprod{u, \Sigma' u}^2$ and that $\tfrac{1}{n} \sum_{i=1}^n (\Iprod{x_i', \theta'} - y_i')^4 \leq \bigO{1}$, we can bound Term B as follows
\begin{align*}
    \calA_\eta &\sststile{\bigO{1}}{} \Biggl\{ \Paren{\tfrac 1 n \sum_{i=1}^n (1-r_i w_i) \Iprod{x_i', \theta' - \thetals} \Paren{\Iprod{x_i',\theta'} - y_i'}}^4 \\
    &\leq \eta^2 \Paren{\tfrac 1 n \sum_{i=1}^n \Iprod{x_i', \theta' - \thetals}^2 \Paren{\Iprod{x_i',\theta'} - y_i'}^2}^2 \\
    &\leq \eta^2 \Paren{\tfrac 1 n \sum_{i=1}^n \Iprod{x_i', \theta' - \thetals}^4} \Paren{ \tfrac 1 n \sum_{i=1}^n\Paren{\Iprod{x_i',\theta'} - y_i'}^4} \\
    &\leq \bigO{\eta^2} \Paren{(\theta' - \thetals)^\top \Sigma' (\theta' - \thetals)}^2\Biggr\}
\end{align*}

Thus, overall we have shown that
\begin{align*}
    \calA_\eta &\sststile{\bigO{1}}{} \Biggl\{\Paren{(\theta' - \thetals)^\top \tilde{\Sigma} (\theta' - \thetals)}^4 \\
    &\leq \bigO{\eta^2} \Paren{\Paren{(\theta' - \thetals)^\top \tilde{\Sigma} (\theta' - \thetals)}^2 + \Paren{(\theta' - \thetals)^\top \Sigma' (\theta' - \thetals)}^2}\Biggr\} \,.
\end{align*}
Applying SoS triangle and~\cref{lem:boostrapping_covariance} (bootstrapping of covariance estimation), we can deduce that $\calA_\eta$ implies at constant degree that the right-hand side of the above is at most
\begin{align*}
    &\bigO{\eta^2} \Paren{\Paren{(\theta' - \thetals)^\top \tilde{\Sigma} (\theta' - \thetals)}^2 + \Paren{(\theta' - \thetals)^\top \Sigma' (\theta' - \thetals)}^2}  \\
    &\leq \bigO{\eta^2} \Paren{\Paren{(\theta' - \thetals)^\top \tilde{\Sigma} (\theta' - \thetals)}^2 + \Paren{(\theta' - \thetals)^\top \Paren{\Sigma' - \tilde{\Sigma}} (\theta' - \thetals)}^2}  \\
    &\leq \bigO{\eta^2} \Paren{(\theta' - \thetals)^\top \tilde{\Sigma} (\theta' - \thetals)}^2 \,,
\end{align*}
as desired.

\end{proof}

It remains to prove~\cref{lem:boostrapping_covariance}.
\begin{proof}[Proof of~\cref{lem:boostrapping_covariance}]
    Note that by SoS triangle inequality, it holds that 
    \[
        \sststile{O(1)}{u} \Biggl\{\eta \Paren{u^\top \Sigma' u}^2 \leq \bigO{\eta} \Paren{u^\top \tilde{\Sigma} u}^2 + \bigO{\eta} \Paren{u^\top \Paren{\Sigma' - \tilde{\Sigma}} u}^2 \Biggr\}\,.
    \]
    Thus, by rearranging, it is enough to show that
    \begin{align}
        \sststile{O(1)}{u} \Biggl\{\Paren{u^\top \Paren{\Sigma' - \tilde{\Sigma}} u}^2 \leq \bigO{\eta} \Paren{u^\top \tilde{\Sigma} u}^2 + \bigO{\eta} \Paren{u^\top \Sigma' u}^2 \Biggr\}\,. \label{eq:to_show_cov_bootstrapping}
    \end{align}
    Plugging in the definition of $\tilde{\Sigma}$ and $\Sigma'$ and using SoS triangle inequality and SoS Cauchy-Schwarz, we obtain
    \begin{align*}
        \calA_\eta \sststile{O(1)}{u} \Biggl\{\Paren{u^\top\Paren{\Sigma' - \tilde{\Sigma}} u}^2 &= \Paren{\tfrac 1 n \sum_{i=1}^n (1-r_iw_i) \Iprod{u, x_i'}^2}^2 + \Paren{\tfrac 1 n \sum_{i=1}^n (1-r_iw_i) \Iprod{u, x_i^*}^2}^2 \\
        &\leq \eta \cdot \tfrac 1 n \sum_{i=1}^n \Iprod{u, x_i'}^4 + \eta \cdot \tfrac 1 n \sum_{i=1}^n \Iprod{u, x_i^*}^4
        \Biggr\} \,.
    \end{align*}
    Now, by our constraint on the fourth moments, we now that there exists an SoS proof in variables $u$ that the sum in the first term is at most $\bigO{1} (u^\top \Sigma' u)^2$.
    Further, by $\eta$-higher-order-goodness and~\cref{fact:cert_bounded_moments_from_concentration} it follows that there also exists an SoS proof in variables $u$ that the sum in the second term is at most $\bigO{1} (u^\top \Sigma' u)^2$.
    Together, this implies~\cref{eq:to_show_cov_bootstrapping}.
\end{proof}

\subsection{Private Regression}
\label{sec:private_regression}

We will now show how to transform the previous robust regression algorithm into a private algorithm.
This will prove~\cref{thm:main_private_regression} restated below.
\restatetheorem{thm:main_private_regression}

Throughout this section, we will assume that $\Sigma$ is invertible. However, it is easy to reduce to this setting when working with private algorithms by adding Gaussian noise to the covariates with variance $\alpha^2/R$. Note that we have that this is equivalent to a model where the noise on labels is equal to $1 + \norm{\theta} \cdot \alpha^2 / R \leq 1 + \alpha^2$ and we note that our robust estimator (and therefore the resulting private estimator) works in the setting where the label noise has variance at least $1$ and at most $1 + \bigO{\alpha / \log(1/\alpha)}$.

Our algorithm will be an instantiation of the exponential mechanism.
The score function we use is based on the following modification of the SoS constraint system introduced in the previous section, where $T \in [n]$. Throughout this section when showing the utility and volume guarantees of our score function we will assume that $x_i$ are an $\alpha/\sqrt{1/\alpha}$-corruption of Gaussian samples and thus our algorithm will also be robust to the same fraction of corruptions.
For technical reasons that will become clear later on, we will also require an SoS variable which encodes the inverse of the ``SoS covariance matrix'' $\Sigma'$.
\begin{equation*}
\calB_{T}(x,y)=
  \left \{
    \begin{aligned}
      &\forall i\in [n].
      & w_i^2
      & = w_i \\
      & &\sum_{i\in [n]} w_i &\geq n - T \\
      &&\tfrac{1} n \sum_{i \in [n]} x_i'(x_i')^\top &= \Sigma' \\
      &\forall i\in [n] & w_i(x'_i - x_i) = 0&\,,\; w_i(y'_i - y_i) =0 \\
      &&\tfrac{1} n \sum_{i \in [n]} \Paren{\Iprod{x_i', \thetaprime} - y_i'}x_i' &= 0 \\
      &&\tfrac{1} n \sum_{i \in [n]} \Paren{y_i' - \Iprod{\thetaprime, x_i'}}^2 &\leq 1+O(\eta) \\
      &\forall v \in \mathbb{R}^d
      & \frac{1}{n}\sum_{i \in [n]}  {\langle x'_i , v\rangle^{4}} 
      &\leq  \Paren{3+\eta \log^2 (1/\eta)} \left(  v^\top \Sigma' v \right)^{2} \\
      & & \tfrac{1} n \sum_{i \in [n]} \Paren{y_i' - \Iprod{\thetaprime, x_i'}}^4 &\leq \bigO{1} \\
      & &Q = Q^\top \,,\quad \Sigma' Q \Sigma' = \Sigma' \,,\quad  Q \Sigma' Q &= Q \,,\quad Q \Sigma' = I_d \,,\quad \Sigma' Q = I_d
    \end{aligned}
  \right \}
\end{equation*}

Roughly speaking, the score candidate point $\tilde{\theta}$ will be the smallest $T$ such that there exists a pseudo-expectation $\pE$ such that the following three conditions are met
\begin{enumerate}
    \item $\pE \sdtstile{}{} \calB_T$ at degree $\bigO{1}$.
    \item $\norm{\pE \theta'}_2 \leq R$.
    \item For all fixed vectors $u \in \R^d$ (i.e., that do depending on indeterminates), it holds that
    \[
        \Iprod{u, \pE \theta' - \tilde{\theta}}^2 \leq \bigO{\alpha^2} \pE \Iprod{u, Q u} \,.
    \]
\end{enumerate}

The last condition is the most crucial one.
On a high level, the score corresponds to the smallest number of input data points we have to change such that the candidate $\tilde{\theta}$ is close the output of our robust estimator in the ``$\Sigma^{1/2}$''-error-metric.
That is, such that $\norm{\Sigma^{1/2}(\pE \theta' - \tilde{\theta})} \leq O(\alpha)$.
Note that this depends on the unknown covariance $\Sigma$ and we can thus not evaluate it.
The closeness constraints above acts as a proxy for this, by instead considering (some version of) the metric induced by $\pE \Sigma'$.
The form above turns out to be the right one that satisfies necessary preconditions of~\cref{thm:pure_dp_reduction} (namely, quasi-convexity and efficient computability).

\paragraph{Certifiable parameters and definition of the score function.}
Before giving the formal definition of our score function, we need one other definition.
\begin{definition}[Certifiable Parameter]
    \label{def:cert_parameter_regression}
    Let $\alpha, \tau \in \R^{\ge 0}$, $n \in \N$ and $T \in [0,n]$.
    Further, let $x_1, \dots, x_n \in \R^d$, $y_1, \dots, y_n \in \R$.
    We call a parameter $\tilde{\theta}$, a $\paren{\alpha, \tau, T}$ certifiable parameter for $\paren{x_1, y_1}, \dots, \paren{x_n, y_n}$ if there exists a linear functional $\mathcal{L}$ which $\tau$-approximately satisfies $\calB_T(x,y)$ (cf.~\cref{def:tau_relaxed_system})\footnote{Our constraint system includes constraints regarding the existence of SoS proofs for all vectors $v$ but these constraints can be encoded as polynomial inequalities and we apply the approximate satisfiability definition to this form of the constraint.} and such that
    \begin{enumerate}
        \item $\norm{\calR \paren{\calL}}_F \le R' + \tau \cdot T$,  where $R' = \poly\paren{n, d, R}$ is sufficiently large and $\calR \paren{\calL}$ is the matrix representation of $\calL$.
        \item $\bigO{\alpha^2} \cdot \calL \brac{Q} - \calL\brac{\theta' - \tilde{\theta}} \calL\brac{\theta' - \tilde{\theta}}^\top \succeq -\tau T \cdot I_d$ 
        \item $\norm{\calL \theta'}_2 \leq 2R + \tau \cdot T$
    \end{enumerate}
    Furthermore, we will refer to the linear functional $\calL$ as a $(\alpha, \tau, T)$ certificate for $((x,y),\tilde{\theta})$.
\end{definition}
For our purposes, we will end up setting $\tau = 1/(n \cdot d \cdot R \cdot \alpha^{-1} \cdot \covscale)^C$, for a large enough absolute constant $C$.
Now we use this definition to define a score function.
\begin{definition}[Score Function]
\label{def:score_function_regression}
Let $\alpha, \tau \in \R^{\ge 0}$, $n \in \N$ and $T \in [0,n]$.
Further, let $x_1, \dots, x_n \in \R^d$, $y_1, \dots, y_n \in \R$. Let $\mathbb{B}^d \paren{2R + n \tau + \alpha \sqrt{\covscale}}$ denote the $\ell_2$-ball of radius $2R + n \tau + \alpha \sqrt{\covscale}$ around the origin in $\R^d$. We define the score function $\calS\paren{\cdot \;; x, y; \alpha, \tau}: \mathbb{B}^d \paren{2R + n \tau + \alpha \sqrt{\covscale}} \to \R$ as follows:
\begin{equation*}
    \calS\paren{\tilde{\theta}; x, y; \alpha, \tau} = \min_{T \in [n]} \quad  \text{such that}\quad \tilde{\theta} \text{ is a } \paren{\alpha, \tau, T} \text{ certifiable parameter for } \paren{x, y}.
\end{equation*}
\end{definition}

\begin{remark}[On the domain of the score function]
    Note that the domain of the score function is chosen such that even if the true parameter $\theta$ is on the edge of the $2R$ ball we are promised it is contained in, all points such that $\norm{\Sigma^{1/2}(\tilde{\theta} - \theta)} \leq \alpha$ are contained within the domain of the score function. This is important because these are our low scoring points and we need a large enough volume of them for our sampling algorithm to output a point with low score.
\end{remark}

\begin{remark}[On the well-definedness of the score function]
    Note that in general for closeness constraints and arbitrary size domains this function would not necessarily be well defined, because there might be \emph{no} $T$ such that a candidate point is a $(\alpha, \tau, T)$ certifiable parameter.
    For example, taking the closeness constraint $\norm{\calL \theta' - \tilde{\theta}}_\infty \leq \alpha/\sqrt{d}$ (which was used in~\cite{Hopkins2023Robustness}) points $\tilde{\theta}$ with norm $2R + \tau T + 2\alpha$ are not necessarily $(\alpha, \tau, n)$-certifiable parameters since there may not be $\theta'$ which is simultaneously norm at most $2R + \tau T$ and also $\alpha$-close to $\tilde{\theta}$. However, in our case every point has score at most $n$, which we will show in~\cref{lem:score-function-ub-regression}.
\end{remark}

In the rest of this section we will show that the score function defined above satisfies the conditions of \Cref{thm:pure_dp_reduction}.
\begin{enumerate}
    \item Bounded Sensitivty (\cref{lem:bounded-sensitivity-regression}: We will show that the score function has bounded sensitivity with respect to the input data (i.e., $(x,y)$).
    \item Quasi Convexity (\cref{lem:regression-quasiconvexity}): We will show that the score function is quasi-convex with respect to the parameter $\tilde{\theta}$.
    \item Volume (\cref{cor:volume_considerations_regression}): We will show that the volume of the set of points $\tilde{\theta}$ that have score at most $\eta \cdot n$ is sufficiently large and the volume of the points with score at most $\eta' \cdot n$ for $\eta' > \eta$ is sufficiently small.
    \item Efficient Computability (\cref{lem:efficient-comp-regression}): We will verify that the score is efficiently computable for a fixed $\tilde{\theta}, (x, y)$.
    \item Robust algorithm finds a low-scoring parameter efficiently (\cref{lem:find-low-score-regression}): We verify that finding $\tilde{\theta}$ that minimizes the score up to error $1$ for a fixed $(x, y)$ can be done efficiently.
\end{enumerate}
We will also show that a low scoring point $\tilde{\theta}$ achieves good accuracy, i.e., such that $\norm{\Sigma^{1/2}(\tilde{\theta} - \theta)}$ is small.
We first focus on quasi-convexity, accuracy and volume, for which our new ``closeness constraint'' is key.
We will then address the other required properties, which are similar to previous applications of this transformation in \cite{Hopkins2023Robustness}.

We introduce some additional notation.
For any given desired accuracy $\alpha$ we let $\eta(\alpha) = \alpha/\sqrt{\log (1/\alpha)}$.
Further, we denote by $\eta^*$ the breakdown point of our robust estimator.
For simplicity, we may omit the dependence of the score function on the parameters $\alpha, \tau$ in the notation and write $\eta = \eta(\alpha)$. 

On a high level, our proof strategy is as follows:
We will show that every point $\tilde{\theta}$ with score $\eta(\alpha) n$ has accuracy at least $O(\alpha)$.\footnote{$\eta(\alpha)$ will also correspond to the fraction of corruptions our algorithm is robust to.}
To show that the exponential mechanism outputs such a point with high probability, we show that the volume of low-scoring points is not too small and the volume of high-scoring points is not too large.
To this end, we analyze points of score less or more than $\eta^* n$ separately, appealing to properties of our robust estimator in the first case, and boundedness of the domain in the second.

\subsubsection{Properties of the Closeness Constraint}

We record the following lemmas we need about our closeness constraint.
We remark that in the lemma below we do note require an SoS proof of the inequality, but just that it be true as an inequality after applying pseudo-expectations.

The two key innovative properties of our closeness constraint is that we are able to enforce closeness in (an adequate proxy of) $\Sigma^{1/2}$-norm without explicitly estimating $\Sigma$ or paying unnecessary factors of $\log L$ in volume computations.

In particular, we show the following lemma, which characterizes $\tilde{\theta}$ which have score at most $T$ via relation to the true parameter $\theta^*$. It will imply both the utility of our final estimator (via accuracy of low-scoring points) as well as the bounded volume of points with score which is large, but below the breakdown point of our estimator. 
\begin{lemma}
\label{lem:closeness_constraint_soundness}
    Let $T \leq \eta^* n$.
    Let $\tilde{\theta} \in \R^d$ be such that there exists a pseudo-distribution $\pE$ such that $\pE \sdtstile{}{} \calB_T$ and it holds that for every fixed vector $u \in \R^d$ the following inequality is true
    \begin{equation}
    \label{eq:reg_closeness}
        \Iprod{u,\pE\theta' - \tilde{\theta}}^2 \leq \bigO{\alpha^2} \cdot \pE \Iprod{u, Q u} \,.
    \end{equation}
    Then, it follows that $\norm{\Sigma^{1/2}(\tilde{\theta} - \theta^*)} \leq \bigO{\alpha + \tfrac T n \sqrt{\log (\tfrac n T)}}$.
    Further, the same holds for linear operators that $\tau$-approximately satisfy $\calB_T$ for $\tau = (ndR)^{-C}$ for $C$ a large enough absolute constant and for which~\cref{eq:reg_closeness} holds up to slack $\tau T \norm{u}^2$.
\end{lemma}

Before we give the proof of~\cref{lem:closeness_constraint_soundness} we state the following lemma which we need in the proof. We defer its proof to~\cref{sec:deferred_privacy_sos_proofs}.
\begin{lemma}
\label{lem:sos_inverse_covariance}
Let $0 < \eta $ be smaller than a sufficiently small constant and suppose $x_1^*, \ldots, x_n^*$ are $\eta$-good and $\eta$-higher-order-good (cf.~\cref{def:eps_goodness,def:higher_order_goodness}).
For any fixed vector $u$, i.e., that is not an indeterminate, we have that
\begin{align*}
    \calB_{\eta n} \sststile{\bigO{1}}{} \Biggl\{ \Iprod{u, Q u}^2 \leq \bigO{1}  \Paren{u^\top \Sigma^{-1} u}^2 \Biggr\} \,.
\end{align*}
\end{lemma}

\begin{proof}[Proof of~\cref{lem:closeness_constraint_soundness}]
    We only show the proof for exact pseudo-expectations.
    The extension to the approximate case is immediate.
    Let $T \leq \eta^*n$ and $\tilde{\theta}$ be such that there exists a pseudo-distribution $\pE$ such that $\pE \sdtstile{\bigO{1}}{} \calB_T$ and such that for every fixed vector $u$ it holds that
    \[
        \Iprod{u, \pE\theta' - \tilde{\theta}}^2 \leq \bigO{\alpha^2} \cdot \pE \Iprod{u, Q u} \,.
    \]
    Recall that we want to show that $\norm{\Sigma^{1/2}(\tilde{\theta}- \theta^*)} \leq \bigO{\alpha + \tfrac T n \log (\tfrac{n}{T})}$.
    It is enough to show that $\norm{\Sigma^{1/2}(\tilde{\theta}- \thetals)} \leq \bigO{\alpha + \tfrac T n \log (\tfrac{n}{T})}$.
    Let $\gamma =  \tfrac T n \log (\tfrac{n}{T})$.
    Note that this is equivalent to show that for every vector $u$, it holds that
    \[
        \Iprod{u, \Sigma^{1/2} (\tilde{\theta} - \thetals)}^2 \leq \bigO{\alpha^2 + \gamma^2} \cdot \norm{u}^2 \,.
    \]
    Which again is equivalent to showing that for every vector $u$, it holds that
    \[
        \Iprod{u, \tilde{\theta} - \thetals}^2 \leq \bigO{\alpha^2 + \gamma^2} \cdot \Iprod{u, \Sigma^{-1} u} \,.
    \]
    
    First, note that by Cauchy-Schwarz for pseudo-expectations (cf.~\cref{fact:pseudo-expectation-cauchy-schwarz}) and since $(a+b)^2 \leq 2a + 2b$ it holds that (for all real numbers)
    \begin{align*}
        \Iprod{u, \tilde{\theta} - \theta^*}^2 = \Iprod{u, \paren{\pE[\theta'] - \tilde{\theta}} - \paren{\thetals - \pE[\theta']}]}^2 \leq 2 \Paren{\pE \Iprod{u, \tilde{\theta} - \theta'}}^2 + 2 \Paren{\pE \Iprod{u, \theta' - \thetals}}^2 \,.
    \end{align*}
    The second term is at most $\bigO{\gamma^2} \Iprod{u, \Sigma^{-1} u}$ by the second part of~\cref{lem:main_sos_proof_regression}.
    By assumption, the first term is at most $\bigO{\alpha^2} \pE \Iprod{u, Q u}$.
    Thus, by~\cref{lem:sos_inverse_covariance} it follows that the second term is at most $\bigO{\alpha^2} \pE \Iprod{u, \Sigma^{-1} u}$ which yields the claim.
\end{proof}

\paragraph{Volume of High and Low Scoring Points.}
\begin{corollary}
    \label{cor:volume_considerations_regression}
    Let $T \leq \eta^* n$ and $n \geq \Omega\Paren{\frac{n^2 + \log^2(1/\beta)}{\alpha^2}}$.
    Then it holds that:
    \begin{enumerate}
        \item Every point in a $\Sigma^{1/2}$-ball of radius $O(\eta \sqrt{\log (1/\eta)})$ around $\thetals$ has score $T \leq \eta n  $
        \item Every point of score at most $\eta^* n$ is contained in a $\Sigma^{1/2}$-ball around $\theta$ of radius $O(1)$.
    \end{enumerate}
\end{corollary}
\begin{proof}
    We have already shown in~\cref{sec:feasibility} that a distribution supported on a single point with $x_i'$ equal to the true samples, $\Sigma'$ equal to the empirical covariance, and $\theta' = \thetals$ is feasible for the original robustness program.
    Note that the additional constraints involving the variable $Q$ can trivially be satisfied by setting $Q$ to be the inverse of the empirical covariance matrix.
    Furthermore, since $\theta$ has norm at most $R$ by the closeness of $\theta, \thetals$ we have that $\thetals$ has norm at most $2R + \tau T$.
    To show that any point $\tilde{\theta}$ such that $\norm{\Sigma^{1/2} (\tilde{\theta} - \theta^*)} \leq \bigO{\eta \sqrt{\log (1/\eta)}}$ has score at most $\eta n$ it suffices to show that for this single point distribution above that the closeness constraint holds for such $\tilde{\theta}$. Let $\bar{\Sigma}$ be the empirical covariance of the true samples and $Q = \bar{\Sigma}^{-1}$. We want to show that for all vectors $u$ it holds that 
    \[ \Iprod{u, \thetals - \tilde{\theta}}^2 \leq \bigO{\alpha^2} \cdot  \Iprod{u, Q u} \]
    or equivalently that 
    \[ \Iprod{u, (\bar{\Sigma})^{1/2} \left(\thetals - \tilde{\theta}\right)}^2 \leq \bigO{\alpha^2} \cdot  \norm{u}^2 \]
    Note that by concentration of the empirical covariance (\cref{fact:empirical_cov_bound}) it suffices to instead show that
    \[ \Iprod{u, \Sigma^{1/2}(\theta' - \tilde{\theta})}^2 \leq \bigO{\alpha^2} \cdot  \norm{u}^2\,.\]
    We have that
    \[ \Iprod{u, \Sigma^{1/2}(\theta' - \tilde{\theta})}^2 \leq \norm{u}^2 \cdot \norm{\Sigma^{1/2}(\theta' - \tilde{\theta})}^2 \leq \bigO{\alpha^2} \cdot \norm{u}^2\,,\]
    so this completes the proof.
    
    Consider any point $\tilde{\theta}$ of score $T < \eta^* n$. We have by~\cref{lem:closeness_constraint_soundness} that 
    \[\norm{\Sigma^{1/2}(\tilde{\theta} - \theta^*)} \leq \bigO{\alpha + \tfrac T n \sqrt{\log (\tfrac n T)}} \leq \bigO{1}\,,\]
    which completes the proof of the second half of the corollary.
\end{proof}

\paragraph{Quasi-Convexity.}
\begin{lemma}[Quasi-Convexity]
\label{lem:regression-quasiconvexity}
The score function $\calS\paren{\tilde{\theta}; x, y; \alpha, \tau}$ as defined in \Cref{def:score_function_regression} is quasi-convex with respect to the parameter $\tilde{\theta}$.
\end{lemma}
\begin{proof}
Suppose $\tilde{\theta}_1, \tilde{\theta}_2 \in \mathbb{B}^d \paren{2R + n \tau + \alpha \sqrt{\covscale}}$, and $\calS\paren{\tilde{\theta}_1; x, y ; \alpha, \tau} = T_1$, and $\calS\paren{\tilde{\theta}_2; x, y ; \alpha, \tau} = T_2$. We will show that for any $\lambda \in [0, 1]$, we have $\calS\paren{\lambda \tilde{\theta}_1 + (1-\lambda) \tilde{\theta}_2; x, y ; \alpha, \tau} \le \max\{T_1, T_2\}$. 

Let $\tilde{\theta}_3 = \lambda \tilde{\theta}_1 + (1-\lambda) \tilde{\theta}_2$. Since $\calS\paren{\tilde{\theta}_1; x, y ; \alpha, \tau} = T_1$, there exists a $\paren{\alpha, \tau, T_1}$ certificate $\mathcal L_1$ for $(x, y)$. Similarly, since $\calS\paren{\tilde{\theta}_2; x, y ; \alpha, \tau} = T_2$, there exists a $\paren{\alpha, \tau, T_2}$ certificate $\mathcal L_2$ for $(x, y)$. We will construct a $\paren{\alpha, \tau, \max\{T_1, T_2\}}$ certificate $\mathcal L_3$ for $(x, y)$. We construct
$\mathcal L_3$ as follows: $\mathcal L_3 = \lambda \mathcal L_1 + (1-\lambda) \mathcal L_2 $. All of the constraints would then be satisfied trivially except for the counting constraint and the closeness constraint. For the closeness constraint, note that we have that
\[ \Iprod{u, \calL_3 \theta' - \tilde{\theta_3}}^2  =  \Iprod{u, \lambda (\calL_1 \theta' - \tilde{\theta_1}) + (1-\lambda) (\calL_2 \theta' - \tilde{\theta_2})}^2 \leq \lambda \Iprod{u, \calL_1 \theta' - \tilde{\theta_1}}^2 + (1-\lambda) \Iprod{u, \calL_2 \theta' - \tilde{\theta_2}}^2\,.\]
Applying that both $\calL_1, \calL_2$ satisfy the closeness constraint we have that
\[ \Iprod{u, \calL_3 \theta' - \tilde{\theta_3}}^2 \leq \bigO{\alpha^2} \cdot \left( \lambda u^\top \calL_1 Q u + (1-\lambda) u^\top \calL_2 Q u\right) = \bigO{\alpha^2} u^\top \calL_3 Q u\,.\]
To verify the counting constraint, we have
to show that for any polynomial $p$ where $\norm{\mathcal R \paren{p}}_2 \le 1$ that $\mathcal L_3 \paren{\sum_{j \in [n]} w_j - n + T_2} p^2 \geq -4 \cdot \max(T_1, T_2) \cdot n$. Without loss of generality assume $T_2 = \max\{T_1, T_2\}$. We have
\begin{align*}
\mathcal L_3 \paren{\sum_{j \in [n]} w_j - n + T_2} p^2 
&= \lambda \mathcal L_1 \paren{\sum_{j \in [n]} w_j - n + T_2 + T_1 - T_1} p^2 + (1-\lambda) \mathcal L_2 \paren{\sum_{j \in [n]} w_j - n + T_2} p^2 \\
&\ge - 5 \lambda \tau \cdot T_1 \cdot n - 5 (1-\lambda) \tau \cdot T_2 \cdot n + \lambda  \paren{T_2 - T_1} \mathcal L_1  p^2 \\
&\ge - 5 \lambda \tau \cdot T_1 \cdot n - 5 (1-\lambda) \tau \cdot T_2 \cdot n - \lambda  \paren{T_2 - T_1} \tau \cdot T_1 \\
&\ge - 5 \lambda \tau \cdot T_1 \cdot n - 5 (1-\lambda) \tau \cdot T_2 \cdot n - 5 \lambda  \tau \cdot  \paren{T_2 - T_1} \cdot n \\
& = - 5 \tau \cdot T_2 \cdot n.
\end{align*}
This verifies that $\mathcal L_3$ is a $\paren{\alpha, \tau, \max\{T_1, T_2\}}$ certificate for $(x, y)$, as desired.
\end{proof}

\paragraph{Efficient Computability.}
\begin{lemma}[Efficient Computation of Regression Score Function]
\label{lem:efficient-comp-regression}
    Let $\tilde{\theta}$ be some point in the domain of the score function and let $T = \calS(\tilde{\theta}, \calY; \alpha, \tau)$ be the score of $\tilde{\theta}$. Then we can compute $T$ in time $\poly(n, \log(R \sqrt{\covscale}), \log(1/\gamma))$ up to accuracy $\bigO{\gamma}$.
\end{lemma}
We will show efficient computability via the ellipsoid algorithm. Note that our original robustness constraints are in the same form as the constraints in~\cite{Hopkins2023Robustness}. Their efficient computability proof (Lemma C.6 of~\cite{Hopkins2023Robustness}) mostly applies to our program, with the exception of requiring a separation oracle for the two new constraints (the $\ell_2$ constraint on $\calL \theta'$ and the closeness constraint). We give a proof of the separation oracle for these constraints in the appendix~\cref{lem:closeness-separation-oracle}.

\subsubsection{Properties of the Score Function}

\paragraph{Well-Definedness of the Score Function.}

We first show that any point must have score bounded by $n$.

\begin{lemma}[Score Function Upper Bound]
\label{lem:score-function-ub-regression}
For any input $(x, y)$ and $\tilde{\theta} \in \mathbb{B}^d \paren{2R + n \tau + \alpha \sqrt{\covscale}}$, we have $\calS\paren{\tilde{\theta}; x, y ; \alpha, \tau} \le n$, for $\calS$ as defined in \Cref{def:score_function_regression}.
\end{lemma}
\begin{proof}
It suffices to show that for $T = n$, and for any $\tilde{\theta} \in \mathbb{B}^d \paren{2R + n \tau + \alpha \sqrt{\covscale}}$, there exists a $\paren{\alpha, \tau, n}$ certificate $\mathcal L$ for $(x, y)$. 

We will define a linear functional $\calL$ which for every monomial $p$ assigns a value equal to the expectation of $p$ over a distribution supported on a single point. Furthermore, we will show that $\calL$ is feasible for our robustness system as well as the additional constraints. Let $\theta$ be $\tilde{\theta}$ projected to the solid $2R + \tau T - \psi$ radius sphere for some choice of $\psi$ to be decided later. Let $x_i' \sim N(0, 1/(2\covscale) \cdot I_d)$, $y_i' = \langle \tilde{\theta}, x_i'\rangle + \zeta_i$ where $\zeta_i \sim N(0,1)$, and $\theta' = \thetals$ for this instance while $\Sigma'$ is the empirical covariance and $Q$ is its inverse.
Furthermore, let all $w_i = 0$. We will show that with non-zero probability this certifies that $\tilde{\theta}$ has score $n$.
Note that with probability $1-\beta$ this system is feasible for the robustness constraints and $\Sigma \preceq 1/\covscale \cdot I_d \preceq \Sigma'$ (and therefore $\covscale I_d \preceq Q$). Furthermore, note that with probability $1-\beta$ we have that $\norm{\theta' - \theta} \leq \bigO{\sqrt{\covscale} \cdot \alpha}$. Thus, if we take $\psi = \sqrt{\covscale} \cdot \alpha$ we have that $\norm{\theta'} \leq 2R + \tau T$. 

It remains to show that the closeness constraint is satisfied. We have that
\[ \langle u, \thetals - \tilde{\theta}\rangle^2 \leq \langle u, \thetals - \theta\rangle^2 + \langle u, \theta - \tilde{\theta}\rangle^2  \,.\]
Note that by closeness of $\theta, \thetals$ we have that the first term is at most $\bigO{\covscale \alpha^2} \cdot \norm{u}^2 \leq \bigO{\alpha^2} \cdot u^\top Q u$. Furthermore, we have that the second term is bounded by $\bigO{\alpha^2 \covscale} \leq \bigO{\alpha^2} \cdot u^\top Q u$. Therefore, we have that the closeness constraint is satisfied and we have given a linear functional $\calL$ which certifies that $\tilde{\theta}$ has score at most $n$.
\end{proof}

\paragraph{Bounded sensitivity.}

\begin{lemma}[Bounded Sensitivity]
\label{lem:bounded-sensitivity-regression}
The score function $\calS\paren{\tilde{\theta}; x, y ; \alpha, \tau}$ as defined in \Cref{def:score_function_regression} has sensitivity $1$ with respect to the input data $(x, y)$.
\end{lemma}
\begin{proof}
Suppose $(x, y), (x', y')$ are two neighbouring datasets, and $\tilde{\theta} \in \mathbb{B}^d \paren{2 R + n \tau + \alpha \sqrt{\covscale}}$. Moreover, assume that $\calS\paren{\tilde{\theta}; x, y; \alpha, \tau} = T$. If we show that $\calS\paren{\tilde{\theta}; x', y'; \alpha, \tau}\le T + 1$, then we are done by symmetry. Since $\calS\paren{\tilde{\theta}; x, y; \alpha, \tau} = T$, there exists a $\paren{\alpha, \tau, T}$ certificate  $\mathcal L$ for $(x, y)$. We will show that there exists a $\paren{\alpha, \tau, T+1}$ certificate $\mathcal L'$ for $(x', y')$. If we show this we conclude that $\calS\paren{\tilde{\theta}; x', y'; \alpha, \tau} \le T + 1$, and we are done.

Without loss of generality, assume $(x, y)$, and $(x', y')$ differ in one point, say $(x_i, y_i)$ and $(x_i', y_i')$. We will construct $\mathcal L'$ from $\mathcal L$ by modifying the value of $\mathcal L$ on the monomials. We will show that $\mathcal L'$ is a $\paren{\alpha, \tau, T+1}$ certificate for $(x', y')$. For any monomial $p$, let $\mathcal L' p = \mathcal L p$, if $p$ does not contain $w_i$, otherwise let $\mathcal L' p = 0$. 

Now let's verify that the conditions hold for this definition of $\mathcal L' p$. Note that if $p = q + w_i r$, where $q$ does not contain $w_i$, then $\mathcal L' s p^2 = \mathcal L s q^2$, for all polynomials $s$. Moreover, $\norm{\mathcal R \paren{q}}_2 \le 1$. Therefore, all of the constraints that do not contain $w_i$ will be satisfied, with the same value $T$. Furthermore, the closeness constraint will also be satisfied since it does not depend on $w_i$.

It remains to verify the constraints that contain $w_i$ from the relaxed version of our robustness program. All of these except the counting constraint are satisfied because $\mathcal L' p = 0$ for all monomials that contain $w_i$. It remains to verify the counting constraint. We have
\begin{align*}
\mathcal L' \paren{\sum_{j \in [n]} w_j - n + T + 1} p^2 &= \mathcal L \paren{\sum_{j \neq i} w_j - n + T + 1} q^2 \\
&= \mathcal L \paren{\sum_{j \in [n]} w_j - n + T} q^2 + \mathcal L q^2 - \mathcal L w_i q^2 \\
& \ge - 5 \tau \cdot T \cdot n + \mathcal L \paren{1 - w_i} q^2  \\
& = - 5 \tau \cdot T \cdot n + \mathcal L' \paren{1 - w_j}^2 q ^2 + \mathcal L' (w_j - w_j^2)q^2.
\end{align*}
Now consider the polynomial $(1-w_j) q$, for the representation of this polynomial we have that $\norm{\mathcal R \paren{(1-w_j) q}}_2 \le 2 \norm{\mathcal R \paren{q}}_2 \le 2 $. Therefore, we have that $\mathcal L' \paren{1 - w_j}^2 q \ge - 4\tau \cdot T$. Moreover, we have that $\mathcal L' (w_j - w_j^2)q^2 \ge  -\tau T$. Therefore we have that $\mathcal L' \paren{\sum_{j \in [n]} w_j - n + T + 1} p^2 \ge - 5 \tau \cdot T \cdot n - 4 \tau \cdot T - \tau \cdot T \ge - 5\tau \cdot \paren{T+ 1} \cdot n$. This verifies that $\mathcal L'$ is a $\paren{\alpha, \tau, T+1}$ certificate for $(x', y')$, as desired.
\end{proof}

\paragraph{Efficiently Finding a Low Scoring Point.}
\begin{lemma}
\label{lem:find-low-score-regression}
There exists an algorithm which runs in time $\poly(n, d, \log(R \sqrt{\covscale})$ that outputs $\hat{\theta}$ such that $\calS(\hat{\theta}, \calY; \alpha, \tau) \leq \min_{\tilde{\theta}} \calS(\tilde{\theta}, \calY; \alpha, \tau) + 1$ and such that all $\theta$ within distance $\alpha/\sqrt{\covscale}$ of $\hat{\theta}$ have score at most $\calS(\hat{\theta}, \calY; \alpha, \tau)$.
\end{lemma}
\begin{proof}
    Consider the constraint system $\calB_T$ which is equal to $\calB_T$ without the closeness constraint. Our algorithm is as follows: we check for $T = 1, \ldots n$ whether $\calB_T$ is satisfiable. This is computable in time $\poly(n, d, \log(R \sqrt{\covscale})$ by~\cref{lem:efficient-comp-regression}. Let $T_{\min}$ be the smallest $T$ for which this is true and let $\calL$ be the corresponding linear operator. We have that $\hat{\theta}$ approximately minimizing score is just $\calL \theta'$. Note that $\calL$ also satisfies $\calB_{T_{\min}}$ with $\tilde{\theta} = \hat{\theta}$ since the additional constraint is the closeness constraint which is trivially satisfied since $\tilde{\theta} = \calL \theta'$. Thus, $\hat{\theta}$ has score at most $T_{\min}$. Furthermore, note that this also holds for all $\tilde{\theta}$ such that $\norm{\Sigma^{1/2} (\tilde{\theta} - \hat{\theta})}_2 \leq \alpha$ which includes all $\tilde{\theta}$ within distance $\alpha / \sqrt{\covscale}$.

    Furthermore, the true minimum score is at most $T_{\min} - 1$. Consider $\calL$ which certifies the minimum score, and note that $\calL$ also satisfies $\calB$ with the same value of $T$. Therefore the true minimum score must be greater than $T_{\min} - 1$ since by definition $T_{\min}$ was the smallest integer $T$ which was feasible for $\calB$.
\end{proof}

\subsubsection{Proof of Main Private Regression Theorem (\cref{thm:main_private_regression})}

We now return to the proof of the main theorem. We have shown all the preconditions for the reduction hold and thus it suffices to bound the number of samples and note that a point of $2\eta$ score has good accuracy.

\begin{proof}
    We choose the parameters in~\cref{thm:pure_dp_reduction} as $\eta = \eta(\alpha)$, $R$ to be the bound on the norm of $\theta$ (also denoted by $R$), and $r = \alpha / \sqrt{L}$.
    Assume that $2\eta \leq \eta^*$.

    Note that for the reduction to succeed it suffices to take 
    \[ n \ge \Omega\left(\max\limits_{\eta': \eta \le \eta' \le 1}\frac{\log(V_{\eta'}(\mathcal{Y})/V_{\eta}(\mathcal{Y})) + \log (1/(\beta \cdot \eta))}{\eps \cdot \eta'}\right)\,.\] 
    Let $\eta^*$ be a sufficiently small absolute constant. We have that the total volume of parameters is at most $\left(2R\right)^d$ by our bound on the domain\footnote{Note that this is technically a ball of radius $R + \alpha \sqrt{L} + \tau T$ however when $R > \Omega(\alpha \sqrt{L}$ this radius is $O(R)$ and if $R$ is small we can simply select a uniformly random point to output as our parameter and this will satisfy the accuracy guarantee.} and by~\cref{cor:volume_considerations_regression} and the upper bound on the covariance we have that the volume of low scoring points is at least $(\eta/\sqrt{\covscale})^d$. Thus,
    \begin{align*}
        \max\limits_{\eta': \eta^* \le \eta' \le 1}\frac{\log(V_{\eta'}(\mathcal{Y})/V_{\eta}(\mathcal{Y})) + \log (1/(\beta \cdot \eta))}{\eps \cdot \eta'} &\leq  \frac{\log(R^d/(\eta/\sqrt{\covscale})^d) + \log (1/(\beta \cdot \eta))}{\eps} \\
        &\leq \frac{d \log(R \sqrt{\covscale} / \eta) + \log (1/(\beta \cdot \eta))}{\eps}\,. \\
    \end{align*}
    Furthermore, note that by invoking~\cref{cor:volume_considerations_regression} we have that
    \begin{align*}
        \max\limits_{\eta': \eta \le \eta' \le \eta^*}\frac{\log(V_{\eta'}(\mathcal{Y})/V_{\eta}(\mathcal{Y})) + \log (1/(\beta \cdot \eta))}{\eps \cdot \eta'} &\leq \bigO{\frac{\log \left( (1/\eta)^d \right) + \log (1/(\beta \cdot \eta))}{\eps \cdot \eta}} \\
        &\leq \bigO{\frac{d \log(1/ \eta) + \log (1/(\beta \cdot \eta))}{\eps \cdot \eta}}\,.
    \end{align*}
    Combining these two bounds we have that
    \[\max\limits_{\eta': \eta \le \eta' \le 1}\frac{\log(V_{\eta'}(\mathcal{Y})/V_{\eta}(\mathcal{Y})) + \log (1/(\beta \cdot \eta))}{\eps \cdot \eta'} \leq \frac{d \log (R \sqrt{\covscale})}{\eps} + \frac{d \log (1 /\eta)}{\eps \cdot \eta}  + \frac{\log (1/(\beta \cdot \eta))}{\eps \cdot \eta}\,.\]
    Furthermore, we have that the algorithm outputs a point with score at most $2 \eta n$. By~\cref{lem:closeness_constraint_soundness} we have that this implies that the outputted point $\hat{\mu}$ satisfies
    \[ \norm{\Sigma^{1/2} (\hat{
    \mu} - \mu)}_2 \leq \bigO{\alpha}\,.\]
\end{proof}

\section{Private and Robust Covariance-Aware Mean Estimation}
\label{sec:mean_estimation}

The goal of this section is to show our main theorem for private covariance-aware mean estimation.
\begin{theorem}[Private Covariance-Aware Mean Estimation, Full Version of~\cref{thm:informal-mean-est}]
\torestate{
\label{thm:main_private_cov_mean_est}
Let $\mu \in \R^d$ such that $\norm{\theta} \leq R$ and $\Sigma$ such that $\tfrac 1 {\covscale} I_d \preceq \Sigma$.
Let $0 <\eta$ be less than a sufficiently small constant and let $0 < \alpha, \beta, \eps$ and $\alpha < 1$.
Let $x_1^*, \ldots, x_n^*$ be $n \geq n_0$ i.i.d. samples from $N(\mu, \Sigma)$.\footnote{More generally, they can be i.i.d. samples from a distribution $\cD$ such that $\Sigma^{-1/2}(\cD - \mu)$ is fourth moment reasonable sub-gaussian.}
Let $\cX = \Set{x_1, \ldots, x_n}$ be an $\eta$-corruption of $x_1^*, \ldots, x_n^*$.
There exists an  $(\eps, 0)$-differentially private algorithm that, given $\eta, \alpha, \eps$, and $\cX$, runs in time $\poly(n,\log L, \log R)$ and with probability at least $1-\beta$ outputs an estimate $\muhat$ satisfying
    \[ \Norm{\Sigma^{-1/2}\left(\muhat - \mu\right)} \leq \bigO{ \alpha } \,,
    \]
    whenever 
    $$n_0 = \tilde{\Omega}\Paren{ \frac{d^2 + \log^2(1/\beta)}{\alpha^2} + \frac{d + \log(1/\beta)}{\alpha \eps} + \frac{d \log(R) + d \log(\covscale)}{\eps} }, $$
and $\alpha \geq \Omega( \eta \log(1/\eta)) $.
}
\end{theorem}

Again, we first obtain a robust estimator achieving (nearly) optimal error.
In particular, we prove the following:
\begin{theorem}[Robust Covariance-Aware Mean Estimation]
\label{thm:mean-est-sample-opt}
Let $\cD$ be a distribution with mean $\mu$ and covariance $\Sigma$ such that $\Sigma^{-1/2}(\cD - \mu)$ is fourth moment matching reasonable sub-gaussian.
Given $0 < \eta$ smaller than a sufficiently small constant$, 0 < \beta$ and $n \geq n_0$ points that are an $\eta$-corruption of $n$ i.i.d.\ samples from $\cD$, there exists an algorithm that runs in time $(n \cdot d)^{\mathcal{O}(1)}$ and with probability at least $1-\beta$, outputs $\muhat$ such that 
\begin{equation*}
    \norm{ \Sigma^{-1/2}\left(\muhat - \mu \right) }_2 \leq \bigO{\eta \sqrt{\log(1/\eta)}} \,,
\end{equation*}
whenever 
\begin{equation*}
    n_0  = \tilde{\Omega}\Paren{ \frac{d^2 + \log^{2}(1/\beta)}{ \eta^2  } }.
\end{equation*}
\end{theorem}

We will use the following notation:
We denote the uncorrupted samples by $\{x^*_1, x^*_2 , \ldots x^*_n\}$.
We denote by $\{x_1, x_2 , \ldots x_n\}$ the $\eta$-corrupted version of $\{x_1^*, \ldots, x_n^*, \}$ we receive as input.
Further, let $r_i \in \{0,1\}$, be 1 if the $i$-th sample is uncorrupted, i.e., $x_i^* = x_i$, and 0 otherwise.

We will use the following constraint system in vector-valued variables $x_1', \ldots, x_n', \mu'$, scalar-valued variables $w_i$, and matrix valued variables $\Sigma'$.

\begin{equation*}
\calA_{\eta}\colon
  \left \{
    \begin{aligned}
      &\forall i\in [n].
      & w_i^2
      & = w_i \\
      & &\sum_{i\in [n]} w_i &= (1-\eta)n \\
      &\forall i\in [n] & w_i(x'_i - x_i) &=0 \\
      &&\tfrac{1} n \sum_{i \in [n]} x_i' &= \mu' \\
    && \frac{1}{n} \sum_{i \in [n]} \Paren{ x_i' - \mu' } \Paren{ x_i' - \mu' }^\top & = \Sigma'  \\ 
    & \exists \text{ the following SoS proof }  & \sststile{4}{v} \Biggl\{\frac{1}{n} \sum_{i \in [n]} \Iprod{x_i'  - \mu', v}^4  \leq & \Paren{3 + \eta \log^2(1/\eta)} \Paren{ v^\top \Sigma' v}^2 \Biggr\}
    \end{aligned}
  \right \}
\end{equation*}

Note that for the inequalities that must be satisfied for all $v$, we require that there is an SoS proof of the inequality for any $v$. These constraints can be efficiently encoded as polynomial inequalities in our program variables by using auxiliary variables (see~\cite{fleming2019semialgebraic}).

\paragraph{Feasibility.} The feasibility of the above program follows by setting $w_i = r_i$ and $\eta$-goodness and $\eta$-higher-order-goodness of the true samples.
We will give a formal proof in \cref{sec:feasibility}.

\begin{mdframed}
  \begin{algorithm}[Optimal Robust Mean Estimation in Polynomial Time]
    \label{algo:mean_estimation}\mbox{}
    \begin{description}
    \item[Input:] $\eta$-corrupted sample $x_1, \ldots, x_n$, corruption level $\eta$, failure probability $\beta$.
    \item[Output:] An estimate $\muhat$ of the mean attaining the guarantees of \cref{thm:mean-est-sample-opt}.
    
    \item[Operations:]\mbox{}
    \begin{enumerate}
    \item Find a degree-$\bigO{1}$ pseudo-distribution $\zeta$ satisfying $\calA_{\eta}$.
    \item Output $ \muhat = \pE_{\zeta}[\muprime]$.
    \end{enumerate}
    \end{description}
  \end{algorithm}
\end{mdframed}

Our main technical lemma will be the following
\begin{lemma}[SoS Mean to True Mean]
\label{lem:main_identity_mean_sos_proof}
Suppose $0 < \eta$ is smaller than a sufficiently small constant.
Let $x_1^*, \ldots, x_n^*$ be the uncurrupted samples and assume $\Sigma^{-1/2}(x_1^* - \mu), \ldots, \Sigma^{-1/2}(x_n^* - \mu)$ are $\eta$-good and $\eta$-higher-order-good.
Then, for any fixed $u \in \mathbb{R}^d$ the following SoS proof exists
\begin{align*}
    \calA_\eta \sststile{6}{w} \Set{ \Iprod{ \Sigma^{-1/2} \Paren{ \mu' - \mu } ,  u }^4 \leq \bigO{\eta^4 \log^2 (1/ \eta) }  \norm{u}^4 } \,.
\end{align*}
\end{lemma}

Given the aforementioned lemmas, the proof of \cref{thm:mean-est-sample-opt} is fairly straight-forward. 
\begin{proof}[Proof of \cref{thm:mean-est-sample-opt}]
    By~\cref{fact:mean_and_cov_of_1-eps,fact:cov_small_subset} and~\cref{lem:higher-order-stability} we have that the $\Sigma^{-1/2}(x_i^* - \mu)$ are $\eta$-good and $\eta$-higher-order good with probability at least $1-\beta$.
    We henceforth condition on this event.
    Let $\zeta$ be a degree-$\bigO{1}$ pseudo-expectation satisfying $\calA_{\eta}$.
    By~\cref{fact:eff-pseudo-distribution}, this can be computed in time $(n \cdot d)^{\mathcal{O}(1)}$.
    Let $\muhat = \pE_\zeta \muprime$ and $u = \tfrac{ \Sigma^{-1/2} \Paren{ \muhat - \mu  } }{\Norm{\Sigma^{-1/2} \Paren{ \muhat - \mu  }}}$.
    Then it follows by \cref{lem:main_identity_mean_sos_proof} that
    \begin{equation*}
    \begin{split}
        \Norm{ \Sigma^{-1/2} \Paren{ \muhat - \mu  }  }^4 = \Iprod{ \Sigma^{-1/2} \Paren{ \muhat - \mu  } , u}^4 & = \Paren{\pE_\zeta \Iprod{ \Sigma^{-1/2} \Paren{ \mu' - \mu  } , u}}^4 \\
        & \leq \pE_\zeta \Iprod{\Sigma^{-1/2} \Paren{ \mu' - \mu  } , u}^4 \\
        & \leq \bigO{\eta^4 \log^2 (1/ \eta)}  \,,
    \end{split}
    \end{equation*}
    where the inequality follows by Cauchy-Schwarz for pseudo-expectations (cf.~\cref{fact:pe_cs}).
    \cref{thm:mean-est-sample-opt} then follows by taking fourth roots.
\end{proof}

\subsection{Proof of \cref{lem:main_identity_mean_sos_proof} }

Our main technical tool for proving \cref{lem:main_identity_mean_sos_proof} will be the following simple fact which allows us to port the notion of $\eta$-goodness (cf.~\cref{def:eps_goodness}) to the sum-of-squares framework. We show that if all subsets of size $\eta n$ are well-behaved for the true samples, then a sum-of-squares indicated subset also inherits these properties.
\begin{lemma}
\label{lemma:small_subset}
Suppose $\Sigma^{-1/2}(x_1^* - \mu), \ldots, \Sigma^{-1/2}(x_n^* - \mu)$ are $\eta$-good.
Then, for any $u\in\mathbb{R}^d$, it holds that:
\begin{align}
    \label{eqn:small_subset_upper_bound}
    \mathcal{A} \sststile{2}{w_1,\ldots,w_n  } \Set{ \frac{1}{n}\sum_{i \in [n]} r_i \Paren{1 - w_i} \Iprod{ \Sigma^{-1/2}\Paren{ x_i^* - \mu} , u}^2 \leq \bigO{\eta \log(1/\eta)  } \norm{u}_2^2  },
\end{align}
where $r_i = 1$ if $x_i^* = x_i$ and $0$ otherwise. Further,
\begin{align}
\label{eqn:small_subset_lower_bound}
    \mathcal{A} \sststile{2}{w_1,\ldots,w_n} \Set{ \frac{1}{n}\sum_{i \in [n]} r_i w_i \Iprod{ \Sigma^{-1/2}\Paren{ x_i^* - \mu}  , u}^2 \geq (1 - \bigO{\eta \log(1/\eta)}  )  \norm{u}_2^2 } \,.
\end{align}
\end{lemma}

\begin{proof}
Let $n' = (1-\eta)n$
We start with \cref{eqn:small_subset_upper_bound}.
Since $r_i \leq 1$, 
\begin{equation}
    \sststile{4}{w} \Set{ \frac{1}{n}\sum_{i \in [n]} r_i \Paren{1 - w_i} \Iprod{ \Sigma^{-1/2}\Paren{ x_i^* - \mu } , u}^2 \leq  \frac 1 {n'} \sum_{i \in [n]} \Paren{1 - w_i} \Iprod{ \Sigma^{-1/2}\Paren{ x_i^* - \mu}  , u}^2  } \,.
\end{equation}

Further, by $\eta$-goodness of $\Sigma^{-1/2}(x_1^* - \mu), \ldots, \Sigma^{-1/2}(x_n^* - \mu)$ (cf.~\cref{fact:cov_small_subset}), we know that for every fixed set $T \subseteq [n]$ of size $\eta n$, it holds that
\begin{equation}
    \begin{split}
        \sststile{4}{w} \Biggl\{  \frac{1}{n'}  \sum_{i \in T} \Iprod{\Sigma^{-1/2}(x_i^* - \mu) , u}^2 & \leq\Norm{\frac{1}{n'} \sum_{i \in T} \left(\Sigma^{-1/2}(x_i^* - \mu)\right) \left(\Sigma^{-1/2}(x_i^* - \mu)\right)^\top}_2 \cdot \norm{u}^2_2 \\
        & \leq \bigO{\eta \log (1/\eta)} \norm{u}_2^2 \Biggr\}  \,,
    \end{split}
\end{equation}
Setting $a_i = \Iprod{\Sigma^{-1/2}(x_i^* - \mu), u}^2$ and using that $\calA_{\eta}$ implies at degree 2 that $\sum_{i=1}^n (1-w_i) = \eta n$ and $0 \leq w_i \leq 1$ the bound of \cref{eqn:small_subset_upper_bound} then follows by Lemma~\ref{lemma:subset_selection}.

Next, we prove \cref{eqn:small_subset_lower_bound}.
Similarly to the above and using that $\sum_{i=1}^n 1 - r_i w_i \leq 2 \eta n$, it follows that 
\begin{equation}
    \sststile{4}{w} \Set{ \frac{1}{n}\sum_{i \in [n]} (1- r_i w_i) \Iprod{\Sigma^{-1/2}(x_i^* - \mu),u}^2 \leq \bigO{\eta \log (1/\eta)} \norm{u}_2^2 \, } .
\end{equation}
Using standard concentration bounds (cf. \cref{fact:empirical_cov_bound}) it also follows that
\[
    \sststile{4}{w} \Set{ \frac{1}{n}\sum_{i \in [n]} \Iprod{\Sigma^{-1/2}(x_i^* - \mu),u}^2\geq \left(1 - \eta \log (1/\eta)\right)\norm{u}_2^2} \,.
\]
We can now see that there is a degree-2 proof that $\calA$ implies
\begin{align*}
    \frac{1}{n}\sum_{i \in [n]} r_i w_i \Iprod{\Sigma^{-1/2}(x_i^* - \mu), u}^2 &= \frac{1}{n}\sum_{i \in [n]} \Iprod{\Sigma^{-1/2}(x_i^* - \mu), u}^2- \frac{1}{n'}\sum_{i \in [n]} (1-r_i w_i) \Iprod{\Sigma^{-1/2}(x_i^* - \mu), u}^2  \\
    &\geq \left(1 - 2\eta \log (1/\eta)\right)\norm{u}_2^2 \,,
\end{align*}
where the last inequality follows from \cref{eqn:small_subset_upper_bound}.
Abusing the $\bigO{\cdot}$-notation, we write $2\eta \log(1/\eta) = \eta \log(1/\eta)$.
\end{proof}

Furthermore, we will also use an intermediate bound on the spectral error of $\Sigma'$, which we prove in~\cref{sec:cov-sos-for-mean}.
\begin{corollary}
\torestate{
\label{cor:cov_mean_sos}
    Suppose $0 < \eta$ is smaller than a sufficiently small constant and assume that $\Sigma^{-1/2}(x_1^* - \mu), \ldots, \Sigma^{-1/2}(x_n^* - \mu)$ are $\eta$-good and $\eta$-higher-order-good.
    Then, for any fixed $u \in \mathbb{R}^d$,
    \begin{align*}
    \calA_\eta \sststile{\bigO{1}}{} \Biggl\{ \eta^2 \Iprod{ \Sigma' - \Sigma,   uu^\top }^2  &\leq \frac{1}{2}\Iprod{\mu' - \mu, u}^4 +   \bigO{\eta^4 \log^2 (1/\eta)}  \Paren{ ( u^\top \Sigma u )^2} \Biggr\} \,.
    \end{align*}
}
\end{corollary}

We can now prove \cref{lem:main_identity_mean_sos_proof} as follows: 
\begin{proof}[Proof of \cref{lem:main_identity_mean_sos_proof}]
We will first prove the following SoS inequality: 
\[
    \calA_\eta \sststile{6}{w} \Set{ \Iprod{ \Sigma^{-1/2} \Paren{ \mu' - \mu } ,  u }^2 \leq \bigO{\eta^2 \log (1/ \eta) }  \norm{u}^2  + \bigO{\eta}\Iprod{\Sigma' - \Sigma, (\Sigma^{-1/2}u)(\Sigma^{-1/2}u)^\top}} \,.
\]
Let $n' = (1-\eta)n$.
Note that for all $i \in [n]$ $r_ix_i = r_i x_i^*$.
Further, our constraints imply, at constant degree, that $w_i x_i' = w_i x_i$ and hence, also that $r_i w_i x_i' = r_i w_i x_i^*$.
Using this and the SoS triangle inequality (cf.~\cref{fact:sos-almost-triangle}), we start to bound
\begin{equation}
\label{eqn:quantity_of_interest}
\begin{split}
     \calA_\eta \sststile{2}{u,w} \Biggl\{ & \Iprod{ \Sigma^{-1/2}\Paren{ \mu' - \mu } ,u }^2 \\
     & =  \Iprod{ \frac{1}{n}\sum_{i \in [n]} \Sigma^{-1/2} \Paren{ x'_i - \mu } ,u }^2 = \Paren{ \frac{1}{n}\sum_{i \in [n]} \Iprod{ \Sigma^{-1/2} \Paren{ x_i' - \mu} , u} }^2\\
     & \leq2 \underbrace{\Paren{ \frac{1}{n}\sum_{i \in [n]} r_i w_i\Iprod{ \Sigma^{-1/2} \Paren{ x_i' - \mu}  , u} }^2}_{\ref{eqn:quantity_of_interest}.1} + 2 \underbrace{\Paren{ \frac{1}{n}\sum_{i \in [n]} (1-r_i w_i)\Iprod{ \Sigma^{-1/2} \Paren{ x_i' - \mu}  , u} }^2}_{\ref{eqn:quantity_of_interest}.2} \Biggr\}
\end{split}
\end{equation}

We bound the above two terms separately.
For the first one, using that $r_i^2 = r_i$, $w_i^2 = w_i$ and $r_i w_i x_i' = r_i w_i x_i^*$, we observe
\begin{equation}
\label{eqn:mean_of_sos_variables_around_mu}
    \begin{split}
        \calA_\eta \sststile{4}{w} \Biggl\{  & \Paren{ \frac{1}{n}\sum_{i \in [n]} r_i w_i\Iprod{ \Sigma^{-1/2}\Paren{ x_i' - \mu } , u} }^2  \\
        &  = \Paren{ \frac{1}{n}\sum_{i \in [n]}r_i w_i \Iprod{ \Sigma^{-1/2} \Paren{ x_i^* -\mu } , u} }^2 \\
        & \leq 2 \Paren{ \frac{1}{n}\sum_{i \in [n]} r_i \Iprod{ \Sigma^{-1/2}\Paren{ x_i^* - \mu } , u} }^2 + 2 \Paren{ \frac{1}{n}\sum_{i \in [n]}r_i \Paren{1 - w_i} \Iprod{ \Sigma^{-1/2}\Paren{ x_i^* - \mu } , u} }^2 \\
        & \leq \frac{2n'}{n}  \Iprod{ \Sigma^{-1/2}\Paren{ \frac{1}{n'}\sum_{i \in [n]} r_i  x^*_i - \mu } , u }^2 +  2\Paren{ \frac{1}{n}\sum_{i \in [n]}r_i \Paren{1 - w_i} \Iprod{ \Sigma^{-1/2}\Paren{ x_i^* - \mu}  , u} }^2\\
        & \leq  \bigO{\eta^2 \log(1/\eta)  } \norm{u}^2   +  2 \underbrace{\Paren{ \frac{1}{n}\sum_{i \in [n]}r_i \Paren{1 - w_i} \Iprod{\Sigma^{-1/2}\Paren{ x_i^* - \mu}  , u} }^2}_{\ref{eqn:mean_of_sos_variables_around_mu}} \Biggr\}
    \end{split}
\end{equation}
where the last inequality follows from $\eta$-goodness of $\Sigma^{-1/2} (x_1^* - \mu), \ldots, \Sigma^{-1/2} (x_n^* - \mu)$ (cf.~\cref{fact:mean_and_cov_of_1-eps}).
For Term \ref{eqn:mean_of_sos_variables_around_mu} it follows by the SoS version of Cauchy-Schwarz (cf.~\cref{fact:sos_cs}) and \cref{eqn:small_subset_upper_bound} in \cref{lemma:small_subset} that
\begin{align*}
    \calA_\eta \sststile{4}{w} \Biggl\{ \Paren{ \frac{1}{n}\sum_{i \in [n]}r_i \Paren{1 - w_i} \Iprod{x_i^* - \mu , u} }^2 &\leq \eta \cdot \Paren{\frac{1}{n}\sum_{i \in [n]}r_i \Paren{1 - w_i} \Iprod{ \Sigma^{-1/2}\Paren{ x_i^* - \mu} , u}^2} \\
    &\leq \bigO{\eta^2 \log(1/\eta)  }\norm{u}^2  \Biggr\} \,.
\end{align*}
Note that we also used that $r_i^2(1-w_i)^2 = r_i (1-w_i)$.
Hence, our constraints imply at degree 2 that Term~\ref{eqn:quantity_of_interest}.1 is at most $\bigO{\eta^2\log(1/\eta)} \norm{u}^2$.

Next, we turn to bounding Term~\ref{eqn:quantity_of_interest}.2.
Again by the SoS version of the Cauchy Schwarz Inequality, it follows that
\begin{equation}
\label{eqn:bad_set}
\begin{split}
    \calA_\eta \sststile{6}{w} \Biggl\{ &\Paren{ \frac{1}{n}\sum_{i \in [n]} (1-r_i w_i)\Iprod{ \Sigma^{-1/2}\Paren{ x_i' - \mu } , u} }^2 \\
    & \leq \eta \Paren{\frac{1}{n}\sum_{i \in [n]} (1-r_i w_i)\Iprod{ \Sigma^{-1/2}\Paren{ x_i' - \mu} , u}^2 } \\
    &= \eta \Bigg(  \underbrace{  \frac{1}{n}\sum_{i \in [n]} \Iprod{ \Sigma^{-1/2}\Paren{ x_i' - \mu} , u}^2 }_{ \eqref{eqn:bad_set}.1 } - \underbrace{ \frac{1}{n}\sum_{i \in [n]} r_iw_i\Iprod{ \Sigma^{-1/2}\Paren{ x_i' - \mu } , u}^2 }_{\eqref{eqn:bad_set}.2} \Bigg) \Biggr\} \,.
\end{split}
\end{equation}
Expanding Term \eqref{eqn:bad_set}.1 and using that $\tfrac 1 n \sum_{i \in [n]}x_i' = \mu'$ it follows that 
\begin{equation}
\label{eqn:bounding-bad-set-part-1}
    \begin{split}
        \calA_\eta \sststile{6}{w} \Biggl\{ &   \frac{1}{n}\sum_{i \in [n]} \Iprod{ \Sigma^{-1/2}\Paren{ x_i' - \mu } , u}^2 \\
        & =   \frac{1}{n}\sum_{i \in [n]} \Iprod{ \Sigma^{-1/2} \Paren{ x_i' - \mu'  + (\mu'- \mu) } , u}^2  \\
    &=  \Paren{ \frac{1}{n} \sum_{i \in [n]} \Iprod{  \Sigma^{-1/2}\Paren{ x_i' - \mu' }  , u}^2 }  + \Iprod{ \Sigma^{-1/2}\Paren{ \mu' - \mu} ,u }^2 \\
    & \hspace{0.2in} +  \frac{2}{n} \sum_{i \in [n]} \Iprod{ x_i' - \mu', \Sigma^{-1/2}u } \Iprod{\mu' - \mu , \Sigma^{-1/2} u}  \\
    & = \frac{1}{n} \sum_{i \in [n]} \Iprod{ \Sigma^{-1/2} \Paren{ x_i' - \mu' } \Paren{ x_i' - \mu' }^\top \Sigma^{-1/2} , u u^\top  }    + \Iprod{ \Sigma^{-1/2} \Paren{ \mu' - \mu } ,u }^2 \\
    & = \Iprod{ \Sigma^{-1/2} \Sigma' \Sigma^{-1/2} , uu^\top  } + \Iprod{ \Sigma^{-1/2} \Paren{ \mu' - \mu } ,u }^2  \\
    & = \norm{u}_2^2 + \Iprod{\Sigma' - \Sigma, (\Sigma^{-1/2} u)(\Sigma^{-1/2}u)^\top} + \Iprod{ \Sigma^{-1/2}\Paren{ \mu' - \mu} , u }^2 
     \Biggr\} \,.
    \end{split}
\end{equation}
 Next, using \cref{eqn:small_subset_lower_bound} in \cref{lemma:small_subset} and that $r_i w_i x_i' = r_i w_i x_i^*$, 

\begin{equation}
\label{eqn:bounding-bad-set-part-2}
    \begin{split}
        \calA_\eta \sststile{6}{w} \Biggl\{   \frac{1}{n}\sum_{i \in [n]} r_iw_i\Iprod{ \Sigma^{-1/2}\Paren{ x_i' - \mu } , u}^2   &  =   \frac{1}{n}\sum_{i \in [n]} r_i w_i\Iprod{ \Sigma^{-1/2}\Paren{ x_i^* - \mu } , u}^2 \\
        & \geq \Paren{1- \bigO{\eta \log(1/\eta) } }\norm{u}_2^2
     \Biggr\},
    \end{split}
\end{equation}
Substituting the bounds obtain in \eqref{eqn:bounding-bad-set-part-1} and \eqref{eqn:bounding-bad-set-part-2}, back into Equation \eqref{eqn:bad_set}, we have the following bound on Term~\ref{eqn:quantity_of_interest}.2
\begin{equation}
\begin{split}
\label{eqn:bound-on-bad-set}
    \calA_\eta \sststile{6}{w} \Biggl\{  &\Paren{ \frac{1}{n}\sum_{i \in [n]} (1-r_i w_i)\Iprod{\Sigma^{-1/2}\Paren{x_i' - \mu} , u} }^2  \leq \bigO{ \eta^2 \log(1/\eta)}  \norm{u}^2 \\
    &+ \bigO{\eta}\Iprod{\Sigma' - \Sigma, (\Sigma^{-1/2} u)(\Sigma^{-1/2}u)^\top} + \bigO{\eta}\Iprod{ \Sigma^{-1/2}\Paren{ \mu' - \mu} , u }^2\Biggr\}
\end{split}
\end{equation}
Combining Equation \eqref{eqn:bound-on-bad-set} with our bound on Term~\ref{eqn:quantity_of_interest}.1 implies that
\begin{align*}
    \calA_\eta \sststile{6}{w} \Biggl\{ \Iprod{ \Sigma^{-1/2 }\Paren{ \mu' - \mu } ,  u }^2 &\leq \bigO{\eta^2 \log (1/ \eta) }  \norm{u}^2 + \eta \Iprod{\Sigma' - \Sigma, (\Sigma^{-1/2} u)(\Sigma^{-1/2}u)^\top} \\
    &+ \eta \cdot \Iprod{ \Sigma^{-1/2 }\Paren{ \mu' - \mu} , u}^2\Biggr\} \,.
\end{align*}

\paragraph{Applying the Covariance Bound.}

Rearranging and renormalizing yields the intermediate SoS inequality we aimed to prove. We now show how to use this SoS inequality to derive~\cref{lem:main_identity_mean_sos_proof}. Note that the RHS is a square and thus there exists an SoS proof that there is non-negative. Thus, applying~\cref{fact:sos-squaring} we have that 
\begin{align*}
    \calA_\eta \sststile{12}{w} \Biggl\{ \Iprod{ \Sigma^{-1/2}\Paren{ \mu' - \mu } ,  u }^4 \leq \Paren{\bigO{\eta^2 \log (1/ \eta) }  \norm{u}^2 + \eta \Iprod{\Sigma' - \Sigma, (\Sigma^{-1/2} u)(\Sigma^{-1/2}u)^\top} }^2 \Biggr\} \,.
\end{align*}
Applying SoS triangle inequality and~\cref{cor:cov_mean_sos} we have that 
\begin{align*}
    \calA_\eta \sststile{12}{w} \Biggl\{ \Iprod{ \Sigma^{-1/2} \Paren{ \mu' - \mu },  u }^4 &\leq \bigO{\eta^4 \log^2 (1/ \eta) }  \norm{u}^4 + \eta^2 \Iprod{\Sigma' - \Sigma, (\Sigma^{-1/2} u)(\Sigma^{-1/2}u)^\top}^2 \\
    &\leq \bigO{\eta^4 \log^2 (1/ \eta)}  \norm{u}^4 + \frac{1}{2} \Iprod{ \Sigma^{-1/2} \Paren{ \mu' - \mu }, u}^4 \Biggr\} \,.
\end{align*}
Once again rearranging and renormalizing completes the proof.
\end{proof}

\subsection{Intermediate Covariance Guarantees}
\label{sec:cov-sos-for-mean}

We now show a key SoS lemma and a corollary which is used to bound the difference between $\Sigma$ and $\Sigma'$ in the above proof.

\begin{lemma}[SoS Covariance to True Covariance]
\label{lem:cov_mean_sos}
Suppose $0 < \eta$ is smaller than a sufficiently small constant and assume $\Sigma^{-1/2}(x_1^* - \mu), \ldots, \Sigma^{-1/2}(x_n^* - \mu)$ are $\eta$-good and $\eta$-higher-order-good.
Then, for any fixed $u \in \mathbb{R}^d$,
\begin{align*}
    \calA_\eta \sststile{6}{w} \Biggl\{ \Iprod{ \Sigma' - \Sigma,   uu^\top }^2  &\leq \bigO{1}\Iprod{\mu' - \mu, u}^4 +   \bigO{\eta^2 \log^2 (1/\eta)}  \Paren{ ( u^\top \Sigma u )^2 + ( u^\top \Sigma' u )^2} \\
    &+ \frac{2}{n} \sum_{i \in [n]} \eta \Iprod{x_i' -\mu', u}^3\Iprod{\mu-\mu',u} \Biggr\} \,.
\end{align*}
\end{lemma}
Before we prove this lemma we will prove an easy corollary of it which is used in the previous section.
\restatecorr{cor:cov_mean_sos}

\begin{proof}
    Note that after applying~\cref{lem:cov_mean_sos} and multiplying by $\eta^2$ it suffices to bound the following expression:
    \[  \bigO{\eta^4 \log^2 (1/\eta)} \cdot ( u^\top \Sigma' u )^2 + \frac{2}{n} \sum_{i \in [n]} \eta^3 \Iprod{x_i' -\mu', u}^3\Iprod{\mu-\mu',u}\,.\]
    We first consider the second term. We have by~\cref{fact:cov-hack-tmp} with $a=10\eta\langle x_i' - \muprime, u\rangle$ and $b = \frac{1}{1000} \langle \mu - \mu', u\rangle$ that
    \[ \sststile{}{} \Set{\frac{2}{n} \sum_{i \in [n]} \eta^3 \Iprod{x_i' -\mu', u}^3\Iprod{\mu-\mu',u} \leq \frac{1}{10^{12}} \Iprod{\mu-\mu',u}^4 + \frac{10000\eta^4}{n} \sum_{i \in [n]} \Iprod{x_i' -\mu', u}^4}\,.\]
    Applying the bounded fourth moment constraint we have that $\frac{\eta^4}{n} \sum_{i \in [n]}  \Iprod{x_i' -\mu', u}^4 \leq 4 \eta^4 \left(u^\top \Sigma' u\right)^2$ and thus it just remains to bound $\bigO{\eta^4 \log^2 (1/\eta)} \cdot ( u^\top \Sigma' u )^2$.

    Note that we have that $u^\top \Sigma' u = u^\top \Sigma u + u^\top \left(\Sigma - \Sigma'\right) u$. Plugging this into the original expression and applying SoS Triangle Inequality (\cref{fact:sos-almost-triangle}) we have that
    \begin{align*}
    \calA_\eta \sststile{}{} \Biggl\{ \eta^2 \Iprod{ \Sigma' - \Sigma,   uu^\top }^2  &\leq \frac{1}{100}\Iprod{\mu' - \mu, u}^4 +   \bigO{\eta^4 \log^2 (1/\eta)}  \Paren{ ( u^\top \Sigma u )^2} \\
    &+ \bigO{\eta^4 \log^2 (1/\eta)}  \Iprod{ \Sigma' - \Sigma,   uu^\top }^2 \,.
    \end{align*}
    Since $\eta^4 \log^2(1/\eta) \leq \eta^2$ when $\eta < 1/e$ we can move the last term to the LHS and renormalize, yielding the claim.
\end{proof}

We now prove~\cref{lem:cov_mean_sos}.
\begin{proof}
Let $r_i$ be the indicators for the uncorrupted samples where $x_i^* = x_i$. Using SoS almost triangle inequality (\cref{fact:sos-almost-triangle}) we have that
\begin{equation}
\label{eqn:cov-spectral-expansion-1}
\begin{split}
    \calA_\eta  \sststile{}{} \Biggl\{ & \Iprod{ \Sigma' - \Sigma , uu^\top }^2  \\
    &= \Paren{  \frac{1}{n } \sum_{i \in [n]} (1 - r_i w_i + r_i w_i) \Paren{ \Iprod{ x_i'  -\mu',u }^2 -  u^\top \Sigma u }  }^2    \\ 
    & \leq 2 \underbrace{ \Paren{  \frac{1}{n } \sum_{i \in [n]}  r_i w_i \Paren{ \Iprod{ x_i' -\mu' ,u }^2 -  u^\top \Sigma u }  }^2}_{\eqref{eqn:cov-spectral-expansion-1}.(1) }  + 2 \underbrace{ \Paren{ \frac{1}{n } \sum_{i \in [n]}  (1-r_i w_i) \Paren{ \Iprod{ x_i' - \mu'  ,u }^2 -  u^\top \Sigma u } }^2}_{\eqref{eqn:cov-spectral-expansion-1}.(2) }   \Biggr\}
\end{split}
\end{equation}
We now handle each term separately, starting with the first one. We have that
\begin{equation}
\label{eqn:expanding-the-first-spect-cov-term}
\begin{split}
    \calA_{\eta}\sststile{}{} \Biggl\{ & \Paren{  \frac{1}{n } \sum_{i \in [n]}  r_i w_i \Paren{ \Iprod{ x_i' - \mu' ,u }^2 -  u^\top \Sigma u }  }^2  \\
    & = \Paren{  \frac{1}{n } \sum_{i \in [n]}  r_i w_i \Paren{ \Iprod{ x_i^* - \mu' \pm \mu  ,u }^2 -  u^\top \Sigma u }  }^2  \\
    & = \Paren{  \frac{1}{n } \sum_{i \in [n]}  r_i w_i \Paren{ \Iprod{ x_i^*  - \mu  ,u }^2 -  u^\top \Sigma u  + \Iprod{\mu - \mu', u}^2 - 2 \Iprod{x_i^* - \mu, u} \Iprod{\mu - \mu', u} }   }^2 \\
    & \leq 8 \Paren{  \frac{1}{n } \sum_{i \in [n]}  r_i w_i \Paren{ \Iprod{ x_i^*  - \mu  ,u }^2 -  u^\top \Sigma u    }  }^2 +  8 \Iprod{\mu - \mu', u}^4  \\
    & \hspace{0.2in }+ 16 \Iprod{\mu - \mu', u}^2  \Paren{  \Iprod{ \frac{1}{n}\sum_{i\in [n]}  r_iw_i (x_i^* - \mu ) , u}  }^2  \\
    & =  \bigO{\eta^2 \log^2(1/\eta) } (u^\top \Sigma u)^2  + \bigO{1} \Iprod{\mu - \mu', u}^4 \Biggr\} \,,
\end{split}
\end{equation}
where the last inequality follows from the SoS triangle inequality.
Note that the terms $\Paren{ \Iprod{ x_i^*  ,u }^2 -  u^\top \Sigma u }^2$ and $\Iprod{x_i^* - \mu, u}$  are fixed scalars and thus we can apply SoS subset selection (c.f.~\cref{lemma:subset_selection}) together with $\eta$-goodness and $\eta$-higher-order-goodness.
We obtain that $\frac{1}{n } \sum_{i \in [n]}  r_i w_i \Paren{ \Iprod{ x_i^*  - \mu  ,u }^2}$ is in between $(\tfrac 1 n \sum_{i=1}^n r_i w_i) u^\top \Sigma u \pm O(\eta \log(1/\eta))$ and (also using that there is an SoS proof that $ab \leq a^2 + b ^2$) that
\[
    \Iprod{\mu - \mu', u}^2  \Paren{  \Iprod{ \frac{1}{n}\sum_{i\in [n]}  r_iw_i (x_i^* - \mu ) , u}  }^2 \leq \Iprod{\mu - \mu', u}^4 +  \Paren{  \Iprod{ \frac{1}{n}\sum_{i\in [n]}  r_iw_i (x_i^* - \mu ) , u}  }^4 \,.
\]
and that the second term is at most $\eta^4 \log^2(1/\eta) (u^\top \Sigma u)^2 \leq \eta^2 \log^2(1/\eta) (u^\top \Sigma u)^2$.
Together, this implies the bound
\[
    \bigO{\eta^2 \log^2(1/\eta) } (u^\top \Sigma u)^2  + \bigO{1} \Iprod{\mu - \mu', u}^4 \,.
\]

We now focus on bounding term \eqref{eqn:cov-spectral-expansion-1}.(2) using SoS Triangle Inequality (\cref{fact:sos-almost-triangle}): 
\begin{equation}
\label{eqn:i-cant-think-of-anything-useful-right-now}
\begin{split}
\calA_\eta \sststile{}{} \Biggl\{ 
& \Paren{ \frac{1}{n } \sum_{i \in [n]}  (1-r_i w_i) \Paren{ \Iprod{ x_i' - \mu' ,u }^2 -  u^\top \Sigma u } }^2 \\
& = \Paren{ \frac{1}{n } \sum_{i \in [n]}  (1-r_i w_i) \Paren{ \Iprod{ x_i' - \mu' \pm \mu ,u }^2 -  u^\top \Sigma u  } }^2 \\
& = \Paren{ \frac{1}{n } \sum_{i \in [n]}  (1-r_i w_i) \Paren{ \Iprod{ x_i' - \mu ,u }^2 + 2\Iprod{ x_i' - \mu ,u } \Iprod{\mu - \mu', u} + \Iprod{\mu - \mu', u}^2 -  u^\top \Sigma u  } }^2 \\
&\leq 4  \Paren{ \frac{1}{n } \sum_{i \in [n]}  (1-r_i w_i) \Paren{ \Iprod{ x_i' -  \mu ,u }^2 -  u^\top \Sigma u  } }^2 + 16 \Iprod{\mu - \mu', u}^4 \\
& \quad + 64 \Paren{\frac{1}{n} \sum_{i \in [n]} (1-r_iw_i) \Iprod{ x_i' - \mu ,u } \Iprod{\mu - \mu', u}}^2 \Biggr\}
\end{split}
\end{equation}
Note that we can bound the cross terms as follows:
\begin{align*}
    \mathcal{A}_\eta \sststile{}{} &\Biggl\{ 64 \Paren{\frac{1}{n} \sum_{i \in [n]} (1-r_iw_i) \Iprod{ x_i' - \mu ,u } \Iprod{\mu - \mu', u}}^2 \\
    &= 64 \Paren{\frac{1}{n} \sum_{i \in [n]} (1-r_iw_i) \Iprod{ x_i' - \mu \pm \mu' ,u } \Iprod{\mu - \mu', u}}^2 \\
    &\leq 128 \eta \Iprod{\mu - \mu', u}^2 \Paren{\frac{1}{n} \sum_{i \in [n]}  \Iprod{ x_i' - \mu' ,u }^2} \\
    &= \bigO{\eta \Iprod{\mu - \mu', u}^2 (u^\top \Sigma' u)} \\
    &\leq \bigO{\Iprod{\mu - \mu', u}^4 + \eta^2 (u^\top \Sigma' u)^2}
    \Biggr\}\,,
\end{align*}
where in the last inequality we used that there is an SoS proof that $ab \leq \bigO{a^2 + b^2}$. Plugging this into the inequalities above we conclude that
\begin{align*}
    \calA_\eta \sststile{}{} \Biggl\{ 
    &\Paren{ \frac{1}{n } \sum_{i \in [n]}  (1-r_i w_i) \Paren{ \Iprod{ x_i' - \mu' ,u }^2 -  u^\top \Sigma u } }^2 \\
    & \leq 4  \Paren{ \frac{1}{n } \sum_{i \in [n]}  (1-r_i w_i) \Paren{ \Iprod{ x_i' -  \mu ,u }^2 -  u^\top \Sigma u  } }^2  + \bigO{1} \Iprod{\mu - \mu', u }^4 + \bigO{\eta^2} (u^\top \Sigma' u)^2   \\
    & \leq 4 \eta \Paren{ \frac{1}{n} \sum_{i \in [n]} ( 1- r_i w_i) \Paren{ \Iprod{x_i' - \mu, u }^2 - u^{\top} \Sigma u }^2   }  + \bigO{1} \Iprod{\mu - \mu', u }^4 + \bigO{\eta^2} (u^\top \Sigma' u)^2 \\
    & =4 \eta \Biggl(  \underbrace{ \frac{1}{n} \sum_{i \in [n]}  \Paren{ \Iprod{x_i' - \mu, u }^2 - u^{\top} \Sigma u }^2  }_{\eqref{eqn:i-cant-think-of-anything-useful-right-now}.(1) }   -      \underbrace{ \frac{1}{n} \sum_{i \in [n]}  r_i w_i \Paren{ \Iprod{x_i^* - \mu, u }^2 - u^{\top} \Sigma u }^2 }_{\eqref{eqn:i-cant-think-of-anything-useful-right-now}.(2) }   \Biggr)  \\
    & \hspace{0.2in} +\bigO{1} \Iprod{\mu - \mu', u }^4 + \bigO{\eta^2} (u^\top \Sigma' u)^2 \Biggr\} \\
\end{align*}
We bound Term \eqref{eqn:i-cant-think-of-anything-useful-right-now}.1 as follows: (where in the last inequality we also use that there is an SoS proof that $\Iprod{\mu' - \mu, u}^2 (u^\top (\Sigma' - \Sigma) u) \leq \Iprod{\mu' - \mu, u}^4 + \Iprod{\Sigma' - \Sigma, uu^\top}^2$)
\begin{equation}
\begin{split}
\calA_\eta \sststile{}{} \Biggl\{ & \frac{1}{n } \sum_{i \in [n]}   \Paren{ \Iprod{ x_i' -\mu ,u }^2 -  u^\top \Sigma u   }^2 \\
& = \frac{1}{n } \sum_{i \in [n]}   \Paren{ \Iprod{ x_i'  - \mu' + \mu' - \mu ,u }^2 -  u^\top \Sigma' u  + (u^\top \Sigma' u - u^\top \Sigma u ) }^2 \\
& = \frac{1}{n } \sum_{i \in [n]}   \Paren{ \Iprod{ x_i'  - \mu'  ,u }^2 + \Iprod{\mu' - \mu, u}^2 + 2 \Iprod{x_i' - \mu' ,u }\Iprod{\mu' - \mu}  -  u^\top \Sigma' u  + ( u^\top \Sigma' u - u^\top \Sigma u ) }^2  \\
& = \frac{1}{n } \sum_{i \in [n]}   \Paren{ \Iprod{ x_i' -\mu'  ,u }^2 -  u^\top \Sigma' u }^2       \\
& \hspace{0.2in}+\frac{1}{n} \sum_{i \in [n]}\Paren{ 2\Iprod{x_i' - \mu' , u}\Iprod{\mu' -\mu}  +\Iprod{\mu'-  \mu, u}^2 +  (u^\top \Sigma' u - u^\top \Sigma u ) }^2   \\
& \hspace{0.2in}+ \frac{2}{n } \sum_{i \in [n]}  \Paren{\Iprod{ x_i' -\mu'  ,u }^2 -  u^\top \Sigma' u } \Paren{2\Iprod{x_i' - \mu' , u}\Iprod{\mu' -\mu}  +\Iprod{\mu'-  \mu, u}^2 +  (u^\top \Sigma' u - u^\top \Sigma u ) }\\
& \leq   \Paren{ 2 + \eta \log^2(1/\eta)  } \Paren{  u^\top \Sigma' u }^2  + 4 u^\top \Sigma' u \Iprod{\mu' - \mu }^2  \\
 & \hspace{0.2in} + 4 \Iprod{\mu' - \mu, u}^4 + 4 
 \Iprod{\Sigma' - \Sigma, uu^\top}^2  + \frac{2}{n} \sum_{i \in [n]} \Iprod{x_i' -\mu', u}^3\Iprod{\mu-\mu',u}  \Biggr\}
\end{split}
\end{equation}
It remains to bound term \eqref{eqn:i-cant-think-of-anything-useful-right-now}.(2).
Note that by the SoS selector lemma and $\eta$-higher-order goodness (cf.~\cref{lemma:subset_selection,lem:higher-order-stability}) we have that this is at least $\eta \log^2 (1/\eta) (u^\top \Sigma u)^2$:
\begin{equation}
\label{eqn:i-cant-think-2}
    \begin{split}
    \calA_\eta \sststile{}{} \Biggl\{ &\frac{1}{n } \sum_{i \in [n]}   r_i w_i \Paren{ \Iprod{ x_i^*  ,u }^2 -  u^\top \Sigma u   }^2  \geq \Paren{2 - \eta \log^2(1/\eta)} \Paren{u^\top \Sigma u }^2 \Biggr\} \,,
    \end{split}    
\end{equation}
Combining \cref{eqn:i-cant-think-of-anything-useful-right-now} and \cref{eqn:i-cant-think-2} we have 
\begin{equation}
    \begin{split}
\calA_\eta \sststile{}{} \Biggl\{ 
& \Paren{ \frac{1}{n } \sum_{i \in [n]}  (1-r_i w_i) \Paren{ \Iprod{ x_i' - \mu' ,u }^2 -  u^\top \Sigma u } }^2 \\
&  \leq \bigO{\eta} \Paren{ (u^\top \Sigma' u)^2  - (u^\top \Sigma u)^2  } + \bigO{\eta^2 \log(1/\eta ) } \Paren{ (u^\top \Sigma' u)^2 +  (u^\top \Sigma u)^2 } \\
& \hspace{0.2in}+ \bigO{\eta} \Iprod{\Sigma' - \Sigma , uu^\top}^2 + \bigO{\eta^2  } (u^\top \Sigma' u)^2 + \bigO{1}\Iprod{\mu - \mu', u}^4  \\
& \leq \bigO{\eta} \Paren{ (u^\top \Sigma' u)^2  - (u^\top \Sigma u)^2    } + \bigO{\eta^2 \log(1/\eta ) } \Paren{ (u^\top \Sigma' u)^2 +  (u^\top \Sigma u)^2 } \\
& \hspace{0.2in} + \bigO{1}\Iprod{\mu - \mu', u}^4  + \bigO{\eta} \Iprod{\Sigma' - \Sigma , uu^\top}^2 \Biggr\} 
\end{split}
\end{equation}

Overall, we get 
\begin{equation*}
\begin{split}
    \calA_{\eta}\sststile{}{} \Biggl\{ \Iprod{\Sigma' - \Sigma, uu^\top}^2 & \leq \bigO{1}\Iprod{\mu - \mu', u}^4  + \bigO{\eta} \Paren{ (u^\top \Sigma' u)^2  - (u^\top \Sigma u)^2    }  \\
    & \hspace{0.2in}+ \bigO{\eta^2 \log(1/\eta ) } \Paren{ (u^\top \Sigma' u)^2 +  (u^\top \Sigma u)^2 }   \Biggr\}\,.
\end{split}
\end{equation*}
This almost matches the guarantee we want except for the second term on the RHS. To finish the proof note that 
\begin{align*}
    \sststile{}{} \Biggl\{ \bigO{\eta} \Paren{ (u^\top \Sigma' u)^2  - (u^\top \Sigma u)^2    } &= \bigO{\eta} \Paren{ (u^\top \Sigma' u)  + (u^\top \Sigma u) } \Paren{(u^\top \Sigma' u)  - (u^\top \Sigma u)} \\
    &\leq 100 \bigO{\eta^2} \Paren{ (u^\top \Sigma' u)  + (u^\top \Sigma u) }^2 + \frac{1}{100} \Iprod{\Sigma - \Sigma', uu^\top}^2 \\
    &\leq \bigO{\eta^2} \Paren{ (u^\top \Sigma' u)^2  + (u^\top \Sigma u)^2 } + \frac{1}{100} \Iprod{\Sigma - \Sigma', uu^\top}^2\Biggr\}\,,
\end{align*}
where we used SoS Triangle Inequality (\cref{fact:sos-almost-triangle}) and that there is an SoS proof that $ab \leq a^2 + b^2$ with $a = 10 \bigO{\eta} \Paren{ (u^\top \Sigma' u)  + (u^\top \Sigma u) }$ and $b = \frac{1}{10} \Iprod{\Sigma - \Sigma', uu^\top}$. We can then move the $\Iprod{\Sigma - \Sigma', uu^\top}^2$ term to the LHS of our inequality and renormalize to get the desired conclusion. 
\end{proof}

\subsection{Private Mean Estimation}

In this section we prove~\cref{thm:main_private_cov_mean_est} restated below.
\restatetheorem{thm:main_private_cov_mean_est}

Our goal is to apply \Cref{thm:pure_dp_reduction}. We transform the constraint system from \Cref{algo:mean_estimation} into a score function. Throughout this section when showing the utility and volume guarantees of our score function we will assume that $x_i$ are an $\alpha/\sqrt{1/\alpha}$-corruption of Gaussian samples and thus our algorithm will also be robust to the same fraction of corruptions. Note that unlike in \cite{Hopkins2023Robustness} we utilize a single one-shot program which allows us to avoid paying a sample complexity of $O(d^2/\alpha \eps)$ to isotropize the samples privately. We will design this program by modifying our robust score function. We will define $\calB_T$ as follows:

\begin{equation*}
\calB_{T}\colon
  \left \{
    \begin{aligned}
      &\forall i\in [n].
      & w_i^2
      & = w_i \\
      & &\sum_{i\in [n]} w_i &\geq n-T \\
      &\forall i\in [n] & w_i(x'_i - x_i) &=0 \\
      &&\tfrac{1} n \sum_{i \in [n]} x_i' &= \mu' \\
    && \frac{1}{n} \sum_{i \in [n]} \Paren{ x_i' - \mu' } \Paren{ x_i' - \mu' }^\top & = \Sigma'  \\ 
    & \exists \text{ the following SoS proof }  & \sststile{4}{v} \Biggl\{\frac{1}{n} \sum_{i \in [n]} \Iprod{x_i'  - \mu', v}^4  \leq & \Paren{3 + \eta \log^2(1/\eta)} \Paren{ v^\top \Sigma' v}^2 \Biggr\}
    \end{aligned}
  \right \}
\end{equation*}

On a high level, and similar to the previous section about private regression, we will be looking for a pseudo-expectation $\pE$ such that the following two conditions are met
\begin{enumerate}
    \item $\pE \sdtstile{}{} \calB_T$ at degree $\bigO{1}$.
    \item $\norm{\pE \mu'}_2 \leq R$.
    \item For all fixed, i.e., not depending on indeterminates, vectors $u \in \R^d$, it holds that
    \[
        \Iprod{u, \pE \mu' - \tilde{\mu}}^2 \leq \bigO{\alpha^2} \pE \Iprod{u, \Sigma' u} \,.
    \]
\end{enumerate}
Again, the last condition acts as a proxy for the condition that $\norm{\Sigma^{-1/2}(\pE \mu' - \tilde{\mu})} \leq O(\alpha)$.

\paragraph{Certifiable parameters and our score function.}

As in the last section, we make the following definition.
\begin{definition}
    \label{def:cert_parameter_mean}
    Let $\alpha, \tau \in \R^{\ge 0}$, $n \in \N$ and $T \in [0,n]$.
    Further, let $x_1, \dots, x_n \in \R^d$.
    We call a parameter $\tilde{\mu}$, a $\paren{\alpha, \tau, T}$ certifiable mean for $x_1, \dots, x_n$ if there exists a linear functional $\mathcal{L}$ which $\tau$-approximately satisfies $\calB_T(x)$ (cf.~\cref{def:tau_relaxed_system})\footnote{Our constraint system includes constraints regarding the existence of SoS proofs for all vectors $v$ but these constraints can be encoded as polynomial inequalities and we apply the approximate satisfiability definition to this form of the constraint.} and such that
    \begin{enumerate}
        \item $\norm{\calR \paren{\calL}}_F \le R' + \tau$,  where $R' = \poly\paren{n, d, R}$ is sufficiently large and $\calR \paren{\calL}$ is the matrix representation of $\calL$.
        \item For all fixed, i.e., not depending on indeterminates, vectors $u \in \R^d$, it holds that
    \[
        \Iprod{u, \calL \mu' - \tilde{\mu}}^2 \leq \bigO{\alpha^2} \calL \Iprod{u, \Sigma' u} + \tau T \,.
    \]
        \item $\norm{\calL \mu'}_2 \leq 2R + \tau$
    \end{enumerate}
    Furthermore, we will refer to the linear functional $\calL$ as a $(\alpha, \tau, T)$ certificate for $(x,\tilde{\mu})$.
\end{definition}
For our purposes, we will end up setting $\tau = 1/(n \cdot d \cdot R \cdot \alpha^{-1} \cdot \covscale)^C$, for a large enough absolute constant $C$.
Now we use this definition to define a score function.

\begin{definition}[Score Function]
    \label{def:score-mean}
    Let $\mathbb{B}^d(2 R + n \tau + \alpha \sqrt{\covscale})$ denote the $\ell_2$-ball of radius $2 R + n \tau + \alpha \sqrt{\covscale}$ in $\R^d$ centered at the origin.
    Let $\alpha, \tau \in \R^{\ge 0}$, $x_1, \dots x_n \in \R^d$ (with $\mathcal{Y} = \{x_1, \dots, x_n\}$) and $\tmu \in \R^d$.
    We define the score function $\cS : \mathbb{B}^d(2 R + n \tau + \alpha \sqrt{\covscale}) \to \N_{\geq 0}$ (viewed as a function of $\tmu$) as
    \begin{equation*}
        \cS\paren{\tmu, \mathcal{Y}; \alpha, \tau} = \min_{T \ge 0} \text{ such that $\tmu$ is a $\paren{\alpha, \tau, T}$ certifiable mean for $\mathcal{Y} = \{y_1, \dots, y_n\}$}.
    \end{equation*}
\end{definition}

Note that for arbitrary domains and closeness constraints it is not necessarily true that this is well defined (ie that there exists $T$ such that $\tilde{\mu}$ is a $(\alpha, \tau, T)$ certifiable mean for every $\tilde{\mu}$ in the domain). However, for our closeness constraint we have that any point $\tilde{\mu}$ in the domain is a $(\alpha, \tau, n)$ certifiable mean as we will show in~\cref{lem:score-function-ub-mean}. 

In the rest of this section we will show that the score function defined above satisfies the conditions of \Cref{thm:pure_dp_reduction}.
\begin{enumerate}
    \item Bounded Sensitivity (\cref{lem:sensitivity-mean}): We will show that the score function has sensitivity at most 1 with respect to the input data (i.e., $x$).
    \item Quasi Convexity (\cref{lem:mean-quasiconvexity}): We will show that for every fixed dataset, the score function is quasi-convex with respect to the parameter $\tilde{\mu}$.
    \item Efficient Computability (\cref{lem:efficient-comp-mean}): We will verify that the score is efficiently computable for a fixed $\tilde{\mu}, x$.
    \item Volume (\cref{cor:volume_considerations_mean}): We will show that the volume of the set of points $\tilde{\mu}$ that have score at most $\eta \cdot n$ is sufficiently large and the volume of the points with score at most $\eta' \cdot n$ for $\eta' > \eta$ is sufficiently small.
    \item Robust Algorithm Finds a Low-scoring Parameter Efficiently (\cref{lem:find-low-score-mean}): We verify that finding $\tilde{\mu}$ that minimizes the score up to error $1$ for a fixed $x$ can be done efficiently.
\end{enumerate}

We will also show that a low scoring point $\tilde{\mu}$ achieves good accuracy, i.e., such that $\norm{\Sigma^{-1/2}(\tilde{\mu} - \mu)}$ is small.
We first focus on quasi-convexity, accuracy and volume, for which our new ``closeness constraint'' is key.
We will then address the other required properties, which are similar to previous applications of this transformation in \cite{Hopkins2023Robustness}.

As for regression, for any given desired accuracy $\alpha$ we let $\eta(\alpha) = \alpha/\sqrt{\log (1/\alpha)}$.
Further, we denote by $\eta^*$ the breakdown point of our robust estimator.
For simplicity, we may omit the dependence of the score function on the parameters $\alpha, \tau$ in the notation and write $\eta = \eta(\alpha)$. 

We follow the same proof strategy as for regression, repeated here for convenience:
We will show that every point $\tilde{\mu}$ with score $\eta(\alpha) n$ has accuracy at least $O(\alpha)$.\footnote{$\eta(\alpha)$ will also correspond to the fraction of corruptions our algorithm is robust to.}
To show that the exponential mechanism outputs such a point with high probability, we show that the volume of low-scoring points is not too small and the volume of high-scoring points is not too large.
To this end, we analyze points of score less or more than $\eta^* n$ separately, appealing to properties of our robust estimator in the first case, and boundedness of the domain in the second.

\subsubsection{Properties of the Closeness Constraint}

In this section, we will show several key features of our closeness constraint and score function, starting with the utility of our score function and its implications for volumes of points.

In particular, we show the following lemma, which characterizes $\tilde{\mu}$ which have score at most $T$ via relation to the true parameter $\mu^*$. It will imply both the utility of our final estimator (via accuracy of low-scoring points) as well as the bounded volume of points with score which is large, but below the breakdown point of our estimator. 

We first show that our robustness program also implies spectral closeness of $\Sigma$ and the ``SoS covariance`` $\Sigma'$.
In the SoS proofs for the robustness algorithms we did not explicitly
derive this bound, but our lemma statements and the final closeness of the means imply closeness of $\Sigma', \Sigma$.
\begin{corollary}
    \label{cor:non_zero_mean_cov_close_sos}
    Suppose $T/n$ is smaller than a sufficiently small constant and assume that $\Sigma^{-1/2}(x_1^* - \mu), \ldots, \Sigma^{-1/2}(x_1^* - \mu)$ are $T/n$-good and $T/n$-higher-order good.
    Then, for any fixed vector $u \in \R^d$ it holds that
    \[
        \calA_{T/n} \sststile{\bigO{1}}{} \Set{\Iprod{\Sigma' - \Sigma, uu^\top}^2 \leq \bigO{(T/n)^2 \log^2(n/T) \Iprod{u, \Sigma u}^2}} \,.
    \]
\end{corollary}
\begin{proof}
    By~\cref{cor:cov_mean_sos} it holds that
    \[
        \calA_{T/n} \sststile{6}{} \Set{(T/n)^2 \Iprod{\Sigma' - \Sigma, uu^\top}^2 \leq \frac 1 2 \Iprod{\mu' - \mu, u}^4 + \bigO{(T/n)^4 \log^2(n/T) \Iprod{u, \Sigma u}^2}} \,.
    \]
    Further, by~\cref{lem:main_identity_mean_sos_proof} it holds that (substituting $u \rightarrow \Sigma^{-1/2} u$)
    \[
        \calA_{T/n} \sststile{}{} \Set{\Iprod{\mu' - \mu, u}^4 \leq \bigO{(T/n)^4 \log^2(n/T) \Iprod{u, \Sigma u}^2}} \,.
    \]
    Substituting this in the equation above and canceling the $(T/n)^2$ factor yields the claim.
\end{proof}

We can now use this to prove bounds on the utility of points based on their score.
\begin{lemma}
\label{lem:closeness_constraint_soundness_mean}
    Let $T \leq \eta^* n$.
    Let $\tilde{\mu} \in \R^d$ be such that there exists a pseudo-distribution $\pE$ such that $\pE \sdtstile{}{} \calB_T$ and it holds that for every fixed vector $u \in \R^d$ the following inequality is true
    \[
        \Iprod{u, \pE \mu' - \tilde{\mu}}^2 \leq \bigO{\alpha^2} \cdot \pE \Iprod{u, \Sigma' u} \,.
    \]
    Then, it follows that $\norm{\Sigma^{-1/2}(\tilde{\mu} - \mu^*)} \leq \bigO{\alpha + \tfrac T n \log (\tfrac n T)}$.
\end{lemma}

\begin{proof}
    Note that by~\cref{cor:non_zero_mean_cov_close_sos} we have that
    \[ \calB_T \sststile{\bigO{1}}{} \Biggl\{ \Iprod{ \Sigma' - \Sigma,   uu^\top }^2  \leq \bigO{(\tfrac T n)^2 \log^2 (\tfrac n T)}  \Paren{ ( u^\top \Sigma u )^2}  \Biggr\} \,.\]
    Therefore, the closeness constraint $\Iprod{u, \pE \mu' - \tilde{\mu}}^2 \leq \bigO{\alpha^2} \cdot \pE \Iprod{u, \Sigma' u}$ also implies that 
    \[ \Iprod{u, \pE \mu' - \tilde{\mu}}^2 \leq \bigO{\alpha^2 u^\top \Sigma u}\,.\]
    From the utility of our robust estimator, we have that for any vector $u$ that
    \[ \Iprod{u, \pE \mu' - \mu^*}^2 \leq \bigO{(\tfrac T n)^2 \log^2 (\tfrac n T)} \cdot u^\top \Sigma u\,,\]
    so by triangle inequality we have that
    \[ \Iprod{u, \tilde{\mu} - \mustar}^2 \leq \bigO{\Paren{\alpha^2 + (\tfrac T n)^2 \log^2 (\tfrac n T)}} \cdot u^\top \Sigma u \,,\]
    which implies that
    \[ \Iprod{u, \Sigma^{-1/2} \Paren{\tilde{\mu} - \mustar}}^2 \leq \bigO{\Paren{\alpha^2 + (\tfrac T n)^2 \log^2 (\tfrac n T)}} \cdot \norm{u}^2\,.\]
    This statement is equivalent to $\norm{\Sigma^{-1/2}(\tilde{\mu} - \mu^*)} \leq \bigO{\alpha + \tfrac T n \log (\tfrac n T)}$ which was the desired bound.
\end{proof}

\paragraph{Volume of High and Low Scoring Points.}
\begin{corollary}
    \label{cor:volume_considerations_mean}
    Let $T \leq \eta^* n$ and $n \geq \tilde{\Omega}\Paren{ \frac{d^2 + \log^{2}(1/\beta)}{ \eta^2  }}$.
    Then it holds that:
    \begin{enumerate}
        \item Every point in a $\Sigma^{-1/2}$-ball of radius $O(\eta \sqrt{\log (1/\eta)})$ around $\bar{\mu} = \frac{1}{n} \sum_{i \in [n]} x_i$ has score $T \leq \eta n  $
        \item Every point of score at most $\eta^* n$ is contained in a $\Sigma^{-1/2}$-ball around $\mu$ of radius $\bigO{1}$.
    \end{enumerate}
\end{corollary}
\begin{proof}
    We have already shown in~\cref{sec:feasibility} that a distribution supported on a single point with $x_i'$ equal to the true samples, $\Sigma'$ equal to the empirical covariance is feasible for the original robustness program.
    Furthermore, note that the empirical mean is sufficiently concentrated such that the empirical mean also has norm at most $2R + \tau T$ when the original mean $\mu$ had norm at most $R$.
    We will show that any point $\tilde{\mu}$ such that $\norm{\Sigma^{-1/2} (\tilde{\mu} - \bar{\mu})} \leq \bigO{\eta \sqrt{\log (1/\eta)}}$ has score at most $\eta n$. It suffices to show that for this single point distribution above that the closeness constraint holds for such $\tilde{\mu}$. If we let $\bar{\Sigma}$ be the empirical covariance of the true samples, equivalently we want to show that for all vectors $u$ it holds that
    \[ \Iprod{u, (\bar{\mu} - \tilde{\mu})}^2 \leq \bigO{\alpha} \cdot  \Iprod{u, \bar{\Sigma} u} = \bigO{\eta \sqrt{\log (1/\eta)}} \cdot  \Iprod{u, \bar{\Sigma} u} \]
    when $\norm{\Sigma^{-1/2} (\tilde{\mu} - \bar{\mu})} \leq \bigO{\eta \sqrt{\log (1/\eta)}}$. 
    Note that by concentration of the empirical covariance (\cref{fact:empirical_cov_bound}) it suffices to instead show that
    \[ \Iprod{u, \bar{\mu} - \tilde{\mu}}^2 \leq \bigO{\alpha} \cdot  \Iprod{u, \Sigma u} = \bigO{\eta \sqrt{\log (1/\eta)}} \cdot  \Iprod{u, \Sigma u}\,.\]
    We have that
    \[ \Iprod{u, \bar{\mu} - \tilde{\mu}}^2 = \Iprod{\Sigma^{1/2} u, \Sigma^{-1/2}(\bar{\mu} - \tilde{\mu})}^2 \leq \norm{\Sigma^{-1/2}(\bar{\mu} - \tilde{\mu})}_2^2 \cdot \norm{\Sigma^{1/2} u}_2^2\,.\]
    Note that $\norm{\Sigma^{-1/2}(\bar{\mu} - \tilde{\mu})}_2^2 \leq \bigO{\eta^2 \log (1/\eta)}$ and $\norm{\Sigma^{1/2} u}_2^2 = \Iprod{u, \Sigma u}$. Furthermore, all such $\tilde{\mu}$ are within the domain of the score function since $\bar{\mu}$ has norm at most $2R + \tau T$ and $\norm{\bar{\mu} - \tilde{\mu}} \leq \alpha \sqrt{\covscale}$. 
    
    Consider any point $\tilde{\mu}$ of score $T < \eta^* n$. We have by~\cref{lem:closeness_constraint_soundness_mean} that 
    \[\norm{\Sigma^{-1/2}(\tilde{\mu} - \mu^*)} \leq \bigO{\alpha + \tfrac T n \sqrt{\log (\tfrac n T)}} \leq \bigO{1}\,,\]
    which completes the proof of the second half of the corollary.
\end{proof}

\paragraph{Quasiconvexity.}
\begin{lemma}[Quasi-Convexity]
\label{lem:mean-quasiconvexity}
The score function $\calS\paren{\tilde{\mu}; x}$ as defined in \Cref{def:score-mean} is quasi-convex with respect to the parameter $\tilde{\mu}$.
\end{lemma}
\begin{proof}
Suppose $\tilde{\mu}_1, \tilde{\mu}_2 \in \mathbb{B}^d \paren{2R + n \tau + \alpha \sqrt{\covscale}}$, and $\calS\paren{\tilde{\mu}_1; x} = T_1$, and $\calS\paren{\tilde{\mu}_2; x} = T_2$. We will show that for any $\lambda \in [0, 1]$, we have $\calS\paren{\lambda \tilde{\mu}_1 + (1-\lambda) \tilde{\mu}_2; x} \le \max\{T_1, T_2\}$. 

Let $\tilde{\mu}_3 = \lambda \tilde{\mu}_1 + (1-\lambda) \tilde{\mu}_2$. Since $\calS\paren{\tilde{\mu}_1; x} = T_1$, there exists a $\paren{\alpha, \tau, T_1}$ certificate $\mathcal L_1$ for $(x, \tilde{\mu}_1)$. Similarly, since $\calS\paren{\tilde{\mu}_2; x} = T_2$, there exists a $\paren{\alpha, \tau, T_2}$ certificate $\mathcal L_2$ for $(x, \tilde{\mu}_2)$. We will construct a $\paren{\alpha, \tau, \max\{T_1, T_2\}}$ certificate $\mathcal L_3$ for $(x, \tilde{\mu}_3)$. We construct
$\mathcal L_3$ as follows: $\mathcal L_3 = \lambda \mathcal L_1 + (1-\lambda) \mathcal L_2 $.
All of the constraints would then be satisfied trivially, except for the counting constraint and the closeness constraint.
Without loss of generality assume $T_2 = \max\{T_1, T_2\}$.
To verify the counting constraint, we have
to show that for any polynomial $p$ where $\norm{\mathcal R \paren{p}}_2 \le 1$ that $\mathcal L_3 \paren{\sum_{j \in [n]} w_j - n + T_2} p^2 \geq -5 \cdot T_2 \cdot n$.
We have
\begin{align*}
\mathcal L_3 \paren{\sum_{j \in [n]} w_j - n + T_2} p^2 
&= \lambda \mathcal L_1 \paren{\sum_{j \in [n]} w_j - n + T_2 + T_1 - T_1} p^2 + (1-\lambda) \mathcal L_2 \paren{\sum_{j \in [n]} w_j - n + T_2} p^2 \\
&\ge - 5 \lambda \tau \cdot T_1 \cdot n - 5 (1-\lambda) \tau \cdot T_2 \cdot n + \lambda  \paren{T_2 - T_1} \mathcal L_1  p^2 \\
&\ge - 5 \lambda \tau \cdot T_1 \cdot n - 5 (1-\lambda) \tau \cdot T_2 \cdot n \\
&\ge - 5 \lambda \tau \cdot T_2 \cdot n - 5 (1-\lambda) \tau \cdot T_2 \cdot n  \\
& = - 5 \tau \cdot T_2 \cdot n.
\end{align*}
We now consider the counting constraint. Note that we have that
\[ \Iprod{u, \calL_3 \mu' - \tilde{\mu}_3}^2 = \Iprod{u, \lambda (\calL_1 \mu' - \tilde{\mu}_1) + (1-\lambda) (\calL_2 \mu' - \tilde{\mu}_2)}^2 \leq \lambda \Iprod{u, \calL_1 \mu' - \tilde{\mu}_1}^2 +  \Iprod{u, \calL_2 \mu' - \tilde{\mu}_2}^2\,,\]
by the convexity of the function $\Iprod{u, \cdot}^2$. Furthermore, applying that $\calL_1, \calL_2$ satisfy the closeness constraints for $\tilde{\mu}_1, \tilde{\mu}_2$ respectively we have that
\[ \lambda \Iprod{u, \calL_1 \mu' - \tilde{\mu}_1}^2 +  \Iprod{u, \calL_2 \mu' - \tilde{\mu}_2}^2 \leq \lambda \calL_1 u^\top \Sigma' u + (1-\lambda) \calL_2 u^\top \Sigma' u = \calL_3 u^\top \Sigma' u\,.\]
This verifies that the closeness constraint holds for $\calL_3$ and $\tilde{\mu}_3$ and that $\mathcal L_3$ is a $\paren{\alpha, \tau, \max\{T_1, T_2\}}$ certificate for $(x, y)$, as desired.
\end{proof}

\paragraph{Efficient Computability.}
\begin{lemma}[Efficient Computation of Mean Score Function]
\label{lem:efficient-comp-mean}
    Let $\tilde{\mu}$ be some point in the domain of the score function and let $T = \calS(\tilde{\mu}, \calY; \alpha, \tau)$ be the score of $\tilde{\mu}$. Then we can compute $T$ in time $\poly(n^d, \log(R \sqrt{\covscale}), \log(1/\gamma))$ up to accuracy $\bigO{\gamma}$.
\end{lemma}
We will show efficient computability via the ellipsoid algorithm. Note that our constraints are the same as the constraints in~\cite{Hopkins2023Robustness} except for the closeness constraint and the norm constraint on $\calL \mu'$. Their efficient computability proof (Lemma C.6 of~\cite{Hopkins2023Robustness}) mostly applies to our program, with the exception of requiring a separation oracle for the two new constraints. We give a proof of the separation oracle for these constraints in the appendix~\cref{lem:closeness-separation-oracle}.

\subsubsection{Other Properties of the Score Function}

\paragraph{Well Definedness of the Score Function.}
\begin{lemma}[Score Function Upper Bound]
\label{lem:score-function-ub-mean}
Let $n \geq \tilde{\Omega}\Paren{ \frac{d^2 + \log^{2}(1/\beta)}{ \eta^2  }}$. For any input $x$ and $\tilde{\mu} \in \mathbb{B}^d \paren{2R + n \tau + \alpha \sqrt{\covscale}}$, we have $\calS\paren{\tilde{\mu}; x} \le n$, for $\calS$ as defined in \Cref{def:score-mean}.
\end{lemma}
\begin{proof}
It suffices to show that for $T = n$, and for any $\tilde{\mu} \in \mathbb{B}^d \paren{2R + n \tau + \alpha \sqrt{\covscale}}$, there exists a $\paren{\alpha, \tau, n}$ certificate $\mathcal L$ for $(x, y)$. 

We will define a linear functional $\calL$ which for every monomial $p$ assigns a value equal to the expectation of $p$ over a distribution supported on a single point. Furthermore, we will show that $\calL$ is feasible for our robustness system as well as the additional constraints. If $\norm{\tilde{\mu}}\leq 2R + \tau T$ we will let $x_i' = \tilde{\mu}$, $w_i = 0$, and $\mu', \Sigma'$ to be the empirical mean and covariance. This trivially satisfies all constraints.

Otherwise, let $\mu$ be $\tilde{\mu}$ projected to the $2R + \tau T - \psi$ radius sphere for some choice of $\psi$ to be decided later. Let $x_i' \sim N(\mu, \covscale/2 \cdot I_d)$ and $\mu', \Sigma'$ be the empirical mean and covariance respectively.
Furthermore, let all $w_i = 0$. We will show that with non-zero probability this certifies that $\tilde{\mu}$ has score $n$.
Note that with probability $1-\beta$ this system is feasible for the robustness constraints and $\Sigma' \preceq \covscale \cdot I_d$. Furthermore, note that with probability $1-\beta$ we have that $\norm{\mu' - \mu} \leq \sqrt{\covscale/2} \cdot \alpha$. Thus, if we take $\psi = \sqrt{\covscale/2} \cdot \alpha$ we have that $\norm{\mu'} \leq 2R + \tau T$. 

It remains to show that the closeness constraint is satisfied. We will bound $\norm{(\Sigma')^{-1/2} (\mu' - \tilde{\mu})}$. By our PSD bounds on $\Sigma'$ and Triangle Inequality we have that 
\begin{align*}
    \norm{(\Sigma')^{-1/2} (\mu' - \tilde{\mu})} &\leq \frac{1}{\sqrt{\covscale}} \norm{\mu' - \mu} \\
    &\leq \frac{1}{\sqrt{\covscale}} \cdot \left( \norm{\mu' - \mu} + \norm{\mu - \tilde{\mu}}\right)\\
    &\leq \frac{1}{\sqrt{\covscale}} \cdot \left( 2\sqrt{\covscale/2} \cdot \alpha + \alpha \sqrt{\covscale} \right) \\
    &\leq \bigO{\alpha} \,.
\end{align*}
Note that this also implies the closeness constraint is satisfied. Therefore, we have given a linear functional $\calL$ which certifies that $\tilde{\mu}$ has score at most $n$.
\end{proof}

\paragraph{Bounded Sensitivity.}
\begin{lemma}[Bounded Sensitivity]
\label{lem:sensitivity-mean}
The score function $\calS\paren{\tilde{\mu}; x}$ as defined in \Cref{def:score_function_regression} has sensitivity $1$ with respect to the input data $x$.
\end{lemma}
\begin{proof}
Suppose $x, x'$ are two neighbouring datasets, and $\tilde{\mu} \in \mathbb{B}^d \paren{2 R + n \tau + \alpha \sqrt{\covscale}}$. Moreover, assume that $\calS\paren{\tilde{\mu}; x} = T$. If we show that $\calS\paren{\tilde{\mu}; x'}\le T + 1$, then we are done by symmetry. Since $\calS\paren{\tilde{\mu}; x} = T$, there exists a $\paren{\alpha, \tau, T}$ certificate  $\mathcal L$ for $x$. We will show that there exists a $\paren{\alpha, \tau, T+1}$ certificate $\mathcal L'$ for $x'$. If we show this we conclude that $\calS\paren{\tilde{\mu}; x'} \le T + 1$, and we are done.

Without loss of generality, assume $x$, and $x'$ differ in one point, say $x_i$ and $x_i'$.
We will construct $\mathcal L'$ from $\mathcal L$ by modifying the value of $\mathcal L$ on the monomials. We will show that $\mathcal L'$ is a $\paren{\alpha, \tau, T+1}$ certificate for $x'$. For any monomial $p$, let $\mathcal L' p = \mathcal L p$, if $p$ does not contain $w_i$, otherwise let $\mathcal L' p = 0$. 

Now let's verify that the conditions hold for this definition of $\mathcal L' p$. Note that if $p = q + w_i r$, where $q$ does not contain $w_i$, then $\mathcal L' s p^2 = \mathcal L s q^2$, for all polynomials $s$. Moreover, $\norm{\mathcal R \paren{q}}_2 \le 1$. Therefore, all of the constraints that do not contain $w_i$ will be satisfied, with the same value $T$. This also includes the closeness constraint.

It remains to verify the constraints that contain $w_i$ from the relaxed version of our robustness program. All of these except the counting constraint are satisfied because $\mathcal L' p = 0$ for all monomials that contain $w_i$. It remains to verify the counting constraint. We have
\begin{align*}
\mathcal L' \paren{\sum_{j \in [n]} w_j - n + T + 1} p^2 &= \mathcal L \paren{\sum_{j \neq i} w_j - n + T + 1} q^2 \\
&= \mathcal L \paren{\sum_{j \in [n]} w_j - n + T} q^2 + \mathcal L q^2 - \mathcal L w_i q^2 \\
& \ge - 5 \tau \cdot T \cdot n + \mathcal L \paren{1 - w_i} q^2  \\
& = - 5 \tau \cdot T \cdot n + \mathcal L' \paren{1 - w_j}^2 q ^2 + \mathcal L' (w_j - w_j^2)q^2.
\end{align*}
Now consider the polynomial $(1-w_j) q$, for the representation of this polynomial we have that $\norm{\mathcal R \paren{(1-w_j) q}}_2 \le 2 \norm{\mathcal R \paren{q}}_2 \le 2 $. Therefore, we have that $\mathcal L' \paren{1 - w_j}^2 q \ge  -4\tau \cdot T$. Moreover, we have that $\mathcal L' (w_j - w_j^2)q^2 \ge  -\tau T$. Therefore we have that 
\[\mathcal L' \paren{\sum_{j \in [n]} w_j - n + T + 1} p^2 \ge - 5 \tau \cdot T \cdot n - 4 \tau \cdot T - \tau \cdot T \ge - 5\tau \cdot \paren{T+ 1} \cdot n\,.\]
This verifies that $\mathcal L'$ is a $\paren{\alpha, \tau, T+1}$ certificate for $x'$, as desired.
\end{proof}

\paragraph{Efficiently Finding a Low Scoring Point.}
\begin{lemma}
\label{lem:find-low-score-mean}
There exists an algorithm which runs in time $\poly(n, d, \log(R \sqrt{\covscale})$ that outputs $\hat{\mu}$ such that $\calS(\hat{\mu}, \calY; \alpha, \tau) \leq \min_{\tilde{\mu}} \calS(\tilde{\mu}, \calY; \alpha, \tau) + 1$ and such that all $\mu$ within distance $\alpha/\sqrt{\covscale}$ have score at most $\calS(\hat{\mu}, \calY; \alpha, \tau)$.
\end{lemma}
\begin{proof}
    Consider the constraint system $\calB_T$ which is equal to $\calB_T$ without the closeness constraint. Our algorithm is as follows: we check for $T = 1, \ldots n$ whether $\calB_T$ is satisfiable. This is computable in time $\poly(n, d, \log(R \sqrt{\covscale})$ by~\cref{lem:efficient-comp-mean}. Let $T_{\min}$ be the smallest $T$ for which this is true and let $\calL$ be the corresponding linear operator. We have that $\hat{\mu}$ approximately minimizing score is just $\calL \mu'$. Note that $\calL$ also satisfies $\calB_{T_{\min}}$ with $\tilde{\mu} = \hat{\mu}$ since the additional constraint is the closeness constraint which is trivially satisfied since $\tilde{\mu} = \calL \mu'$. Thus, $\hat{\mu}$ has score at most $T_{\min}$. Furthermore, note that this also holds for all $\tilde{\mu}$ such that $\norm{\Sigma^{-1/2} (\tilde{\mu} - \hat{\mu})}_2 \leq \alpha$ which includes all $\tilde{\mu}$ within distance $\alpha / \sqrt{\covscale}$.

    Furthermore, the true minimum score is at most $T_{\min} - 1$. Consider $\calL$ which certifies the minimum score, and note that $\calL$ also satisfies $\calB$ with the same value of $T$. Therefore the true minimum score must be greater than $T_{\min} - 1$ since by definition $T_{\min}$ was the smallest integer $T$ which was feasible for $\calB$.
\end{proof}

\subsubsection{Proof of Main Private Mean Estimation Theorem (\cref{thm:main_private_cov_mean_est})}

We now return to the proof of the main theorem. We have shown all the preconditions for the reduction hold and thus it suffices to bound the number of samples and note that a point of $2\eta$ score has good accuracy.

\begin{proof}
    Again, we let $r = \alpha / \sqrt{L}, \eta = \eta(\alpha)$ and $R$ be the bound that we are promised on $\mu$.
    Assume that $2 \eta \leq \eta^*$.
    Note that for the reduction to succeed it suffices to take 
    \[ n \ge \Omega\left(\max\limits_{\eta': \eta \le \eta' \le 1}\frac{\log(V_{\eta'}(\mathcal{Y})/V_{\eta}(\mathcal{Y})) + \log (1/(\beta \cdot \eta))}{\eps \cdot \eta'}\right)\,.\] 
    Let $\eta^* = 1/e$. We have that the total volume of parameters is at most $\left(2R\right)^d$ by our bound on the domain\footnote{Note that this is technically a ball of radius $R + \alpha \sqrt{L} + \tau T$ however when $R > \Omega(\alpha \sqrt{L}$ this radius is $O(R)$ and if $R$ is small we can simply select a uniformly random point to output as our parameter and this will satisfy the accuracy guarantee.} and by~\cref{cor:volume_considerations_mean} and the lower bound on the covariance we have that the volume of low scoring points is at least $(\eta/\sqrt{\covscale})^d$. Thus,
    \begin{align*}
        \max\limits_{\eta': \eta^* \le \eta' \le 1}\frac{\log(V_{\eta'}(\mathcal{Y})/V_{\eta}(\mathcal{Y})) + \log (1/(\beta \cdot \eta))}{\eps \cdot \eta'} &\leq  \frac{\log(R^d/(\eta/\sqrt{\covscale})^d) + \log (1/(\beta \cdot \eta))}{\eps} \\
        &\leq \frac{d \log(R \sqrt{\covscale} / \eta) + \log (1/(\beta \cdot \eta))}{\eps}\,. \\
    \end{align*}
    Furthermore, note that by invoking~\cref{cor:volume_considerations_mean} we have that
    \begin{align*}
        \max\limits_{\eta': \eta \le \eta' \le \eta^*}\frac{\log(V_{\eta'}(\mathcal{Y})/V_{\eta}(\mathcal{Y})) + \log (1/(\beta \cdot \eta))}{\eps \cdot \eta'} &\leq \bigO{\frac{\log \left( (1/\eta)^d \right) + \log (1/(\beta \cdot \eta))}{\eps \cdot \eta}} \\
        &\leq \bigO{\frac{d \log(1/ \eta) + \log (1/(\beta \cdot \eta))}{\eps \cdot \eta}}\,.
    \end{align*}
    Combining these two bounds we have that
    \[\max\limits_{\eta': \eta \le \eta' \le 1}\frac{\log(V_{\eta'}(\mathcal{Y})/V_{\eta}(\mathcal{Y})) + \log (1/(\beta \cdot \eta))}{\eps \cdot \eta'} \leq \frac{d \log (R \sqrt{\covscale})}{\eps} + \frac{d \log (1 /\eta)}{\eps \cdot \eta}  + \frac{\log (1/(\beta \cdot \eta))}{\eps \cdot \eta}\,.\]
    Furthermore, we have that the algorithm outputs a point with score at most $2 \eta n$. By~\cref{lem:closeness_constraint_soundness_mean} we have that this implies that the outputted point $\hat{\mu}$ satisfies
    \[ \norm{\Sigma^{-1/2} (\hat{
    \mu} - \mu)}_2 \leq \bigO{\alpha}\,.\]
\end{proof}

\section*{Acknowledgements}
The authors thank Samuel B. Hopkins and Shyam Narayanan for helpful discussions.
ST was supported by the European Union’s Horizon research and innovation programme (grant agreement no. 815464). AB is supported by the NSF TRIPODS program (award DMS-2022448). 

\clearpage

\bibliography{bib/scholar}
\bibliographystyle{alpha}

\newpage 
\appendix

\section{Robust Covariance Estimation}
\label{sec:cov-estimation}

\subsection{Estimation in (Relative) Spectral Norm}
\label{sec:cov-estimation-spectral}

In this section, we provide an efficient algorithm for robustly estimating the covariance of a Gaussian using $\tilde{\Omega}(\tfrac{d^2 + \log^2(1/\beta)}{\eta^2})$ samples.
Again, our results also holds for fourth-moment matching reasonable sub-gaussian distributions (cf.~\cref{def:moment_match_sub_gaussian}).
Note that throughout we will assume the mean is zero, which can be achieved without loss of generality by pairing samples and subtracting them to sample from a distribution with mean zero and covariance $2\Sigma$.

\begin{theorem}[Robust Covariance Estimation in (Relative) Spectral Norm]
\label{thm:cov-est-sample-opt}
Let $\cD$ be a mean-zero distribution with covariance $\Sigma$ such that $\Sigma^{-1/2} \cD$ is fourth-moment matching reasonable sub-gaussian.
Let $\eta$ be smaller than a sufficiently small constant and let $0 < \beta$.
Given $0 < \eta $, and $n \geq \tilde{\Omega}(\tfrac{d^2 + \log^2(1/\beta)}{\eta^2})$ points that are an $\eta$-corruption of $n$ i.i.d.\ samples from $\cD$.
Then, there exists an algorithm that runs in time $(n \cdot d)^{\mathcal{O}(1)}$ and with probability at least $1-\beta$, outputs $\hat{\Sigma}$ such that 
\begin{equation*}
    \norm{ \hat{\Sigma} - \Sigma }_{\textrm{op}} \leq \bigO{\eta  \log(1/\eta) } \,.
\end{equation*}
and 
\begin{equation*}
    \norm{ { \Sigma}^{-1/2} \cdot  \hat{\Sigma}  \cdot { \Sigma}^{-1/2} - I }_{\textrm{op}} \leq \bigO{\eta  \log(1/\eta) } \,.
\end{equation*}
\end{theorem}

We will use the following constraint system in vector-valued variables $x_1', \ldots, x_n', \mu'$, scalar-valued variables $w_i$, and matrix valued variables $\Sigma'$.

\begin{equation*}
\calA_{\eta}\colon
  \left \{
    \begin{aligned}
      &\forall i\in [n].
      & w_i^2
      & = w_i \\
      & &\sum_{i\in [n]} w_i &= (1-\eta)n \\
      &\forall i\in [n] & w_i(x'_i - x_i) &=0 \\
      &&\tfrac{1} n \sum_{i \in [n]} x_i' {x_i'}^{\top} &= \Sigma' \\
    &\exists \text{ the following SoS proof }  & \sststile{4}{v} \Biggl\{\frac{1}{n} \sum_{i \in [n]} \Iprod{x_i'  - \mu', v}^4  \leq & \Paren{3 + \eta \log^2(1/\eta)} \Paren{ v^\top \Sigma' v}^2 \Biggr\}
    \end{aligned}
  \right \}
\end{equation*}
We denote by $x_1, \ldots, x_n$ the input to the algorithm.
Further, let $r_i \in \Set{0,1}$ be the indicator that $x_i^* = x_i$.

\paragraph{Feasibility.} 
The feasibility of the above program follows by setting $w_i = r_i$ and $\eta$-goodness and $\eta$-higher-order-goodness of the true samples.
We will give a formal proof in \cref{sec:feasibility}. 

\begin{mdframed}
  \begin{algorithm}[Robust Covariance Estimation in (Relative) Spectral Norm]
    \label{algo:cov_estimation}\mbox{}
    \begin{description}
    \item[Input:] $\eta$-corrupted sample $x_1, \ldots, x_n$.
    \item[Output:] An estimate $\hat{\Sigma}$ of the covariance attaining the guarantees of \cref{thm:cov-est-sample-opt}.
    
    \item[Operations:]\mbox{}
    \begin{enumerate}
    \item Find a degree-$\bigO{1}$ pseudo-distribution $\zeta$ satisfying $\calA_{\eta}$.
    \item Output $ \hat{\Sigma} = \pE_{\zeta}[\Sigma']$.
    \end{enumerate}
    \end{description}
  \end{algorithm}
\end{mdframed}

The key lemma is the following.
\begin{lemma}[SoS Covariance to True Covariance]
\label{lem:main_identity_cov_sos_proof}
Suppose $\eta$ is a sufficiently small constant and $\Sigma^{-1/2} x_1^*, \ldots, \Sigma^{-1/2} x_n^*$ are $\eta$-good and $\eta$-higher-order-good. %
Then, for any fixed $u \in \mathbb{R}^d$,
\begin{align*}
    \calA_\eta \sststile{6}{w} \Set{ \Iprod{ \Sigma' - \Sigma,   uu^\top }^2  \leq \bigO{\eta^2 \log^2 (1/\eta)}  \Paren{ u^\top \Sigma u }^2 } \,.
\end{align*}
\end{lemma}

Given this lemma, we can prove \cref{thm:cov-est-sample-opt}.
\begin{proof}[Proof of \cref{thm:cov-est-sample-opt}]
Note that by~\cref{fact:mean_and_cov_of_1-eps,fact:cov_small_subset} and~\cref{lem:higher-order-stability}, with probability at least $1-\beta$, $\Sigma^{-1/2} x_1^*, \ldots, \Sigma^{-1/2} x_n^*$ are $\eta$-good and $\eta$-higher-order-good.
Thus, by~\cref{sec:cov-feasibility} we have that with probability $1-\delta$ the program is feasible.
Next, observe since we computed a degree-$6$ pseudo-distribution $\mu$ consistent with $\calA_\eta$, it follows from Cauchy-Schwarz for pseudo-expectations (cf.~\cref{fact:pe_cs}) that 
\begin{equation*}
    \Iprod{ \pexpecf{\mu}{\Sigma'} - \Sigma, u u^\top }^2 \leq \pexpecf{\mu}{ \Iprod{ \Sigma' - \Sigma, vv^\top }^2 }  \leq \bigO{\eta^2 \log^2(1/\eta)} (u^\top \Sigma u)^2 
\end{equation*}
Setting $u = \Sigma^{-1/2} v$ such that $v$ is the top singular vector of $\Sigma^{-1/2} \pexpecf{\mu}{\Sigma'} \Sigma^{-1/2} - I$,  we can conclude that 
\begin{align*}
    \Norm{  \Sigma^{-1/2} \pexpecf{\mu}{\Sigma'} \Sigma^{-1/2} - I }_{\textrm{op}} &\leq \bigO{\eta \log(1/\eta)} \,.
\end{align*}
\end{proof}

We next prove~\cref{lem:main_identity_cov_sos_proof}.
\begin{proof}
Let $r_i$ be the indicators for the uncorrupted samples where $x_i^* = x_i$. Using SoS almost triangle inequality (\cref{fact:sos-almost-triangle}) we have that
\begin{equation}
\label{eqn:cov-spectral-expansion-1}
\begin{split}
    \calA_\eta  \sststile{}{} \Biggl\{ & \Iprod{ \Sigma' - \Sigma , uu^\top }^2  \\
    &= \Paren{  \frac{1}{n } \sum_{i \in [n]} (1 - r_i w_i + r_i w_i) \Paren{ \Iprod{ x_i',u }^2 -  u^\top \Sigma u }  }^2    \\ 
    & \leq 2 \underbrace{ \Paren{  \frac{1}{n } \sum_{i \in [n]}  r_i w_i \Paren{ \Iprod{ x_i' ,u }^2 -  u^\top \Sigma u }  }^2}_{\eqref{eqn:cov-spectral-expansion-1}.(1) }  + 2 \underbrace{ \Paren{ \frac{1}{n } \sum_{i \in [n]}  (1-r_i w_i) \Paren{ \Iprod{ x_i' ,u }^2 -  u^\top \Sigma u } }^2}_{\eqref{eqn:cov-spectral-expansion-1}.(2) }   \Biggr\}
\end{split}
\end{equation}
We now handle each term separately, starting with the first one. We have that
\begin{equation}
\label{eqn:expanding-the-first-spect-cov-term}
\begin{split}
    \calA_{\eta}\sststile{}{} \Biggl\{ & \Paren{  \frac{1}{n } \sum_{i \in [n]}  r_i w_i \Paren{ \Iprod{ x_i' ,u }^2 -  u^\top \Sigma u }  }^2  \\
    & = \Paren{  \frac{1}{n } \sum_{i \in [n]}  r_i w_i \Paren{ \Iprod{ x_i^*   ,u }^2 -  u^\top \Sigma u }  }^2  \\
    & \leq 2 \Paren{  \frac{1}{n } \sum_{i \in [n]}  r_i  \Paren{ \Iprod{ x_i^*   ,u }^2 -  u^\top \Sigma u }  }^2 + 2 \underbrace{ \Paren{  \frac{1}{n } \sum_{i \in [n]}  r_i (1-w_i) \Paren{ \Iprod{ x_i^* ,u }^2 -  u^\top \Sigma u }  }^2 }_{\eqref{eqn:expanding-the-first-spect-cov-term}.(1) }  \,.
\end{split}
\end{equation}
Note that the first term is just the covariance over a large subset of samples and does not contain any indeterminates. Thus, by the $\eta$-goodness of the samples (\cref{fact:mean_and_cov_of_1-eps}) we have that it is at most $\bigO{\eta^2 \log^2 (1/\eta)} \cdot \Paren{u^\top \Sigma u}^2$.
We now bound \eqref{eqn:expanding-the-first-spect-cov-term}.(1). Using the SoS Cauchy Schwarz inequality (\cref{fact:sos_cs}) we have that 
\begin{equation}
    \begin{split}
        \calA_{\eta}\sststile{}{} \Biggl\{ \Paren{  \frac{1}{n } \sum_{i \in [n]}  r_i (1-w_i) \Paren{ \Iprod{ x_i^*   ,u }^2 -  u^\top \Sigma u }  }^2 \leq \eta \cdot \frac{1}{n } \sum_{i \in [n]}  r_i (1-w_i) \Paren{ \Iprod{ x_i^*   ,u }^2 -  u^\top \Sigma u }^2  \Biggr\} 
    \end{split}
\end{equation}
Note that the terms $\Paren{ \Iprod{ x_i^*   ,u }^2 -  u^\top \Sigma u }^2$ are fixed scalars and thus we can apply~\cref{lemma:subset_selection}.
In combination with the $\eta$-higher-order-goodness of our samples and the fact that $\sum_{i \in [n]} r_i(1-w_i) \leq \eta n$, the above is at most $\bigO{\eta^2 \log^2(1/\eta)} (u^\top \Sigma u)^2$.

We now focus on bounding term \eqref{eqn:cov-spectral-expansion-1}.(2): 
\begin{equation}
\label{eqn:futher-cov-spectral-expansion}
\begin{split}
& \Paren{ \frac{1}{n } \sum_{i \in [n]}  (1-r_i w_i) \Paren{ \Iprod{ x_i'  ,u }^2 -  u^\top \Sigma u } }^2 \\
& \leq \eta \cdot   \Paren{ \frac{1}{n } \sum_{i \in [n]}  (1-r_i w_i) \Paren{ \Iprod{ x_i'  ,u }^2 -  u^\top \Sigma u   }^2  } \\
&  =   \eta \cdot   \Biggl(  \underbrace{ \frac{1}{n } \sum_{i \in [n]}   \Paren{ \Iprod{ x_i'  ,u }^2 -  u^\top \Sigma u   }^2 }_{\eqref{eqn:futher-cov-spectral-expansion}.(1)}   - \underbrace{ \frac{1}{n } \sum_{i \in [n]}   r_i w_i \Paren{ \Iprod{ x_i'  ,u }^2 -  u^\top \Sigma u   }^2}_{\eqref{eqn:futher-cov-spectral-expansion}.(2)}    \Biggr)
\end{split}
\end{equation}
We bound Term \eqref{eqn:futher-cov-spectral-expansion}.1 as follows: 
\begin{equation}
\begin{split}
\calA_\eta \sststile{}{} \Biggl\{ & \frac{1}{n } \sum_{i \in [n]}   \Paren{ \Iprod{ x_i'  ,u }^2 -  u^\top \Sigma u   }^2 \\
& = \frac{1}{n } \sum_{i \in [n]}   \Paren{ \Iprod{ x_i'  ,u }^2 -  u^\top \Sigma' u  + (u^\top \Sigma' u - u^\top \Sigma u ) }^2 \\
& = \frac{1}{n } \sum_{i \in [n]}   \Paren{ \Iprod{ x_i'  ,u }^2 -  u^\top \Sigma' u }^2   +\Paren{ (u^\top \Sigma' u - u^\top \Sigma u ) }^2  \\
&+  \frac{2}{n} \sum_{i \in [n]} \Paren{ \Iprod{ x_i'  ,u }^2 -  u^\top \Sigma' u } \Paren{ (u^\top \Sigma' u - u^\top \Sigma u )  }  \\
& = \frac{1}{n } \sum_{i \in [n]}   \Paren{ \Iprod{ x_i'  ,u }^2 -  u^\top \Sigma' u }^2  +\Paren{ (u^\top \Sigma' u - u^\top \Sigma u ) }^2   \\
& \leq \Paren{ 2 + \eta \log^2(1/\eta) } \Paren{ u^\top \Sigma' u}^2  + \Paren{ (u^\top \Sigma' u - u^\top \Sigma u ) }^2\Biggr\}
\end{split}
\end{equation}\
Note that the first term is equivalent to $\frac{1}{n} \sum_{i\in[n]} \langle x_i', u\rangle^4 - (u^\top \Sigma' u)^2$ and thus by the last constraint we have that it is bounded by $(2 + \eta \log^2(1/\eta)) (u^\top \Sigma' u)^2$. 

We continue with term \eqref{eqn:futher-cov-spectral-expansion}.(2).
Note that we have that
\begin{equation}
\label{eqn:even-more-cov-expansion}
    \begin{split}
    \calA_\eta \sststile{}{} \Biggl\{ &\frac{1}{n } \sum_{i \in [n]}   r_i w_i \Paren{ \Iprod{ x_i'  ,u }^2 -  u^\top \Sigma u   }^2 = \frac{1}{n } \sum_{i \in [n]}   r_i w_i \Paren{ \Iprod{ x_i^*  ,u }^2 -  u^\top \Sigma u   }^2 \\
        &= \frac{1}{n } \sum_{i \in [n]}  \Paren{ \Iprod{ x_i^*  ,u }^2 -  u^\top \Sigma u   }^2 - \frac{1}{n } \sum_{i \in [n]}  (1-r_iw_i) \Paren{ \Iprod{ x_i  ,u }^2 -  u^\top \Sigma u   }^2 \\
        &\geq (2- \eta \log^2(1/\eta)) (u^\top \Sigma u)^2 - \frac{1}{n } \sum_{i \in [n]}  (1-r_iw_i) \Paren{ \Iprod{ x_i^*  ,u }^2 -  u^\top \Sigma u   }^2 \Biggr\} \,,
    \end{split}    
\end{equation}
where in the last inequality we applied the $\eta$-higher-order goodness of the $\Sigma^{-1/2}x_i^*$.
Further, again by $\eta$-higher-order-goodness and the SoS selector lemma (cf.~\cref{lemma:subset_selection}) it holds that
\[ \calA_\eta \sststile{}{} \Biggl \{ \frac{1}{n } \sum_{i \in [n]}   (1-r_i w_i) \Paren{ \Iprod{ x_i^*  ,u }^2 -  u^\top \Sigma u   }^2 \leq \bigO{\eta \log^2(1/\eta)} (u^\top \Sigma u)^2\Biggr\}\]
Combining these bounds we conclude that
\begin{equation}
    \begin{split}
        \calA_\eta \sststile{}{} \Biggl\{ &\eta \cdot   \Biggl( \frac{1}{n } \sum_{i \in [n]}   \Paren{ \Iprod{ x_i'  ,u }^2 -  u^\top \Sigma u   }^2    -  \frac{1}{n } \sum_{i \in [n]}   r_i w_i \Paren{ \Iprod{ x_i'  ,u }^2 -  u^\top \Sigma u   }^2 \\
        &\leq \bigO{\eta} \cdot \left(u^\top \Sigma' u - u^\top \Sigma u\right)^2 \\ &+ \eta \left( (2 + \eta \log^2(1/\eta) ) (u^\top \Sigma' u)^2 - (2  - \eta \log^2(1/\eta))(u^\top \Sigma' u)^2 \right) \\
        &\leq \bigO{\eta^2 \log^2 (1/\eta)} (u^\top \Sigma' u)^2 + \bigO{\eta} \cdot \left(u^\top \Sigma' u - u^\top \Sigma u\right)^2 \Biggr\}\,.
    \end{split}
\end{equation}
Finally we observe that if we combine the bounds on Terms \eqref{eqn:cov-spectral-expansion-1}.1 and \eqref{eqn:cov-spectral-expansion-1}.2 and move the $\bigO{\eta} \cdot \left(u^\top \Sigma' u - u^\top \Sigma u\right)^2$ terms over to the LHS and renormalize we get the desired SoS inequality.
\end{proof}

\subsection{Covariance Estimation in Relative Frobenius Norm}
\label{sec:covariance_rel_frob_norm}

In this section, we will prove the following theorem.
\begin{theorem}
    \label{thm:covariance_rel_frob_norm}
    Let $\cD$ be a mean-zero distribution with covariance $\Sigma$ such that $\Sigma^{-1/2} \cD$ is fourth-moment matching reasonable sub-gaussian.
    Let $\eta$ be smaller than a sufficiently small constant and let $\psi = O(\eta \log(1/\eta))$ and $\Sigma$ be such that $(1-\psi) I_d \preceq \Sigma \preceq (1+\psi) I_d$.
    Let $\beta > 0$.
    Given $0 < \eta$, and $n \geq \tilde{\Omega}(\tfrac{d^2 + \log^2(1/\beta)}{\eta^2})$ points that are an $\eta$-corruption of $n$ i.i.d.\ samples from $\cD$, there exists an algorithm that runs in time $(n \cdot d)^{\mathcal{O}(1)}$ and with probability at least $1-\beta$, outputs $\hat{\Sigma}$ such that
    \[
        \Norm{(\Sigma)^{-1/2} \hat{\Sigma} (\Sigma)^{-1/2} - I_d}_F \leq O(\eta \log(1/\eta)) \,.
    \]
\end{theorem}

We denote the input to the algorithm by $x_1, \ldots, x_n$.
We consider the following SoS program in variables $x_i', w_i, \Sigma'$.
The $x_i'$ are $d$-dimensional vector-valued variables and $\Sigma'$ is a $d\times d$-dimensional matrix-valued variables.
Note that in the SoS proof we search for in the last constraint, $P$ is a $d \times d$ matrix-valued variable.
\begin{equation*}
\calA \colon
  \left \{
    \begin{aligned}
      &\forall i\in [n].
      & w_i^2
      & = w_i \\
      & &\sum_{i\in [n]} w_i &= (1-\eta)n \\
      &\forall i\in [n] & w_i(x'_i - x_i) &=0 \\
      && \Sigma' &= \tfrac{1} n \sum_{i \in [n]} x_i' {x_i'}^{\top} \\
      &\exists \text{ the following SoS proof }  & \sststile{4}{P} \Biggl\{\frac{1}{n} \sum_{i \in [n]} \Iprod{P, x_i'(x_i')^\top  - \Sigma'}^2  \leq & \Paren{2 + \bigO{\eta \log^2(1/\eta)}} \Norm{P}_F^2 \Biggr\}
    \end{aligned}
  \right \}
\end{equation*}
Further, let $r_i \in \Set{0,1}$ be the indicator that $x_i^* = x_i$.
Note that $\sum_{i=1}^n (1-r_i) \leq \eta n$ and $\calA \sststile{O(1)}{} \Set{r_i w_i x_i' = r_i w_i x_i^*}$.
Further, by triangle inequality it also holds that $\calA \sststile{O(1)}{} \Set{\sum_{i=1}^n (1-r_i w_i) \leq 2 \eta n}$.
\paragraph{Feasibility.} 
The feasibility of the above program follows by setting $w_i = r_i$ and $\eta$-goodness and $\eta$-higher-order-goodness of the true samples.
We will give a formal proof in \cref{sec:feasibility}. 

\begin{mdframed}
  \begin{algorithm}[Robust Covariance Estimation in Frobenius Norm]
    \label{algo:cov_estimation}\mbox{}
    \begin{description}
    \item[Input:] $\eta$-corrupted sample $x_1, \ldots, x_n$, corruption level $\eta$.
    \item[Output:] An estimate $\hat{\Sigma}$ of the covariance attaining the guarantees of \cref{thm:covariance_rel_frob_norm}.
    
    \item[Operations:]\mbox{}
    \begin{enumerate}
    \item Find a degree-$\bigO{1}$ pseudo-distribution $\zeta$ satisfying $\calA_{\eta}$.
    \item Output $ \hat{\Sigma} = \pE_{\zeta}[\Sigma']$.
    \end{enumerate}
    \end{description}
  \end{algorithm}
\end{mdframed}

The following is the key lemma.
\begin{lemma}
    \label{thm:covariance_rel_frob_norm_sos}
    Suppose $\eta$ is smaller than a sufficiently small constant and assume that $\Sigma^{-1/2} x_1^*, \ldots, \Sigma^{-1/2} x_n^*$ are $\eta$-good and $\eta$-higher-order-good.
    Then, with probability at least $1-\beta$, for all $P \in \R^{d\times d}$ such that $\norm{P}_F = 1$, the following SoS proof exists
    \[
        \calA \sststile{O(1)}{} \Set{\Iprod{P, (\Sigma)^{-1/2} \Sigma' (\Sigma)^{-1/2} - I_d}^2 \leq O(\eta^2 \log^2(1/\eta)) }\,.
    \]
\end{lemma}
Given this lemma, we can deduce~\cref{thm:covariance_rel_frob_norm}.
\begin{proof}[Proof of~\cref{thm:covariance_rel_frob_norm}]
    Note that by~\cref{fact:mean_and_cov_of_1-eps,fact:cov_small_subset} and~\cref{lem:higher-order-stability}, with probability at least $1-\beta$, $\Sigma^{-1/2} x_1^*, \ldots, \Sigma^{-1/2} x_n^*$ are $\eta$-good and $\eta$-higher-order-good.
    Thus, by~\cref{sec:cov-feasibility} we have that with probability $1-\delta$ the program is feasible.
    Let $\zeta$ be a degree-$O(1)$ pseudo-expectation satisfying $\calA$ and let $\hat{\Sigma} = \pE_\zeta \Sigma'$.
    Let
    \[
        U = \frac{(\Sigma)^{-1/2} \hat{\Sigma}(\Sigma)^{-1/2} - I_d}{\norm{(\Sigma)^{-1/2} \hat{\Sigma}(\Sigma)^{-1/2} - I_d}_F} \,.
    \]
    Then by H\"older's inequality for pseudo-expectations (cf.~\cref{fact:pseudo-expectation-holder}) and~\cref{thm:covariance_rel_frob_norm_sos}, it holds that
    \begin{align*}
        \Norm{(\Sigma)^{-1/2} \hat{\Sigma}(\Sigma)^{-1/2} - I_d}_F^2 &= \Iprod{U, (\Sigma)^{-1/2} \hat{\Sigma}(\Sigma)^{-1/2} - I_d}^2 = \Paren{\pE_\zeta \Iprod{U, (\Sigma)^{-1/2} \Sigma' (\Sigma)^{-1/2} - I_d}}^2 \\
        &\leq \pE_\zeta \Iprod{U, (\Sigma)^{-1/2} \Sigma'(\Sigma)^{-1/2} - I_d}^2 \leq O(\eta^2 \log^2(1/\eta)) \,.
    \end{align*}
    The claim now follows by taking square roots.
\end{proof}

The following lemma captures the essence of the proof of~\cref{thm:covariance_rel_frob_norm_sos}.
It shows that our constraints allow us to derive (an approximate version of) the second-order stability conditions for both the true samples as well as the "SoS" samples.
\begin{lemma}
    \label{lem:cov_rel_helper}
    Suppose $\eta$ is smaller than a sufficiently small constant and assume that $\Sigma^{-1/2} x_1^*, \ldots, \Sigma^{-1/2} x_n^*$ are $\eta$-good and $\eta$-higher-order-good.
    For a matrix $M \in \R^{d \times d}$, let $\tilde{M} = (\Sigma)^{-1/2} M (\Sigma)^{-1/2}$.
    With probability at least $1-\beta$ for all $P \in \R^{d \times d}$ such that $\norm{P}_F = 1$, then the following two SoS proof exists
    \begin{align*}
        \calA &\sststile{O(1)}{} \tfrac 1 n \sum_{i=1}^n (1-r_i w_i) \Iprod{\tilde{P}, x_i^* (x_i^*)^\top - \Sigma}^2 \leq O(\eta \log^2(1/\eta)) \,, \\
        \calA &\sststile{O(1)}{} \tfrac 1 n \sum_{i=1}^n (1-r_i w_i) \Iprod{\tilde{P}, x_i' (x_i')^\top - \Sigma'}^2 \leq O(\eta \log^2(1/\eta)) + (2 + O(\eta)) \Iprod{\tilde{P}, \Sigma' - \Sigma}^2 \,.
    \end{align*}
\end{lemma}
\begin{proof}
    By $\eta$-higher-order goodness we know that for every set $T \subseteq [n]$ of size at most $\eta n$ it holds that
    \[
        \tfrac 1 n \sum_{i \in T} \Iprod{\tilde{P}, x_i^* (x_i^*)^\top - \Sigma}^2 = \tfrac 1 n \sum_{i \in T} \Iprod{P, [(\Sigma)^{-1/2}x_i^*] [(\Sigma)^{-1/2}x_i^*]^\top - I_d}^2 \leq O\Paren{\eta \log^2 (1/\eta)} \,.
    \]
    The first SoS proof now follows directly by the SoS Selector lemma (\cref{lemma:subset_selection}).

    For the second one, first note that by the SoS triangle inequality
    \begin{align*}
        \calA \sststile{O(1)}{} &\tfrac 1 n \sum_{i=1}^n (1-r_i w_i) \Iprod{\tilde{P}, x_i' (x_i')^\top - \Sigma'}^2 \\
        &\leq \tfrac 2 n\sum_{i=1}^n (1-r_i w_i) \Iprod{\tilde{P}, x_i' (x_i')^\top - \Sigma}^2 + O(\eta) \Iprod{\tilde{P}, \Sigma - \Sigma'}^2 \,. 
    \end{align*}
    We focus on bounding the first sum
    \begin{equation}
    \begin{split}
    \label{eqn:cov_helper}
        \calA &\sststile{O(1)}{} \tfrac 1 n\sum_{i=1}^n (1-r_i w_i) \Iprod{\tilde{P}, x_i' (x_i')^\top - \Sigma}^2 \\
        &= \underbrace{\tfrac 1 n\sum_{i=1}^n  \Iprod{\tilde{P}, x_i' (x_i')^\top - \Sigma}^2}_{\ref{eqn:cov_helper}.1} - \underbrace{\tfrac 1 n\sum_{i=1}^n r_i w_i \Iprod{\tilde{P}, x_i' (x_i')^\top - \Sigma}^2}_{\ref{eqn:cov_helper}.2} \,.
    \end{split}
    \end{equation}
    We will first derive a lower bound on \ref{eqn:cov_helper}.2 and then an upper bound on \ref{eqn:cov_helper}.1.
    We start by observing that
    \begin{align*}
        \calA &\sststile{O(1)}{} \tfrac 1 n\sum_{i=1}^n r_i w_i \Iprod{\tilde{P}, x_i^* (x_i^*)^\top - \Sigma}^2 \\
        &= \tfrac 1 n\sum_{i=1}^n \Iprod{\tilde{P}, x_i^* (x_i^*)^\top - \Sigma}^2 - \tfrac 1 n\sum_{i=1}^n (1 - r_i w_i) \Iprod{\tilde{P}, x_i^* (x_i^*)^\top - \Sigma}^2 \,.
    \end{align*}
    By $\eta$-higher-order-goodness the first term is at least $2 - O(\eta \log^2(1/\eta))$.
    The second term is at most $O(\eta \log^2(1/\eta))$ by what we proved above.
    Hence, in total the term \ref{eqn:cov_helper}.2 is at least $2-O(\eta \log^2 (1/\eta))$.
    For the the term \ref{eqn:cov_helper}.1, we notice that
    \begin{align*}
        \calA &\sststile{O(1)}{} \tfrac 1 n\sum_{i=1}^n  \Iprod{\tilde{P}, x_i' (x_i')^\top - \Sigma}^2 = \tfrac 1 n\sum_{i=1}^n  \Iprod{\tilde{P}, x_i' (x_i')^\top - \Sigma' + \Sigma' - \Sigma}^2 \\
        &= \tfrac 1 n\sum_{i=1}^n  \Iprod{\tilde{P}, x_i' (x_i')^\top - \Sigma'}^2 + 2 \Paren{\tfrac 1 n \sum_{i=1}^n \Iprod{\tilde{P}, x_i' (x_i')^\top - \Sigma'}} \Iprod{\tilde{P}, \Sigma' - \Sigma} + \Iprod{\tilde{P}, \Sigma' - \Sigma}^2 \,.
    \end{align*}
    The middle term is zero by definition of $\Sigma'$.
    By our last constraint, the first term is at most
    \[
        \Paren{2+ O(\eta \log(1/\eta))} \Norm{\tilde{P}}_F^2 \leq \Paren{2+ O(\eta \log(1/\eta))} \Paren{1+ O(\eta \log(1/\eta))}^2 \leq 2+ O(\eta \log(1/\eta)) \,.
    \]
    Hence, putting everything together, we have shown that
    \[
        \calA \sststile{O(1)}{} \tfrac 1 n \sum_{i=1}^n (1-r_i w_i) \Iprod{\tilde{P}, x_i' (x_i')^\top - \Sigma'}^2 \leq  O(\eta \log^2(1/\eta)) + (2 + O(\eta)) \Iprod{\tilde{P}, \Sigma' - \Sigma}^2 \,.
    \]
\end{proof}

With this hand, we will now proof~\cref{thm:covariance_rel_frob_norm_sos}:
\begin{proof}[Proof of~\cref{thm:covariance_rel_frob_norm_sos}]
    Let $P \in \R^{d \times d}$ be a fixed matrix of unit Frobenius norm and define $\tilde{P} = (\Sigma)^{-1/2} P (\Sigma)^{-1/2}$.
    Further, let $\bar{\Sigma} = \tfrac 1 n \sum_{i=1}^n x_i^* (x_i^*)^\top$.
    Then by SoS triangle inequality (cf.~\cref{fact:sos-almost-triangle})
    \[
        \calA \sststile{O(1)}{}\Iprod{P, (\Sigma)^{-1/2} \Sigma' (\Sigma)^{-1/2} - I_d}^2 = \Iprod{\tilde{P}, \Sigma' - \Sigma}^2 \leq 2 \Iprod{\tilde{P}, \Sigma' - \bar{\Sigma}}^2 + 2 \Iprod{\tilde{P}, \Sigma - \bar{\Sigma}}^2 \,.
    \]
    Note that by~\cref{lem:higher-order-stability} (Point 1), it holds that
    \begin{align*}
        \Iprod{\tilde{P}, \Sigma - \bar{\Sigma}}^2 &= \Iprod{P, I_d - \tfrac 1 n \sum_{i=1}^n \Brac{(\Sigma)^{-1/2} x_i^*} \Brac{(\Sigma)^{-1/2} x_i^*}^\top}^2 \\
        &\leq \Norm{I_d - \tfrac 1 n \sum_{i=1}^n \Brac{(\Sigma)^{-1/2} x_i^*} \Brac{(\Sigma)^{-1/2} x_i^*}^\top}_F^2 \leq O(\eta^2 \log^2(1/\eta)) \,.
    \end{align*}
    Thus, it is enough to show that $\calA \sststile{O(1)}{} \Iprod{\tilde{P}, \Sigma' - \bar{\Sigma}}^2 \leq O(\eta^2 \log^2(1/\eta))$.
    We will try to reduce to a state in which we can apply~\cref{lem:cov_rel_helper}.
    We start by observing that
    \begin{align*}
        &\calA \sststile{O(1)}{} \Iprod{\tilde{P}, \Sigma' - \bar{\Sigma}}^2 = \Paren{\tfrac 1 n \sum_{i=1}^n \Iprod{\tilde{P}, x_i' (x_i')^\top} - \tfrac 1 n \sum_{i=1}^n \Iprod{\tilde{P}, x_i^* (x_i^*)^\top}}^2 \\
        &= \Paren{\tfrac 1 n \sum_{i=1}^n (1- w_i r_i) \Brac{\Iprod{\tilde{P}, x_i' (x_i')^\top - x_i^* (x_i^*)^\top}}}^2 \\
        &= (\tfrac 1 n \sum_{i=1}^n (1- w_i r_i) \Brac{\Iprod{\tilde{P}, x_i' (x_i')^\top - \Sigma'}} + \tfrac 1 n \sum_{i=1}^n (1- w_i r_i)\Brac{\Iprod{\tilde{P},\Sigma - x_i^* (x_i^*)^\top}} \\
        &+ \Paren{\tfrac 1 n \sum_{i=1}^n (1- w_i r_i)} \Brac{\Iprod{\tilde{P}, \Sigma' - \Sigma}})^2 \\
        &\leq 2 \Paren{\tfrac 1 n \sum_{i=1}^n (1- w_i r_i) \Brac{\Iprod{\tilde{P}, x_i' (x_i')^\top - \Sigma'}}}^2 + 2 \Paren{\tfrac 1 n \sum_{i=1}^n (1- w_i r_i)\Brac{\Iprod{\tilde{P},\Sigma -  x_i^* (x_i^*)^\top}}}^2 \\
        &+ 2\Paren{\Paren{\tfrac 1 n \sum_{i=1}^n (1- w_i r_i)} \Brac{\Iprod{\tilde{P}, \Sigma' - \Sigma}}}^2 \,.
    \end{align*}
    The last term is at most $2 \eta^2 \Iprod{\tilde{P}, \Sigma' - \Sigma}^2$ which is a small multiple of the left-hand side we started with.
    By subtracting this and renormalizing, we can hence ignore the last term.
    Combining all of the above and using the SoS Cauchy-Schwarz inequality, it follows that 
    \begin{equation}
    \label{eqn:cov_rel_main_terms}
    \begin{split}
    \calA \sststile{O(1)}{} \Iprod{\tilde{P}, \Sigma' - \Sigma}^2 &\leq O(1) \cdot \Paren{\tfrac 1 n \sum_{i=1}^n (1-r_i w_i)}\Paren{\tfrac 1 n \sum_{i=1}^n (1-r_iw_i) \Iprod{\tilde{P}, x_i' (x_i')^\top - \Sigma'}^2} \\
    &+ O(1) \cdot \Paren{\tfrac 1 n \sum_{i=1}^n (1-r_i w_i)} \Paren{\tfrac 1 n \sum_{i=1}^n (1-r_iw_i) \Iprod{\tilde{P}, x_i^* (x_i^*)^\top - \Sigma}^2} \\
    &= O(\eta) \cdot \underbrace{\Paren{\tfrac 1 n \sum_{i=1}^n (1-r_iw_i) \Iprod{\tilde{P}, x_i' (x_i')^\top - \Sigma'}^2}}_{\ref{eqn:cov_rel_main_terms}.1} \\
    &+ O(\eta) \cdot \underbrace{\Paren{\tfrac 1 n \sum_{i=1}^n (1-r_iw_i) \Iprod{\tilde{P}, x_i^* (x_i^*)^\top - \Sigma}^2}}_{\ref{eqn:cov_rel_main_terms}.2}
    \end{split}
    \end{equation}
    By~\cref{lem:cov_rel_helper} it follows  \ref{eqn:cov_rel_main_terms}.2 is at most $O(\eta \log^2 (1/\eta))$ and \ref{eqn:cov_rel_main_terms}.1 is at most $O(\eta \log^2(1/\eta)) + O(1) \Iprod{\tilde{P}, \Sigma' - \Sigma}^2$.
    Hence,
    \[
        \calA \sststile{O(1)}{} \Iprod{\tilde{P}, \Sigma' - \Sigma}^2 \leq O(\eta^2 \log^2(1/\eta)) + O(\eta) \Iprod{\tilde{P}, \Sigma' - \Sigma}^2 \,.
    \]
    Thus, by subtracting the last term on the right-hand side and renormalizing, we obtain the claim.
\end{proof}

\section{Feasibility and Concentration bounds}

\subsection{$\eta$-goodness}
\label{sec:eps_goodness_proof}

In this section, we will prove \cref{fact:mean_and_cov_of_1-eps,fact:cov_small_subset}.
We start with \cref{fact:cov_small_subset} restated below.
\restatefact{fact:cov_small_subset}

In fact, it will follow by the more general version below (\cref{lem:moment_control}) which also gives a guarantees for the sample mean over small subsets which will be used in the proof of \cref{fact:mean_and_cov_of_1-eps}.
\cref{fact:cov_small_subset} follows from the first part and the triangle inequality.
Note that in the proof below we prove that the bound in \cref{fact:cov_small_subset} holds with probability $1-\delta$, this of course does not make a difference.
\begin{proof}[Proof of \cref{fact:cov_small_subset}]
    We condition on the event that the conclusion of \cref{lem:moment_control} holds.
    Let $\calT \subset \calS$ be of size $\eta n$.
    Using the triangle inequality, $\sqrt{ab} \leq a + b$ for $a > 0, b > 0$ and for $\eta \leq 1/e$ we have that $\eta \leq \eta \log(1/\eta)$, we obtain
    \begin{align*}
        \Norm{\frac{1}{(1-\eta) n} \sum_{x_i^* \in \calT} \Paren{x_i^* - \mu} \Paren{x_i^* - \mu}^\top }_{\textrm{op}} &\leq \bigO{\eta \log(1/\eta) + \sqrt{\eta}\Paren{ \sqrt{\frac{d}{n}} + \sqrt{\frac{\log(1/\delta)}{n}}} + {\frac{d}{n}} + {\frac{\log(1/\delta)}{n}}} \\
&\leq \bigO{\eta \log(1/\eta) + {\frac{d + \log(1/\delta)}{n}}}
    \end{align*}
    as desired.
\end{proof}

\begin{lemma}\label{lem:moment_control}
Let $0 \leq \psi \leq O(\eta \sqrt{\log(1/\eta)})$. Let $\calS = \Set{x_1^*, \ldots, x_n^*}$ be \iid samples from a $d$-dimensional sub-gaussian distribution  with mean $\mu$ and covariance $\Sigma$ such that $(1-\psi) I_d \preceq \Sigma \preceq (1+\psi)I_d$. Let $\calT_{\eta} = \Set{ \calT \subset S \suchthat \Card{\calT} = \eta n }$ be the collection of subsets of $\calS$ size $\eta n$. Then, for $\eta \leq 1/e$, we have that with probability at least $1 - \delta$,
\[ \sup_{\calT \in \calT_\eta} \apnorm{\frac{1}{\eta n} \sum_{x_i^* \in \calT} (x_i^* - \mu)(x_i^* - \mu)^T- I_d}{op}  \leq \bigO{\sqrt{\frac{d}{\eta n}} + \sqrt{\frac{\log(1/\delta)}{\eta n}} + {\frac{d}{\eta n}} + {\frac{\log(1/\delta)}{\eta n}} + \log(1/\eta) + \psi} \]
and
\[ \sup_{\calT \in \calT_\eta} \apnorm{\frac{1}{(1-\eta)n} \sum_{x_i^* \in \calT} (x_i^* - \mu)}{2} \leq \bigO{\eta \sqrt{\log(1/\eta)} + \sqrt{\frac{d +\log(1/\delta)}{n}}} \,. \]
\end{lemma}
\begin{proof}
Using a union bound on Fact~\ref{fact:empirical_cov_bound} for $|\calT| = \binom{n}{\eta n}$ subsets each of size $\eta n$, we get that with probability at least $1 - \tfrac \delta 2$,
\begin{align*}
&\sup_{\calT \in \calT_\eta} \apnorm{\frac{1}{\eta n} \sum_{x_i^* \in \calT}  (x_i^* - \mu)(x_i^* - \mu)^T - I_d}{op} \\
&\leq \bigO{\sqrt{\frac{d}{\eta n}} + \sqrt{\frac{\log(1/\delta)}{\eta n}} + {\frac{d}{\eta n}} + {\frac{\log(1/\delta)}{\eta n}} + \sqrt{\log(1/\eta)} + \log(1/\eta)} \\
&\leq \bigO{\sqrt{\frac{d}{\eta n}} + \sqrt{\frac{\log(1/\delta)}{\eta n}} + {\frac{d}{\eta n}} + {\frac{\log(1/\delta)}{\eta n}} + \log(1/\eta)} \,,
\end{align*}
where we additionally used that $\log(\binom{n}{\eta n}) \leq 2\eta n \log(1/\eta)$ and that $\eta \leq 1/e$ implies 
\[ 1 \leq \sqrt{\log(1/\eta)} \leq \log(1/\eta). \]

An analogous union bound argument relying on \cref{fact:empirical_mean_bound} instead of \cref{fact:empirical_cov_bound}, and using $\eta \leq 1$ instead of $\sqrt{ab} \leq a + b$, gives us that with probability $1 - \tfrac \delta 2$,
\[ \sup_{T \in \calT_\eta} \apnorm{\frac{1}{(1-\eta)n} \sum_{x_i^* \in \calT} (x_i^* - \mu)}{2} \leq \bigO{\eta \sqrt{\log(1/\eta)} + \sqrt{\frac{d + \log(1/\delta)} n}} \,.\]
\end{proof}

We next give the proof of \cref{fact:mean_and_cov_of_1-eps}.
\restatefact{fact:mean_and_cov_of_1-eps}
Again, we will prove that the bounds hold with probability $1-\delta$.
\begin{proof}[Proof of \cref{fact:mean_and_cov_of_1-eps}]
We know that with probability $1-\delta$ the bounds in \cref{lem:moment_control} hold.
We henceforth condition on this event.
Let $\calT \subseteq \calS$ be of size $(1-\eta)n$ and denote by $\calT^c$ its complement in $\calS$.
Then it holds that
\begin{align*}
    \Norm{\mu - \frac{1}{(1-\eta)n}\sum_{x_i^* \in \calT} x_i^* }_2 &= \Norm{\frac 1 {(1-\eta)n} \sum_{i=1}^n (x_i^* - \mu) - \frac 1 {(1-\eta)n} \sum_{x_i^* \in \calT^c} (x_i^* - \mu)}_2 \\
    &\leq \Norm{\frac 1 {(1-\eta)n} \sum_{i=1}^n (x_i^* - \mu)}_2 + \Norm{\frac 1 {(1-\eta)n} \sum_{x_i^* \in \calT^c} (x_i^* - \mu)}_2 \\
    &\leq \bigO{\eta \sqrt{\log(1/\eta)} + \sqrt{\frac{d+\log(1/\delta)}{n}}} \,, 
\end{align*}
where the last inequality follows from the second part of \cref{lem:moment_control} and \cref{fact:empirical_mean_bound}.
We obtain the following bound for the empirical covariance over $\calT$ in a similar way.
However, first observe that by an argument analogous to that in the proof of \cref{fact:cov_small_subset} it follows that
\[
\Norm{\frac{1}{(1-\eta) n} \sum_{x_i^* \in \calT^c} \left[\Paren{x_i^* - \mu} \Paren{x_i^* - \mu}^\top - I_d\right]}_{\textrm{op}} \leq \bigO{\eta \log(1/\eta) + \frac{d+\log(1/\delta)}{n} + \psi} \,.
\]
Using \cref{fact:empirical_cov_bound}, it follows that 
\begin{align*}
    &\Norm{\frac{1}{(1-\eta) n} \sum_{x_i^* \in \calT} \Paren{x_i^* - \mu} \Paren{x_i^* - \mu}^\top  - I_d}_{\textrm{op}} \\
    = &\Norm{\frac{1}{(1-\eta) n} \sum_{i=1}^n \left[\Paren{x_i^* - \mu} \Paren{x_i^* - \mu}^\top - I_d\right] - \frac{1}{(1-\eta) n} \sum_{x_i^* \in \calT^c} \left[\Paren{x_i^* - \mu} \Paren{x_i^* - \mu}^\top - I_d\right] }_{\textrm{op}} \\
    \leq &\bigO{\eta \log(1/\eta) + \sqrt{\frac d n} + \sqrt{\frac {\log(1/\delta)} n} + \frac {d + \log(1/\delta)} n + \psi} \,.
\end{align*}
\end{proof}

The proofs above use the following two facts.
The first one follows from \cite[Theorem 6.5]{wainwright_2019}.
\begin{fact}
\label{fact:empirical_cov_bound}
Let $0 \leq \psi \leq O(\eta \sqrt{\log(1/\eta)})$. Let $x_1^*, \ldots, x_n^*$ be \iid samples from a $d$-dimensional sub-gaussian distribution  with mean $\mu$ and covariance $\Sigma$ such that $(1-\psi) I_d \preceq \Sigma \preceq (1+\psi) I_d$. Then we have that with probability at least $1-\delta$,
\[ \apnorm{\frac{1}{n} \sum_{i=1}^n (x_i^* - \mu)(x_i^* - \mu)^T - I_d}{op} \leq \bigO{\sqrt{\frac{d}{n}} + \sqrt{\frac{\log(1/\delta)}{n}} + {\frac{d}{n}} + {\frac{\log(1/\delta)}{n}} + \psi} \,.  \]
\end{fact}
The second one follows from standard concentration bounds of sub-gaussian distributions and an $\eta$-net argument
\begin{fact}
    \label{fact:empirical_mean_bound}
    Let $0\leq \psi \leq O(\eta \sqrt{\log(1/\eta)})$ and let $x_1^*, \ldots, x_n^*$ be \iid samples from a $d$-dimensional sub-gaussian distribution  with mean $\mu$ and covariance $\Sigma$ such that $(1-\psi)I_d \preceq \Sigma \preceq (1+\psi)I_d$. Then we have that with probability at least $1-\delta$,
    \[
    \Norm{ \Paren{\frac 1 n \sum_{i=1}^n x_i^*} - \mu}_2 \leq \bigO{\sqrt{\frac {d + \log(1/\delta)} n}} \,.
    \]
\end{fact}

\subsection{Feasibility}
\label{sec:feasibility}

\subsubsection{Certifiably Hypercontractivity via Concentration}
In this section we show that the samples $x_i$ are certifiably hypercontractive whenever higher order goodness holds. In particular, we show that they satisfy $4-2$ certifiable hypercontractivity or that there exists an SoS proof that
\[ \sum_{i=1}^n \langle x_i, v \rangle^4 \leq \left(3 + \bigO{\eta \log^2(1/\eta)} \right) \cdot \left(v^\top \left(\frac{1}{n} \sum_{i \in [n]} x_i x_i^\top\right) v\right)^2\,. \]
This will be used in the feasibility arguments for regression, covariance aware mean estimation, and covariance estimation in spectral norm. 

\begin{fact}[Certifiably Hypercontractivity from Concentration]
\label{fact:cert_bounded_moments_from_concentration}
    Let $\Sigma^{-1/2}x_1, \ldots, \Sigma^{-1/2}x_n$ be $\eta$-higher-order-good.
    I.e., in particular let them satisfy
    \[ \abs{\sum_{i \in [n]} \iprod{(\Sigma^{-1/2}x_i)(\Sigma^{-1/2}x_i)^\top - I_d}{P}^2 - 2} \leq O(\eta \log^2 ( 1 /\eta)) \,,\]
    for all fixed $P \in \mathbb{R}^{d \times d}$ such that $\norm{P}_F = 1$. Then we have that
    \[ \sststile{v}{} \Set{\frac{1}{n} \sum_{i=1}^n \langle x_i, v \rangle^4 \leq \left(3 + \bigO{\eta \log^2(1/\eta)} \right) \cdot \left(v^\top \left(\frac{1}{n} \sum_{i \in [n]} x_i x_i^\top\right) v\right)^2} \,.\]
\end{fact}
\begin{proof}
     Let $\opvec(\cdot)$ be the operation that flattens a $d\times d$ matrix into a $d^2$-dimensional vector. We have that
     \begin{align*}
         \sststile{v}{4} \Biggl\{\frac{1}{n} \sum_{i=1}^n \langle x, v \rangle^4 &= \Paren{v^{\otimes 2}}^\top \Paren{\frac{1}{n} \sum_{i=1}^n \Paren{x_i^{\otimes 2}}\Paren{x_i^{\otimes 2}}^\top} \Paren{v^{\otimes 2}} \\
         &= \Paren{v^{\otimes 2}}^\top \Paren{\frac{1}{n} \sum_{i=1}^n \Paren{x_i^{\otimes 2} - \opvec(\Sigma) + \opvec(\Sigma)}\Paren{x_i^{\otimes 2} - \opvec(\Sigma) + \opvec(\Sigma)}^\top} \Paren{v^{\otimes 2}} \Biggr\} \,.
     \end{align*}
     We now expand the inner product. We thus have that
     \begin{align*}
         \sststile{v}{} \Biggl\{ \frac{1}{n} \sum_{i=1}^n \langle x, v \rangle^4 &= \underbrace{\Paren{v^{\otimes 2}}^\top \Paren{\frac{1}{n} \sum_{i=1}^n \Paren{x_i^{\otimes 2} - \opvec(\Sigma)}\Paren{x_i^{\otimes 2} - \opvec(\Sigma)}^\top} \Paren{v^{\otimes 2}}}_{\text{Term A}} \\
         &\qquad + \underbrace{\Paren{v^{\otimes 2}}^\top \Paren{\frac{1}{n} \sum_{i=1}^n \opvec(\Sigma)\Paren{x_i^{\otimes 2} - \opvec(\Sigma)}^\top} \Paren{v^{\otimes 2}}}_{\text{Term B}} \\
         &\qquad + \underbrace{\Paren{v^{\otimes 2}}^\top \Paren{\frac{1}{n} \sum_{i=1}^n \Paren{x_i^{\otimes 2} - \opvec(\Sigma)}\opvec(\Sigma)^\top} \Paren{v^{\otimes 2}}}_{\text{Term C}} \\
         &\qquad + \underbrace{\Paren{v^{\otimes 2}}^\top \Paren{\frac{1}{n} \sum_{i=1}^n \opvec(\Sigma)\opvec(\Sigma)^\top} \Paren{v^{\otimes 2}}}_{\text{Term D}} \Biggr\}\,.
     \end{align*}
     We first observe that Term D is exactly equal to $\left(v^\top \Sigma v\right)^2$. We now bound Term C. Rewriting this term we have that
     \begin{align*}
         \sststile{v}{}& \Biggl\{\Paren{v^{\otimes 2}}^\top \Paren{\frac{1}{n} \sum_{i=1}^n \Paren{x_i^{\otimes 2} - \opvec(\Sigma) }\opvec(\Sigma)^\top} \Paren{v^{\otimes 2}} \\
         &= \Paren{v^{\otimes 2}}^\top \Paren{\frac{1}{n} \sum_{i=1}^n \Paren{x_i^{\otimes 2} - \opvec(\Sigma)}} \left(v^\top \Sigma v\right) \Biggr\}\,.
     \end{align*}
     The inner product of the first two terms is equal to $v^\top \left(\frac{1}{n} \sum_{i=1}^n x_i x_i^\top - \Sigma\right)v$ and is thus bounded by $\bigO{\eta \cdot (v^\top \Sigma v)}$ by~\cref{fact:empirical_cov_bound}. We conclude that
     \[ \sststile{v}{}\Set{\Paren{v^{\otimes 2}}^\top \Paren{\frac{1}{n} \sum_{i=1}^n \Paren{x_i^{\otimes 2} - \opvec(\Sigma)}\opvec(\Sigma)^\top} \Paren{v^{\otimes 2}} \leq \bigO{\eta \cdot \left(v^\top \Sigma v\right)^2}}\,.\]
     The bound on Term B is analogous and we have that 
     \[ \sststile{v}{} \Set{\Paren{v^{\otimes 2}}^\top \Paren{\frac{1}{n} \sum_{i=1}^n \opvec(\Sigma)\Paren{x_i^{\otimes 2} - \opvec(\Sigma)}^\top} \Paren{v^{\otimes 2}} \leq \bigO{\eta \cdot \left(v^\top \Sigma v\right)^2}}\,.\]
     Finally, we bound Term A. Note that it suffices to bound the spectral norm of the following matrix:
     \[ \frac{1}{n} \sum_{i=1}^n \Paren{x_i^{\otimes 2} - \opvec(\Sigma)}\Paren{x_i^{\otimes 2} - \opvec(\Sigma)}^\top\,.\]
     Consider the expression $u^\top \Paren{\frac{1}{n} \sum_{i=1}^n \Paren{x_i^{\otimes 2} - \opvec(\Sigma)}\Paren{x_i^{\otimes 2} - \opvec(\Sigma)}^\top} u$ for any fixed unit vector $u\in\mathbb{R}^{d^2}$. Note that we have that 
     \[ u^\top \Paren{\frac{1}{n} \sum_{i=1}^n \Paren{x_i^{\otimes 2} - \opvec(\Sigma)}\Paren{x_i^{\otimes 2} - \opvec(\Sigma)}^\top} u = \frac{1}{n} \sum_{i=1}^n \langle x_i^{\otimes 2} - \opvec(\Sigma), u \rangle^2\,.\]
     Let $U$ be the reshaping of the $d^2$-dimensional vector $u$ into a $d \times d$ matrix. We thus have that this expression is also equivalent to $\frac{1}{n} \sum_{i=1}^n \langle x_ix_i^\top - \Sigma, U\rangle^2$. We can then apply the concentration assumption to bound this term by $\left(2 + \eta \log^2(1/\eta)\right) \cdot \left(v^\top \Sigma v\right)^2$. Thus, using the spectral norm bound we can conclude that
     \[ \sststile{v}{} \Set{\Paren{v^{\otimes 2}}^\top \Paren{\frac{1}{n} \sum_{i=1}^n \Paren{x_i^{\otimes 2} - \opvec(\Sigma)}\Paren{x_i^{\otimes 2} - \opvec(\Sigma)}^\top} \Paren{v^{\otimes 2}} \leq \left(2 + \eta \log^2(1/\eta)\right) \cdot \left(v^\top \Sigma v\right)^2}\,,\]
     and combining the bounds on Terms A, B, C, and D we get that 
     \[ \sststile{v}{} \Set{\frac{1}{n} \sum_{i=1}^n \langle x, v \rangle^4 \leq \left(3 + 3\eta \log^2(1/\eta) \right) \cdot \left(v^\top \Sigma v\right)^2} \,.\]
     Finally we note that $\Sigma$ and the empirical covariance $\frac{1}{n} \sum_{i \in [n]} x_i x_i^\top$ have difference at most $\bigO{\eta}$ in operator norm and thus we have that
     \[ \sststile{v}{} \Set{\frac{1}{n} \sum_{i=1}^n \langle x, v \rangle^4 \leq \left(3 + 3\eta \log^2(1/\eta) \right) \cdot \left(v^\top \left(\frac{1}{n} \sum_{i \in [n]} x_i x_i^\top\right) v\right)^2} \,.\]
\end{proof}

\subsubsection{Regression}
\label{sec:feasibility_regression}

We now show that the program in~\cref{sec:robust_regression_oneshot} is feasible with probability at least $1-\delta$. Simply, we let $x_1^*, \ldots, x_n^*$ be the uncorupted samples and set $w_i = r_i$, $\theta' = \thetals$. The first four constraints are trivially satisfied and we have that the fifth constraint is satisfied by the definition of $\thetals$. We now consider the sixth constraint.

\begin{fact}
    Let $x_1^*, \ldots, x_n^*$ be $n \geq \tilde{\Omega}\paren{\frac{d^2 + \log^2(1/\beta)}{\eta^2}}$ i.i.d samples from $N(0, \Sigma)$ and let $y_i^* = \langle x_i^*, \thetastar \rangle + \zeta_i$ where $\zeta_i$ are i.i.d. samples from $N(0,1)$. Then we have that with probability $1-\beta$
    \[ \frac{1}{n} \sum_{i \in [n]} \Paren{y_i^* - \langle \thetals, x_i^* \rangle}^2 \leq 1 + O(\eta) \,.\]
\end{fact}
\begin{proof}
    We have that 
    \begin{align*}
        \frac{1}{n} \sum_{i \in [n]} \Paren{y_i^* - \langle \thetals, x_i^* \rangle}^2 &= \frac{1}{n} \sum_{i \in [n]} \Paren{y_i^* - \langle \thetastar, x_i^* \rangle}^2 + \frac{2}{n} \sum_{i \in [n]} \Paren{y_i^* - \langle \thetastar, x_i^* \rangle}\langle \thetastar - \thetals, x_i^*\rangle \\
        &+ \frac{1}{n} \sum_{i \in [n]} \langle \thetastar - \thetals, x_i^*\rangle^2\,.
    \end{align*}
    The first term is equal to $\frac{1}{n}\sum_{i \in [n]} \zeta_i^2$ which is bounded by $1+ \bigO{\eta}$ by~\cref{fact:empirical_cov_bound}. The last term is equal to $(\thetastar - \thetals)^\top \left(\frac{1}{n} \sum_{i \in [n]} x_i^* (x_i^*)^\top \right) (\thetastar - \thetals)$ which is bounded by $\bigO{1} \cdot \norm{\Sigma^{1/2}(\thetastar - \thetals)}_2^2$ by~\cref{fact:empirical_cov_bound}. Note that for the true least-squares solution we have that \[\norm{\Sigma^{1/2}(\thetastar - \thetals)} \leq \bigO{\sqrt{\frac{d+\log(1/\delta)}{n}}} \leq \bigO{\eta}\] and thus we conclude that the last term is bounded by $\bigO{\eta}$.
    
    Finally, we consider the cross-term. Applying Cauchy Schwartz we have that 
    \[\frac{2}{n} \sum_{i \in [n]} \Paren{y_i^* - \langle \thetastar, x_i^* \rangle}\langle \thetastar - \thetals, x_i^*\rangle \leq \sqrt{\frac{1}{n} \sum_{i\in[n]} \zeta_i^2} \cdot \left\Vert\left(\frac{1}{n}\sum_{i \in [n]} x_i^*(x_i^*)^\top\right)^{1/2}  (\theta^*-\thetals) \right\Vert,.\]
    By~\cref{fact:empirical_cov_bound} we have that $\left\Vert\left(\frac{1}{n}\sum_{i \in [n]} x_i^*(x_i^*)^\top\right)^{1/2}  (\theta^*-\thetals) \right\Vert \leq \bigO{\norm{\Sigma^{1/2}(\theta^*-\thetals)}}$ which is at most $\bigO{\eta}$. Furthermore, the sum in the square root is just the empirical variance of the noise which as previously mentioned is at most $\bigO{1}$ by~\cref{fact:empirical_cov_bound}. Thus, we bound the cross term by $\bigO{\eta}$ which completes the proof.
\end{proof}

Finally, we consider the last constraint. Assuming the inequalities from~\cref{lem:higher-order-stability} hold, we have that we have the empirical distribution has certifiably hypercontractivity. 

\begin{fact}[Certifiably Hypercontractivity from Concentration]
\label{fact:cert_bounded_moments_from_concentration-regression}
    Let $x_1, \ldots, x_n$ be $n \geq \bigO{\frac{d^2 + \log^2(1/\beta)}{\eta^2}}$ i.i.d samples from $N(0,\Sigma)$ satisfying
    \[ \abs{\sum_{i \in [n]} \iprod{(\Sigma^{-1/2} x_i)(\Sigma^{-1/2} x_i)^\top - I_d}{P}^2 - 2} \leq O(\eta \log^2 ( 1 /\eta)) \,,\]
    for all fixed $P \in \mathbb{R}^{d \times d}$ such that $\norm{P}_F = 1$. Then with probability $1-\beta$ we have that
    \[ \sststile{v}{} \Set{\frac{1}{n} \sum_{i=1}^n \langle x, v \rangle^4 \leq \bigO{1} \cdot \left(v^\top \left(\frac{1}{n} \sum_{i \in [n]} x_ix_i^\top\right) v\right)^2} \,.\]
\end{fact}
\begin{proof}
    Applying~\cref{fact:cert_bounded_moments_from_concentration} immediately yields the claim.

\end{proof}

Finally, we consider the last constraint.

\begin{fact}
    Let $x_1^*, \ldots, x_n^*$ be $n \geq \tilde{\Omega}\paren{\frac{d^2 + \log^2(1/\beta)}{\eta^2}}$ i.i.d samples from $N(0, \Sigma)$ and let $y_i^* = \langle x_i^*, \thetastar \rangle + \zeta_i$ where $\zeta_i$ are i.i.d. samples from $N(0,1)$. Then we have that with probability $1-\beta$
    \[ \frac{1}{n} \sum_{i \in [n]} \Paren{y_i^* - \langle \thetals, x_i^* \rangle}^4 \leq O(1) \,.\]
\end{fact}
\begin{proof}
    Note that we have that
    \[ \frac{1}{n} \sum_{i \in [n]} \Paren{y_i^* - \langle \thetals, x_i^* \rangle}^4 \leq 4 \frac{1}{n} \sum_{i \in [n]} \Paren{\Paren{y_i^* - \langle \theta^*, x_i^* \rangle}^4 + \Iprod{x_i^*, \thetals - \theta^*}}\,.\]
    Note that the sum of the first terms is just the fourth moment of the Gaussian noise in the model. Thus, with probability $1-\beta$ it is bounded by $\bigO{1}$ by~\cref{lem:higher-order-stability}. Furthermore, by~\cref{lem:higher-order-stability} we also have with probability $1-\beta$ that
    \[ \frac{1}{n} \sum_{i \in [n]} \Iprod{x_i^*, \thetals - \theta^*} \leq \bigO{\Paren{(\thetals - \theta^*)^\top \Sigma (\thetals - \theta^*)}^2} \leq \bigO{\norm{\Sigma^{1/2} (\thetals - \theta^*)}_2^2} \leq \bigO{\alpha^2} \,.\]
    Thus we have that 
    \[ \frac{1}{n} \sum_{i \in [n]} \Paren{y_i^* - \langle \thetals, x_i^* \rangle}^4 \leq O(1) \,,\]
    as desired.
\end{proof}

\subsubsection{Mean Estimation}

In this section, we will show that the program in \cref{sec:mean_estimation} is feasible with probability at least $1-\delta$.
Indeed, let $x_1^*, \ldots, x_n^*$ be the uncorrupted samples and set $w_i = r_i, \mu' = \tfrac 1 n \sum_{i \in [n]} x_i^*$.
Then all but the last constraints are trivially satisfied.
That the last one (certifiable hypercontractivity) is satisfied with probability at least $1-\delta$ is a simple correlary of~\cref{fact:cert_bounded_moments_from_concentration}.

\begin{lemma}
    Let $x_1, \ldots, x_n$ be $n \geq \bigO{\frac{d^2 + \log^2(1/\beta)}{\eta^2}}$ i.i.d samples from $N(\mu,\Sigma)$ satisfying
    \[ \abs{\frac 1 n \sum_{i \in [n]} \iprod{(\Sigma^{-1/2} (x_i-\mu))(\Sigma^{-1/2} (x_i-\mu))^\top - I_d}{P}^2 - 2} \leq O(\eta \log^2 ( 1 /\eta)) \,,\]
    for all fixed $P \in \mathbb{R}^{d \times d}$ such that $\norm{P}_F = 1$. Let $\bar{\mu} = \frac{1}{n} \sum_{i \in [n]} x_i$. Then with probability $1-\beta$ we have that
    \[ \sststile{v}{} \Set{\frac{1}{n} \sum_{i=1}^n \langle x - \frac{1}{n} \sum_{i \in [n]} x_i, v \rangle^4 \leq (3 + \eta \log^2(1/\eta) \cdot \left(v^\top \left(\frac{1}{n} \sum_{i \in [n]} (x_i- \bar{\mu}) (x_i - \bar{\mu})^\top\right) v\right)^2} \,.\]
\end{lemma}
\begin{proof}           
    Let $\bar{\mu} = \frac{1}{n} \sum_{i=1}^n x_i$ be the empirical mean. We first will argue that if 
    \[ \abs{\frac 1 n \sum_{i \in [n]} \iprod{(\Sigma^{-1/2} (x_i-\mu))(\Sigma^{-1/2} (x_i-\mu))^\top - I_d}{P}^2 - 2} \leq O(\eta \log^2 ( 1 /\eta)) \,,\]
    then with probability $1-\beta$ we also have that 
    \[ \abs{\frac 1 n \sum_{i \in [n]} \iprod{(\Sigma^{-1/2} (x_i-\bar{\mu}))(\Sigma^{-1/2} (x_i-\bar{\mu}))^\top - I_d}{P}^2 - 2} \leq O(\eta \log^2 ( 1 /\eta)) \,.\]
    Note that we can consider the samples $y_i = \Sigma^{-1/2} x_i$ which are distributed i.i.d. according to $N(\bar{\mu}, I_d)$. Therefore, the statement reduces to bounding
    \[ \abs{\frac 1 n \sum_{i \in [n]} \Iprod{(y_i - \mu)(y_i - \mu)^\top - I_d, P}^2 - \Iprod{(y_i - \bar{\mu})(y_i - \bar{\mu})^\top - I_d, P}^2} \,.\]
    We have that 
    \begin{align*}
        &\frac 1 n \sum_{i \in [n]} \Iprod{(y_i - \mu \pm \bar{\mu})(y_i - \mu \pm \bar{\mu})^\top - I_d, P}^2 \\
        &= \frac 1 n \sum_{i \in [n]} \left( \Iprod{y_i - \bar{\mu})(y_i - \bar{\mu})^\top - I_d, P} + 2 \Iprod{(y_i - \bar{\mu})(\bar{\mu} - \mu)^\top, P} + \Iprod{(\bar{\mu} - \mu)(\bar{\mu} - \mu)^\top, P}\right)^2 \,.
    \end{align*}
    Note that if we expand out the square the first term is $\frac 1 n \sum_{i \in [n]} \Iprod{y_i - \bar{\mu})(y_i - \bar{\mu})^\top - I_d, P}^2$ so it suffices to bound the other terms. First, we consider the term $\frac 1 n \sum_{i \in [n]} \Iprod{(y_i - \bar{\mu})(\bar{\mu} - \mu)^\top, P}^2$. Note that we have that
    \[ \frac 1 n \sum_{i \in [n]} \Iprod{(y_i - \bar{\mu})(\bar{\mu} - \mu)^\top, P}^2 \leq \frac 1 n \sum_{i \in [n]} \left( \norm{\bar{\mu} - \mu} \norm{y_i - \bar{\mu}} \norm{P}_F\right)^2 \leq \norm{\bar{\mu} - \mu}  \cdot \frac{1}{n} \sum_{i \in [n]} \norm{y_i - \bar{\mu}}^2 \,,\]
    and that the empirical covariance is bounded by $\bigO{1}$ (\cref{fact:empirical_mean_bound} and~\cref{fact:empirical_cov_bound}). Thus, applying~\cref{fact:empirical_mean_bound} to bound the difference between the true and empirical mean we get that this is at most $\bigO{\eta^2}$.

    Now, consider the term $\Iprod{(\bar{\mu} - \mu)(\bar{\mu} - \mu)^\top, P}^2$. Note that this is at most $\norm{\bar{\mu} - \mu}^2 \cdot \norm{P}_F \leq \norm{\bar{\mu} - \mu}^2$ and thus by the same reasoning as above it is bounded by $\bigO{\eta^2}$.

    Next, we consider the cross terms. First, we consider all cross terms containing $\Iprod{(\bar{\mu} - \mu)(\bar{\mu} - \mu)^\top, P}$. Note that we have that
    \begin{align*}
        &\frac{1}{n} \sum_{i=1}^n \Iprod{y_i - \bar{\mu})(y_i - \bar{\mu})^\top - I_d, P} \Iprod{(\bar{\mu} - \mu)(\bar{\mu} - \mu)^\top, P} \\
        & \qquad \leq \sqrt{\frac 1 n \sum_{i \in [n]} \Iprod{y_i - \bar{\mu})(y_i - \bar{\mu})^\top - I_d, P}^2} \cdot \sqrt{\Iprod{(\bar{\mu} - \mu)(\bar{\mu} - \mu)^\top, P}^2} \,.
    \end{align*}
    Thus, it suffices to bound $\frac 1 n \sum_{i \in [n]} \Iprod{y_i - \bar{\mu})(y_i - \bar{\mu})^\top - I_d, P}^2 \leq \bigO{1}$, which holds by the assumption as well as~\cref{fact:empirical_mean_bound}. We therefore have that this is bounded by $\bigO{\eta}$ by applying our bound on $\Iprod{(\bar{\mu} - \mu)(\bar{\mu} - \mu)^\top, P}^2$ from above. Similarly, we can also bound 
    \[\frac{1}{n} \sum_{i=1}^n \Iprod{(y_i - \bar{\mu})(\bar{\mu} - \mu)^\top, P} \Iprod{(\bar{\mu} - \mu)(\bar{\mu} - \mu)^\top, P} \leq \bigO{\eta} \,.\]
    Finally, it remains to bound $\frac 1 n \sum_{i \in [n]} \Iprod{y_i - \bar{\mu})(y_i - \bar{\mu})^\top - I_d, P} \Iprod{(y_i - \bar{\mu})(\bar{\mu} - \mu)^\top, P}$. Similarly, to the previous two terms, we have that this is bounded by $\bigO{1} \cdot \sqrt{\frac 1 n \sum_{i \in [n]} \Iprod{(y_i - \bar{\mu})(\bar{\mu} - \mu)^\top, P}^2}$ which is at most $\bigO{\eta}$. Thus, we have bounded all error terms and we can conclude that 
    \[ \abs{\frac 1 n \sum_{i \in [n]} \Iprod{(y_i - \mu)(y_i - \mu)^\top - I_d, P}^2 - \Iprod{(y_i - \bar{\mu})(y_i - \bar{\mu})^\top - I_d, P}^2}  \leq \bigO{\eta} \,.\]
    This implies that 
    \[ \abs{\frac 1 n \sum_{i \in [n]} \iprod{(\Sigma^{-1/2} (x_i-\bar{\mu}))(\Sigma^{-1/2} (x_i-\bar{\mu}))^\top - I_d}{P}^2 - 2} \leq O(\eta \log^2 ( 1 /\eta)) \,.\]
    Then, the SoS proof follows a similar way to~\cref{fact:cert_bounded_moments_from_concentration}. The main differences are that in Terms B, C we have $(x_i - \bar{\mu})^{\otimes 2}$ instead of $(x_i - \mu)^{\otimes 2}$. However, this is fine because $\bar{\mu}, \mu$ are close by~\cref{fact:empirical_mean_bound}. 
\end{proof}

\subsubsection{Covariance Estimation}
\label{sec:cov-feasibility}

\paragraph{Estimation in Spectral Norm.} We first argue that the first program in~\cref{sec:cov-estimation} is feasible. Let $x_i'$ be the true samples $x_i^*$, $w_i$ be the indicators of whether a given sample is corrupted, and $\Sigma'$ be the empirical covariance. Note that that all but the last constraint are satisfied trivially. Finally, by~\cref{fact:cert_bounded_moments_from_concentration} we have that the last constraint is satisfied whenever the concentration bounds from~\cref{lem:higher-order-stability} hold.

\paragraph{Estimation in Frobenius Norm.} 

Note that that the system $\calA$ is feasible since $\Sigma^{-1/2}x_1^*, \ldots, \Sigma^{-1/2}x_n^*$ are $\eta$-good and $\eta$-higher-order-good.
    In particular, let $x_i' = x_i^*$ and $w_i = r_i$.
    Since the input $x_i$ is an $\eta$-corruption of $x_1^*, \ldots, x_n^*$, clearly the first three sets of constraints are satisfied.
    Setting $\Sigma' = \bar{\Sigma} \coloneqq \tfrac 1 n \sum_{i=1}^n x_i^* (x_i^*)^\top$, the fourth constraint is also satisfied.
    It remains to argue that the following SoS proof exists
    \[
        \sststile{4}{P} \Biggl\{\frac{1}{n} \sum_{i \in [n]} \Iprod{P, x_i'(x_i')^\top  - \Sigma'}^2  \leq \Paren{2 + \bigO{\eta \log^2(1/\eta)}} \Norm{P}_F^2 \,.
    \]
    Note that when we can view $P$ as a $d^2$-dimensioanl vector-variable and thus, it suffices to show that the matrix $\tfrac{1}{n}\sum_{i=1}^n \opvec(x_i^*(x_i^*)^\top - \tilde{\Sigma}) \opvec(x_i^*(x_i^*)^\top - \tilde{\Sigma})^\top$ is PSD-dominated by $(2+\bigO{\eta \log^2(1/\eta)}) I_{d^2}$
    This is equivalent to showing that for all matrices $P \in \R^{d \times d}$ of Frobenius norm 1, it holds that
    \[
        \tfrac 1 n \sum_{i=1}^n \Iprod{P, x_i^* (x_i^*)^\top - \bar{\Sigma}}^2 \leq 2 + O(\eta \log^2(1/\eta)) \,.
    \]
    Fix such a matrix $P$ and let $\hat{P} = (\Sigma)^{1/2} P (\Sigma)^{1/2}$ and $\tilde{x}_i = (\Sigma)^{-1/2}x_i^*$.
    Then our condition is equivalent to showing that
    \[
        \tfrac 1 n \sum_{i=1}^n \Iprod{\hat{P}, \tilde{x}_i \tilde{x}_i^\top - \tfrac 1 n \sum_{i=1}^n \tilde{x}_i \tilde{x}_i^\top }^2 \leq 2 + O(\eta \log^2(1/\eta)) \,.
    \]
    Note that
    \begin{align*}
        &\tfrac 1 n \sum_{i=1}^n \Iprod{\hat{P}, \tilde{x}_i \tilde{x}_i^\top - I_d + I_d - \tfrac 1 n \sum_{i=1}^n \tilde{x}_i \tilde{x}_i^\top}^2 \\
        &= \tfrac 1 n \sum_{i=1}^n \Iprod{\hat{P}, \tilde{x}_i \tilde{x}_i^\top - I_d}^2 + 2 \Paren{\tfrac 1 n \sum_{i=1}^n \Iprod{\hat{P}, \tilde{x}_i \tilde{x}_i^\top - I_d}}\Iprod{\hat{P}, I_d - \tfrac 1 n \sum_{i=1}^n \tilde{x}_i \tilde{x}_i^\top} + \Iprod{\hat{P}, I_d - \tfrac 1 n \sum_{i=1}^n \tilde{x}_i \tilde{x}_i^\top}^2\,.
    \end{align*}
    Note that the $\tilde{x}_i$ are $\eta$-higher-order-good by assumption.
    Also, since $\Sigma \preceq (1+O(\eta \log(1/\eta))) I_d$, it holds that $\Norm{\hat{P}}_F \leq \norm{\Sigma} \norm{P}_F \leq 1 + O(\eta \log(1/\eta))$.
    Hence, the first term is at most
    \[
    (2 + O(\eta \log(1/\eta))) \norm{\hat{P}}_F^2 \leq 2 + O(\eta \log(1/\eta))
    \]
    by~\cref{lem:higher-order-stability} (Point 2).
    Similarly, by Point 1 of~\cref{lem:higher-order-stability}, the last term is at most
    \[
        \Norm{\hat{P}}_F^2 \Norm{I_d - \tfrac 1 n \sum_{i=1}^n \tilde{x}_i \tilde{x}_i^\top}_F^2 \leq (1+O(\eta \log(1/\eta)))^2 O(\eta^2 \log^2(1/\eta)) = O(\eta^2 \log^2(1/\eta)) \leq O(\eta \log(1/\eta))\,.
    \]
    Finally, again by Point 1 of~\cref{lem:higher-order-stability}, the middle term is at most
    \[
        2\Norm{\hat{P}}_F^2 \cdot O(\eta \log(1/\eta)) \cdot  O(\eta \log(1/\eta)) \leq O(\eta \log(1/\eta)) \,.
    \]

\section{Lower Bounds}
In this section, we will present pure DP lower bounds for private regression and covariance-aware mean estimation under differential privacy. The following theorems are information theoretic lower bounds which follow from packing lower bounds \cite{hardt2010geometry,hopkins2022efficient}. 

The first theorem we show gives an information-theoretic lower bound for regression with unknown covariance under pure differential privacy.

\begin{theorem}[Pure DP Lower Bound for Private Regression]
\label{thm:pure-dp-lb-regresion}
Let $R > 8\alpha /\sqrt{\covscale}> 0$, $\covscale > 1$, and $ \epsilon, \beta \in \paren{0, 1}$.
Suppose there exists an $\epsilon$-DP algorithm $M$ such that for every distribution $D$ on $\R^d \times \R$ distributed as $(x, y)$, where $x$ is distributed as $\mathcal{N}(0, \Sigma)$, where $\Sigma \preccurlyeq \covscale \cdot I$ and $y = \iprod{\theta}{x} + \zeta$, where $\norm{\theta}_2 \le R$, and $\zeta \sim \mathcal N (0,1)$, given $(x_1, y_1), \dots, (x_n, y_n) \sim D$, with probability $1-\beta$ the algorithm $M$ outputs $\hat{\theta}$ such that $\norm{\Sigma^{1/2} \paren{\theta - \hat{\theta}}}_2 \le \alpha$. Then we have that
\begin{equation*}
n = \Omega \Paren{
\frac{d + \log\paren{1/\beta}}{\alpha^2}
+
\frac{d + \log\paren{1/\beta}}{\alpha \eps}
+
\frac{d \log \paren{R \sqrt{\covscale}/\alpha}+ \log\paren{1/\beta}}{\eps}
} \, .
\end{equation*}
\end{theorem}

The second theorem we show gives an information-theoretic lower bound for covariance-aware mean estimation under pure differential privacy.

\begin{theorem}[Pure DP Lower Bound for Covariance-Aware Gaussian Mean Estimation]
\label{thm:pure-dp-lb-mean-etimation}
Let $R > 8 \alpha  / \sqrt{\covscale}>  0$, $\covscale \ge 1$, and, $\eps, \beta \in \paren{0,1}$. Suppose there exists an $\eps$-DP algorithm $M$ such that for every distribution $\mathcal{N} (\mu, \Sigma)$ on $\R^d$, where $\norm{\mu}_2 \le R$, and $\Sigma \succcurlyeq \frac{1}{\covscale} \cdot I$, given $X_1 \dots, X_n \sim \mathcal{N} (\mu, \Sigma)$, with probability $1-\beta$ the algorithm $M$ outputs $\hat{\mu}$ such that $\norm{\Sigma^{-1/2}\paren{\mu - \hat{\mu}}}_2 \le \alpha$. Then we have that
\begin{equation*}
n = \Omega \Paren{
\frac{d + \log\paren{1/\beta}}{\alpha^2}
+
\frac{d + \log\paren{1/\beta}}{\alpha \eps}
+
\frac{d \log \paren{R \sqrt{\covscale}/\alpha}+ \log\paren{1/\beta}}{\eps}
} \, .
\end{equation*}
\end{theorem}

\subsection{Preliminaries}

We will make use of the following theorem in proving packing lower bounds.

\begin{theorem}[Theorem~7.1 in \cite{hopkins2022efficient}]
\label{thm:lb-meta}
Let $\mathcal P=\left\{P_1, \ldots, P_m\right\}$ be a set of distributions, and $P_O$ be a distribution such that for every $P_i \in \mathcal{P},\left\|P_i-P_O\right\|_{\mathrm{TV}} \leq \gamma$. Let $\mathcal{G}=\left\{G_1, \ldots, G_m\right\}$ be a collection of disjoint subsets of some set $\mathcal{Y}$. If there is an $\varepsilon$-DP algorithm $M$ such that $\mathbb{P}_{X \sim P_i^n}\left[M(X) \in G_i\right] \geq 1- \beta$ for all $i \in[m]$, then
$$
n \geq \Omega\left(\frac{\log m+\log (1 / \beta)}{\gamma\left(e^{2 \varepsilon}-1\right)}\right) .
$$
Note that for the usual regime $\varepsilon \leq 1$, we can replace the $e^{2 \varepsilon}-1$ in the denominator with $\varepsilon$.
\end{theorem}

In our analysis of the packing lower bound for regression, we will need to bound the TV distance of two Gaussians with mean zero and different covariances. The following theorem will help us in that analysis.

\begin{theorem}[Total variation distance between Gaussians with the same mean \cite{devroye2018total}]
\label{thm:tv-distance-gaussians}
 Let $\mu \in \mathbb{R}^d, \Sigma_1$ and $\Sigma_2$ be positive definite $d \times d$ matrices, and $\lambda_1, \ldots, \lambda_d$ denote the eigenvalues of $\Sigma_1^{-1} \Sigma_2-I_d$. Then,

$$
\frac{1}{100} \leq \frac{\operatorname{TV}\left(\mathcal{N}\left(\mu, \Sigma_1\right), \mathcal{N}\left(\mu, \Sigma_2\right)\right)}{\min \left\{1, \sqrt{\sum_{i=1}^d \lambda_i^2}\right\}} \leq \frac{3}{2} \,.
$$
\end{theorem}

\subsection{Proof of \Cref{thm:pure-dp-lb-regresion}}
Before proving \Cref{thm:pure-dp-lb-regresion}, it is useful to prove the following lemma that bounds the change in the distribution of the input in TV distance when the regression parameter $\theta$ changes.
\begin{lemma}[Regression TV Distance]
\label{lem:regression-tv-distance}
Let $D_\theta$ be the distribution over $\R^d \times \R$ distributed as $(x,y)$ where $x$ is distributed as $\mathcal{N} \paren{0, I_d}$, and $y = \iprod{\theta}{x} + \zeta$, where $\zeta \sim \mathcal N \paren{0,1}$. Then for any $\theta_1, \theta_2 \in \R^d$, such that $\norm{\theta_1 - \theta_2}\le 1$, we have
\begin{equation*}
\operatorname{TV}(D_{\theta_1}, D_{\theta_2})
\le 
3
\norm{\theta_1 - \theta_2}_2 \, .
\end{equation*}
\end{lemma}
\begin{proof}
Since $x$ is a Gaussian and $y = \iprod{\theta}{x} + \zeta$, where $\zeta \sim \mathcal N (0,1)$, $(x, y)$ are jointly Gaussian. In order to understand the distribution $D_\theta$ it remains to understand the mean and covariance of $(x,y)$. $D_\theta$ is mean zero, and its covariance is the $(d+1) \times (d+1)$ block matrix
\begin{equation*}
\Sigma_\theta \coloneqq
\begin{pmatrix}
I_d & \theta \\
\theta^\top & \|\theta\|^2 + 1
\end{pmatrix} \, .
\end{equation*}
To see why note that $\E\brac{x_i y} = \E\brac{x_i \paren{\sum_i \theta_i x_i + \zeta}} = \theta_i$, and $\E\brac{y^2} = \E\brac{\paren{\sum_i \theta_i x_i + \zeta}^2} = \norm{\theta}_2^2 + 1$.
Now it remains to bound the TV distance of two Gaussians $\mathcal N \paren{0, \Sigma_{\theta_1}}$ and $\mathcal N \paren{0, \Sigma_{\theta_2}}$. We aim to apply \Cref{thm:tv-distance-gaussians}. We need to compute $\Sigma_{\theta_2}^{-1}\Sigma_{\theta_1} -I_{d+1}$. First it is easy to verify that
\begin{equation*}
\Sigma_\theta^{-1} =
\begin{pmatrix}
I_d+ \theta \theta^\top & -\theta \\
-\theta^\top &  1
\end{pmatrix} \, .
\end{equation*}
Let $z = \theta_1 - \theta_2$, then

\begin{equation*}
M \coloneqq
\Sigma_{\theta_2}^{-1}\Sigma_{\theta_1} -I_{d+1} = 
\begin{pmatrix}
\theta_2 \theta_2^{\top}-\theta_2 \theta_1^{\top} & \theta_1-\theta_2+\theta_2\left(\theta_2^{\top} \theta_1-\left\|\theta_1\right\|^2\right) \\
\left(\theta_1-\theta_2\right)^{\top} & \left\|\theta_1\right\|^2-\theta_2^{\top} \theta_1
\end{pmatrix}
=
\begin{pmatrix}
-\theta_2 z^{\top} & z-\theta_2\theta_1^\top z \\
z^{\top} & \theta_1^\top z
\end{pmatrix}\, .
\end{equation*}
We are interested in computing the sum of the squared eigenvalues of $M$. We know that the eigenvalues of $M^2$ are the squared eigenvalues of $M$, and that the trace of a matrix gives us the sum of its eigenvalues. Therefore it remains to compute $\operatorname{Tr}\paren{M^2}$.
We have 
\begin{align*}
\operatorname{Tr} \paren{M^2} &=
\operatorname{Tr} (\theta_2 z^\top \theta_2 z^\top + z z^\top - \theta_2 \theta_1^\top z z^\top) + \operatorname{Tr} \paren{z^\top z - z^\top \theta_2 \theta_1^\top z + \theta_1^\top z\theta_1^\top z} \\
&=
2 \norm{z}_2^2 +
\paren{\iprod{z}{\theta_1} - \iprod{z}{\theta_2}}^2 \\
&= 2 \norm{z}_2^2 + \norm{z}_2^4 \,.
\end{align*}
Applying \Cref{thm:tv-distance-gaussians} and noting $\norm{z}_2^2\le 1$, finishes the proof.
\end{proof}

Now using the lemma above we can prove \Cref{thm:pure-dp-lb-regresion}.
\begin{proof}[Proof of \Cref{thm:pure-dp-lb-regresion}]
The first term is folklore and is implied by non-private lower bounds. 
The second term is the cost of fine estimation under pure differential privacy. A lower bound of $d\alpha \eps$ is known even under the weaker approximate differential privacy \cite{cai2021cost}. 
Here we focus on showing lower bounds that include the failure probability $\beta$.

We first show the second term is a lower bound. Let $\{\theta_i\}$ be such that $\norm{\theta_i-\theta_j} \ge 2\alpha$, for $i\neq j \in I$, and $\norm{\theta_i} \le 8 \alpha$ for all $i$. Bounds on metric entropy imply that such a set with log-size at least $d\log 4$ exists. Let $P_O = \mathcal{N}\paren{0, I_d}$, and 
$P_i = \mathcal{N}\paren{0, \Sigma_\theta}$, where $\Sigma_\theta$ is as defined in \Cref{lem:regression-tv-distance}. Then \Cref{lem:regression-tv-distance} implies $\operatorname{TV}(P_O, P_i) \le 3 \norm{\theta_i} \le 24 \alpha$. Let $G_i$ be the ball of radius $\alpha$ around $\theta_i$, then applying \Cref{thm:lb-meta} we obtain
\begin{equation*}
n = \Omega\paren{\frac{d + \log\paren{1/\beta}}{\alpha \eps}} \,.
\end{equation*}
Finally, we prove the last term is a lower bound. Let $\{\theta_i\}$ be such that $\norm{\theta_i - \theta_j} \ge 2\alpha / \sqrt{\covscale}$, for $i \neq j \in I$, and $\norm{\theta_i} \le R$ for all $i$. Bounds on metric entropy imply that such a set with log-size at least $d \log\paren{R\sqrt{L}/ 2 \alpha}$ exists. Let $P_O$ and be $P_{i}$ be the distributions over $\R^d \times \R$ distributed as $(x,y)$ where $x$ is distributed as $\mathcal{N}(0, \covscale I_d)$, and $y = \iprod{0}{x} +\zeta$ and $y = \iprod{\theta_i}{x}$ respectively, where $\zeta \sim \mathcal{N}(0,1)$. Then $\operatorname{TV}(P_O, P_i) \le 1$. Let $G_i$ be the ball of radius $\alpha/ \sqrt{\covscale}$ around $\theta_i$. Note that $G_i$'s are disjoint, and $M$'s output when applied to sample from $P_i$ is inside $G_i$ with probability $1-\beta$. Applying \Cref{thm:lb-meta}, we obtain 
\begin{equation}
n= \Omega\paren{\frac{d \log \paren{R \sqrt{L} / \alpha} +\log\paren{1/\beta}}{\eps}} \, .
\end{equation}
\end{proof}

\subsection{Proof of \Cref{thm:pure-dp-lb-mean-etimation}}
\begin{proof}[Proof of \Cref{thm:pure-dp-lb-mean-etimation}]
The first three terms follow from the lower bound in the case where $\Sigma$ is known to be identity. For the theorem in the case where $\Sigma$ is known to be identity see Theorem~7.3 in \cite{hopkins2022efficient}. Therefore, the only term we need to prove here is the last term.

We show the last term is a lower bound by applying \Cref{thm:lb-meta}. Let $\{\mu_i \}_{i \in I}$ be such that $\norm{\mu_i - \mu_j} \ge 2 \alpha / \sqrt{\covscale}$, for $i \neq j \in I$, and $\norm{\mu_i} \le R$ for all $i$. Bounds on metric entropy imply that such a set with log-size at least $d \log (R \sqrt{\covscale} / 2\alpha)$ exists. 
Let $P_O = \mathcal{N}\paren{0, \frac{1}{\covscale}I_d}$, and $P_i = \mathcal{N} \paren{\mu_i, \frac{1}{\covscale}I_d}$. The utility guarantees of the algorithm $M$ imply that $M$ takes as input samples from $P_i$ and outputs $\hat{\mu}$ such that $\norm{\hat{\mu} - \mu_i}_2 \le \alpha / \sqrt{\covscale}$, with probability $1-\beta$. Let $G_i$ be the ball of radius $\alpha / \sqrt{\covscale}$ around $\mu_i$. Note that $G_i$'s are disjoint, and $M$'s output on an input sampled from $P_i$ falls inside $G_i$ with probability $1-\beta$. Therefore, we may apply \Cref{thm:lb-meta}, and we obtain that
\begin{equation*}
n = \Omega \Paren{\frac{d \log \paren{R \sqrt{\covscale} / \alpha} + \log \paren{1/\beta}}{\eps}} \, ,
\end{equation*}
as desired.
\end{proof}

\section{Deferred Proofs for Private Algorithms}
\subsection{Efficient Computability of Score Functions}

In this section we will show that our new constraints for our privacy programs admit a separation oracle. 
Recall, that it is of the following form (exemplified for covariance-aware mean estimation)
\[
    \bigO{\alpha^2} \pE \Sigma' - \pE[\mu' - \tilde{\mu}]\pE[\mu' - \tilde{\mu}]^\top \succeq 0 \,.\footnote{We include the version of this constraint without slack for simplicity, but a similar argument works when including the version with slack stated in the main body.}
\]
Note that this is indeed convex in $\pE$ (for a fixed $\tilde{\mu}$).
Consider $\pE_1, \pE_2$ both satisfying the constraint.
Fix an arbitrary unit vector $u$, then by convexity of the function $z \mapsto z^2$ it follows that\footnote{Note that formally, we would need to verify convexity when representing pseudo-expectations using there vector representation as described below. This argument is analogous and we give the one below for simplicity.}
\begin{align*}
    \Iprod{\tfrac 1 2 \pE_1 \mu' + \tfrac 1 2 \pE_2 \mu' - \tilde{\mu}, u}^2 &= \Paren{\tfrac 1 2 \Iprod{\pE_1 \mu' - \tilde{\mu}, u} + \tfrac 1 2 \Iprod{\pE_2 \mu' - \tilde{\mu}, u}}^2 \\
    &\leq \tfrac 1 2 \Iprod{\pE_1 \mu' - \tilde{\mu}, u}^2 + \tfrac 1 2 \Iprod{\pE_2 \mu' - \tilde{\mu}, u}^2 \leq \bigO{\alpha^2} \Paren{\tfrac 1 2 \pE_1 + \tfrac 1 2 \pE_2} \Iprod{u, \Sigma' u} \,.
\end{align*}
(And similarly for the one used for regression)

Recall that we run the ellipsoid algorithm to find a degree-$d$ pseudo-expectation that approximately satisfy our constraints.
We represent these pseudo-expectations by the values they assign to each multi-set of variables up to size $d$, excluding the empty set.
That is if we have variables $X_1, \ldots, X_n$, we are searching for a vector whose entries correspond to $\pE \prod_{i \in S} X_i$, where $S \subseteq [n]$ is a multi-set of size at most $d$
By considering the entries corresponding to (the upper-triangular part of) $\pE \Sigma'$ and $\pE \mu'$ in this representation of $\pE$, it is thus enough to show the following claim.
The second part corresponds to constraints of the form $\norm{\pE \mu'}^2 \leq R^2$ (which is obviously convex).
For a $d(d+1)/2$-dimensional vector $z$, we denote by $M(z)$ the $d \times d$-dimensional matrix that is symmetric and whose upper triangular part is populated by the entries of $z$.
\begin{lemma}
    \label{lem:closeness-separation-oracle}
    Consider constraints of the following form, where $x\in\mathbb{R}^d,y\in\mathbb{R}^d,z\in\mathbb{R}^{d (d+1)/2}$ are variables in a convex program and $c$ is some fixed constant:
    \begin{enumerate}
        \item $xx^\top  - c M(z) \preceq 0$\,.
        \item $\norm{y}_2^2 \leq c$\,.
    \end{enumerate}
    Then these constraints all admit a $\poly(d)$-time separation oracle.
\end{lemma}
\begin{proof}
    Note that we can easily verify if these constraints are satisfied in both cases. Thus, it suffices to show that computing a separating hyperplane can be done efficiently when the constraints are not satisfied. 

    We start with the first constraint.
    For simplicity, we assume that $c = 1$, the argument for general $c$ is completely analogous.
    Thus, for the constraints on $x,M(z_0)$ we can consider some $x_0, M(z_0)$ such that $x_0x_0^\top \not\preceq M(z_0)$. We aim to find a separating hyperplane, given by $w \in \mathbb{R}^{d+d^2}$, $a \in \mathbb{R}$, such that $\Iprod{w, (x_0, M(z_0))} > a$ but $\Iprod{w, (x',M'(z))} \leq a$ for all $x',M'(z)$ satisfying the constraints.\footnote{Technically, we would have to convert back to the parametrization where we only have $(x_0,z_0)$ and $(x',z')$ instead of $(x_0,M(z_0))$ and $(x',M(z'))$. We omit this for simplicity.}

    Note that since the constraint is violated there exists some $v$ such that $v^\top M(z_0) v < \Iprod{v, x_0}^2$ and we can compute such a $v$ efficiently by computing eigendecomposition of the matrix $M(z_0) - x_0x_0^\top$ and taking $v$ to be the eigenvector associated with the minimum eigenvalue. We will let $w = (2\langle v, x_0\rangle v, - \opvec\Paren{vv^\top})$ and $a = \Iprod{v, x_0}^2$. For any $x,M$ we have that
    \[ \Iprod{w, (x, M)} = 2\langle v, x_0\rangle \langle v, x\rangle - v^\top M v\,. \]
    Plugging in $(x_0, M(z_=))$ and applying that the PSD constraint is violated we have that
    \[ \Iprod{w, (x_0, M_0)} = 2\langle v, x_0\rangle^2  - v^\top M(z_0) v  = \langle v, x_0\rangle^2 + \left(\langle v, x_0\rangle^2  - v^\top M(z_0) v\right) > \langle v, x_0\rangle^2\,, \]
    as desired. We now show that all points that satisfy the constraint must have that $\Iprod{w, (x, M(z))} \leq \Iprod{v, x_0}^2$. We have that
    \begin{align*}
        \Iprod{w, (x, M(z))} &= 2\langle v, x_0\rangle \langle v, x\rangle - v^\top M(z) v \\
        &\leq \langle v, x_0\rangle^2 + \langle v, x\rangle^2 - v^\top M(z) v \\
        &\leq \langle v, x_0\rangle^2\,,
    \end{align*}
    using that $2ab \leq a^2 + b^2$ and that $M(z) \succeq xx^\top$. Thus, we have given a separation oracle for the first two constraints.

    We now consider the constraint on $y$. We can simply take $w = \frac{\sqrt{c}}{\norm{y_0}} \cdot y_0$ and $a = c$. We have that
    \[ \Iprod{w, y_0} = \frac{\sqrt{c}}{\norm{y_0}} \cdot \norm{y_0}^2 > c\,,\]
    when $y_0$ violates the constraint. Now consider $y$ which satisfies this constraint. We have that
    \[ \Iprod{w, y} \leq \norm{w}_2 \cdot \norm{y}_2 \,.\]
    Note that $w$ has norm exactly $\sqrt{c}$ by construction and since $y$ satisfies the constraints it also has norm at most $\sqrt{c}$. Therefore, this inner product is bounded by $c$, as desired.
\end{proof}

\subsection{Deferred Privacy SoS Proofs}
\label{sec:deferred_privacy_sos_proofs}

In this section we will complete the proof from~\cref{sec:private_regression} that points of low score have good utility. Specifically, we need to show that our SoS variable $Q$ is close to $\Sigma^{-1}$.

    \paragraph{Accuracy of SoS inverse covariance matrix.}
    It remains to show~\cref{lem:sos_inverse_covariance}.
    Recall that $\tilde{\Sigma} = \tfrac 1 n \sum_{i=1}^n x_i^* (x_i^*)^\top$.
    We will use the following lemma about closeness of $\tilde{\Sigma}$ and $\Sigma'$.
    Note that we take $u$ and $v$ to be SoS indeterminates.
    \begin{lemma}
        \label{lem:sos_assymetric_cov_estimation}
        Let $\eta > 0$ be smaller than a sufficiently small absolute constant and assume that $x_1^*, \ldots, x_n^*$ are $\eta$-good and $\eta$-higher-order-good.
        Then it holds that
        \[
            \calA_{\eta n} \sststile{\bigO{1}}{u,v} \Iprod{\tilde{\Sigma} - \Sigma', u v^\top}^2 \leq \bigO{\eta} \paren{u^\top \Sigma' u}^2 + \bigO{\eta} \paren{v^\top \Sigma' v}^2 \,.
        \]
    \end{lemma}
    We will give the proof at the end of this section.
    \begin{proof}[Proof of~\cref{lem:sos_inverse_covariance}]
        First, note that by $\eta$-goodness it holds that, say, $\tfrac 1 2 \Sigma \preceq \tilde{\Sigma} \preceq 2 \Sigma$.
        And in particular, $\tilde{\Sigma}$ is invertible.
        Further, this implies that it is enough to show that 
        \begin{align*}
            \calA_{\eta n} \sststile{\bigO{1}}{} \Biggl\{ \Iprod{u, Q u}^2 \leq \bigO{1}  \Paren{u^\top \tilde{\Sigma}^{-1} u}^2 \Biggr\} \,.
        \end{align*}
        We start by observing that by SoS triangle inequality it holds that
        \[
            \calA_{\eta n} \sststile{\bigO{1}}{} \Biggl\{ \Iprod{u, Q u}^2 \leq 2 \Paren{u^\top \tilde{\Sigma}^{-1} u}^2 + 2 \Paren{u^\top \brac{Q - \tilde{\Sigma}^{-1} } u}^2 \Biggr\} 
        \]
        Thus, it remains to bound the second term.
        Repeatedly using the constraints that $\Sigma' Q \Sigma' = \Sigma', Q^\top = Q$, and $Q \Sigma' = I_d, \Sigma' Q = I_d$ we obtain that 
        \begin{align*}
            \calA_{\eta n} &\sststile{\bigO{1}}{} \Biggl\{ \Paren{u^\top \brac{Q - \tilde{\Sigma}^{-1} } u}^2 = \Iprod{Q - \tilde{\Sigma}^{-1}, u u^\top}^2 = \Iprod{Q \Sigma' Q \Sigma' Q - Q \Sigma' \tilde{\Sigma}^{-1} \Sigma' Q, u u^\top}^2 \\
            &= \Iprod{\Sigma' Q \Sigma' - \Sigma' \tilde{\Sigma}^{-1} \Sigma', \brac{Qu} \brac{Qu}^\top}^2 = \Iprod{\Sigma' Q \Sigma' - \Sigma' \tilde{\Sigma}^{-1} \Sigma', \brac{Qu} \brac{Qu}^\top}^2 \\
            &= \Iprod{\Sigma' Q \Sigma' - \Sigma' \tilde{\Sigma}^{-1} \Sigma', \brac{Qu} \brac{Qu}^\top}^2 = \Iprod{Q \Sigma' - \tilde{\Sigma}^{-1} \Sigma', \brac{\Sigma' Qu} \brac{Qu}^\top}^2 \\
            &= \Iprod{I_d - \tilde{\Sigma}^{-1} \Sigma', u \brac{Qu}^\top}^2 = \Iprod{\tilde{\Sigma}^{-1} \tilde{\Sigma} - \tilde{\Sigma}^{-1} \Sigma', u \brac{Qu}^\top}^2 = \Iprod{\tilde{\Sigma} -\Sigma', \brac{\tilde{\Sigma}^{-1} u} \brac{Qu}^\top}^2 \Biggr\} \,.
        \end{align*}
        Thus, applying~\cref{lem:sos_assymetric_cov_estimation} with $u = \tilde{\Sigma}^{-1} u$ and $v = Q u$ (overloading notation) and using the constraint $Q \Sigma' Q = Q$, it follows that
        \begin{align*}
            \calA_{\eta n} \sststile{\bigO{1}}{} \Biggl\{ \Paren{u^\top \brac{Q - \tilde{\Sigma}^{-1} } u}^2 &\leq \bigO{\eta} \paren{u^\top \brac{\tilde{\Sigma}^{-1} \Sigma' \tilde{\Sigma}^{-1}} u}^2 + \bigO{\eta} \paren{u^\top \brac{Q \Sigma' Q} u}^2  \\
            &= \bigO{\eta} \paren{u^\top \brac{\tilde{\Sigma}^{-1} \Sigma' \tilde{\Sigma}^{-1}} u}^2 + \bigO{\eta} \paren{u^\top Q u}^2\Biggr\} \,.
        \end{align*}
        The second term is a small multiple of the initial left-hand side and thus, we can ignore it by rearranging.
        For the first term, note that by~\cref{lem:main_identity_cov_sos_proof} our constraints imply that it is at most $\bigO{1} \paren{u^\top \tilde{\Sigma}^{-1} \Sigma \tilde{\Sigma}^{-1}}^2$.
        Using that $\Sigma \preceq 2 \tilde{\Sigma}$ now implies the claim.
    \end{proof}

    It remains to show~\cref{lem:sos_assymetric_cov_estimation}.
    \begin{proof}[Proof of~\cref{lem:sos_assymetric_cov_estimation}]
        Using the definition of $\tilde{\Sigma}$, SoS triangle and Cauchy-Schwarz inequality, and that $\sststile{2}{X,Y} \Set{XY \leq X^2 + Y^2}$, it follows that (recall that $r_i = 1$ if and only if $x_i^* = x_i$, i.e., the $i$-th sample of the input is uncorrupted)
        \begin{align*}
            \calA_{\eta n} \sststile{\bigO{1}}{u,v} \Iprod{\tilde{\Sigma} - \Sigma', u v^\top}^2 &= \paren{\tfrac{1}{n} \sum_{i=1}^n \Iprod{x_i^*, u}\Iprod{x_i^*, v} - \tfrac{1}{n} \sum_{i=1}^n \Iprod{x_i', u}\Iprod{x_i', v}}^2 \\
            &= \paren{\tfrac{1}{n} \sum_{i=1}^n (1 - r_i w_i)\Iprod{x_i^*, u}\Iprod{x_i^*, v} - \tfrac{1}{n} \sum_{i=1}^n (1 - r_i w_i)\Iprod{x_i', u}\Iprod{x_i', v}}^2 \\
            &\leq \paren{\tfrac{1}{n} \sum_{i=1}^n (1 - r_i w_i)\Iprod{x_i^*, u}\Iprod{x_i^*, v}}^2 +  \paren{\tfrac{1}{n} \sum_{i=1}^n (1 - r_i w_i)\Iprod{x_i', u}\Iprod{x_i', v}}^2 \\
            &\leq \eta \cdot \tfrac{1}{n} \sum_{i=1}^n \Iprod{x_i^*, u}^2\Iprod{x_i^*, v}^2 + \eta \cdot \tfrac{1}{n} \sum_{i=1}^n \Iprod{x_i', u}^2\Iprod{x_i', v}^2 \\
            &= \eta \cdot \tfrac{1}{n} \sum_{i=1}^n \paren{\Iprod{x_i^*, u}^4 + \Iprod{x_i^*, v}^4} + \eta \cdot \tfrac{1}{n} \sum_{i=1}^n \paren{\Iprod{x_i', u}^4 + \Iprod{x_i', v}^4 } \,.
        \end{align*}
        By our hypercontractive fourth moment constraint, we know that our constraints imply that there exists an SoS proof (also in variables $u$ and $v$) that$\tfrac{1}{n} \sum_{i=1}^n \Iprod{x_i', u}^4 \leq \bigO{1} \Paren{u^\top \Sigma' u}^2$ and the same for the term involving $v$.
        Similarly, it follows by $\eta$-higher-order goodness that there exists SoS proofs in variables $u$ and $v$ that (cf.~\cref{sec:cov-feasibility}) $\tfrac{1}{n} \sum_{i=1}^n \Iprod{x_i^*, u}^4 \leq \bigO{1} \paren{u^\top \tilde{\Sigma} u}^2$ and the same for the term involving $v$.
        Further, by~\cref{lem:boostrapping_covariance} we can conclude that there is an SoS proof in variables $u$ and $v$ respectively, that $\paren{u^\top \tilde{\Sigma} u}^2 \leq \bigO{1} \paren{u^\top \Sigma' u}^2$ and $\paren{v^\top \tilde{\Sigma} v}^2 \leq \bigO{1} \paren{v^\top \Sigma' v}^2$.
    \end{proof}

\section{Approximate DP via Product Pseudo-Expectations}
\label{sec:approx_dp_product_pe}

In this section, we prove our results for covariance-aware mean estimation (\cref{sec:approx_mean_est}) and linear regression (\cref{sec:approx_lin_reg}) under approximate differential privacy.
As opposed to the pure DP case, we will also need to show bounds on the volume ratio when the input is worst-case, instead of random.
\cite{Hopkins2023Robustness} showed such statements by cleverly constructing "substitute" parameters (for say the mean) also for worst-case inputs.
We show that these statements can also be derived by (modifications of) existing utility proofs of robust estimators.\footnote{We thank Samuel Hopkins for pointing this out to us.}

\subsection{Product Pseudo-Expectations}
\label{sec:product_pe}

In this section, we define the \emph{product} of two pseudo-distributions and derive useful properties that we will use in \cref{sec:approx_mean_est,sec:approx_lin_reg}.
The product of two pseudo-distribution is defined analogously to the product of two actual distributions.
For simplicity, we only define it when the degrees and number of variables of both distributions are the same.
\begin{definition}
\label{def:product_pE}
    Let $d, n\in \N$.
    Let $\pE_x$ and $\pE_y$ degree-$2d$ pseudo-expectations over variables $x = (x_1, \ldots, x_n) $ and $y = (y_1, \ldots, y_n)$, respectively.
    We define the \emph{product} of $\pE_1$ and $\pE_2$, denoted by $\pE$, as follows.
    For $\alpha \in \N^{n}, \beta \in \N^{n}$ such that $\card{\alpha} + \card{\beta} \leq 2d$, we define
    \[
        \pE x^\alpha y^\beta = \pE_x x^\alpha \cdot \pE_y y^\beta 
    \]
    and extend it linearly to all polynomials of degree at most $2d$ (jointly in $x$ and $y$). 
\end{definition}
Products of pseudo-expectations are themselves pseudo-expectations.
Further, if $\pE_x$ and $\pE_y$ satisfy constraints $\calA_x$ and $\calA_y$, then their product satisfies their union.
In particular, we have the following lemma.
\begin{lemma}
    \label{lem:product_pEs}
    Let $d, n \in \N$ and let $\pE_x$ and $\pE_y$ be degree-$2d$ pseudo-expectations over variables $x = (x_1, \ldots, x_n) $ and $y = (y_1, \ldots, y_n)$, respectively.
    Let $\pE$ be the product of $\pE_x$ and $\pE_y$.
    Then the following statements are true.
    \begin{enumerate}
        \item $\pE$ is a pseudo-expectation of degree $2d$ over the $2n$ variables $(x,y)$.
        \item Suppose $\pE_x$ and $\pE_y$ satisfy $\calA_x$ and $\calA_y$ at degree $2d$, respectively.
        Then $\pE$ satisfies $\calA_x \cup \calA_y$ at degree $2d$.\footnote{The proof straightforwardly extends to the setting when $\pE_x$ and $\pE_y$ only approximately satisfy the constraints.}
    \end{enumerate}
\end{lemma}
\begin{proof}
    We start with the first property.
    \paragraph{$\pE$ is a valid pseudo-expectation.}
    Normalization and linearity follow by definition.
    We have to show positive semidefiniteness.
    Note that $\pE_x$ and $\pE_y$ can be represented via the matrices $M_x = \pE_x(1,x)^{\otimes d} ((1,x)^{\otimes d})^\top$ and $M_y = \pE_x(1,y)^{\otimes d} ((1,y)^{\otimes d})^\top$.
    We represent $\pE$ by the matrix $M = M_x \otimes M_y$.
    By construction both $M_x$ and $M_y$ are positive semidefinite, and hence so is $M$.
    Now consider any degree at most $d$ polynomial $s$.
    Then there exists a vector representation $\vec{s}$ such that
    \[
        \pE s(x,y)^2 = \vec{s}^\top M \vec{s} \geq 0.
    \]

    \paragraph{$\pE$ satisfies $\calA_x \cup \calA_y$.}

    For the second property, without loss of generality we can assume that $\calA_x$ and $\calA_y$ only contain equality constraints.\footnote{Inequality constraints can be encoded as equality constraints using one additional variable per constraint.}
    That is, $\calA_x = \Set{q_1(x) = 0 , \ldots, q_{m_x}(x) = 0}$ and $\calA_y = \Set{q_1'(y) = 0 , \ldots, q_{m_y}'(y) = 0}$.
    Let $p(x,y)$ be an arbitrary polynomial and $S = S_x \cup S_y \subseteq [m_x] \cup [m_y]$ such that the degree of $\tilde{p}(x,y) = p(x,y) \cdot \prod_{j \in S_x} q_j(x)\cdot \prod_{k \in S_y} q_k'(y)$ is at most $2d$.
    We have to show that $\pE \tilde{p} = 0$.
    By lineary, we can assume that $p(x,y) = x^\alpha y^\beta$.
    But then
    \[
        \pE \tilde{p} = \Paren{\pE_x x^\alpha \prod_{j \in S_x} q_j(x)} \cdot \Paren{\pE_y y^\beta \prod_{k \in S_x} q_k(x)} = 0 \,.
    \]
\end{proof}

\paragraph{Product Pseudo-Expectations in Our Settings.}

The following facts will be useful to us.
Let $x_1, \ldots, x_n$ be arbitrary points and $0 \leq \eta < 1/2$.
Consider a constraint set $\calA$ over variables $w_1, \ldots, w_n, x_1', \ldots, x_n'$ (and potentially more) that contains constraints
\[
    \Set{w_i^2 = w_i\,, \sum_{i=1}^n w_i \geq (1-\eta)n\,, w_i x_i' = w_i x_i}
\]
(and potentially more).
We will use the following fact
\begin{fact}
    \label{fact:product_pE_overlap}
    Consider two ``independent'' copies of $\calA$ described above.
    Denote the copies by $\calA^{(1)}$ and $\calA^{(1)}$ and the variables by $w_i^{(1)},w_i^{(2)}$ and $x_i^{(1)},x_i^{(2)}$.
    Define $z_i = w_i^{(1)} \cdot w_i^{(2)}$.
    Then it holds that
    \begin{enumerate}
        \item For all $i \in [n] \colon \calA^{(1)} \cup \calA^{(2)} \sststile{4}{} z_i^2 = z_i$ and $\calA^{(1)} \cup \calA^{(2)} \sststile{3}{} z_ix_i^{(1)} = z_ix_i^{(2)}$.
        \item $\calA^{(1)} \cup \calA^{(2)} \sststile{2}{} \sum_{i=1}^n z_i \geq (1-2\eta)n$.
    \end{enumerate}
\end{fact}
\begin{proof}
    The first part of the first item follows directly since $(w_i^{(1)})^2 = w_i^{(1)}$ and $(w_i^{(2)})^2 = w_i^{(2)}$.
    The second part follows since
    \[
        z_i x_i^{(1)} = w_i^{(1)} w_i^{(2)} x_i^{(1)} = w_i^{(1)} w_i^{(2)} x_i = w_i^{(1)} w_i^{(2)} x_i^{(2)} = z_i x_i^{(2)} \,.
    \]
    For the second property we use that $\Set{x^2 = x, y^2 = y} \sststile{2}{} xy \geq x + y -1$.
    It follows that
    \[
        \sum_{i=1}^n z_i = \sum_{i=1}^n w_i^{(1)} w_i^{(2)} \geq \sum_{i=1}^n \Paren{w_i^{(1)} + w_i^{(2)} - 1} \geq (1-2\eta) n \,.
    \]
\end{proof}

\subsection{Approximate DP Covariance-Aware Mean Estimation}
\label{sec:approx_mean_est}

In this section, we prove our main result for covariance-aware mean estimation under approximate differential privacy.
\begin{theorem}
    \label{thm:approx_dp_mean_full}
    Let $\mu \in \R^d$ such that $\norm{\mu} \leq R$.
    Let $0 <\eta$ be less than a sufficiently small constant and let $0 < \alpha, \beta, \eps, \delta$ and $\alpha < 1$.
    Let $x_1^*, \ldots, x_n^*$ be $n \geq n_0$ i.i.d. samples from $N(\mu, \Sigma)$.\footnote{More generally, they can be i.i.d. samples from a distribution $\cD$ such that $\Sigma^{-1/2}(\cD - \mu)$ is fourth moment reasonable sub-gaussian.}
    Let $\cX = \Set{x_1, \ldots, x_n}$ be an $\eta$-corruption of $x_1^*, \ldots, x_n^*$.
    There exists an  $(\eps, \delta)$-differentially private algorithm that, given $\eta, \alpha, \eps, \delta$, and $\cX$, runs in time $\poly(n, \log R)$ and with probability at least $1-\beta$ outputs an estimate $\muhat$ satisfying
    \[ \Norm{\Sigma^{-1/2}\left(\muhat - \mu\right)} \leq \bigO{ \alpha } \,,
    \]
    whenever 
    $$n_0 = \tilde{\Omega}\Paren{ \frac{d^2 + \log^2(1/\beta)}{\alpha^2} + \frac{d + \log(1/\beta)}{\alpha \eps} + \frac{d  + \log(1/\delta)}{\eps} }, $$
and $\alpha \geq \Omega( \eta \log(1/\eta)) $.
\end{theorem}

Similarly to our pure-DP algorithm (cf.~\cref{thm:main_private_cov_mean_est}), this theorem follows from the reduction in \cite{Hopkins2023Robustness}, see \cref{thm:approx-dp-reduction}.
The main difference is that we also need a bound on the volume ratios of low-scoring points when given worst-case instead of (corrupted) Gaussian data.

Let $\eta^*$ be a sufficiently small absolute constant.
Also let $\calX = \Set{x_1, \ldots, x_n}$ be a dataset such that there exists a point $\tilde{\mu}$ of score $0.7\eta^* n$, witnessed by a pseudo-expectation $\pE$.
Let $V_{\gamma}(\calX)$ be the volume of the set of candidates that achieve score $\gamma n$ with respect to $\calX$.
It will be enough to show that $\log(V_{\eta^*}(\calX)/ V_{0.8\eta^*}(\calX)) \leq d \log (1/\alpha)$.
We show this in two parts.
First, we give a lower bound on $V_{0.8\eta^*}(\calX)$ and then an upper bound on $V_{\eta^*}(\calX)$.

\paragraph{Lower bound.}
By definition, $\pE$ is a certificate of score $0.7\eta^*n \leq 0.9 \eta^* n$ for any point $\bar{\mu}$ such that 
\begin{align*}
    \Iprod{u, \pE \mu' - \bar{\mu}}^2 \leq \bigO{\alpha^2} \pE \Iprod{u, \Sigma' u} \,.
\end{align*}
Let $\hat{\Sigma} = \pE \Sigma'$ and $\hat{\mu} = \pE \mu'$.
Then the above is equivalent to $\norm{\hat{\Sigma}^{-1/2} (\hat{\mu} - \bar{\mu})} \leq \bigO{\alpha}$.

\paragraph{Upper bound.}
Thus, if we can show that for all points $\bar{\mu}$ of score $\eta^*$ it holds that $\norm{\hat{\Sigma}^{-1/2} (\hat{\mu} - \bar{\mu})} \leq \bigO{1}$ we are done.
We do this in several steps.
First, we introduce some notation.
Let $\pE_1$ and $\pE_2$ be the witnessing pseudo-expectations of $\tilde{\mu}$ and $\bar{\mu}$ respectively.
Denote the variables associated with $\pE_1$ using a superscript $(1)$ and the ones of $\pE_2$ using a superscript $(2)$.
Let
\begin{align*}
\Sigma_1 = \pE_1 \Sigma^{(1)} = \hat{\Sigma}\,, \Sigma_2 = \pE_2 \Sigma^{(2)} \quad\quad&\text{and}\quad\quad\mu_1 = \pE_1 \mu^{(1)} = \hat{\mu}\,, \mu_2 = \pE_2 \mu^{(2)} \,.
\end{align*}
We will first show that $\Omega(1) \Sigma_1 \preceq \Sigma_2 \preceq \bigO{1} \Sigma_1$ and then that for all vectors $u$ it holds that $\Iprod{u, \mu_1 - \mu_2}^2 \leq \bigO{1} u^\top (\Sigma_1 + \Sigma_2)u$.
Together they imply the claim.

We show this as follows.
Let $\pE$ be the product of $\pE_1$ and $\pE_2$ as defined in \cref{sec:approx_dp_product_pe}.
Recall that both $\pE_1$ and $\pE_2$ satisfy
\begin{equation*}
\calB_{\eta^* n}\colon
  \left \{
    \begin{aligned}
      &\forall i\in [n].
      & w_i^2
      & = w_i \\
      & &\sum_{i\in [n]} w_i &\geq (1-\eta^*)n \\
      &\forall i\in [n] & w_i(x'_i - x_i) &=0 \\
      &&\tfrac{1} n \sum_{i \in [n]} x_i' &= \mu' \\
    && \frac{1}{n} \sum_{i \in [n]} \Paren{ x_i' - \mu' } \Paren{ x_i' - \mu' }^\top & = \Sigma'  \\ 
    & \exists \text{ the following SoS proof }  & \sststile{4}{v} \Biggl\{\frac{1}{n} \sum_{i \in [n]} \Iprod{x_i'  - \mu', v}^4  \leq & \Paren{3 + \eta \log^2(1/\eta)} \Paren{ v^\top \Sigma' v}^2 \Biggr\}
    \end{aligned}
  \right \}
\end{equation*}
And thus $\pE$ satisfies $\calB_{\eta^* n}^{(1)} \cup \calB_{\eta^* n}^{(2)}$.
We deviate from the notation above and denote the variables in $\calB_{\eta^* n}^{(1)}$ and $\calB_{\eta^* n}^{(1)}$ using a superscript of 1 and 2 respectively.
Let $z_i = w_i^{(1)} \cdot w_i^{(2)}$.
We will show that for any vector $u$,
\begin{align}
\label{eq:cov_error}
\calB_{\eta^* n}^{(1)} \cup \calB_{\eta^* n}^{(2)} \sststile{}{} \Iprod{u,(\Sigma^{(1)}-\Sigma^{(2)})u}^2 \leq \bigO{1} \cdot \eta^* \Iprod{u, (\Sigma^{(1)}+\Sigma^{(2)})u}^2
\end{align}
and
\begin{align}
\label{eq:mean_error}
\calB_{\eta^* n}^{(1)} \cup \calB_{\eta^* n}^{(2)} \sststile{}{} \Iprod{u, \mu^{(1)} - \mu^{(2)}}^2 \leq \bigO{1}\cdot  \eta^* \Iprod{u, (\Sigma^{(1)}+\Sigma^{(2)})u} \,.
\end{align}
Using that $\Sigma_1 - \Sigma_2 = \pE [\Sigma^{(1)} - \Sigma^{(2)}]$ and $\mu_1 - \mu_2 = \pE [\mu^{(1)} - \mu^{(2)}]$ and the same arguments as before, this implies the claims above.

The proofs of this are basically the same as in \cite{kothari2018robust} and we only sketch them.
We will first show that 
\begin{align}
\label{eq:cov_err_intermediate}
\calB_{\eta^* n}^{(1)} \cup \calB_{\eta^* n}^{(2)} \sststile{}{} \Paren{u^\top \Paren{\Sigma^{(1)} - \Sigma^{(2)}} u}^2 \leq \bigO{1} \cdot \eta^* \Iprod{u, (\Sigma^{(1)}+\Sigma^{(1)})u}^2 + \bigO{1} \cdot \Iprod{u, \mu^{(1)} - \mu^{(2)}}^4 \,.
\end{align}
Recall that $\calB_{\eta^* n}^{(1)} \cup \calB_{\eta^* n}^{(2)}$ implies $z_i x_i^{(1)} = z_i x_i^{(2)}$.
Thus, using that
\[
    \Sigma^{(1)} = \tfrac 1 n \sum_{i=1}^n (x_i^{(1)} - \mu^{(2)})(x_i^{(1)} - \mu^{(2)})^\top + (\mu^{(1)} - \mu^{(2)})(\mu^{(1)} - \mu^{(2)})^\top \,,
\]
it follows by almost triangle inequality that
\begin{align*}
    \calB_{\eta^* n}^{(1)} \cup \calB_{\eta^* n}^{(2)} \sststile{}{} & \Paren{u^\top \Paren{\Sigma^{(1)} - \Sigma^{(2)}} u}^2 \\
    &\leq \bigO{1} \Paren{\tfrac 1 n \sum_{i=1}^n (1-z_i) \Iprod{u,x_i^{(1)} - \mu^{(2)}}^2}^2 + \bigO{1}  \Paren{\tfrac 1 n \sum_{i=1}^n (1-z_i) \Iprod{u,x_i^{(2)} - \mu^{(2)}}^2}^2 \\
    &+ \bigO{1}  \Iprod{\mu^{(1)} - \mu^{(2)},u}^4 \\
    &\leq \bigO{1} \Paren{\tfrac 1 n \sum_{i=1}^n (1-z_i) \Iprod{u,x_i^{(1)} - \mu^{(1)}}^2}^2 + \bigO{1}  \Paren{\tfrac 1 n \sum_{i=1}^n (1-z_i) \Iprod{u,x_i^{(2)} - \mu^{(2)}}^2}^2 \\
    &+ \bigO{1}  \Iprod{\mu^{(1)} - \mu^{(2)},u}^4
\end{align*}
By Cauchy-Schwarz and hypercontractivity, the first two terms are at most a constant multiple of $\eta^* \Iprod{u, (\Sigma^{(1)}+\Sigma^{(1)})u}^2$.

We next show \cref{eq:mean_error}.
Together with \cref{eq:cov_err_intermediate}, this immediately implies \cref{eq:cov_error}.
As before, we can bound using Cauchy-Schwarz
\begin{align*}
    \calB_{\eta^* n}^{(1)} \cup \calB_{\eta^* n}^{(2)} \sststile{}{} &\Iprod{u, \mu^{(1)} - \mu^{(2)}}^2 \\
&\leq \bigO{1} \Paren{\tfrac 1 n \sum_{i=1}^n (1-z_i) \Iprod{u, x_i^{(1)} - \mu^{(1)}}} + \bigO{1}\Paren{\tfrac 1 n \sum_{i=1}^n (1-z_i) \Iprod{u, x_i^{(2)} - \mu^{(2)}}} \\
&+ \bigO{1} \cdot \eta^* \Iprod{u, \mu^{(1)} - \mu^{(2)}}^2 \\
&\leq \bigO{1} \cdot \eta^* \Iprod{u, (\Sigma^{(1)} + \Sigma^{(1)} )u} + \bigO{1} \cdot \eta^* \Iprod{u, \mu^{(1)} - \mu^{(2)}}^2 \,,
\end{align*}
which implies the claim by taking $\eta^*$ small enough.

\subsection{Approximate DP Linear Regression}
\label{sec:approx_lin_reg}

In this section, we prove our main result for linear regression under approximate differential privacy.
\begin{theorem}
    \label{thm:approx_dp_regression_full}
    Let $\theta \in \R^d$ such that $\norm{\theta} \leq R$.
    Let $0 < \eta$ be less than a sufficiently small constant and let $0 < \alpha, \beta, \eps, \delta$ and $\alpha < 1$.
    Let $(x_1^*, y_1^*), \ldots, (x_n^*, y_n^*)$ be $n \geq n_0$ i.i.d. samples from a linear regression instance with parameter $\theta$ and covariance $\Sigma$ (as defined in~\cref{model:robust_regression}).
    Let $(\cX, \calY) = \Set{(x_1, y_1), \ldots, (x_n, y_n)}$ be an $\eta$-corruption of $(x_1^*, y_1^*), \ldots, (x_n^*, y_n^*)$.
    There exists an  $(\eps, \delta)$-differentially private algorithm that, given $\eta, \alpha, \eps, \delta$, and $(\cX, \calY)$, runs in time $\poly(n, \log R)$ and with probability at least $1-\beta$ outputs an estimate $\muhat$ satisfying
    \[ \Norm{\Sigma^{1/2}\left(\muhat - \mu\right)} \leq \bigO{ \alpha } \,,
    \]
    whenever 
    $$n_0 = \tilde{\Omega}\Paren{ \frac{d^2 + \log^2(1/\beta)}{\alpha^2} + \frac{d + \log(1/\beta)}{\alpha \eps} + \frac{d  + \log(1/\delta)}{\eps} }, $$
and $\alpha \geq \Omega( \eta \log(1/\eta)) $.
\end{theorem}
The score function we will consider is very similar to the one used in the pure DP sections but contains a few extra variables, which we will need for technical reasons. The fact that the new variables do not affect the key properties of the score function established in the pure DP section easily follows from the existing analysis.
\begin{equation*}
\calB_{T}(x,y)=
  \left \{
    \begin{aligned}
      &\forall i\in [n].
      & w_i^2
      & = w_i \\
      & &\sum_{i\in [n]} w_i &\geq n - T \\
      &&\tfrac{1} n \sum_{i \in [n]} x_i'(x_i')^\top &= \Sigma' \\
      &\forall i\in [n] & w_i(x'_i - x_i) = 0&\,,\; w_i(y'_i - y_i) =0 \\
      &&\tfrac{1} n \sum_{i \in [n]} \Paren{\Iprod{x_i', \thetaprime} - y_i'}x_i' &= 0 \\
      &&\tfrac{1} n \sum_{i \in [n]} \Paren{y_i' - \Iprod{\thetaprime, x_i'}}^2 &\leq 1+O(\eta) \\
      &\forall v \in \mathbb{R}^d
      & \frac{1}{n}\sum_{i \in [n]}  {\langle x'_i , v\rangle^{4}} 
      &\leq  \Paren{3+\eta \log^2 (1/\eta)} \left(  v^\top \Sigma' v \right)^{2} \\
      & & \tfrac{1} n \sum_{i \in [n]} \Paren{y_i' - \Iprod{\thetaprime, x_i'}}^4 &\leq \bigO{1} \\
      & &Q = Q^\top \,,\quad \Sigma' Q \Sigma' = \Sigma' \,,\quad  Q \Sigma' Q &= Q \,,\quad Q \Sigma' = I_d \,,\quad \Sigma' Q = I_d \\
      & &B = B^{\top} \, \quad BB = \Sigma' \, \quad D &= D^\top \, \quad DD = Q \, \quad BD = I_d
    \end{aligned}
  \right \}
\end{equation*}

The score candidate point $\tilde{\theta}$ will be the smallest $T$ such that there exists a pseudo-expectation $\pE$ such that the following three conditions are met
\begin{enumerate}
    \item $\pE \sdtstile{}{} \calB_T$ at degree $\bigO{1}$.
    \item $\norm{\pE \theta'}_2 \leq R$.
    \item For all fixed vectors $u \in \R^d$ (i.e., that do not depend on indeterminates), it holds that
    \[
        \Iprod{u, \pE \theta' - \tilde{\theta}}^2 \leq \bigO{\alpha^2} \pE \Iprod{u, Q u} \,.
    \]
\end{enumerate}

Similarly to our pure-DP algorithm (cf.~\cref{thm:main_private_regression}), this theorem follows from the reduction in \cite{Hopkins2023Robustness}, see \cref{thm:approx-dp-reduction}.
The main difference is that we also need a bound on the volume ratios of low-scoring points when given worst-case instead of (corrupted) Gaussian regression data.

Let $\eta^*$ be a sufficiently small absolute constant and $(\calX, \calY) = \Set{(x_1, y_1), \ldots, (x_n, y_n)}$ be a dataset such that there exists a point $\tilde{\theta}$ of score $0.7\eta^* n$, witnessed by a pseudo-expectation $\pE$.
Let $V_{\gamma}(\calX, \calY)$ be the volume of the set of candidates that achieve score $\gamma n$ with respect to $(\calX, \calY)$.
It will be enough to show that $\log(V_{\eta^*}(\calX, \calY)/ V_{0.8\eta^*}(\calX, \calY)) \leq d \log (1/\alpha)$.
We show this in two parts.
First, we give a lower bound on $V_{0.8\eta^*}(\calX)$ and then an upper bound on $V_{\eta^*}(\calX)$. Both these upper and lower bounds will be in terms if $(\pE Q)^{-1/2}$-balls, where the psuedoexpectation is the one which certifies that there exists a point of score $0.7\eta^* n$.

\paragraph{Lower Bounds on Volume.}
We will show that if a single point has score $0.7\eta^* n$, then a non-trivial volume of points also have score at most $0.7\eta^* n$. Specificially, we prove the following lemma:
\begin{lemma}
    Let $(\calX, \calY)$ be a worst-case set of data points such that there exists a point $\tilde{\theta}$ satisfying the robust regression constraints with score $0.7\eta^* n$ and let the certifying pseudo expectation be $\pE$. Then there exists a $\left(\pE Q\right)^{-1/2}$ ball of radius $\Omega(\alpha)$ such that all $\tilde{u}'$ in the ball also have score at most $0.7 \eta^* n$.
\end{lemma}
\begin{proof}
    The proof is analogous to the lower bounds on volume in the pure DP section. We note that the same exact pseudoexpectation which certifies that $\tilde{\theta}$ has score at most $0.7 \eta^* n$ also certifies that any point such that 
    \[ \forall u, \,\Iprod{u, \pE \theta' - \tilde{\theta}'}^2 \leq \bigO{\alpha^2} \pE \Iprod{u, Q u}\]
    also has score at most $0.7 \eta^* n$. However, rearranging this via substituting $\left(\pE Q\right)^{-1/2} u$ for $u$ shows that this is equivalent to $\norm{\left(\pE Q\right)^{-1/2} \left( \pE \theta' - \tilde{\theta}\right)}^2 \leq \bigO{\alpha^2}$ which completes the proof.
\end{proof}

\paragraph{Upper Bounds on Volume.}
We will now describe how to modify the identifiability proofs for regression to show that two pseudoexpectations for the robust regression constraints must have close parameters, even when the original dataset $x_i$ are worst-case points. Note that it suffices to show the following closeness bound for $\tilde{\theta}_1, \tilde{\theta}_2$, where $\pE_1, \pE$ are pseudoexpectations which certify that they have a sufficiently small constant score.
\[ \norm{\left(\pE_1 Q \right)^{-1/2} \left(\tilde{\theta}_1 - \tilde{\theta}_2\right)}^2 \leq \bigO{1}\,. \]
We will proceed via defining a product psuedoexpecation (\cref{def:product_pE}) and executing the original identifiability proofs on the product pseudoexpectation. 

As in~\cref{fact:product_pE_overlap} if we have two pseudoexpectations satisfying the robust regression constraints we will let the variables in the first be $x^{(1)}_i, w^{(1)}_i,$ etc and the variables in the second to be $x^{(2)}_i, w^{(2)}_i,$ and we will also define $z_i$ as in the statement of~\cref{fact:product_pE_overlap}. Note that this is different from the original notation for these variables in the previous sections.

We will prove the following lemma:
\begin{lemma}
\label{lem:approx-dp-reg-identifiability}
    Let $\pE_1$ and $\pE_2$ be two pseudoexpectations satisfying the robust regression constraints $\calA_1, \calA_2$ with score $\eta n < n/e$. Let $\pE$ be the product of the two pseudoexpecations, which satisfies the constraints $\calA_1 \cup \calA_2$. Then we have that
    \[ \calA_1 \cup \calA_2 \sststile{}{} \Set{(\theta^{(2)} - \theta^{(1)})^\top \Sigma^{(1)} (\theta^{(2)} - \theta^{(1)}) \leq \bigO{\sqrt{\eta}}}\,.\]
    Furthermore, we also have that for all $u$
    \[ \Iprod{u, \pE \theta^{(1)} - \theta^{(2)}}^2 \leq \bigO{\sqrt{\eta}} \pE \Iprod{u, Q^{(1)} u} \,.\]
\end{lemma}
We will also need the following lemma about the closeness of the two $Q$ parameters in our SoS system.
\begin{lemma}
\label{lem:sos_inverse_worst_case}
    Let $\pE_1$ and $\pE_2$ be two pseudoexpectations satisfying the robust regression constraints $\calA_1, \calA_2$ with score $\eta n < n/e$. Let $\pE$ be the product of the two pseudoexpecations, which satisfies the constraints $\calA_1 \cup \calA_2$. Then we have that
    \[  \calA_1 \cup \calA_2 \sststile{\bigO{1}}{} \Biggl\{ \Iprod{u, Q^{(2)} u} \leq \bigO{1}  \Paren{u^\top Q^{(1)} u} \Biggr\} \,.\]
\end{lemma}

The combination of the above two lemmas easily yields the desired bound.

\paragraph{\cref{lem:approx-dp-reg-identifiability} and~\cref{lem:sos_inverse_worst_case} implies bounded volume.}
Observe that it suffices to prove that for every $u$,
\[ \Iprod{u, \left(\pE_1 Q \right)^{-1/2} \left(\tilde{\theta}_1 - \tilde{\theta}_2\right)}^2 \leq \bigO{\norm{u}^2}\,,\]
which is equivalent to 
\[ \Iprod{u, \tilde{\theta}_1 - \tilde{\theta}_2}^2 \leq \bigO{u^\top \pE Q^{(1)} u}\,.\]
We have by Triangle Inequality that 
\begin{align*}
    \Iprod{u, \tilde{\theta}_1 - \tilde{\theta}_2}^2 &\lesssim \Iprod{u, \tilde{\theta}_1 - \pE \theta^{(1)}}^2 + \Iprod{u, \pE \theta^{(1)} - \theta^{(1)}}^2 +  \Iprod{u,  \pE \theta^{(2)} - \tilde{\theta}_2}^2 \\
    &\leq \bigO{u^\top \pE Q^{(1)} u + u^\top \pE Q^{(2)} u}\,,
\end{align*}
using both closeness constraints as well as~\cref{lem:approx-dp-reg-identifiability}. Finally, using~\cref{lem:sos_inverse_worst_case} completes the proof.

We will now prove the main parameter closeness lemma (\cref{lem:approx-dp-reg-identifiability}). In the proof of~\cref{lem:approx-dp-reg-identifiability} we will also need the following covariance closeness result:
\begin{lemma}
\label{lem:approx-dp-cov}
    Let $\pE_1$ and $\pE_2$ be two pseudoexpectations satisfying the robust regression constraints $\calA_1, \calA_2$ with score $\eta n \leq n/e$. Let $\pE$ be the product of the two pseudoexpecations, which satisfies the constraints $\calA_1 \cup \calA_2$. Then we have that    
    \[\calA_1 \cup \calA_2 \sststile{O(1)}{u} \Biggl\{\Paren{u^\top\Paren{\Sigma^{(2)} - \Sigma^{(1)}} u}^2 \leq \bigO{\eta} \Paren{u^\top \Sigma^{(1)} u}^2 \Biggr\}\,.\]
\end{lemma}

Note that in contrast to the covariance closeness shown in~\cite{Hopkins2023Robustness} this shows that there is a \emph{SoS proof} that the two parameters are close, rather than just showing the closeness of the two moments. We will now prove the regression identifiability result using the covariance closeness lemma.

\begin{proof}[Proof of~\cref{lem:approx-dp-reg-identifiability}]
    Note that since for $C > 0$, $\Set{a^4 \leq C a^2} \sststile{\bigO{1}}{a} \Set{a^2 \leq C}$ and $\Set{a^2 \leq C} \sststile{\bigO{1}}{a} \Set{a \leq \sqrt{C}}$ (cf.~\cref{lem:square_root_sos,fact:cancellation_two_sos}), it is enough to show that
    \[
        \calA_1 \cup \calA_2 \sststile{O(1)}{} \Paren{(\theta^{(2)} - \theta^{(1)})^\top \Sigma^{(1)} (\theta^{(2)} - \theta^{(1)})}^4 \leq \bigO{\eta} \Paren{(\theta^{(2)} - \theta^{(1)})^\top \Sigma^{(1)} (\theta^{(2)} - \theta^{(1)})}^2 \,.
    \]
    Analogously as in the proof of~\cref{lem:sos_regression_without_bootstrapping} and using that $z_i x_i^{(1)} = z_i x_i^{(2)}$ (cf~\cref{fact:product_pE_overlap}) it holds that
    \begin{align*}
        \calA_1 \cup \calA_2 \sststile{\bigO{1}}{} \Biggl\{ &\Sigma^{(1)} (\theta^{(2)} - \theta^{(1)}) \\
        &= \tfrac 1 n \sum_{i=1}^n (1- z_i) x_i^{(1)} \cdot \Paren{\Iprod{x_i^{(1)}, \theta^{(2)} } - y_i^{(1)}} - \tfrac{1} n \sum_{i \in [n]} (1- z_i) x_i^{(2)} \cdot \Paren{\Iprod{x_i^{(2)}, \theta^{(2)}} - y_i^{(2)}} \Biggr\} \,.
    \end{align*}
    Adding and subtracting $\theta^{(1)}$ in the inner product of the first sum, we obtain that $\calA$ implies at constant degree that
    \begin{align*}
       \Sigma^{(1)} (\theta^{(2)} - \theta^{(1)})  &= \tfrac 1 n \sum_{i=1}^n (1- z_i) x_i^{(1)} \cdot \Paren{\Iprod{x_i^{(1)}, \theta^{(1)} } - y_i^{(1)}} - \tfrac{1} n \sum_{i \in [n]} (1- z_i) x_i^{(2)} \cdot \Paren{\Iprod{x_i^{(2)}, \theta^{(2)}} - y_i^{(2)}}  \\
        &+ \tfrac 1 n \sum_{i=1}^n (1-z_i) x_i^{(1)} \cdot \Iprod{x_i^{(1)}, \theta^{(2)} - \theta^{(1)} } \,.
    \end{align*}
    Thus, by SoS triangle inequality we obtain that
    \begin{align*}
        \calA_1 \cup \calA_2 \sststile{\bigO{1}}{} \Biggl\{ &\Paren{(\theta^{(2)} - \theta^{(1)})^\top \Sigma^{(1)} (\theta^{(2)} - \theta^{(1)})}^4 \\
        &\leq \bigO{1} \underbrace{\Paren{\tfrac 1 n \sum_{i=1}^n (1-z_i) \Iprod{x_i^{(1)}, \theta^{(2)} - \theta^{(1)}} \Paren{\Iprod{x_i^{(1)},\theta^{(1)}} - y_i^{(1)}}}^4}_{\text{Term A}} \\
        &+\bigO{1} \underbrace{\Paren{\tfrac 1 n \sum_{i=1}^n (1-z_i) \Iprod{x_i^{(2)}, \theta^{(2)} - \theta^{(1)}} \Paren{\Iprod{x_i^{(2)},\theta^{(2)}} - y_i^{(2)}}}^4}_{\text{Term B}} \\
        &+\bigO{1} \underbrace{\Paren{\tfrac 1 n \sum_{i=1}^n (1-z_i) \Iprod{x_i^{(1)}, \theta^{(2)} - \theta^{(1)}}^2}^4}_{\text{Term C}}\Biggr\} \,.
    \end{align*}
By~\cref{fact:sos_cs} (Cauchy-Schwartz), our constraint on SoS proof of hypercontractive fourth moments, and the large overlap of the two pseudoexpectations (cf~\cref{fact:product_pE_overlap}) it follows that Term C above is at most $\eta^2 \Paren{(\theta^{(2)} - \theta^{(1)})^\top \Sigma^{(1)} (\theta^{(2)} - \theta^{(1)})}^4$ which is a small multiple of the left-hand side.
Thus, we can ignore the last summand by rearranging.

\paragraph{Bounding Terms A and B.}
We continue to bound Terms A and B.
For Term A, we obtain using again our constraint of hyper contractive fourth moments, SoS Cauchy Schwarz, our SoS constraint on bounded moments of noise, and the large overlap of the two pseudoexpectations (cf~\cref{fact:product_pE_overlap})
\begin{align*}
    \calA_1 \cup \calA_2 &\sststile{\bigO{1}}{} \Biggl\{ \Paren{\tfrac 1 n \sum_{i=1}^n (1-z_i) \Iprod{x_i^{(1)}, \theta^{(2)} - \theta^{(1)}} \Paren{\Iprod{x_i^{(1)},\theta^{(1)}} - y_i^{(1)}}}^4 \\
    &\leq \eta ^2 \Paren{\tfrac 1 n \sum_{i=1}^n \Iprod{x_i^{(1)}, \theta^{(2)} - \theta^{(1)}}^2 \Paren{\Iprod{x_i^{(1)},\theta^{(1)}} - y_i^{(1)}}^2}^2 \\
    &\leq \eta ^2 \Paren{\tfrac 1 n \sum_{i=1}^n \Iprod{x_i^{(1)}, \theta^{(2)} - \theta^{(1)}}^4} \Paren{ \tfrac 1 n \sum_{i=1}^n \Paren{\Iprod{x_i^{(1)},\theta^{(1)}} - y_i^{(1)}}^4}  \\
    &\leq \bigO{\eta^2} \Paren{(\theta^{(2)} - \theta^{(1)})^\top \Sigma^{(1)} (\theta^{(2)} - \theta^{(1)})}^2\Biggr\}
\end{align*}

Similarly, using the constraints that there is an SoS proof in variables $v$ that $\tfrac{1}{n} \sum_{i=1}^n \Iprod{x_i^{(2)},v}^4 \leq \bigO{1} \Iprod{u, \Sigma^{(2)} u}^2$, that $\tfrac{1}{n} \sum_{i=1}^n (\Iprod{x_i^{(2)}, \theta^{(2)}} - y_i^{(2)})^4 \leq \bigO{1}$, and the large overlap of the two pseudoexpectations (cf~\cref{fact:product_pE_overlap}), we can bound Term B as follows
\begin{align*}
    \calA_1 \cup \calA_2 &\sststile{\bigO{1}}{} \Biggl\{ \Paren{\tfrac 1 n \sum_{i=1}^n (1-z_i) \Iprod{x_i^{(2)}, \theta^{(2)} - \theta^{(1)}} \Paren{\Iprod{x_i^{(2)},\theta^{(2)}} - y_i^{(2)}}}^4 \\
    &\leq \eta^2 \Paren{\tfrac 1 n \sum_{i=1}^n \Iprod{x_i^{(2)}, \theta^{(2)} - \theta^{(1)}}^2 \Paren{\Iprod{x_i^{(2)},\theta^{(2)}} - y_i^{(2)}}^2}^2 \\
    &\leq \eta^2 \Paren{\tfrac 1 n \sum_{i=1}^n \Iprod{x_i^{(2)}, \theta^{(2)} - \theta^{(1)}}^4} \Paren{ \tfrac 1 n \sum_{i=1}^n\Paren{\Iprod{x_i^{(2)},\theta^{(2)}} - y_i^{(2)}}^4} \\
    &\leq \bigO{\eta^2} \Paren{(\theta^{(2)} - \theta^{(1)})^\top \Sigma^{(2)} (\theta^{(2)} - \theta^{(1)})}^2\Biggr\}
\end{align*}

Thus, overall we have shown that
\begin{align*}
    \calA_1 \cup \calA_2 &\sststile{\bigO{1}}{} \Biggl\{\Paren{(\theta^{(2)} - \theta^{(1)})^\top \Sigma^{(1)} (\theta^{(2)} - \theta^{(1)})}^4 \\
    &\leq \bigO{\eta^2} \Paren{\Paren{(\theta^{(2)} - \theta^{(1)})^\top \Sigma^{(1)} (\theta^{(2)} - \theta^{(1)})}^2 + \Paren{(\theta^{(2)} - \theta^{(1)})^\top \Sigma^{(2)} (\theta^{(2)} - \theta^{(1)})}^2}\Biggr\} \,.
\end{align*}
Applying SoS triangle and~\cref{lem:approx-dp-cov} (covariance closeness for worst case points), we can deduce that $\calA_\eta$ implies at constant degree that the right-hand side of the above is at most
\begin{align*}
    &\bigO{\eta^2} \Paren{\Paren{(\theta^{(2)} - \theta^{(1)})^\top \Sigma^{(1)} (\theta^{(2)} - \theta^{(1)})}^2 + \Paren{(\theta^{(2)} - \theta^{(1)})^\top \Sigma^{(2)} (\theta^{(2)} - \theta^{(1)})}^2}  \\
    &\leq \bigO{\eta^2} \Paren{\Paren{(\theta^{(2)} - \theta^{(1)})^\top \Sigma^{(1)} (\theta^{(2)} - \theta^{(1)})}^2 + \Paren{(\theta^{(2)} - \theta^{(1)})^\top \Paren{\Sigma^{(2)} - \Sigma^{(1)}} (\theta^{(2)} - \theta^{(1)})}^2}  \\
    &\leq \bigO{\eta^2} \Paren{(\theta^{(2)} - \theta^{(1)})^\top \Sigma^{(1)} (\theta^{(2)} - \theta^{(1)})}^2 \,,
\end{align*}
as desired. Now we note that for any $u$, we have via Cauchy-Schwartz (cf.~\cref{fact:sos_cs}) and our bound on the quadratic form of the difference of $\theta^{(1)}, \theta^{(2)}$ on $\Sigma^{(1)}$ that 
\begin{align*}
    \calA_1 \cup \calA_2 \sststile{}{} \Biggl\{ \Iprod{u, \theta^{(1)} - \theta^{(2)}}^2 &= \Iprod{B^{(1)}D^{(1)}u, \theta^{(1)} - \theta^{(2)}}^2 \\
    &= \Iprod{D^{(1)}u, B^{(1)} \left( \theta^{(1)} - \theta^{(2)}\right)}^2  \\
    &\leq \left(u^\top D^{(1)}D^{(1)} u\right) \cdot \left( \theta^{(1)} - \theta^{(2)}\right)^\top B^{(1)}B^{(1)} \left( \theta^{(1)} - \theta^{(2)}\right) \\
    &\leq \bigO{\eta} \left(u^\top Q^{(1)} u\right)\Biggr\}\,.
\end{align*}
Noting that the LHS is convex and using this to move the psuedoexpectation inside the square completes the proof.
\end{proof}

It remains to prove~\cref{lem:approx-dp-cov}.
\begin{proof}[Proof of~\cref{lem:approx-dp-cov}]
    Note that for any $u$ by Triangle Inequality (cf~\cref{fact:sos-almost-triangle}) and Cauchy-Schwartz (cf~\cref{fact:sos_cs}) we have that 
    \begin{align*}
        \calA_1 \cup \calA_2 &\sststile{O(1)}{u} \Biggl\{\Paren{u^\top\Paren{\Sigma^{(2)} - \Sigma^{(1)}} u}^2 \\
        &\leq \bigO{1} \Paren{\tfrac 1 n \sum_{i=1}^n (1-z_i) \Iprod{u, x_i^{(2)}}^2}^2 + \bigO{1}\Paren{\tfrac 1 n \sum_{i=1}^n (1-z_i) \Iprod{u, x_i^{(1)}}^2}^2 \\
        &\leq \bigO{\eta} \cdot \tfrac 1 n \sum_{i=1}^n \Iprod{u, x_i^{(2)}}^4 + \bigO{\eta} \cdot \tfrac 1 n \sum_{i=1}^n \Iprod{u, x_i^{(1)}}^4
        \Biggr\} \,.
    \end{align*}
    Now, by our constraint on the fourth moments, we now that there exists an SoS proof in variables $u$ that the sum in the first term is at most $\bigO{1} (u^\top \Sigma^{(2)} u)^2$ and similarly for the second with $\bigO{1} (u^\top \Sigma^{(1)} u)^2$.
    Together, this implies the lemma.
\end{proof}

Now the last remaining piece is to prove~\cref{lem:sos_inverse_worst_case}. In the proof we will need the following lemma, whose proof is inspired by~\cref{lem:sos_assymetric_cov_estimation}.
\begin{lemma}
\label{lem:assymetric-approx-dp}
    Let $\pE_1$ and $\pE_2$ be two pseudoexpectations satisfying the robust regression constraints $\calA_1, \calA_2$ with score $\eta n < n/e$. Let $\pE$ be the product of the two pseudoexpecations, which satisfies the constraints $\calA_1 \cup \calA_2$. Then we have that
    \[ \calA_1 \cup \calA_2 \sststile{}{} \Set{u^\top \left(\Sigma^{(1)} - \Sigma^{(2)}\right) v \leq \bigO{1} \cdot u^\top \Sigma^{(2)} u + \frac{1}{2} v^\top \Sigma^{(2)} v}\,.\]
\end{lemma}
We now proceed to the proof of~\cref{lem:sos_inverse_worst_case} and then prove ~\cref{lem:assymetric-approx-dp} afterwards.

\begin{proof}[Proof of~\cref{lem:sos_inverse_worst_case}]
    We start by observing that we have that
    \[
        \calA_1 \cup \calA_2 \sststile{\bigO{1}}{} \Biggl\{ \Iprod{u, Q^{(2)} u} = \Paren{u^\top Q^{(1)} u} +  \Paren{u^\top \brac{Q^{(2)} - Q^{(1)} } u} \Biggr\} 
    \]
    Thus, it remains to bound the second term.
    Repeatedly using the constraints on $\Sigma^{(1)}, \Sigma^{(2)}, Q^{(1)}, Q^{(2)}$ we have that
    \begin{align*}
        \calA_1 \cup \calA_2 &\sststile{\bigO{1}}{} \Biggl\{ u^\top \brac{Q^{(2)} - Q^{(1)} } u = \Iprod{Q^{(2)} - Q^{(1)}, u u^\top} \\
        &= \Iprod{Q^{(2)} \Sigma^{(2)} Q^{(2)} \Sigma^{(2)} Q^{(2)} - Q^{(2)} \Sigma^{(2)} Q^{(1)} \Sigma^{(2)} Q^{(2)}, u u^\top} \\
        &= \Iprod{\Sigma^{(2)} Q^{(2)} \Sigma^{(2)} - \Sigma^{(2)} Q^{(1)} \Sigma^{(2)}, \brac{Q^{(2)}u} \brac{Q^{(2)}u}^\top} \\
        &=\Iprod{Q^{(2)} \Sigma^{(2)} - Q^{(1)} \Sigma^{(2)}, \brac{\Sigma^{(2)} Q^{(2)}u} \brac{Q^{(2)}u}^\top} \\
        &= \Iprod{I_d - Q^{(1)} \Sigma^{(2)}, u \brac{Q^{(2)}u}^\top} \\
        &= \Iprod{Q^{(1)} \Sigma^{(1)} - Q^{(1)} \Sigma^{(2)}, u \brac{Q^{(2)}u}^\top} = \Iprod{\Sigma^{(1)} -\Sigma^{(2)}, \brac{Q^{(1)} u} \brac{Q^{(2)}u}^\top} \Biggr\} \,.
    \end{align*}
    Applying~\cref{lem:assymetric-approx-dp} with $u = Q^{(1)} u$ and $v = Q^{(2)} u$ (overloading notation) and using the constraint $Q^{(2)} \Sigma^{(2)} Q^{(2)} = Q^{(2)}$, it follows that
    \begin{align*}
        \calA_{\eta n} \sststile{\bigO{1}}{} \Biggl\{ \Paren{u^\top \brac{Q^{(2)} - Q^{(1)} } u} &\leq \bigO{1} \paren{u^\top \brac{Q^{(1)} \Sigma^{(2)} Q^{(1)} }u} + \frac 1 2  \paren{u^\top \brac{Q^{(2)} \Sigma^{(2)} Q^{(2)}} u}  \\
        &= \bigO{1} \paren{u^\top \brac{Q^{(1)} \Sigma^{(2)} Q^{(1)}} u} + \frac 1 2 \paren{u^\top Q^{(2)} u}\Biggr\} \,.
    \end{align*}
    The second term is a small multiple of the initial left-hand side and thus, we can ignore it by rearranging.
    For the first term, note that by~\cref{lem:approx-dp-cov} our constraints imply that it is at most $\bigO{1} \paren{u^\top Q^{(1)} \Sigma^{(1)} Q^{(1)}u} = \bigO{1} \paren{u^\top Q^{(1)} u}$.
\end{proof}

Finally, we have to prove~\cref{lem:assymetric-approx-dp}.
\begin{proof}[Proof of~\cref{lem:assymetric-approx-dp}]
    Expanding the equality constraints for the covariance variables we have that
    \begin{align*}
         \calA_1 \cup \calA_2 \sststile{}{} \Biggl\{ &u^\top \left(\Sigma^{(1)} - \Sigma^{(2)}\right) v  \\
         &= \frac{1}{n} \sum_{i \in [n]} (1-z_i) \Iprod{x_i^{(1)}, u} \Iprod{x_i^{(1)}, v} - \frac{1}{n} \sum_{i \in [n]} (1-z_i) \Iprod{x_i^{(2)}, u} \Iprod{x_i^{(2)}, v}\Biggr\}\,.
    \end{align*}
    Using that there exists an SoS proof that $ab \leq a^2 + b^2$ we have that for any $\gamma$ the first term is bounded as follows:
    \begin{align*}
        \calA_1 \cup \calA_2 \sststile{}{} \Biggl\{ \frac{1}{n} \sum_{i \in [n]} (1-z_i) \Iprod{x_i^{(1)}, u} \Iprod{x_i^{(1)}, v} &\leq \gamma \cdot \frac{1}{n} \sum_{i \in [n]} (1-z_i)^2 \Iprod{x_i^{(1)}, u}^2  + \gamma^{-1} \cdot \frac{1}{n} \sum_{i \in [n]} \Iprod{x_i^{(1)}, v}^2 \\
        &\leq \gamma u^\top \Sigma^{(1)} u + \gamma^{-1} v^\top \Sigma^{(1)} v \Biggr\} \,.
    \end{align*}
    Applying a similar bound on the second term we can conclude that there exists an SoS proof that
    \[ u^\top \left(\Sigma^{(1)} - \Sigma^{(2)}\right) v \leq \gamma u^\top \Sigma^{(1)} u + \gamma^{-1} v^\top \Sigma^{(1)} v + \gamma u^\top \Sigma^{(2)} u + \gamma^{-1} v^\top \Sigma^{(2)} v\,.\]
    Applying~\cref{lem:approx-dp-cov} to bound $u^\top \Sigma^{(1)} u \leq \bigO{1} u^\top \Sigma^{(2)} u$ (along with bounding the quadratic form on $v$) and choosing $\gamma$ to be a sufficiently large constant yields the lemma. 
\end{proof}

\end{document}